\documentclass[pra,
aps,
twocolumn,
superscriptaddress,
nofootinbib,
floatfix,
10pt
]{revtex4-2}

\usepackage{amsmath, amsfonts, amssymb, amsthm}
\usepackage{bbm, physics, mathtools, xfrac, dsfont}
\usepackage{graphicx}
\usepackage{thmtools,thm-restate}
\usepackage{tikz}
\usepackage{lipsum}
\usepackage[normalem]{ulem}
\usepackage{xcolor, hyperref, enumerate, enumitem, lipsum}

\hypersetup{
    colorlinks,
    linkcolor={red!50!black},
    citecolor={blue!50!black},
    urlcolor={blue!80!black}
}

\newcommand{\nocontentsline}[3]{}
\let\origcontentsline\addcontentsline
\newcommand\stoptoc{\let\addcontentsline\nocontentsline}
\newcommand\resumetoc{\let\addcontentsline\origcontentsline}

\makeatletter    
\renewcommand\onecolumngrid{
\do@columngrid{one}{\@ne}%
\def\set@footnotewidth{\onecolumngrid}
\def\footnoterule{\kern-6pt\hrule width 1.5in\kern6pt}%
}

\renewcommand\twocolumngrid{
        \def\footnoterule{
        \dimen@\skip\footins\divide\dimen@\thr@@
        \kern-\dimen@\hrule width.5in\kern\dimen@}
        \do@columngrid{mlt}{\tw@}
}%
\makeatother

\newtheorem{thm}{Theorem}

\newtheorem{theorem}{Theorem}
\newtheorem{lemma}[theorem]{Lemma}
\newtheorem{proposition}[theorem]{Proposition}

\newtheorem{definition}[theorem]{Definition}

\newcommand{\bes} {\begin{subequations}}
\newcommand{\ees} {\end{subequations}}
\newcommand{\bea} {\begin{eqnarray}}
\newcommand{\eea} {\end{eqnarray}}
\newcommand{\be} {\begin{equation}}
\newcommand{\ee} {\end{equation}}
\newcommand{\beq} {\begin{equation}}
\newcommand{\eeq} {\end{equation}}

\newcommand{\B}{\mathcal{B}}
\newcommand{\D}{\mathcal{D}}
\newcommand{\E}{\mathcal{E}}
\newcommand{\F}{\mathcal{F}}

\newcommand{\I}{\mathcal{I}}
\renewcommand{\H}{\mathcal{H}}
\newcommand{\M}{\mathcal{M}}
\newcommand{\N}{\mathcal{N}}

\newcommand{\T}{\mathcal{T}}
\newcommand{\V}{\mathcal{V}}

\newcommand{\Sp}{\operatorname{Sp}}
\newcommand{\Or}{\operatorname{O}}
\newcommand{\Un}{\operatorname{U}}
\newcommand{\SU}{\operatorname{SU}}

\renewcommand{\tr}{{\rm{tr}}}
\renewcommand{\det}{{\rm{det}}}
\newcommand{\diag}{\mathrm{diag}}

\newcommand{\overbar}[1]{\mkern 1.5mu\overline{\mkern-1.5mu#1\mkern-1.5mu}\mkern 1.5mu}

\newcommand{\bmo}{\boldsymbol{0}}
\newcommand{\bma}{\boldsymbol{a}}
\newcommand{\bmb}{\boldsymbol{b}}

\newcommand{\bmd}{\boldsymbol{d}}

\newcommand{\bmr}{\boldsymbol{r}}

\newcommand{\bmalpha}{\boldsymbol{\alpha}}
\newcommand{\bmxi}{\boldsymbol{\xi}}

\begin{document}

\title{Symmetry and Asymmetry in Bosonic Gaussian Systems:\texorpdfstring{\\}{ }A Resource-Theoretic Framework}
\author{Nikolaos Koukoulekidis}
\email{nikolaos.koukoulekidis@duke.edu}
\affiliation{Duke Quantum Center and Departments of Physics and Electrical and Computer Engineering, Duke University, Durham, NC 27708, USA}
\author{Iman Marvian}
\email{iman.marvian@duke.edu}
\affiliation{Duke Quantum Center and Departments of Physics and Electrical and Computer Engineering, Duke University, Durham, NC 27708, USA}

\begin{abstract}
We study the interplay of symmetries and Gaussianity in bosonic systems, under closed and open dynamics, and develop a resource theory of Gaussian asymmetry. 
Specifically, we focus on Gaussian symmetry-respecting (covariant) operations, which serve as the free operations in this framework. 
We prove that any such operation can be realized via Gaussian Hamiltonians that respect the symmetry under consideration, coupled to an environment prepared in a symmetry-respecting pure Gaussian state. 
We further identify a family of tractable monotone functions of states that remain non-increasing under Gaussian symmetry-respecting dynamics, and are exactly conserved in closed systems.
We demonstrate that these monotones are not generally respected under non-Gaussian symmetry-respecting dynamics.
Along the way, we provide several technical results of independent interest to the quantum information and optics communities, including a new approach to the Stinespring dilation theorem, and an extension of Williamson’s theorem for the simultaneous normal mode decomposition of Gaussian systems and conserved charges.
\end{abstract}

\maketitle

\stoptoc

\section{Introduction}
\label{sec:intro}

Symmetries lie at the foundation of the fundamental principles of nature.
Noether's theorem, for instance,  establishes that every continuous symmetry gives rise to a conservation law~\cite{Noether1918nachrichten}.
In practice, exploiting symmetries  often simplifies the characterization of a system's dynamics.
Conversely, if one has access only to physical operations that respect a certain symmetry, e.g. spin rotations in the $\hat{z}$ direction, it becomes impossible to transform symmetric states, e.g. spin-$\hat{z}$ particles, into asymmetric states, e.g. spin-$\hat{x}$ particles.
From this perspective, modern treatments of symmetry in quantum physics~\cite{reference2007bartlett,gour2008resource,marvian2013theory,marvian2014symmetry,marvian2014asymmetry,cirstoiu2020robustness,yamaguchi2025quantum} view asymmetry as a resource that permits reaching certain quantum states that would otherwise remain inaccessible under symmetric dynamics.
This resource-theoretic approach~\cite{chitamber2019quantum,Gour2025quantum} has uncovered deep links between asymmetry and various areas of quantum physics, such as quantum coherence~\cite{ marvian2014asymmetry, yadin2016general, marvian2016quantum, marvian2016quantify,  lostaglio2015description, lostaglio2015quantum,  lostaglio2019coherence, streltsov2017colloquium},  quantum metrology~\cite{chiribella2004efficient,hall2012does,marvian2016quantum,marvian2020coherence,yadavalli2024optimal}, and athermality in quantum thermodynamics~\cite{brandao2013resource,brandao2015second,gour2018quantum,alexander2022infinitesimal, lostaglio2015description, lostaglio2015quantum,  lostaglio2019coherence}. More recently, in the context of many-body systems, the resource theory of asymmetry has been shown to be useful for understanding thermal relaxation processes in quantum systems, and in particular for elucidating the Mpemba effect~\cite{summer2025resource,ares2023entanglement}.
However, the resource theory of asymmetry has been developed primarily for finite-dimensional Hilbert spaces, leaving the subtleties of continuous-variable (CV) systems with infinite-dimensional Hilbert spaces largely unaddressed.

An important class of CV systems that is ubiquitous in nature is the class of bosonic Gaussian systems, i.e. bosonic systems described by Gaussian states and Gaussian dynamics.
The Gaussian formalism often sufficiently captures or approximates quantum correlations~\cite{adesso2007entanglement} which are crucial in quantum optics, field theories and solid state physics~\cite{cerf2007quantum}.

Gaussian systems are also ubiquitous in the laboratory, as they can be efficiently prepared and manipulated~\cite{weedbrook2012gaussian,adesso2014continuous,serafini2017quantum}.
Therefore, there have been extensive analysis of the role of Gaussianity in operational tasks in the contexts of quantum communication~\cite{caves1982quantum,caves2004fidelity,caves2012quantum,adesso2012measuring}, quantum thermodynamics~\cite{brown2016passivity,macchiavelo2020quantum,serafini2020gaussian,singh2023gaussian,medina2024anomalous,rodriguez2025extracting}, quantum simulation~\cite{cramer2008exact,somhorst2023quantum,schweigler2021decay}, quantum coherence~\cite{anderson2005experimental,zhao2017quantum,du2023incoherent,yadavalli2024optimal}, quantum sensing~\cite{aasi2013enhanced,caves1980on,joo2011quantum,ohliger2012continuous,mccormick2019quantum}, quantum learning theory~\cite{gandhari2023precision,becker2024classical,mele2024learning,fanizza2025efficient} and quantum computing~\cite{lloyd1999quantum,gottesman2001encoding,niset2009no,chabaud2023resources,brenner2024factoring,upreti2025interplay}.
Gaussian operations are sufficient for quantum advantage in several of these scenarios, but other times they prove too limited compared to the much larger class of non-Gaussian operations~\cite{zhuang2018resource}.
Viewing Gaussianity as an operational constraint within the set of all possible operations on CV systems has motivated the development of Gaussian resource theories~\cite{lami2018gaussian,takagi2018convex} that attempt to characterize and quantify the degree of non-Gaussianity present in a CV system.

In this work, we bridge the notions of symmetry and Gaussianity in the context of bosonic CV systems, and study the effect of Gaussian constraints on symmetry-respecting (covariant) operations, and vice versa. Specifically, we study the consequences of symmetry in both closed and open Gaussian time dynamics.  For closed systems, we identify quantities that remain conserved if the dynamics is governed by Gaussian Hamiltonians that respect a symmetry, but that are not conserved if the Hamiltonian is not Gaussian or does not respect the symmetry. Similarly, for open systems, we find functions of quantum states that remain monotone under Gaussian covariant operations.

We take an approach which is akin to quantum resource theories, and in particular to the resource theory of asymmetry.
In this context, the motivation and scope of our results are twofold. Firstly, we investigate the standard resource theory of asymmetry in the setting of bosonic systems, where the set of free operations is defined as all transformations that respect the underlying symmetry. This establishes a systematic characterization of asymmetry within a well-established theoretical framework with broad applications. Secondly, we develop and analyze a new resource-theoretic framework in which the free operations are restricted to Gaussian covariant operations. This refinement provides a framework that is both physically motivated and experimentally relevant, as it naturally incorporates the structural features of Gaussian processes while respecting the imposed symmetry constraints. In analogy with the well-studied notion of Gaussian entanglement, we refer to this new framework as \emph{Gaussian asymmetry}.

An immediate application of our framework lies in the context of quantum thermodynamics, where the free operations are those realizable by energy-conserving unitary transformations together with ancillas in the thermal state of the system's intrinsic Hamiltonians.
Recently, Gaussian thermal operations have been studied~\cite{serafini2020gaussian} as a restriction of the free operations that is desirable in certain practical settings.
To understand the connection with our framework, we note that energy-conserving unitaries can equivalently be defined as unitaries that respect a certain symmetry, namely, as unitaries that commute with the time-translation symmetry generated by the intrinsic Hamiltonian, i.e. $\exp(-i \hat{H} t)$ for all $t \in \mathbb{R}$.
If the time evolution of the system is periodic, then the corresponding family of unitary transformations forms a Gaussian unitary representation of the group $\mathrm{U}(1)$, which we study extensively in this work.

After a brief look at preliminary notation in Section~\ref{sec:prelim}, we introduce and characterize the representations of interest in Section~\ref{sec:sym_repr}.
We then establish the core ingredients of the resource theory of Gaussian asymmetry in Section~\ref{sec:res_theory}, discussing monotones in Section~\ref{sec:monotones}, and conservation laws for closed systems in Section~\ref{sec:conserve_laws}.
We provide a dilation theorem for our free operations in Section~\ref{sec:dilation}, and we conclude in Section~\ref{sec:summary}.

\subsection*{Summary of main results}
\textbf{Covariant Gaussian dilation.} We show that any symmetry-respecting (i.e. covariant) Gaussian channel can be realized using Gaussian Hamiltonians that respect the symmetry and couple the system to an environment initially prepared in a pure, symmetry-respecting Gaussian state (Theorem~\ref{thm:cov_gauss_dilation}).
This result generalizes the well-known Gaussian dilation of Gaussian channels ~\cite{caruso2008multi,serafini2017quantum}.
Remarkably, our argument relies on a different approach based on the ideas of Gaussian purification and decoupling (see Fig.~\ref{fig:dilation}).
This result strengthens the operational significance of the resource theory of Gaussian asymmetry, by providing an implementation of free operations via symmetry-respecting Gaussian unitaries and ancillary states.\\

\textbf{Monotones for Gaussian asymmetry.}
In order to quantify asymmetry in open covariant dynamics, in Section~\ref{sec:monotones}, we define and calculate quantities that are monotone under the action of Gaussian covariant channels.
Given a representation of interest $\hat{S}(g): g \in G$ for a symmetry group $G$, we derive families of tractable monotones of Gaussian asymmetry as well as of asymmetry in the displacement vector only (type--1 asymmetry) and in the covariance matrix only (type--2 asymmetry) of a quantum state $\hat{\rho}$.

We introduce a family of type--1 asymmetry monotones that take the form
\begin{equation*}
    f_\mu(\hat{\rho}) \coloneqq{\rm rank}\left(T^{(\mu)}\left(\bmd\bmd^\dag\right)\right) \,,
\end{equation*}
where $\bmd$ is the displacement vector of any, possibly non-Gaussian, state and $T^{(\mu)}$ is a completely positive map defined in Eq.(\ref{eq:Tmu}) that projects $\bmd\bmd^\dag$ to an irreducible representation $\mu$ of the symmetry.

We then give a recipe for constructing type--2 asymmetry monotones. 
In the case of a Gaussian state $\hat{\rho} = e^{-\beta \hat{H}} / \tr[e^{-\beta \hat{H}}]$ with covariance matrix $\sigma$, we find the following closed form for a family of such monotones,
\begin{equation*}
    f_{\alpha, g}^{(2)}(\hat{\rho}) \coloneqq -\frac{1}{2}\log\frac{\det[(\sigma_\alpha + \sigma_{1-\alpha})/2]}{\det[\left(\sigma_\alpha + S(g)\sigma_{1-\alpha}S(g)^{\rm T}\right)/2]} \,,
\end{equation*}
for any $\alpha \in (0,1)$, $g \in G$.
Here, $S(g)$ is the matrix representation of $\hat{S}(g)$ in terms of position and momentum operators, and the quantity $\sigma_{\alpha}$ is the covariance matrix of $\hat{\rho}^\alpha/\tr[\hat{\rho}^\alpha]$, i.e. the thermal state of Hamiltonian $\hat{H}$ at temperature $(\alpha\beta)^{-1}$. 
Monotones $f_{\alpha, g}^{(2)}$ are essentially ratios of overlaps between Gaussian states obtained from $\hat{\rho}$.
Specifically, the numerator corresponds to the overlap of two thermal states of Hamiltonian $\hat{H}$ with linear terms removed, but at temperatures higher than the temperature of $\hat{\rho}$, while the denominator corresponds to the same overlap with one of the states rotated according to $g \in G$.

In the important case of a continuous 1--parameter family of unitary transformations $\hat{S}(t) = e^{it\widehat{Q}}: t \in \mathbb{R}$, where $\widehat{Q}$ is a Gaussian observable representing a conserved charge, we obtain additional monotones,
\begin{equation*}
    F_Q(\hat{\rho}) \coloneqq \frac{1}{4}\tr\Big[\Omega Q\Omega\sigma_{1/2}^{-1}[\Omega Q, \sigma_{1/2}\Omega]\Big] \,,
\end{equation*}
where $\Omega$ is the symplectic form, and $Q$ is the matrix representation of $\widehat{Q}$ in terms of position and momentum operators.
\\

\textbf{Conservation laws for closed Gaussian systems.} 
All the above monotones remain conserved under closed Gaussian dynamics that respect the symmetry.
In Section~\ref{sec:conserve_laws}, we further introduce new conservation laws that hold under closed Gaussian dynamics, but do not hold under non-Gaussian dynamics.

Remarkably, we find that for Gaussian states lacking coherence with respect to the eigenspaces of the conserved charge, a simple finite set of conservation laws captures all the consequences of symmetry and Gaussianity (Theorem~\ref{thm:interconvert}): if these laws are satisfied, then there exists a Gaussian charge-conserving unitary that transforms one, possibly mixed, state into the other.\\

\textbf{Diagonalization with symplectic transformations (extension of Williamson's theorem).}  
As a key technical ingredient to derive the conservation laws, we introduce an extension of Williamson's theorem~\cite{williamson1936on} (Theorem~\ref{thm:nmd_q}) that, under certain conditions, allows for the simultaneous normal mode decomposition of a positive-definite matrix and a symmetric matrix representing the conserved charge.
That is, we show that these two matrices can be simultaneously diagonalized by congruence via a symplectic transformation.

\section{Preliminaries}
\label{sec:prelim}

We first provide preliminary definitions on symmetry representations and  Gaussian systems.

\subsection{Symmetry representations in the Hilbert space}

We assume any system under consideration is equipped with a unitary representation $ \hat{S}(g): g\in G$ of a symmetry group $G$, which, in general, can be a projective representation, such that
\begin{equation}\label{eq:group_rule_hilbert}
    \hat{S}(g_2)\hat{S}(g_1)=\omega(g_1,g_2) \hat{S}(g_2g_1) \,, 
\end{equation}
for all $g_1,g_2\in G$ and $\hat{S}(g^{-1})=\hat{S}(g)^\dag$, where $\omega(g_1,g_2)$ is a phase. 
Then, under symmetry transformation $g \in G$, the state $\ket{\psi} \in \H$ of the system is transformed as $\ket{\psi} \mapsto \hat{S}(g)\ket{\psi}$. 
This, in turn, implies that a density operator $\hat{\rho} \in \B(\H)$ transforms as $\hat{\rho}\mapsto \hat{S}(g)\hat{\rho}\hat{S}(g)^\dag$, and a channel, i.e. a completely positive trace-preserving linear map $\E: \B(\H_A) \rightarrow \B(\H_B)$ transforms as
\begin{equation*}
    \mathcal{E}(\cdot) \mapsto \hat{S}_B(g)\mathcal{E}\Big(\hat{S}_A(g)(\cdot)\hat{S}_A(g)^\dag\Big) \hat{S}_B(g)^\dag \,,
\end{equation*}
where subscripts $A$ and $B$ label the input and output systems, respectively, which in general can be different. 
As usual, rather than applying the transformation on states, we can apply the dual transformation on observables, i.e. an observable $\hat{K}$ transforms as $\hat{K} \mapsto \hat{S}(g)^\dag\hat{K}\hat{S}(g)$.

\subsection{Gaussian systems}

Any system under consideration is a collection of a finite number $n$ of bosonic modes with position and momentum operators $\hat{\bmr} \coloneqq (\hat{x}_1, \hat{p}_1, \dots, \hat{x}_n, \hat{p}_n)^{\rm T}$ which satisfy the canonical commutation relations $[\hat{x}_j, \hat{x}_k] = [\hat{p}_j, \hat{p}_k] = 0$, and $[\hat{x}_j, \hat{p}_k] = i\delta_{j,k}$, written compactly as
\begin{equation}\label{comm}
    \left[\hat{\bmr}, \hat{\bmr}^{\rm T}\right] = i\Omega_n \coloneqq i \bigoplus_{j=1}^n \begin{pmatrix} 0 & 1 \\ -1 & 0 \end{pmatrix} \,,
\end{equation}
where $\Omega \equiv \Omega_n$ is the symplectic form on $n$ modes\footnote{If the number of modes of system $A$ is unclear, we write $\Omega_A$ instead. Notation is further explained in Appendix~\ref{app:notation}.}.
The annihilation and creation operators are $\hat{a}_j \coloneqq (\hat{x}_j + i\hat{p}_j)/\sqrt{2}$, and $\hat{a}_j^\dag$, respectively, and the $n$--mode vacuum state $\ket{0}^{\otimes n}$ is the unique state annihilated by all $n$ annihilation operators.

The Gaussian formalism can be defined based on the notion of the second-order Hamiltonian  
\begin{equation}\label{eq:ham}
    \hat{H} = \frac{1}{2} \hat{\bmr}^{\rm T} H \hat{\bmr} + \hat{\bmr}^{\rm T}\bmxi \,,
\end{equation}
where $H$ is a real symmetric $(2n \times 2n)$ matrix and $\bmxi \in \mathbb{R}^{2n}$ is a displacement vector.
We refer to such Hamiltonians as \emph{Gaussian Hamiltonians}, and the group of unitaries that can be realized by them, i.e. the group generated by unitaries of the form $e^{i \hat{H}}$, as \emph{Gaussian unitaries}. 

Any Gaussian unitary can be brought to the standard form $\hat{U} = \hat{D}_{\bmxi}\hat{V}$, where
\begin{equation}\label{eq:disp_op}
    \hat{D}_{\bmxi} \coloneqq e^{i\bmxi^{\rm T} \Omega \hat{r}} \,,
\end{equation}
for some $\bmxi \in \mathbb{R}^{2n}$, is a Weyl displacement operator, and $\hat{V}$ is a symplectic Gaussian unitary that can be realized with purely quadratic Hamiltonians.
Gaussian unitary $\hat{U}$ transforms the position and momentum operators as
\begin{equation}\label{eq:gauss_basis_transformation}
    \hat{U}^\dag \hat{\bmr} \hat{U} 
    = \hat{V}^\dag \big( \hat{D}_{\bmxi}^\dag \hat{\bmr} \hat{D}_{\bmxi} \big) \hat{V} 
    = V \hat{\bmr} - \bmxi \,,
\end{equation}
where $V = e^{-\Omega H}$ is a $(2n \times 2n)$ real matrix\footnote{We denote the phase space matrix associated with a Hilbert space operator by simply dropping the hat, unless otherwise stated.} satisfying $V \Omega V^{\rm T} = \Omega$.
The set of $(2n \times 2n)$ real matrices satisfying this condition forms the \emph{symplectic group} denoted $\Sp(2n, \mathbb{R})$.
In particular, symplectic Gaussian unitaries form a subgroup of the full Gaussian unitary group\footnote{
Indeed, the correspondence $-\Omega H \mapsto i\hat{H}$ for purely quadratic $\hat{H}$ yields a faithful representation of the Lie algebra. 
Specifically,
\begin{equation*}
    i\hat{H}_A = [\,i\hat{H}_B, i\hat{H}_C\,]
    \quad \Longleftrightarrow \quad 
    -\Omega H_A = [-\Omega H_B, -\Omega H_C] \,.
\end{equation*}}.

A Gaussian quantum density matrix $\hat{\rho}$ is a thermal state $e^{-\beta\hat{H}}/\tr[e^{-\beta\hat{H}}]$ of a positive second-order Hamiltonian $\hat{H}$ with $\beta>0$ and it is fully specified by its displacement vector $\bmd$ and covariance matrix $\sigma$,
\begin{equation}\label{eq:moments}
    \bmd \coloneqq \tr\left[ \hat{\rho} \hat{\bmr} \right] \,,\quad
    \sigma \coloneqq \tr\left[ \hat{\rho} \left\{ (\hat{\bmr}-\bmd), (\hat{\bmr}-\bmd)^{\rm T} \right\} \right] \,.
\end{equation}
Here, $\sigma$ satisfies $\sigma = \sigma^{\rm T}$ and the uncertainty relation $\sigma + i\Omega \ge 0$.
The limit of zero temperature, i.e. $\beta\rightarrow \infty$, corresponds to pure Gaussian states.

A bosonic quantum channel is a completely-positive trace-preserving (CPTP) map $\E: \B(\H_A) \rightarrow \B(\H_B)$ from input bosonic system $A$ to output bosonic system $B$.
A Gaussian channel $\E$ is a subclass of bosonic channels~\cite{holevo2001evaluating,caruso2008multi} that transform the displacement vector $\bmd$ and covariance matrix $\sigma$ of the input as  
\begin{equation}\label{eq:xychannel}
    \bmd \mapsto X\bmd + \bmxi \,,\quad \sigma \mapsto X\sigma X^{\rm T} + Y \,,
\end{equation}
for an arbitrary real vector $\bmxi$ and two arbitrary real matrices $X$ and $Y$ satisfying
\begin{equation}\label{eq:channel_uncertainty}
    Y + i(\Omega_B - X \Omega_A X^{\rm T}) \ge 0 \,,
\end{equation}
where $\Omega_A$ and $\Omega_B$ are the input and output symplectic forms, respectively.
Gaussian channels map Gaussian states to Gaussian states, even when they act on a subsystem of a composite system, i.e. $(\E \otimes \I_E) (\hat{\rho}_{AE})$ is Gaussian whenever $\hat{\rho}_{AE}$ is Gaussian, for any environment $E$ with arbitrary number of modes.
We provide a formal definition in Appendix~\ref{app:channel_def}.

Appendix~\ref{app:elements} overviews additional details on Gaussian formalism that are relevant to our proofs.

\section{Invariance under Symplectic Representations of Symmetry}
\label{sec:sym_repr}

In our framework, all systems under consideration are equipped with a (possibly projective) Gaussian unitary representation of a group $G$, denoted as $g \mapsto \hat{U}(g)$, for all group elements $g\in G$. This means that the symmetry transformation acts linearly on position and momentum operators as described by Eq.(\ref{eq:gauss_basis_transformation}). 
In particular, we are often interested in representations which do not contain displacements, which correspond to the case $\bmxi=0$ in Eq.(\ref{eq:gauss_basis_transformation}). That is $\hat{U}(g)=\hat{S}(g)$, where $\hat{S}(g)$ is in the symplectic subgroup of Gaussian unitaries. Let $S(g)$ be the symplectic matrix associated with $\hat{S}(g)$ via Eq.(\ref{eq:gauss_basis_transformation}). 
Then, for any (possibly projective) symplectic Gaussian representation $g \mapsto \hat{S}(g)$ in the Hilbert space, the map $g \mapsto S(g)$ defines its associated symplectic representation of symmetry in the phase space, namely it satisfies
\begin{equation}\label{eq:group_rule_phase}
    \begin{split}
    S(g_2g_1) &= S(g_2)S(g_1) \,, \quad\text{and}\\ 
    S(g^{-1}) &= S(g)^{-1}  \,, \hspace{28pt}\text{for all } g_1, g_2, g \in G \,.
    \end{split}
\end{equation}

Consider systems A and B with $n_{\rm A}$ and $n_{\rm B}$ modes, and position and momentum operators $\hat{\bmr}_A$ and $\hat{\bmr}_B$, respectively. 
Then, the composite system AB has $n_A+n_B$ modes and we write the position and momentum operators as $\hat{\bmr}_A \oplus \hat{\bmr}_B$.
The symmetry representation on the composite system is the tensor product of representations, $g \mapsto \hat{S}_{AB}(g) = \hat{S}_{A}(g) \otimes \hat{S}_{B}(g)$, in the Hilbert space, with associated phase space representation $g \mapsto S_{AB}(g) = S_{A}(g) \oplus S_{B}(g)$. 
Note that each of the subsystem representations $\hat{S}_{A}, \hat{S}_{B}$ can generally entangle modes within the subsystem.

\subsection{Symmetry-respecting Gaussian operators and channels in phase space}
\label{sec:sym_op_ch}

Operators and channels that remain unchanged under the action of a symmetry group $G$ take a central role in the study of asymmetry.
Table~\ref{tab:transf_rules} shows how different object of interest in phase space transform under symmetry transformations, allowing us to characterize invariant Gaussian operators and covariant Gaussian channels in the phase space.

\begin{table}[ht]
{
\renewcommand{\arraystretch}{2}
\setlength{\tabcolsep}{10pt}
\begin{tabular}{c|l|l}
    {\bf Tensor type} & {\bf Tensor} & {\bf Transformation} \\
    \hline
    $(1,0)$ & $\bmd, \bmxi$ & $\bmd \rightarrow S\bmd$  \\
    $(1,1)$ & $V, X$ & $S \rightarrow S V S^{-1}$  \\
    $(2,0)$ & $\sigma, Y$ & \hspace{0.5pt}$\sigma \rightarrow S \sigma S^{\rm T}$  \\
    $(0,2)$ & $H$ & \hspace{-2pt}$H \rightarrow S^{\rm -T} H S^{-1}$ 
\end{tabular}
}
\caption{\textbf{The action of symmetry transformations in phase space.} 
This table summarizes how different objects of interest transform under a symmetry transformation described by a symplectic matrix $S$, namely, 
displacement vectors $\bmd$ and covariance matrices $\sigma$ associated with a quantum state via Eq.(\ref{eq:moments}), displacement vectors $\bmxi$ and symplectic matrices $V$ associated with a Gaussian unitary $\hat{U}$ via Eq.(\ref{eq:gauss_basis_transformation}), 
the $X,Y$ matrices associated with a Gaussian channel via Eq.(\ref{eq:xychannel}), 
and symmetric matrix $H$ associated with a second-order Hamiltonian $\hat{H}$ via Eq.(\ref{eq:ham}). 
The type of tensor associated to each object is determined by the number of matrices $S$ or $S^{\rm T}$ and their inverses that appear in the transformation.
}
\label{tab:transf_rules}
\end{table}

Given a symplectic Gaussian representation $g \mapsto \hat{S}(g)$ of group $G$, state $\hat{\rho}$ is called $G$--invariant if $[\hat{\rho}, \hat{S}(g)] = 0$ for all $g \in G$.
This implies invariance of its displacement vector $\bmd$ and covariance matrix $\sigma$ under the associated phase space representation $g \mapsto S(g)$,
\begin{equation}\label{eq:sp_state}
\begin{split}
    S(g) \bmd &= \bmd \,,\text{ and }\\
    S(g) \sigma S(g)^{\rm T} &= \sigma \,,\text{ for all } g \in G \,.
\end{split}
\end{equation}
The above constraints are satisfied by all invariant states, whether Gaussian or not. 
However, for Gaussian states, they guarantee that the state is $G$--invariant.

A Gaussian unitary $\hat{U} = \hat{D}_{\bmxi}\hat{V}$ is called $G$--invariant when $[\hat{U}, \hat{S}(g)] = 0$ for all $g \in G$.
This is equivalent to invariance of vector $\bmxi$ and symplectic matrix $V$,
\begin{equation}\label{eq:sp_unitary}
\begin{split}
    S(g) \bmxi &= \bmxi \,,\text{ and }\\
    S(g) V S(g)^{-1} &= V \,,\text{ for all } g \in G \,.
\end{split}
\end{equation}
In particular, the displacement vector $\bmxi$ is invariant if and only if it is restricted to modes where the symmetry representation acts trivially.
In Appendix~\ref{app:ps_invariance}, we show that any $G$--invariant symplectic Gaussian unitary can be decomposed as $\hat{V} = e^{i\hat{H}_1}e^{i\hat{H}_2}$ for $G$--invariant Hamiltonians $\hat{H}_1$ and $\hat{H}_2$. 
It follows that the subgroup of symplectic Gaussian unitaries respecting the symmetry is connected.

Finally, consider a Gaussian channel $\E: \B(\H_A) \rightarrow \B(\H_B)$ with input and output symplectic Gaussian representations $g \mapsto \hat{S}_A(g)$ and $g \mapsto \hat{S}_B(g)$ of group $G$. 
Channel $\E$ is called $G$--covariant when
\begin{equation}
    \E\left(\hat{S}_A(g)(\cdot)\hat{S}^\dag_A(g)\right) = \hat{S}_B(g)\E(\cdot)\hat{S}_B(g)^\dag \,, \text{ for all } g \in G \,,
\end{equation}
which is equivalent to the $G$--invariance of $\bmxi, X$ and $Y$,
\begin{equation}\label{eq:sp_channel}
\begin{split}
    S_A(g) \bmxi &= \bmxi \,,\\
    XS_A(g) &= S_B(g)X \,,\text{ and } \\
    S_B(g)YS_B(g)^{\rm T} &= Y \,,\text{ for all } g \in G \,.
\end{split}
\end{equation}

Appendix~\ref{app:ps_invariance} offers a proof of the equivalence between Hilbert space and phase space invariance conditions based on the transformation rules in Table~\ref{tab:transf_rules}.

\subsection{Orthogonal symplectic (passive) representations}
\label{sec:sym_repr_pass}

A particularly physically relevant family of Gaussian representations is defined by the so-called \emph{passive unitaries},  
a terminology borrowed from the quantum optics literature~\cite{serafini2017quantum,gerry2023introductory}.  
We call an operator \emph{passive} if it commutes with   
the free Hamiltonian of harmonic oscillators with identical frequencies,
\begin{equation}\label{eq:ham_free}
    \hat{H}_{\rm free} \coloneqq \frac{1}{2} \sum_{j=1}^n(\hat{x}_j^2 + \hat{p}_j^2) = \sum_{j=1}^n \left(\hat{a}_j^\dag\hat{a}_j + \frac{1}{2}\right) \,.
\end{equation}
Specifically, a symplectic Gaussian unitary $\hat{S}$ and its associated symplectic matrix $S$ are passive if $S$ is also orthogonal, i.e. $S \Omega S^{\rm T} = \Omega$ and $SS^{\rm T} = I$, that is, 
\begin{equation}
  S \in \Sp(2n,\mathbb{R}) \cap \mathrm{O}(2n,\mathbb{R})
  \;\cong\; \mathrm{U}(n) \,.
\end{equation}
The isomorphism with $\mathrm{U}(n)$ can be explained as a consequence of the fact that passive transformations do not mix 
annihilation operators with creation operators, and act unitarily within each set. That is, under any passive representation of group $G$,
\begin{equation}\label{eq:unitary_repr_annih}
    \hat{a}_j \mapsto \hat{S}(g)^\dag \hat{a}_j \hat{S}(g)= \sum_{k} U_{jk}(g) \hat{a}_k
\end{equation}
where $g \mapsto U(g)$ defines an $n$--dimensional unitary representation of group $G$ (see Appendix~\ref{app:passive} for details).

In the case of passive representations, since $S(g)^{\rm T} = S(g)^{-1}$, invariance by congruence is equivalent to invariance by similarity.
Hence, the invariance conditions on all matrices $\sigma$, $V$, $X$, $Y$, and $H$, are equivalent to commutation with $S(g)$ for all $g \in G$. 
The following proposition, proved in Appendix~\ref{app:elements}, gathers frequently used properties of passive unitaries and quadratic observables.

\begin{restatable}{proposition}{propQ}\label{prop:Qprop}
For any purely quadratic observable 
$\widehat{Q} = \sum_{j,k} Q_{jk} \hat{r}_j \hat{r}_k$, 
where $Q = Q^{\rm T}$ is symmetric, the following statements are equivalent:
\begin{enumerate}
    \item[(1)] $e^{i \widehat{Q} s}$ is a passive symplectic Gaussian unitary for all $s\in\mathbb{R}$;
    \item[(2)] $e^{-\Omega Q s}$ is orthogonal symplectic for all $s\in\mathbb{R}$;
    \item[(3)] $\widehat{Q}$ commutes with the free Hamiltonian;
    \item[(4)] $Q$ commutes with $\Omega$;
    \item[(5)] the vacuum state is an eigenvector of $\widehat{Q}$.
\end{enumerate}
\end{restatable}
\noindent Apart from the pure vacuum state $\ket{0}^{\otimes n}$, passive operators also leave invariant any free thermal state $e^{-\beta \hat{H}_{\rm free}}/\tr[e^{-\beta \hat{H}_{\rm free}}]$ with $\beta^{-1} > 0$.

\subsection{Invariant Hamiltonians}
\label{sec:ham_inv}

In terms of creation and annihilation operators, a Gaussian Hamiltonian can be written as $\hat{H} = \hat{H}^{\text{p}}+\hat{H}^{\text{np}}$, where
\begin{equation}
    \hat{H}^{\text{p}} \coloneqq \sum_{j,k} {H}_{jk}^{\text{p}}\ \hat{a}_j^\dag\hat{a}_k \,,\
    \hat{H}^{\text{np}} \coloneqq \sum_{j,k} {H}_{jk}^{\text{np}}\ \hat{a}_j\hat{a}_k + \text{H.c.} \,,
\end{equation}
denote the passive and non-passive terms, respectively. 
If $\hat{H}^{\mathrm{np}} \neq 0$, then under the Hamiltonian $\hat{H}$ the vacuum evolves into a pure Gaussian state exhibiting squeezing. 
Equivalently, its density operator does not commute with $\hat{H}_{\rm free}$. 
While there always exists a passive Gaussian Hamiltonian that respects any symmetry, namely $\hat{H}_{\text{free}}$, a representation may not allow the existence of $G$--invariant non-passive Gaussian Hamiltonians.
 
A passive representation transforms a general $n$--mode second-order Hamiltonian $\hat{H} \mapsto \hat{S}(g)\hat{H}\hat{S}(g)^\dag$ as $H^{\text{p}}\mapsto U(g)  H^{\text{p}} U(g)^\dag $ and $H^{\text{np}}\mapsto U(g)  H^{\text{np}} U(g)^{\rm T}$. 
We thus identify necessary and sufficient conditions for $\hat{H}$ to be invariant under the given representation,
\begin{equation}\label{eq:quadH_cond}
    \begin{split}
    U(g) H^{\text{p}} U(g)^\dag &= H^{\text{p}} \,,\\
    U(g) H^{\text{np}} U(g)^{\rm T} &= H^{\text{np}} \,,
    \end{split}
\end{equation}
for all $g \in G$.
The matrices $H^{\text{p}}$ and $H^{\text{np}}$ satisfying the conditions in Eq.(\ref{eq:quadH_cond}) are characterized in Appendix~\ref{app:irreps}.
In particular, using the vectorization of these equations, one finds that invariant Hamiltonians correspond to subspaces that are invariant under the representations $U(g)\otimes U(g)^*$ and $U(g)\otimes U(g)$ for $g\in G$, respectively.
In general, the form of invariant states and operations depends on the choice of representation, as we illustrate in Fig.~\ref{fig:intro} by considering composite systems with different numbers of modes on each subsystem, necessitating subsystem representations of different dimensions.

By decomposing the representation $g \mapsto U(g)$ to irreps and applying Schur's lemma, one can fully characterize the matrices that satisfy these conditions, as we further explain in Section~\ref{sec:irreps}.

\subsection*{Example: U(1) symmetry}
\label{sec:u1}

We start with the important case of $\Un(1)$ symmetry, which can have different physical interpretations. 
These include invariance under rotations around a fixed axis in real space,  time-translational invariance in periodic systems, or, equivalently, energy conservation, and insensitivity under phase shifts.

\begin{figure}[t]
\centering{\includegraphics[width=0.4\textwidth]{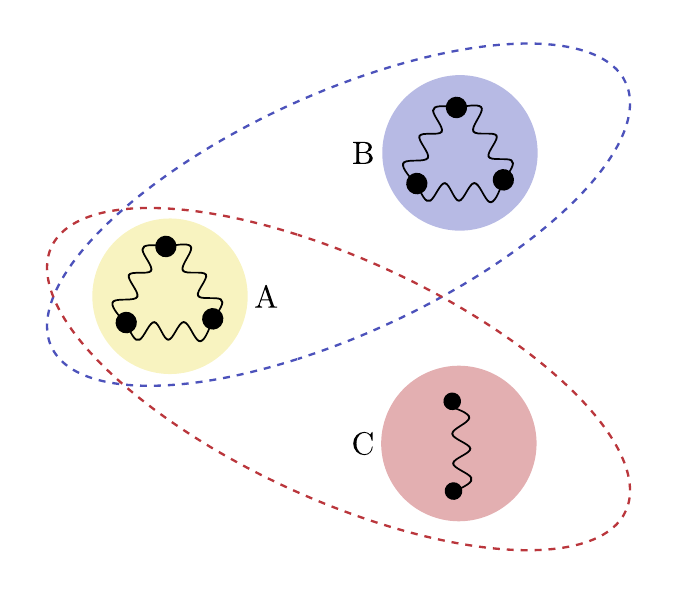}}
\caption{
    \textbf{Gaussian systems and symmetry representations}. 
    Each system under consideration is a collection of bosonic modes carrying a representation of a symmetry group $G$, which, in general, can differ between subsystems. In this example, systems $A$, $B$, and $C$ have 3, 3, and 2 bosonic modes, respectively, and carry an arbitrary representation of group $G$ with dimension equal to the number of modes of the system.
    Since all Gaussian Hamiltonians are quadratic, interactions in this framework can only couple pairs of systems, i.e. they are 2--local. The symmetry further restricts these 2--local interactions (see Appendix~\ref{app:examples} for further discussion and examples including the group $\SU(2)$ and the permutation group $\operatorname{S}_3$).
}
\label{fig:intro}
\end{figure}

For a system with $n$ modes, consider the $\Un(1)$ symmetry associated with the observable
\begin{equation}
    \widehat{Q} = \sum_{j=1}^n q_j \hat{a}^\dag_j\hat{a}_j = \sum_{j=1}^n q_j \left( \hat{x}_j^2 + \hat{p}_j^2 - \frac{1}{2} \right) \,,
\end{equation}
where integer $q_j$ determines the charge associated with mode $j$, and constant $\frac{1}{2}$ can be ignored as it corresponds to a global phase.
If $\hat{H} = \omega \widehat{Q}$ is the intrinsic Hamiltonian of a system of $n$ modes, then $q_j\omega$ is the energy of an excitation in mode $j$, e.g. a photon or a phonon, or, equivalently, the frequency associated with that mode.
 
More precisely, assume the action of the $\Un(1)$ symmetry for any phase $\theta \in [0, 2\pi)$ is represented by operator
\begin{align}\label{eq:u1_repr_hat}
    \hat{O}(\theta) = \exp({i\theta\widehat{Q}}) = \bigotimes_{j=1}^n \exp({i\theta q_j \hat{a}^\dag_j\hat{a}_j}) \,,
\end{align}
with associated symplectic matrix
\begin{equation}\label{eq:u1_repr}
    O(\theta) = \exp({-\theta \Omega Q}) = \bigoplus_{j=1}^n \begin{pmatrix} \cos(q_j \theta) & -\sin(q_j \theta) \\ \sin(q_j \theta) & \cos(q_j \theta) \end{pmatrix} \,,
\end{equation}
where $Q \coloneqq \bigoplus_{j=1}^n q_j I$ is the total conserved charge, and each mode is treated as its own subsystem with representation $\theta \mapsto O_j(\theta) \coloneqq \exp({-\theta q_j \Omega})$.
In general, charges $q_j$ can be arbitrary integers.
In the context of quantum thermodynamics, when $\widehat{Q}$ is interpreted as the intrinsic Hamiltonian of a system, the charges $q_j$ are all positive, otherwise, if $q_j < 0$, for some $j$, the Hamiltonian of mode $j$ will be unbounded and so will not admit a proper thermal state. 
However, as we describe in Appendix~\ref{app:neg_freq}, engineering Hamiltonians with negative frequencies is a common technique in quantum optics, used, for example, to implement 2--mode squeezing~\cite{caves1985new,tsang2012evading}.

Modes with charge $q_j=0$ are not restricted by the $\Un(1)$ symmetry 
and may have arbitrary Gaussian Hamiltonians. 
However, they cannot be coupled to modes with non-zero charge.
For any mode with charge $q_j \neq 0$, the only 1--mode Hamiltonian that respects the symmetry is $\hat{a}_j^\dag \hat{a}_j$, up to normalization.
 This Hamiltonian generates a phase-shifter, with associated symplectic matrix
\begin{equation}\label{eq:Vps}
    V_{\rm ps}(\phi) = e^{-\phi\Omega} = \begin{pmatrix} \cos{\phi} & -\sin{\phi} \\ \sin{\phi} & \cos{\phi} \end{pmatrix} \,,\quad \phi \in [0,2\pi) \,.
\end{equation}
Two modes $j$ and $k$ with non-zero charges $q_j, q_k$ are coupled by a U(1)--invariant Hamiltonian $\hat{H}$ if and only if $[H, \Omega Q] = 0$, i.e. all $(2 \times 2)$ blocks $H_{jk}$ that  couple mode $j$ and mode $k$ satisfy
\begin{equation}\label{eq:u1Hblock}
    H_{jk} = -\left(q_j q_k^{-1}\right) \Omega H_{jk} \Omega \,.
\end{equation}
In particular, modes with charges of unequal magnitude are uncoupled, i.e. $H_{jk} = 0$ unless $q_j = \pm q_k$.
This can be seen by considering any unitary-invariant norm of both sides of Eq.(\ref{eq:u1Hblock}), e.g. the operator norm.
Hence, up to normalization and a phase shift, the U(1)--invariant Gaussian Hamiltonians are
\begin{equation}
    \hat{H}_{\rm bs} = i\left(\hat{a}_j^\dag \hat{a}_k - \text{H.c.}\right) \,,\ \hat{H}_{\rm 2sq} = i\left(\hat{a}_j \hat{a}_k - \text{H.c.}\right) \,,
\end{equation} 
for $q_j = q_k \neq 0$ and for $q_j = -q_k \neq 0$, respectively,
which generate the beam-splitter and 2--mode squeezing operators, with associated symplectic matrices
\begin{align}
    V_{\rm bs}(\phi) &= e^{\phi(\Omega \otimes I)} = \begin{pmatrix} \cos{\phi} I & \sin{\phi} I \\ -\sin{\phi} I & \cos{\phi} I \end{pmatrix} \,,\quad \phi \in [0,2\pi) \,, \label{eq:Vbs}\\
    V_{\rm{2sq}}(r) &= e^{r Z(\Omega \otimes I)} = \begin{pmatrix} \cosh{r} I & \sinh{r} Z \\ \sinh{r} Z & \cosh{r} I \end{pmatrix} \,,\quad r \in \mathbb{R} \,. \label{eq:V2sq}
\end{align}
Unless $\phi \in \{0,\pi\}$ and $r=0$, these operators do not commute with the free Hamiltonians $\hat{a}_j^\dag \hat{a}_j$ and $ \hat{a}_k^\dag \hat{a}_k$, which means they exchange the conserved charge between modes $j$ and $k$.
Together with phase-shifters, they generate all U(1)--invariant Gaussian unitaries, as we prove in Appendix~\ref{app:local}.
\begin{restatable}{proposition}{local}
    All U(1)--invariant Gaussian unitaries are generated by displacement operators on zero-charge sectors, phase-shifters, beam-splitters between sectors of equal charge and 2--mode squeezers between sectors of opposite charge.
\end{restatable}
\noindent When all charges $q_j$ are equal, then invariant symplectic matrices coincide with passive symplectic matrices, and there exists a known algorithm~\cite{reck1994experimental} that decomposes them into phase-shifters and beam-splitters.

\subsection*{SU(2): Schwinger model of angular momentum}

Suppose that a 2--mode system is equipped with the defining representation of the $\SU(2)$ symmetry, generated by angular momentum operators
\begin{equation}\label{eq:jzjx}
    \hat{J}_z = \frac{1}{2}\left(\hat{a}^\dagger \hat{a} - \hat{b}^\dagger \hat{b}\right) \,,\quad
    \hat{J}_x = \frac{1}{2}\left(\hat{a}^\dagger \hat{b}+ \hat{a}\hat{b}^\dagger\right) \,.
\end{equation}
This construction is known as the Schwinger boson representation of angular momentum~\cite{schwinger1952angular}. 
Using the Euler decomposition of $\SU(2)$ elements, this immediately defines a representation of $\SU(2)$ as
\begin{equation}\label{eq:repSU}
    \hat{S}(U) = e^{-i\alpha_1 \hat{J}_z}e^{-i\alpha_2 \hat{J}_x}e^{-i\alpha_3 \hat{J}_z} \,,
\end{equation}
where $\alpha_1,\alpha_2,\alpha_3$ are the Euler angles associated to unitary $U\in\mathrm{SU}(2)$.

This system does not admit any single-body non-passive $\SU(2)$--invariant operator, which follows from the fact that the defining irreducible representation of $\SU(2)$ is pseudo-real (see discussion in Section~\ref{sec:irreps}).
Nevertheless, it does allow for two-body non-passive $\SU(2)$--invariant operators.
In particular, the quadratic $\SU(2)$--invariant Hamiltonians 
\begin{align}
    \hat{H}^{\text{p}}_1 &= \hat{a}^\dagger \hat{a} + \hat{b}^\dagger \hat{b} \,, \\
    \hat{H}^{\text{p}}_2 &= \hat{a}_1^\dagger \hat{a}_2 + \hat{b}_1^\dagger \hat{b}_2 \,, \\
    \hat{H}^{\text{np}} &= e^{i\phi}(\hat{a}_1 \hat{b}_2 - \hat{b}_1 \hat{a}_2) + \text{H.c.} \,, \label{eq:pseudoreal_ent}
\end{align}
are, respectively, single-body passive, two-body passive, and two-body non-passive, for arbitrary $\phi \in [0,2\pi)$. 
Note that, unlike the passive Hamiltonians, $\hat{H}^{\text{np}}$ allows interactions between the modes associated with $\hat{a}$ and $\hat{b}$.

\subsection{Symplectic representations with invariant states: Squeezed passive representations}

An important feature of orthogonal symplectic representations, which is crucial in our work, is the existence of invariant states: the vacuum and, more generally, all thermal states of the free Hamiltonian $\hat{H}_{\rm free}$. 
In contrast, a general symplectic representation $S: G \rightarrow \Sp(2n,\mathbb{R})$ of a symmetry group $G$ does not necessarily admit an invariant state.
A prominent example is given by the full symplectic group itself.
In particular, there is no state that remains invariant under the 1--mode squeezing operator $\exp(r(\hat{a}^{\dagger 2} - \hat{a}^2))$, when $r \neq 0$, described by its associated symplectic matrix
\begin{equation}
    V_{\rm 1sq}(r) = e^{rZ} = \diag(e^r, e^{-r}) \,.
\end{equation}
To see this note that, according to Eq.(\ref{eq:sp_state}), if a state with covariance matrix $\sigma$ is invariant, then $\sigma^{-1/2} S(g) \sigma^{1/2}$ is orthogonal. 
In other words, the representation $g \mapsto S(g)$ is similar to an orthogonal representation. 
The eigenvalues of $V_{\rm 1sq}(r)$ are $e^r$ and $e^{-r}$, so the 1--mode squeezing operator cannot be similar to an orthogonal transformation when $r \neq 0$.

However, this phenomenon cannot occur in the important case of finite groups, 
and more generally of compact groups such as $\Un(1)$ and $\SU(d)$.
 Moreover, all representations that admit a physical invariant state are equivalent to a passive representation.
\begin{restatable}{proposition}{squrepr}\label{prop:squ}
    Let $G$ be a group with a symplectic representation $S: G \rightarrow \Sp(2n,\mathbb{R})$ that admits a $G$--invariant quantum state with finite covariance matrix $\sigma$, i.e. such that $S(g) \sigma S(g)^{\rm T} = \sigma$ for all $g \in G$ (this condition is automatically satisfied if $G$ is compact).
    Then, the following statements are true.
    \begin{enumerate}
        \item[(i)] Representation $S$ is symplectically equivalent to a passive (orthogonal symplectic) one; that is, there exists a symplectic matrix $R \in \Sp(2n,\mathbb{R})$ such that, for all $g \in G$,
        \begin{equation}\label{eq:orthogonalized}
            O(g) \coloneqq R^{-1} S(g) R \in \Sp(2n,\mathbb{R}) \cap \mathrm{O}(2n) \,.
        \end{equation}
        \item[(ii)] The passive representation $O$ is unique up to conjugation by an orthogonal symplectic transformation; that is, if $\widetilde{O}(g) = \widetilde{R}^{-1} S(g) \widetilde{R}$ is passive for some $g \in G$ and symplectic matrix $\widetilde{R}$, then $\widetilde{O}(g) = T O(g) T^{\rm T}$ for some orthogonal symplectic matrix $T$.
        \item[(iii)] Representation $S$ admits a pure Gaussian $G$--invariant state, which is characterized by zero displacement and covariance matrix $R R^{\rm T}$.
    \end{enumerate}
\end{restatable}
\noindent In Appendix~\ref{app:squ_repr}, we present a self-contained proof.
For any symplectic matrix $R$, and orthogonal symplectic representation $G \rightarrow \Sp(2n,\mathbb{R})\cap \mathrm{O}(2n)$, we shall refer to the family of representations defined by 
\begin{equation}\label{eq:sq_passive}
     S(g) = R O(g) R^{-1}: g \in G  \,,
\end{equation}
as \emph{squeezed passive representations}.

A few remarks on this important proposition are in order.
First, statement (i) for compact groups is an immediate consequence of the Cartan-Iwasawa-Malcev theorem, which states that any compact subgroup of a connected Lie group is contained in a maximal compact subgroup, and all maximal compact subgroups are conjugate to each other~\cite{malcev1945theory,iwasawa1949some,onishchik2012lie}.

Second, the uniqueness of the associated passive representation $O(g)$ up to a passive symplectic transformation follows from the symplectic singular value decomposition, also known as the Bloch-Messiah decomposition~\cite{arvind1995real,braunstein2005squeezing}. As a simple example, consider two different representations of $\Un(1)$ symmetry on a single mode, as 
\begin{equation}\label{eq:exU(1)}
    O_1(\theta) = e^{\Omega\theta} \text{ or } O_2(\theta) = e^{-\Omega\theta} \,,
\end{equation}
for $\theta\in[0,2\pi)$. 
While these representations are equivalent as real representations, i.e. they are related as 
\begin{equation}
    O_2(\theta) = V O_1(\theta) V^{-1} \,,
\end{equation}
for all $\theta \in [0,2\pi)$ with e.g. $V=Z$, they are \emph{not} symplectically equivalent; that is, there exists no symplectic transformation $V$ satisfying the above equation. 
According to statement (ii) of Proposition~\ref{prop:squ}, this can be seen by simply noting that no \emph{passive} symplectic transformation $V$ can convert one representation into the other.
In this example, all such passive transformations are of the form $e^{\phi\Omega}$ for some $\phi\in[0,2\pi)$, which commutes with the representation, and therefore cannot change it.

Third, statement (iii) of Proposition~\ref{prop:squ} implies that the existence of a $G$--invariant state under a symplectic Gaussian representation $\hat{S}(g): g\in G$ guarantees the existence of a pure $G$--invariant Gaussian state $\ket{\eta}$, satisfying
\begin{equation}\label{G-inv}
    \hat{S}(g) \ket{\eta} = \ket{\eta} \,,
\end{equation}
for all $g \in G$.
This statement is not generally true for non-Gaussian representations.
For example, although the maximally mixed qubit state is invariant under $\SU(2)$ rotations, no pure qubit state is $\SU(2)$--invariant.

Analysis of passive representations is directly applicable to squeezed passive representations via the basis change induced by $\hat{R}$.
In the Hilbert space, this means that any symplectic Gaussian representation with invariant states is equivalent to $\hat{S}(g) = \hat{R} \hat{O}(g) \hat{R}^{\dag}: g \in G$, where $\hat{R}$ is a symplectic Gaussian unitary. 
Equivalently, there exist annihilation operators $\hat{b}_j \coloneqq \hat{R} \hat{a}_j \hat{R}^\dag$, that can be written as a linear combination of operators $\{\hat{a}_j, \hat{a}^\dag_j\}$, via what is known as a Bogoliubov transformation, and, relative to this new basis, $\hat{S}(g)$ does not mix creation and annihilation operators,
\begin{equation}\label{eq:unitary_repr_annih2}
    \hat{S}(g)^\dag \hat{b}_j \hat{S}(g) = \sum_{k} U_{jk}(g) \hat{b}_k \,,
\end{equation}
where $g \mapsto U(g)$ represents group $G$ unitarily, generalizing Eq.(\ref{eq:unitary_repr_annih}) to all squeezed passive representations. 
We note that the above decomposition is not unique and depends on the choice of the annihilation operators $\{\hat{b}_j\}$. However, according to statement (ii) in Proposition~\ref{prop:squ}, different choices of these operators result in the same unitary $U(g)$ on annihilation operators, up to a unitary change of basis.

For example, for the single-mode $\mathrm{U}(1)$ representations given in Eq.(\ref{eq:exU(1)}), the corresponding unitaries are simple phase factors, i.e. $(1 \times 1)$ unitaries, equal to $e^{-i\theta}$ and $e^{i\theta}$, respectively.

\begin{figure}[t]
    \centering{\includegraphics[width=0.49\textwidth]{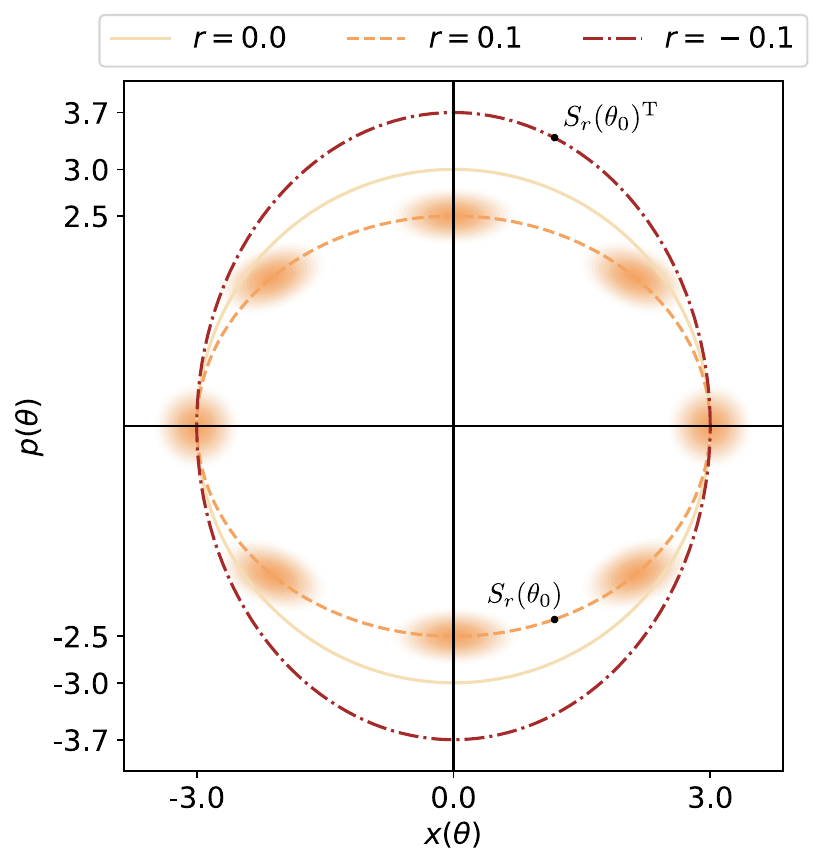}}
    \caption{\textbf{Squeezed passive representations}.
    We sketch the action of squeezed passive representations of U(1), $\theta \mapsto S_r(\theta) \coloneqq V_{\rm 1sq}(r)\exp(-\theta \Omega)V_{\rm 1sq}(-r)$ for different squeezing parameters $r$, on the displaced vacuum state $\hat{D}_{(3,0)^{\rm T}}\ket{0}$.
    These transformations move phase space points anti-clockwise on an ellipsoid, with the passive representation ($r = 0$), corresponding to a circle.
    Each element $S_r(\theta_0)$ for fixed $\theta_0$ commutes with representation $S_{r'}$ if and only if $r'=r$, whereas its transpose $S_r(\theta_0)^{\rm T} = S_{-r}(-\theta_0)$ commutes with $S_{r'}$ if and only if $r'=-r$. 
    We discuss when a $(1,1)$--tensor and its transpose simultaneously commute with the same representation in Appendix~\ref{app:passive}.
    The complex conjugate representations $\theta \mapsto S_r(-\theta)$ move points clockwise instead.
    We also depict the state's covariance matrix under the action of the representation with $r = 0.1$.
    The covariance matrix picks up squeezing along its orbit, but returns to the same state after a $\pi$ rotation, i.e. it respects the $\mathbb{Z}_2$ symmetry.
    We generalize this observation for any representation of $\Un(1)$ in Appendix~\ref{app:u1breaking}, where we show that a Gaussian operator which breaks the $\Un(1)$ symmetry always breaks the symmetry of a finite $\Un(1)$ subgroup of sufficiently high order.
    Instead, the state $V_{\rm 1sq}(r)\ket{0}$ is an invariant state of representation $S_r$ for all $r$.
    }
    \label{fig:repr}
\end{figure}

As a consequence, squeezed passive representations can, up to conjugation by a symplectic transformation, be decomposed into unitary irreducible representations of the group. 
These can then be used to fully characterize invariant Hamiltonians, invariant quantum states, and covariant channels as we discuss in Section~\ref{sec:irreps}.
Another desirable property of squeezed passive representations is that invariant Gaussian unitaries act transitively on invariant pure Gaussian states, which we prove in Appendix~\ref{app:res_transverse} by constructing the unitary in terms of the states' displacement vectors and covariance matrices.

As an example, Fig.~\ref{fig:repr} illustrates the squeezed passive representation $\theta \mapsto \exp(i\theta (e^{-2r}\hat{x}^2+e^{2r}\hat{p}^2))$, obtained by conjugating a 1--mode $\Un(1)$ representation with charge $q=1$ by a 1--mode squeezing operator.

\subsection*{Complex Conjugate Representation}

Given any unitary representation $\hat{U}(g): g\in G$ in the Hilbert space, one can consider the complex conjugate representation $\hat{U}(g)^*: g\in G$, where complex conjugation ($*$) is defined with respect to a fixed arbitrary basis, and corresponds to an anti-unitary transformation in the Hilbert space.
We use the convention that this fixed basis is the common eigenbasis of the position operators $\{\hat{x}_j\}$, which means under complex conjugation the momentum operators transform as $\hat{p}_j\mapsto -\hat{p}_j$, i.e. $\hat{\bmr} \mapsto (I \otimes Z)\hat{\bmr}$. 
With this convention, the complex conjugation can be interpreted as time reversal.

Suppose $\hat{S}(g): g\in G$ is a symplectic Gaussian representation with associated phase space representation $g \mapsto S(g)$.
Then, $\hat{S}(g)^*: g\in G$ is also a symplectic Gaussian representation with associated representation
\begin{equation}\label{eq:conj_repr}
    g \mapsto (I \otimes Z) S(g) (I \otimes Z) \,.
\end{equation}

\subsection{Irrep decomposition \& the trichotomy of irreps}
\label{sec:irreps}

A standard method for characterizing the set of invariant states and Hamiltonians is via decomposition into irreducible representations.
Consider a group $G$ with squeezed passive representation $g \mapsto S(g) = RO(g)R^{-1}$ for symplectic $R$ and orthogonal symplectic representation $g \mapsto O(g)$. Then, according to Eq.(\ref{eq:unitary_repr_annih2}), we can choose a proper basis for the annihilation, such that they transform unitarily under the action of symmetry, without mixing with creation operators.
By decomposing this unitary representation into irreducible representations, we find a basis $\{\hat{b}_{\lambda,k,\alpha}\}$ such that
\begin{equation}\label{eq:irrep_dec}
    \hat{S}(g)^\dag\ \hat{b}_{\lambda,j,\alpha}\ \hat{S}(g) = \sum_{k=1}^{d_\lambda} [U_\lambda(g)]_{jk}\ \hat{b}_{\lambda,k,\alpha} \,,
\end{equation}
for all $g \in G$.
Here, $\lambda$ labels the inequivalent irreps of $G$, each with dimension $d_\lambda$, $[U_\lambda(g)]_{jk}$ are the matrix elements of irrep $\lambda$ with indices $j,k = 1, \dots, d_\lambda$, and $\alpha$ is a multiplicity index.
The overall transformation mapping the set $\{\hat{a}_j, \hat{a}^\dag_j\}$ of the original annihilation and creation operators to the set $\{\hat{b}_{\lambda,j,\alpha}, \hat{b}^\dag_{\lambda,j,\alpha}\}$ is a Bogoliubov transformation, which is described by a symplectic matrix.

For a given squeezed passive representation $g \mapsto S(g)$, the multiplicity $m(\lambda)$ of an irrep $\lambda$ appearing in the decomposition in Eq.(\ref{eq:irrep_dec}) remains invariant under any symplectic change of basis, i.e. under a similarity transformation $R S(g) R^{-1}$ by a symplectic matrix $R$.
\begin{proposition}
    Any symplectic representation $S: G \rightarrow \Sp(2n,\mathbb{R})$ of a group $G$ that admits an invariant state, is uniquely specified, up to a similarity transformation $S(g) \rightarrow R S(g) R^{-1}: g\in G$ by a symplectic matrix $R\in \Sp(2n,\mathbb{R})$, in terms of the integer multiplicities  $m(\lambda)$ defined above, for all unitary irreps $\lambda$ of group $G$.
\end{proposition}
\noindent The proof is presented in Appendix~\ref{app:squ_repr}. It is worth noting that similarity transformations by general, possibly non-symplectic, real matrices, such as the time-reversal $I \otimes Z$ in Eq.(\ref{eq:conj_repr}), can change $m(\lambda)$.
However, it is still guaranteed that the sum $m(\lambda) + m(\lambda^*)$ remains conserved, where $\lambda^*$ denotes the dual irrep of $\lambda$~\cite{fulton2013representation, hall2013lie}, i.e. the irrep equivalent to\footnote{Since the representations are unitary, the notions of complex conjugate and dual irreps coincide.} $U_\lambda^*$ (see Appendix~\ref{app:irreps}).
 
In general, symmetry forbids correlations between modes transforming under inequivalent irreps, except when the irreps are related by complex conjugation. 
In particular, based on the decomposition in Eq.(\ref{eq:irrep_dec}), we present a full characterization of symmetry-respecting Gaussian Hamiltonians, unitaries, and channels (see Proposition~\ref{prop:irreps} in Appendix~\ref{app:irreps}).
Here, we highlight an important implication of this characterization. 
\begin{proposition}\label{prop:irreps_corollary}
    Let $\hat{\rho}$ be a, possibly non-Gaussian, $G$--invariant state. 
     Modes transforming under irreps $\lambda$ and $\lambda'$, i.e. modes associated with annihilation operators $\hat{b}_{\lambda,j,\alpha}$ and $\hat{b}_{\lambda',j',\alpha'}$, can be correlated only if $\lambda' = \lambda$ or $\lambda' = \lambda^*$. 
    Likewise, $G$--invariant Hamiltonians and unitaries couple modes associated with $\hat{b}_{\lambda,j,\alpha}$ and $\hat{b}_{\lambda',j',\alpha'}$ only if $\lambda' = \lambda$ or $\lambda' = \lambda^*$.
    
    Finally, for any Gaussian $G$--covariant channel, the output modes with annihilation operator transforming under an irrep $\lambda$ can depend non-trivially only on input modes with annihilation operator transforming under irreps $\lambda$ or $\lambda^*$.
\end{proposition}

This proposition, for instance, implies that a non-passive quadratic Hamiltonian can be $G$--invariant only if there exist two distinct sets of modes carrying dual irreps, or if there is a set of modes carrying an irrep that is self-dual (i.e. equivalent to its own dual). 
As another consequence, the only U(1)--covariant channels on 1--mode systems, for which the output depends non-trivially on the input, correspond to frequencies $q_{\text{out}} = \pm q_{\text{in}}$.
These channels are called phase-covariant when $q_{\text{out}}=q_{\text{in}}$,  and phase-contravariant when $q_{\text{out}}=-q_{\text{in}}$.

An important special case arises when $n$ subsystems carry the same representation, and this representation is irreducible, that is, only a single $\lambda$ appears in Eq.(\ref{eq:irrep_dec}), occurring with multiplicity $n$, i.e. $U(g) = U_\lambda(g) \otimes I_n$.
In this case, the corresponding real symplectic representation $g \mapsto S(g)$ is irreducible over $\mathbb{R}$ but reducible over $\mathbb{C}$. 
This reducibility can be seen from the fact that, for all $g \in G$, $S(g)$ does not mix annihilation and creation operators $\hat{b}_{\lambda,j,\alpha}$ and  $\hat{b}^\dag_{\lambda,j,\alpha}$. 
When $U_\lambda$ is irreducible, its dual representation $U_\lambda^*$ is also irreducible. 
In general, $U_\lambda^*$ may or may not be equivalent to $U_\lambda$.
That is, there may or may not exist an invertible matrix $B$ such that 
\begin{align}\label{eq:dual_equiv}
    U_\lambda(g)^* = B^{-1} U_\lambda(g) B \,,
\end{align}
for all $g \in G$.
It turns out that if such a matrix $B$ exists, it is either symmetric or anti-symmetric (see Appendix~\ref{app:irreps}), so the relationship between these two irreps falls into three cases, each with important consequences for invariant Hamiltonians and states:
\begin{enumerate}
    \item[(i)] {Complex irreps.} $U_\lambda$ and $U_\lambda^*$ are inequivalent; that is, there exists no invertible matrix $B$ satisfying Eq.(\ref{eq:dual_equiv}).  
    This occurs, for example, for the $\Un(1)$ symmetry on a single mode. 
    In this case, if all subsystems carry the same irrep, then there is no non-passive $G$--invariant Gaussian Hamiltonian.
    
    \item[(ii)] {Real irreps.} $U_\lambda$ and $U_\lambda^*$ are equivalent via a symmetric matrix $B$.
    This occurs, for example, for the permutation group $\operatorname{S}_n$, where all irreps are real, i.e. $B = I$. 
    In this case, if all subsystems carry the same irrep, both single-body and two-body $G$--invariant non-passive Hamiltonians exist.
    In particular, the Hamiltonian
    \begin{equation}\label{eq:realH}
        \hat{H}^{\text{np}} = e^{i\phi} \sum_{j,k=1}^{n} B_{jk} \hat{a}_j \hat{a}_k + \text{H.c.} \,,
    \end{equation}
    is $G$--invariant for any $\phi\in[0,2\pi)$, where $B_{jk}$ are the matrix elements of $B$.
    
    \item[(iii)] {Pseudo-real irreps.} $U_\lambda$ and $U_\lambda^*$ are equivalent via an anti-symmetric matrix $B$.
    This occurs, for example, for $\SU(2)$, where $B = \Omega_1$.
    In this case, if all systems carry the same irrep, there are two-body $G$--invariant non-passive Hamiltonians, such as the one in Eq.(\ref{eq:pseudoreal_ent}), but no single-body non-passive Hamiltonians. 
    In particular, the Hamiltonian in Eq.(\ref{eq:realH}) vanishes.
\end{enumerate} 
A more detailed discussion of this irrep classification is provided in Appendix~\ref{app:irreps}.

\subsection*{Entangled $G$--invariant Gaussian states}

As another example of important consequences of the above trichotomy of irreps, consider a system composed of $n$ subsystems, each containing $d$ modes that transform under the same $d$--dimensional irrep of a symmetry group $G$. 
Consider the set of pure $G$--invariant Gaussian density operators. 
In the case of a passive representation, this set contains the state in which all modes are in the vacuum.
In the case of $n=1$, i.e. when there is only one subsystem, then the vacuum is the only invariant pure state if the irrep $U$ is complex or pseudo-real.
On the other hand, for $n>1$ subsystems, the vacuum is the only invariant state if the irrep $U$ is complex. 
In other words, when the irrep is real or pseudo-real, $G$--invariant states can contain entanglement, e.g. by evolving the vacuum under the two-body non-passive Hamiltonians in Eq.(\ref{eq:realH}), such as the Hamiltonian in Eq.(\ref{eq:pseudoreal_ent}) for $\SU(2)$.

In Appendix~\ref{app:state_irreps}, we characterize the covariance matrix of certain invariant quantum states.

\section{Resource theory of \\ Gaussian asymmetry}
\label{sec:res_theory}

\subsection{Free states and free operations}

All systems under consideration are equipped with a squeezed passive representation of a symmetry group $G$, as defined in Proposition~\ref{prop:squ}.
Then, we specify a resource theory with the following ingredients:\\

\noindent\textbf{Free states:} Gaussian states that are invariant under a representation of $G$, as characterized by Eq.(\ref{eq:sp_state}) in phase space.\\

\noindent\textbf{Free operations:} Gaussian channels that are covariant under the input and output representations of $G$, as characterized by Eq.(\ref{eq:sp_channel}) in phase space.\\

For brevity, we sometimes refer to this resource theory as \emph{Gaussian asymmetry}, in analogy to Gaussian entanglement.
We emphasize, however, that this terminology is imprecise, since in this framework both asymmetry and non-Gaussianity can play the role of a resource.

Restricting to squeezed passive representations guarantees that the set of free states is not empty, and contains both pure and mixed states.
We can also define free unitary operations as a subset of the free operations.\\

\noindent\textbf{Free unitary operations:} Gaussian unitaries that are invariant under representations of $G$, as characterized by Eq.(\ref{eq:sp_unitary}) in phase space.\\

As discussed in Section~\ref{sec:intro}, restricting to such operations and states is motivated by physical and practical considerations.
In particular, it is less expensive to implement such operations in various operational contexts compared to other operations.

The resource theory of Gaussian asymmetry inherits many desirable properties from the resource theories of asymmetry and non-Gaussianity.
These properties include that the theory is closed under partial trace and tensoring a free state, while the tensor product of free operations is a free operation.
It is a non-convex theory, as in the case of a trivial symmetry it reduces to the non-convex resource theory of non-Gaussianity.

Free operations act transitively on free states. 
That is, for any pair of $G$--invariant Gaussian states $\hat{\rho}_1$ and $\hat{\rho}_2$ there exists a Gaussian channel that converts $\hat{\rho}_1$ to $\hat{\rho}_2$.
Moreover, free unitaries act transitively on free pure states as we show in Appendix~\ref{app:res_transverse}.
\begin{restatable}{proposition}{pureinterconvert}\label{prop:pure_interconversion}
    Let $G$ be a symmetry group with squeezed passive representation $g \mapsto \hat{R}\hat{O}(g)\hat{R}^\dag$ in Hilbert space $\H$, where $\hat{O}$ is passive. 
    Let $\ket{\psi_1}, \ket{\psi_2} \in \H$ be any pair of pure invariant Gaussian states, and denote the displacement vector and covariance matrix of $\ket{\psi_k}$ as $\bmd_k, \sigma_k$, respectively, for $k = 1,2$. 
        
    Then, there exists an invariant Gaussian unitary $\hat{U}$, such that $\hat{U} \ket{\psi_1} = \ket{\psi_2}$.
    Furthermore, the unitary can be realized as $\hat{U} = \hat{D}_{\bmxi} \hat{V}$ with
    \begin{equation}\label{eq:rho1to2}
       \begin{split}
          \bmxi &= \bmd_2-\bmd_1 \,, \text{ and} \\
          V &= R(R^{-1}\sigma_2 R^{\rm -T})^{\frac{1}{2}}(R^{-1}\sigma_1 R^{\rm -T})^{-\frac{1}{2}}R^{-1} \,.
        \end{split}
    \end{equation}
\end{restatable}
\noindent Combined with Proposition~\ref{prop:squ}, this implies that all $G$--invariant pure Gaussian states of a composite system can be obtained from a fixed $G$--invariant pure Gaussian state that is unentangled across subsystems.

\subsection*{Example: U(1) asymmetry}

Recall that a $(2n \times 2n)$ real matrix $M$ commutes with $\Un(1)$ representation $O(\theta) = \exp({-\theta \Omega Q})$, where $Q = \bigoplus_{j=1}^n q_j I$ if and only if its $(2 \times 2)$ blocks $M_{jk}$ satisfy the condition of Eq.(\ref{eq:u1Hblock}), leading to
\begin{equation}\label{eq:u1M}
    M_{jk} = 
    \begin{cases}
        m_{jk}\exp(\theta_{jk} \Omega) \,, 
        &\text{if } q_j = +q_k \,, \\
        m_{jk}Z\exp(\theta_{jk} \Omega) \,, 
        &\text{if } q_j = -q_k \,,
    \end{cases}
\end{equation}
for real coefficients $m_{jk}$ and $\theta_{jk} \in [0, 2\pi)$.
Then, we can specify a resource theory of $\Un(1)$ asymmetry with the following ingredients:\\

\noindent\textbf{Free states:} Gaussian states with displacement vector $\bmd$ that is supported on modes with zero charge, and covariance matrix $\sigma$ of the form of Eq.(\ref{eq:u1M}).
Note that in order for $\sigma$ to be symmetric, we require that $m_{jk} = m_{kj}$ and $\theta_{jk} = \mp \theta_{kj}$ for $q_j = \pm q_k$.
In particular, the diagonal blocks are proportional to the identity, $\sigma_{jj} = a_{jj}I$.\\

\noindent\textbf{Free Channels:} Gaussian channels characterized by Eq.(\ref{eq:sp_channel}) in phase space, with $\bmxi$ supported on modes with zero charge, and $X,Y$ of the form of Eq.(\ref{eq:u1M}).
Again, for $Y$ to be symmetric, we require that $m_{jk} = m_{kj}$ and $\theta_{jk} = \mp \theta_{kj}$ for $q_j = \pm q_k$.\\

\noindent\textbf{Free Unitaries:} Gaussian unitaries $\hat{D}_{\bmxi} \hat{V}$, with $\bmxi$ supported on modes with zero charge, and symplectic matrix $V$ of the form of Eq.(\ref{eq:u1M}).

\subsection{Two types of Gaussian asymmetry}
\label{sec:twotypes}

Gaussian channels do not mix the displacement vector and the covariance matrix of a quantum state $\hat{\rho}$, as seen in Eq.(\ref{eq:xychannel}), which has implications in the manipulation of Gaussian asymmetry.

Given a symplectic Gaussian representation $g \mapsto \hat{S}(g)$, we say that a state $\hat{\rho}$ contains type--1 asymmetry if $S(g)\bmd \neq \bmd$ for some $g \in G$, i.e. its displacement vector breaks the symmetry condition satisfied by $G$--invariant states. 
Similarly, state $\hat{\rho}$ contains type--2 asymmetry if $S(g) \sigma S(g)^{\rm T} \neq \sigma$ for some $g \in G$, i.e. its covariance matrix breaks the symmetry condition satisfied by $G$--invariant states.
Crucially, these two types of asymmetry remain distinct under Gaussian covariant channels, that is, the output state contains each type, only if the input state contains it.
Non-Gaussian covariant operations, on the other hand, can convert type--1 to type--2 asymmetry, but not the other way around. 
Specifically, we prove the following statement in Appendix~\ref{app:types}.
\begin{restatable}{proposition}{asymtypes}\label{prop:types}
    Covariant channels cannot transform a Gaussian state $\hat{\rho}$ with no type--1 asymmetry to a state which contains type--1 asymmetry, regardless of the type--2 asymmetry of $\hat{\rho}$.
    
    On the contrary, for any pure Gaussian state which contains type--1 asymmetry, there exists a, necessarily non-Gaussian, covariant channel that maps it to a state which contains type--2 asymmetry.
\end{restatable}

In Section~\ref{sec:type1mono}, we present monotones that are defined exclusively through either the displacement vector or the covariance matrix, which will be called type--1 and type--2 Gaussian asymmetry monotones, respectively.

While for Gaussian states, lack of asymmetry in the displacement and covariance matrix implies that state is invariant under the action of symmetry, this is not the case for  non-Gaussian states, for which higher moments are not determined by the first two.
For instance, consider the non-Gaussian states $\hat{\rho}_1 = \frac{1}{3}\left(\ket{0} + \ket{3} + \ket{6}\right)\left(\bra{0} + \bra{3} + \bra{6}\right)$ and $\hat{\rho}_2 = \frac{1}{3}\left(\ketbra{0} + \ketbra{3} + \ketbra{6}\right)$. 
They have indistinguishable first two moments, with zero displacement and covariance matrix $7I$, but, under phase shifts, $\hat{\rho}_1$ transforms non-trivially, whereas $\hat{\rho}_2$ is incoherent and remains invariant.
Therefore, the asymmetry content of $\hat{\rho}_1$ lies entirely in its higher-order central moments.
Consequently, to understand and characterize the asymmetry of non-Gaussian states, one can extend the definitions of type--1 and type--2 asymmetry to higher central moments.

\subsection{Complete characterization of first moments under Gaussian covariant dynamics}
\label{sec:char_d}

Consider the set of reachable states by Gaussian covariant channels from a given initial state.
Here, we characterize the displacement vectors of the reachable states.
Recall that for any, possibly non-Gaussian, input state with displacement vector $\bmd$, the displacement vector of the output state of a Gaussian channel is $\bmd' = X\bmd + \bmxi$, where, in general, $\bmxi$ and $X$ are unrestricted real vector and matrix, respectively.
For a covariant Gaussian channel, Eq.(\ref{eq:sp_channel}) implies that the displacement vector $\bmxi$ vanishes on modes where the symmetry acts non-trivially, but remains fully unrestricted on modes carrying the trivial representation of symmetry.
Hence, in the following discussion, we restrict attention to the modes where the symmetry acts non-trivially.
Moreover, for covariant channels, the transformation matrix $X$ must satisfy the commutation relation $[X, S(g)] = 0$ for all $g \in G$, where, for simplicity, we assume that the input and output systems have the same number of modes and carry the same representation of the symmetry.

Consider a group $G$ with squeezed passive representation $g \mapsto S(g)$.
The set of displacement vectors of reachable states is
\begin{equation}
\begin{split}
    \mathcal{S}(\bmd) \coloneqq {\rm span}\{X \bmd:\ &X \in \mathbb{M}_{\mathbb{R}}(2n,2n), \\
    &[X, S(g)] = 0 \text{ for all } g \in G\} \,,
\end{split}
\end{equation}
where $\mathbb{M}_{\mathbb{R}}(2n,2n)$ denotes the set of all real $(2n \times 2n)$ matrices.
To characterize the subspace $\mathcal{S}(\bmd)$, we consider the irrep decomposition of the representation given by Eq.(\ref{eq:irrep_dec}) to irreps.
In particular, we write
\begin{equation}\label{eq:uustar_decomp}
    K S(g) K^{-1} = \bigoplus_{\mu \in M} U_\mu(g) \otimes I_{\N_\mu} \,,
\end{equation}
where $M \coloneqq \Lambda \cup \Lambda^*$ is the union of inequivalent irreps $\Lambda$ in the space of annihilation operators and their duals $\Lambda^*$, and $\N_\mu$ is the multiplicity space of irrep $\mu$.
Note that, in general, matrix $K$ is neither unitary nor symplectic, but is invertible. 

Let $T^{(\mu)}: \mathbb{M}_\mathbb{C}(2n \times 2n) \rightarrow \mathbb{M}_\mathbb{C}(d_\mu \times d_\mu)$ be the completely positive map that projects $(2n \times 2n)$ complex matrices to irrep $\mu$. 
Formally, this map is given by
\begin{equation}\label{eq:Tmu}
    T^{(\mu)}(\cdot) \coloneqq \tr_{\mathcal{N}_\mu}[\Pi_\mu K(\cdot)K^\dag \Pi_\mu] \,,
\end{equation}
where $\Pi_\mu$ is the projector onto irrep $\mu$. 
Roughly speaking, for modes carrying irrep $\mu$, the support of $T^{(\mu)}(\bmd\bmd^{\dag})$ determines\footnote{Note that we have used $\bmd^\dag$ instead of $\bmd^{\rm T}$ to include the case where $\bmd$ is expressed in the basis defined by creation and annihilation operators, where, in general, $\bmd$ is complex.} the subspace $\mathcal{S}(\bmd)$ of displacement vectors that can be reached from the initial state with displacement vector $\bmd$. 

The above intuition is formulated in the following proposition, which fully characterizes  $\mathcal{S}(\bmd)$, the set of displacement vectors of reachable states.
\begin{restatable}{proposition} {suppsubset}\label{prop:suppd}
    Let $G$ be a group with squeezed passive representations $\hat{S}_A(g): g \in G$ and $\hat{S}_B(g): g \in G$ on input system $A$ and output system $B$ with respective sets of inequivalent irreps $M_A$ and $M_B$ in the space of annihilation and creation operators.
    
    Given any quantum state $\hat{\rho}$ with first moment $\bmd_A$, there exists a Gaussian channel that is covariant under input/output representations $\hat{S}_A, \hat{S}_B$ and maps $\hat{\rho}$ to a state with first moment $\bmd_B$ if and only if
    \begin{equation}\label{eq:suppsubseq}
        {\rm supp}\left(T^{(\mu)}\left(\bmd_B\bmd_B^{\dag}\right)\right) \subseteq {\rm supp}\left(T^{(\mu)}\left(\bmd_A\bmd_A^\dag\right)\right) \,,
    \end{equation}
    for all $\mu \in (M_A \cap M_B)\setminus\{\mu_{\rm trivial}\}$, where $\mu_{\rm trivial}$ is the trivial irrep and $T^{(\mu)}$ is defined in Eq.(\ref{eq:Tmu}).
\end{restatable}
\noindent We defer the full proof to Appendix~\ref{app:suppd}.

\subsection{Example: SU(2)--covariant amplifiers and attenuators}
\label{sec:exSU2}

Recall the Schwinger boson representation of $\mathrm{SU}(2)$ in Eq.(\ref{eq:repSU}), defined on two modes with annihilation operators $\hat{a}_1$ and $\hat{a}_2$. 
Consider general Gaussian bosonic channels on this system, and let $X, Y$ be matrices describing the action of such channel in phase space via Eq.(\ref{eq:xychannel}). 
For $\SU(2)$--covariant channels, the displacement vector $\bmxi$ is zero, and the matrices $X$ and $Y$ satisfy
\begin{equation}\label{eq:XYSU}
    [X, S(U)] = 0 \,, \quad S(U) Y S(U)^{\rm T} = Y \,,
\end{equation}
for all $U \in \mathrm{SU}(2)$, where $S(U)$ is the symmetry representation defined in Eq.(\ref{eq:repSU}), which in the $(\hat{a}_1, \hat{a}_2, \hat{a}^\dag_1, \hat{a}^\dag_2)$ basis is given by  
\begin{equation}\label{eq:SU}
    S(U) = U \oplus U^* :\ U \in \SU(2) \,.
\end{equation}

Using the fact that $U^* = \Omega^{\rm T} U \Omega$, this representation can be decomposed into two copies of the defining representation of $\SU(2)$,
\begin{equation}\label{eq:uustar}
    (I_2 \oplus \Omega)(U \oplus U^*)(I_2 \oplus \Omega^{\rm T}) = U \oplus U \cong U \otimes I_2 \,,
\end{equation}
for all $U \in \SU(2)$.
As we show in Appendix~\ref{app:su2cov}, a general matrix $X$ that commutes with the action of symmetry takes the following form in the $(\hat{x}_1,\hat{p}_1,\hat{x}_2,\hat{p}_2)$ basis,
\begin{equation}\label{eq:Xchar}
    X = e^{\phi\Omega_2} \left( x_+ I_4 + x_- e^{\theta\Omega_2} X_- \right)
\end{equation}
for arbitrary $x_\pm \in \mathbb{R}$ and $\theta,\phi \in [0,2\pi)$, where
\begin{align}
    X_- &= \Omega \otimes Z = (Z \oplus Z)V_{\rm bs}\left(\frac{\pi}{2}\right) = (Z \oplus -Z)\text{SW}_{12} \nonumber\\
    &= \begin{pmatrix}
         &  & 1 &  \\
         &  &  & -1 \\
        -1 &  & &  \\
         & 1 &  &  \\
    \end{pmatrix} \,,
\end{align}
in the $(\hat{x}_1,\hat{p}_1,\hat{x}_2,\hat{p}_2)$ basis.
Here, $Z = \diag(1,-1)$ is a Pauli matrix, and $\text{SW}_{12}$ is the swap of modes 1 and 2.

Note that $X_-$ is an orthogonal but non-symplectic matrix, satisfying $X_-^4=I_4$.
In the Hilbert space, it corresponds to the anti-unitary transformation
\begin{equation}\label{eq:hatX-}
    \hat{X}_- = \mathcal{C}\ \hat{V}_{\rm bs}\left(\frac{\pi}{2}\right) = \mathcal{C} \exp(i\pi\hat{J}_y) \,,
\end{equation}
where $\mathcal{C}$ is the anti-unitary corresponding to complex conjugation in the common eigenbasis of $\hat{x}_1, \hat{x}_2$, and $\hat{J}_y = i[\hat{J}_z, \hat{J}_x] = \frac{i}{2}(\hat{a}^\dag_1\hat{a}_2-\hat{a}_1\hat{a}^\dag_2)$. 
The second expression in Eq.(\ref{eq:hatX-}) has the same form as the time-reversal operator on spin 1/2 systems. 
Indeed, under $\hat{X}_-$, the angular momentum operators acquire a minus sign,
\begin{equation}
    \hat{X}_- \hat{J}_w \hat{X}_-^{-1} = -\hat{J}_w \,,\quad \text{for } w \in \{x, y, z\} \,,
\end{equation}
which means that $\hat{X}_-$ acts as time reversal if the two modes describe an angular momentum degree of freedom.
On the other hand, if the two modes describe, e.g. two independent mechanical oscillators, then the time-reversal transformation is simply complex conjugation in the $\hat{x}$ basis, i.e. $\mathcal{C}$, whereas $\hat{X}_-$ acts as $(\hat{X}_- \psi)(x_1,x_2)=\psi(x_2,-x_1)^*$ on the wavefunction.

As an anti-unitary, $\hat{X}_-$ can be physically realized only when the process adds a sufficient amount of noise to the system, which is captured by matrix $Y$.
In fact, the only symmetric matrices $Y$ that satisfy the invariance condition in Eq.(\ref{eq:XYSU}) are of the form $Y = y I_4,\ y \in \mathbb{R}$ in the $(\hat{x}_1,\hat{p}_1,\hat{x}_2,\hat{p}_2)$ basis, as can be seen by imposing that Eq.(\ref{eq:Xchar}) is equal to its transpose.
This leads to the following complete-positivity condition for 2--mode $\SU(2)$--covariant channels,
\begin{equation}\label{eq:SU2channel_uncertainty}
    y \ge \sqrt{\left((x_+ - 1)^2 + x_-^2\right)\left((x_+ + 1)^2 + x_-^2\right)} \,,
\end{equation}
which is independent of phases $\phi, \theta$.

Up to a global phase shift $e^{\phi\Omega_2}$ at the output of the channel, the matrix $X$ is specified by two real parameters $x_\pm$ and a phase $\theta$.
Two special cases of interest are channels with $X = x_+ I_4$ and $X = x_- X_-$, which are analogous to phase-covariant and phase-contravariant channels discussed in Section~\ref{sec:irreps}.
In fact, Eq.(\ref{eq:SU2channel_uncertainty}) reduces to the complete-positivity condition for these channels when $x_- = 0$ and $x_+ = 0$, respectively.
However, unlike in the case of $\mathrm{U}(1)$ symmetry, since $\SU(2)$ is self-dual, both types of channels are covariant for the same system.
These channels can be interpreted as $\SU(2)$--covariant amplifiers and attenuators. In particular, for $x_-=0$  channels with $|x_+|>1$ and $|x_+|<1$ are $\SU(2)$--covariant  amplifiers and attenuators, respectively.

It can then be shown that under such channels, an input with a non-zero displacement vector can be mapped to an output with an arbitrary displacement vector. In other words, as long as the input displacement vector is non-zero, $\SU(2)$ symmetry and Gaussianity do not restrict the displacement of the output state.
Note that the space of displacement vectors on two modes is a 4D real vector space, which is consistent with the fact that matrices $X$ satisfying Eq.(\ref{eq:Xchar}) also form a 4D real space (indeed, this space carries a quaternionic structure). Of course, these constraints do restrict the covariance matrix of the output, ensuring that monotonicity of asymmetry and Gaussian asymmetry monotones, as discussed below, is not violated.  

We conclude by applying Proposition~\ref{prop:suppd} to this example, which immediately implies that, for inputs with non-zero displacement, the displacement vector of the outputs is unrestricted.  
The matrix $K$ in Eq.(\ref{eq:uustar_decomp}) takes the form $I \oplus \Omega$. 
Acting on the displacement vector $\bmd=(\langle\hat{a}_1\rangle, \langle\hat{a}_2\rangle, \langle\hat{a}_1\rangle^*, \langle\hat{a}_2\rangle^*)^{\rm T}$, this gives
\begin{equation}
K \bmd=(\langle\hat{a}_1\rangle, \langle\hat{a}_2\rangle, -\langle\hat{a}_2\rangle^*, \langle\hat{a}_1\rangle^*)^{\rm T} \,.
\end{equation}
The matrix $T(\bmd\bmd^{\dag})$ is obtained by taking the partial trace over the multiplicity subsystem with respect to the tensor product decomposition defined in Eq.(\ref{eq:uustar}):  
\begin{equation}\label{eq:col}
    {\rm supp}(T(\bmd\bmd^{\dag})) = \text{span}\left\{ 
    \begin{pmatrix}
    \langle \hat{a}_1\rangle \\ 
    \langle \hat{a}_2\rangle
    \end{pmatrix} 
    \,, \quad
    \begin{pmatrix}
    -\langle \hat{a}_2\rangle^* \\ 
    \langle \hat{a}_1\rangle^*
    \end{pmatrix}    
    \right\} \,.
\end{equation}
Note that $T(\bmd\bmd^\dag)$ has full support, i.e. ${\rm rank}(T(\bmd\bmd^\dag))=2$, unless both $\langle \hat{a}_1\rangle$ and $\langle \hat{a}_2\rangle$ vanish, in which case $T(\bmd\bmd^\dag)$ is zero, as can be seen by taking the determinant of the matrix whose columns are given in Eq.(\ref{eq:col}).
As we discuss in the next section, ${\rm rank}(T(\bmd\bmd^\dag))$ is a Gaussian asymmetry monotone.

\section{Gaussian asymmetry monotones}
\label{sec:monotones}

Monotones are used in resource theories to quantify how resourceful a state is.
Within the context of Gaussian asymmetry, we define monotones as follows.
\begin{definition}\label{def:monotone}
    Suppose all bosonic systems under consideration are equipped with Gaussian unitary representations of a group $G$. We call a function from the set of quantum states of these systems to the non-negative real numbers a \emph{Gaussian asymmetry monotone} if it is non-increasing under Gaussian $G$--covariant operations.

    If a Gaussian asymmetry monotone depends solely on the displacement vector or the covariance matrix of the quantum state, we further call it a \emph{type--1} or \emph{type--2} asymmetry monotone, respectively.
\end{definition}
It immediately follows that all Gaussian asymmetry monotones remain conserved under (reversible) closed Gaussian dynamics that respect the symmetry.

For instance, all asymmetry monotones~\cite{vaccaro2008tradeoff, marvian2014extending, Marvian2022operational, wigner1963information, skotiniotis2012alignment, toloui2011constructing, marvian2014symmetry}, such as the relative entropy of asymmetry~\cite{vaccaro2008tradeoff, gour2009measuring} and quantum Fisher information metrics~\cite{marvian2020coherence, Marvian2022operational}, are non-increasing under all covariant operations and, therefore, also under Gaussian covariant operations, which are a special case. Consequently, they are automatically Gaussian asymmetry monotones.
Similarly, any non-Gaussianity monotone, such as the Wigner negativity~\cite{kenfack2004negativity}, relative entropy of non-Gaussianity~\cite{genoni2008quantifying}, and stellar rank~\cite{chabaud2020stellar} are monotones for Gaussian asymmetry.
Note that monotones that are faithful, i.e. achieve a minimum if and only if the state is free, in the theories of asymmetry or non-Gaussianity are not generally faithful in the theory of Gaussian asymmetry.

\subsection{Relative entropy of asymmetry}
\label{sec:relent}

An entropic monotone, commonly considered in resource theories~\cite{chitamber2019quantum}, is the relative entropy of resource $\Gamma$.
In the case of asymmetry, consider the twirling operation~\cite{marvian2014symmetry} $\T(\hat{\rho}) = \int {\rm d}g\ \hat{S}(g) \hat{\rho} \hat{S}(g)^\dag$, where the integral is taken with respect to the Haar measure ${\rm d}g$ of group $G$ with unitary representation $g \mapsto \hat{S}(g)$.
Then,
\begin{align}\label{eq:relent_asym}
    \Gamma(\hat{\rho}) 
    &\coloneqq \inf_{\hat{\omega}\in\F} S_v(\hat{\rho}\|\hat{\omega}) \nonumber\\
    &= S_v(\hat{\rho}\|\T(\hat{\rho})) 
    = S_v(\T(\hat{\rho})) - S_v(\hat{\rho}) \,,
\end{align}
where $\F$ is the set of bosonic quantum states that are invariant under the action of the symmetry group $G$, and the von Neumann entropy $S_v$ of state $\hat{\rho}$ is
\begin{equation}\label{eq:vonneumann}
    S_v(\hat{\rho}) \coloneqq -\tr[\hat{\rho} \log(\hat{\rho})] \,.
\end{equation}
In analogy with other quantum resource theories, function $\Gamma$ is often called the relative entropy of asymmetry\footnote{The quantity $S_v(\T(\hat{\rho})) - S_v(\hat{\rho})$ was first introduced in Ref.~\cite{vaccaro2008tradeoff}, and was called \emph{asymmetry}. Then, in Ref.~\cite{gour2009measuring}, it was shown that it can be expressed in terms of the relative entropy.}.

Note that if, instead, we define $\Gamma$ by taking $\F$ to be the set of Gaussian $G$--invariant states, then $\Gamma$ is faithful in the resource theory of Gaussian asymmetry, but it becomes harder to compute, as the last two equations in Eq.(\ref{eq:relent_asym}) do not hold in general.

Consider the 1--mode $\Un(1)$ representation $e^{i\theta} \mapsto e^{i\theta \hat{a}^\dag\hat{a}}$.
In Appendix~\ref{app:monotone}, we calculate the relative entropy of asymmetry of $\hat{\rho}$ in terms of its mean photon number $\langle n \rangle \coloneqq \tr[\hat{\rho}\hat{a}^\dag\hat{a}]$ for pure coherent and squeezed states. 
In particular, for a coherent state $\ket{\alpha} \coloneqq D_{\sqrt{2}({\rm Re}(\alpha),{\rm Im}(\alpha))^{\rm T}}\ket{0}$,
\begin{align}
    \Gamma(\ketbra{\alpha})
    = \frac{1}{2}\log\langle n \rangle + \log(\sqrt{2\pi e}) + O(\langle n \rangle^{-1}) \,, \label{eq:asym_coherent}
\end{align}
where $\langle n \rangle = |\alpha|^2$.
Similarly, for a squeezed state $\ket{r} \coloneqq \exp(r(\hat{a}^{\dagger 2} - \hat{a}^2))\ket{0}$,
\begin{align}
    \Gamma(\ketbra{r})
    = \log\langle n \rangle + \log(2^{-3}\sqrt{\pi}e^{2-\gamma/2}) + o(1) \,, \label{eq:asym_1squeezed}
\end{align}
where $\langle n \rangle = \sinh^2r$, and $\gamma \approx 0.577$ is the Euler– Mascheroni constant.
We plot both monotones in Fig.~\ref{fig:monotones}.

\begin{figure}[t]
    \centering{\includegraphics[width=0.48\textwidth]{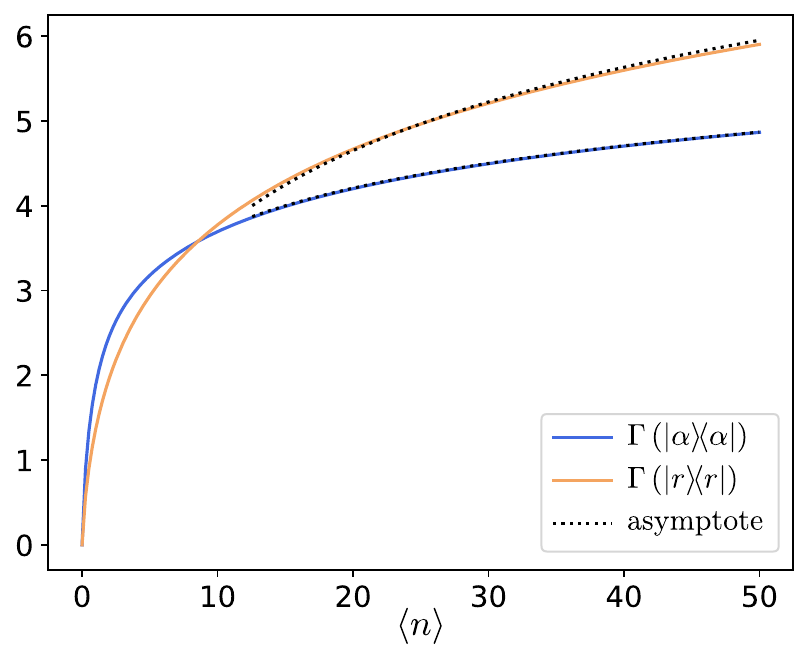}}
    \caption{\textbf{Relative entropy of asymmetry for 1--mode states}. 
    We consider 1--mode $\Un(1)$ symmetry and calculate the Holevo asymmetry $\Gamma$ of coherent states $\ket{\alpha}$ and squeezed vacuum states $\ket{r}$ as a function of mean photon number $\langle n \rangle$.
    The maximum mean photon number $\langle n \rangle = 50$ in the plot corresponds to $|\alpha| = \sqrt{50} \approx 7.07$ and $r = \sinh^{-1}\left(\sqrt{50}\right) \approx 2.65$.
    The dashed lines represent the asymptotic behavior of the monotone as $r, |\alpha| \rightarrow \infty$.
    For $\langle n \rangle \lesssim 8.54$, coherent states have a larger relative entropy of asymmetry, whereas squeezed states exhibit the dominant asymptotic behavior.
    }
    \label{fig:monotones}
\end{figure}

The relative entropy of asymmetry can be extended by considering the weighted twirling operation,
\begin{equation}
    \T_p(\hat{\rho}) = \int {\rm d}g\ p(g) \hat{S}(g) \hat{\rho} \hat{S}(g)^\dag \,,
\end{equation}
where $p$ is a distribution over group elements.
If $p$ is sharp at some group element $g \in G$, then $\T_p(\hat{\rho}) =  \hat{S}(g) \hat{\rho} \hat{S}(g)^\dag$.
On the other hand, if $p$ is the uniform distribution over $G$, then $\T_p \equiv \T$ is the twirling operation.
Then, one can define the so-called Holevo asymmetry monotones~\cite{marvian2014symmetry},
\begin{equation}\label{eq:holevo_monotone}
    \Gamma_p(\hat{\rho}) \coloneqq S_v(\T_p(\hat{\rho})) - S_v(\hat{\rho}) \,,
\end{equation}
for some distribution $p$ over the group $G$.
In addition to the relative entropy of asymmetry, and its generalization to Holevo asymmetry measures, an important class of asymmetry monotones with clear operational interpretations is obtained from the quantum Fisher information~\cite{Marvian2022operational, yadavalli2024optimal, marvian2020coherence}. 
For further approaches to constructing asymmetry monotones, see~\cite{marvian2014symmetry}.

Based on the split of Gaussian asymmetry into type--1 and type--2 asymmetry, we may ask whether we can similarly construct separate monotones for each type.
Indeed, we provide such examples in the following sections.

\subsection{Type--1 asymmetry monotones}
\label{sec:type1mono}

We present a family of functions that depend solely on the first moments of states and are monotone under covariant Gaussian channels. 
Recall the setting of Proposition~\ref{prop:suppd}.  
For each irrep $\mu$ of a symmetry group $G$ and for any, possibly non-Gaussian, state $\hat{\rho}$ with displacement vector $\bmd$, define the integer
\begin{equation}\label{eq:typ1monotone}
    f_\mu(\hat{\rho}) \coloneqq{\rm rank}\left(T^{(\mu)}\left(\bmd\bmd^\dag\right)\right) \,.
\end{equation}
The monotonicity of this quantity under covariant Gaussian channels follows immediately from Eq.(\ref{eq:suppsubseq}). 
Since the support of $T^{(\mu)}(\bmd_B\bmd_B^{\dag})$ for the output state is contained in the support of $T^{(\mu)}(\bmd_A\bmd_A^\dag)$ for the input state, its dimension cannot be larger.

It is also clear that $f_\mu$ is an integer bounded by $d_\mu$, the dimension of irrep $\mu$. 
In particular, for Abelian groups, it can only take values $0$ or $1$, indicating whether or not the first moment of the state has a component in a mode that supports irrep $\mu$.

\subsection{Type--2 asymmetry monotones}
\label{sec:type2mono}

Here, we present a general recipe for constructing type--2 Gaussian asymmetry monotones, based on the notion of type--1 asymmetry--destroying maps.
One can define the type--1 asymmetry-destroying map $\D^{(1)}_{\bmd}: (\bmd, \sigma) \mapsto (\bmo, \sigma)$ that removes the first moment $\bmd$ of a state $\hat{\rho}$.
This map explicitly depends on input $\bmd$ and can be realized as a displacement $\D^{(1)}_{\bmd}(\hat{\rho}) = \hat{D}_{-\bmd}\hat{\rho}\hat{D}_{\bmd}$.

Consider a Gaussian $G$--covariant channel $\E$ whose action in phase space is described by $X,Y,\bmxi$ according to Eq.(\ref{eq:xychannel}), such that $\E(\hat{\rho}_1) = \hat{\rho}_2$ for states $\hat{\rho}_1,\hat{\rho}_2$ with first moments $\bmd_1,\bmd_2$.
We can fully remove the type--1 asymmetry in the input/output states by applying appropriate displacements, obtaining states $\D^{(1)}_{\bmd_1}(\hat{\rho}_1)$ and $\D^{(1)}_{\bmd_2}(\hat{\rho}_2)$ with vanishing first moments. 
Then, the Gaussian $G$--covariant channel $\E'$ described by $X,Y,\bmo$ according to Eq.(\ref{eq:xychannel}), maps $\D^{(1)}_{\bmd_1}(\hat{\rho}_1)$ to $\D^{(1)}_{\bmd_2}(\hat{\rho}_2)$, as depicted\footnote{One can similarly define a type--2 asymmetry-destroying map that sends the covariance matrix $\sigma$ of a state $\hat{\rho}$ to a fixed covariance matrix that respects the symmetry such as the identity $I$ for passive representations. 
However, in general, one cannot guarantee the existence of a covariant channel $\mathcal{E}'$ that maps the transformed input state to the corresponding transformed output state, as depicted in Fig.~\ref{fig:type1destroy} (for instance, this is the case for 1-mode phase-insensitive amplifiers with $\Un(1)$ symmetry).} in Fig.~\ref{fig:type1destroy}.

Given an asymmetry monotone $f$, we can construct a type--2 asymmetry monotone $f^{(2)}: \hat{\rho} \mapsto f(\D^{(1)}_{\bmd}(\hat{\rho}))$, where $\D^{(1)}_{\bmd}$ is the type--1 asymmetry-destroying map that depends on the first moment $\bmd$ of input state $\hat{\rho}$.
Then, using the notation of Fig.~\ref{fig:type1destroy},
\begin{equation}\label{eq:type2monotone}
    f(\D^{(1)}_{\bmd_2}(\E(\hat{\rho}_1))) = f(\E'(\D^{(1)}_{\bmd_1}(\hat{\rho}_1))) \le f(\D^{(1)}_{\bmd_1}(\hat{\rho}_1)) \,,
\end{equation}
because $\E'$ is a $G$--covariant operation.
We proceed to give concrete examples of type--2 asymmetry monotones in the next section.

\begin{figure}[t]
    \centering{\includegraphics[width=0.38\textwidth]{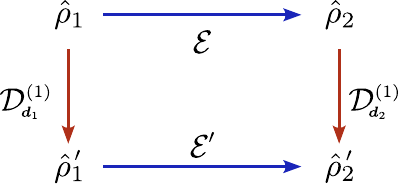}}
    \caption{\textbf{Type--1 asymmetry}.
    Suppose a $G$--covariant Gaussian channel $\E$ sends an input state $\hat{\rho}_1$ to an output state $\hat{\rho}_2$.
    Then, there exists a $G$--covariant Gaussian channel $\E'$ that maps $\hat{\rho}'_1 = \D^{(1)}_{\bmd}(\hat{\rho}_1)$ to $\hat{\rho}'_2 = \D^{(1)}_{\bmd}(\hat{\rho}_2)$.
    }
    \label{fig:type1destroy}
\end{figure}

\subsection{Petz-R{\'e}nyi relative entropy of asymmetry}
\label{sec:PRmono}

Here, we give a recipe for deriving Gaussian asymmetry monotones and corresponding type--2 asymmetry monotones, providing some examples that take explicit closed forms for Gaussian states.

Given any divergence $D$ and any group element $g \in G$, the quantity $f_g(\hat{\rho}) \coloneqq D(\hat{\rho} \| \hat{S}(g) \hat{\rho} \hat{S}^\dag (g))$ is an asymmetry monotone for unitary representation $\hat{S}$ of $G$. 
Concretely, for any $\hat{\rho} \in \B(\H)$ and $G$--covariant $\E: \B(\H) \rightarrow \B(\H)$,
\begin{align}
    f_g(\E(\hat{\rho}))
    &= D(\E(\hat{\rho}) \| \E(\hat{S}(g) \hat{\rho} \hat{S}^\dag (g))) \nonumber\\
    &\le D(\hat{\rho} \| \hat{S}(g) \hat{\rho} \hat{S}^\dag (g))
    = f_g(\hat{\rho}) \,. \label{eq:fg_monotone}
\end{align}
where the first equality follows because $\E$ is covariant and the inequality is the data-processing inequality that divergence $D$ needs to satisfy.

Consider the Petz-R{\'e}nyi $\alpha$--divergence~\cite{petz1986quasi},
\begin{equation}
    D_\alpha(\hat{\rho} \| \hat{\tau}) \coloneqq \frac{1}{\alpha-1} \log \tr[\hat{\rho}^\alpha \hat{\tau}^{1-\alpha}] \,,
\end{equation}
defined for $\alpha \in (0,1) \cup (1,\infty)$.
This quantity has been calculated explicitly~\cite{seshadreesan2018renyi} in terms of the displacement vector and covariance matrix of any two Gaussian states $\hat{\rho}, \hat{\tau}$, for $\alpha \in (0,1)$.
For our purposes, consider $f_{\alpha, g} \coloneqq D_\alpha(\hat{\rho} \| \hat{S}(g)\hat{\rho}\hat{S}(g)^\dag)$, where $\hat{\rho}$ has first moment $\bmd$ and covariance matrix $\sigma$.
We find in Appendix~\ref{app:PRentropy} that
\begin{align}
    f_{\alpha, g} &= \frac{1}{1-\alpha} \left(f_{\alpha,g}^{(1)}(\hat{\rho}) + f_{\alpha,g}^{(2)}(\hat{\rho})\right) \,, \label{eq:f} \\
    f_{\alpha,g}^{(1)}(\hat{\rho}) &\coloneqq \bmd_g^{\rm T} \left(\sigma_\alpha + S(g)\sigma_{1-\alpha}S(g)^{\rm T}\right)^{-1}\bmd_g \,, \label{eq:f1}\\
    f_{\alpha,g}^{(2)}(\hat{\rho}) &\coloneqq -\frac{1}{2}\log\frac{\det[(\sigma_\alpha + \sigma_{1-\alpha})/2]}{\det\left[\left(\sigma_\alpha + S(g)\sigma_{1-\alpha}S(g)^{\rm T}\right)/2\right]} \,, \label{eq:f2}
\end{align}
where $\bmd_g \coloneqq (I - S(g))\bmd$ is non-zero if $S(g) \neq I$, and
\begin{equation}\label{eq:sigmaalpha}
    \sigma_{\alpha} \coloneqq \coth(\alpha\ {\rm coth}^{-1}(\sigma i\Omega))i\Omega \,,
\end{equation}
for $\alpha \in (0,1)$, is the covariance matrix of $\hat{\rho}^\alpha/\tr[\hat{\rho}^\alpha]$, i.e. the thermal state of Hamiltonian $\hat{H}$ at temperature $(\alpha\beta)^{-1}$, where $\hat{\rho}$ is the thermal state of Hamiltonian $\hat{H}$ at temperature $\beta^{-1}$.
This follows from the well-known relation $\sigma i\Omega = \coth\left(\frac{1}{2}i\Omega H\right)$ between the covariance matrix $\sigma$ of state $\hat{\rho}$ and the Hamiltonian matrix $H$ of $\hat{H}$.
Inserting this relation in Eq.(\ref{eq:sigmaalpha}) gives $\sigma_{\alpha} i\Omega = \coth\left(\frac{1}{2}i\Omega (\alpha H)\right)$. 

Note that if $\hat{\rho}$ is pure, so is $\hat{\rho}^\alpha/\tr[\hat{\rho}^\alpha]$ and $\sigma_\alpha = \sigma$ for all $\alpha \in (0,1)$.
Then, Eqs.(\ref{eq:f1},\ref{eq:f2}) simplify to
\begin{align}
    f_{\alpha,g}^{(1)}(\hat{\rho}) &= \bmd_g^{\rm T} \left(\sigma + S(g) \sigma S(g)^{\rm T}\right)^{-1}\bmd_g \,, \label{eq:f1_pure}\\
    f_{\alpha,g}^{(2)}(\hat{\rho}) &= \frac{1}{2}\log\det\left[\left(\sigma + S(g) \sigma S(g)^{\rm T}\right)/2\right] \,. \label{eq:f2_pure}
\end{align}
We emphasize that the expressions in Eqs.(\ref{eq:f}--\ref{eq:f2},\ref{eq:f1_pure},\ref{eq:f2_pure}) are not valid for non-Gaussian states.

Under closed covariant dynamics, the monotone $f_{\alpha,g}$ is conserved, since Ineq.(\ref{eq:fg_monotone}) holds as equality.
If the dynamics is Gaussian, both $f_{\alpha,g}^{(1)}, f_{\alpha,g}^{(2)}$ are independently conserved, i.e. for any invariant Gaussian unitary $\hat{U}$, it is evident that $f_{\alpha,g}^{(k)}(\hat{U}\hat{\rho}\hat{U}^\dag) = f_{\alpha,g}^{(k)}(\hat{\rho})$, for $k=1,2$.

Under open dynamics, the quantity $f_{\alpha,g}^{(2)}$ is a monotone for type--2 asymmetry due to Eq.(\ref{eq:type2monotone}).
We can rewrite $f_{\alpha,g}^{(2)}$ for $\alpha = 1/2$ as follows,
\begin{align}
    f_{1/2, \hspace{1pt}g}^{(2)}(\hat{\rho}) 
    &= -\frac{1}{2}\log\frac{\det[\sigma_{1/2}]}{\det[\left(\sigma_{1/2} + S(g)\sigma_{1/2}S(g)^{\rm T}\right)\hspace{-2pt}/2]} \\
    &= \log(\frac{\tr[\hat{\rho}_{1/2}^2]}{\tr\left[\D_{\bmd}^{(1)}(\hat{\rho}_{1/2})\D_{S(g)\bmd}^{(1)}(\hat{S}(g)\hat{\rho}_{1/2}\hat{S}(g)^\dag)\right]}) \,,\nonumber
\end{align}
which expresses the ratio of the purity of state $\hat{\rho}_{1/2}$ over the overlap between state $\hat{\rho}_{1/2}$ with removed displacement vector and its rotation by $g \in G$.

Consider the $\Un(1)$ symmetry generated by a purely quadratic observable $\widehat{Q}$, given by $\hat{S}(t) = e^{i\widehat{Q}t}$ with associated symplectic matrix $S(t) = e^{-\Omega Qt}$.
In Appendix~\ref{app:smallt}, we show that, for this continuous symmetry, we obtain additional monotones of Gaussian asymmetry and type--2 asymmetry as the second derivatives $F_Q$ of $f_{1/2,t}$ and $F_Q^{(2)}$ of $f_{1/2,t}^{(2)}$, respectively.
In particular,
\begin{align}
    F_Q 
    &\coloneqq F_Q^{(1)} + F_Q^{(2)} \\
    F_Q^{(1)}(\hat{\rho}) 
    &\coloneqq -\bmd^{\rm T} Q\Omega \sigma_{1/2}^{-1} \Omega Q \bmd \,, \label{eq:FQ1} \\
    F_Q^{(2)}(\hat{\rho}) 
    &\coloneqq \frac{1}{4}\tr\Big[\Omega Q\Omega\sigma_{1/2}^{-1}[\Omega Q, \sigma_{1/2}\Omega]\Big] \,. \label{eq:FQ2}
\end{align}
Note that if $Q$ commutes with $\Omega$, then $F_Q^{(2)}(\hat{\rho})$ attains a zero value for every state with $\sigma_{1/2} = I$, e.g. the vacuum. 
On the contrary, if $Q$ does not commute with $\Omega$, then the representation is not passive (see Fig.~\ref{fig:repr}), so, for example, the vacuum is no longer a free state.

As an example, consider free evolution with frequency $\omega$, i.e. $Q = \omega I$, on 1 mode.
For a coherent state $\ket{\gamma}$, we have $f_{\alpha, t}^{(2)}(\ket{\gamma}) = 0$, and
\begin{equation}
    f_{\alpha, t}(\ket{\gamma}) = f_{\alpha, t}^{(1)}(\ket{\gamma}) = 2(1-\cos(\omega t))|\gamma|^2 \,,
\end{equation}
for all $\alpha \in (0,1)$.
Therefore, $F_Q^{(2)}(\ket{\gamma}) = 0$, and
\begin{equation}
    F_Q(\ket{\gamma}) = F_Q^{(1)}(\ket{\gamma}) = 2\omega^2|\gamma|^2 \,.
\end{equation}
For a squeezed state $\ket{r}$, we have $f_{\alpha, t}^{(1)}(\ket{r}) = 0$, and
\begin{equation}
    f_{\alpha, t} = f_{\alpha, t}^{(2)}(\ket{r}) = \frac{1}{2}\log\left(1 + \sin^2(\omega t) \sinh^2(2r)\right) \,,
\end{equation}
for all $\alpha \in (0,1)$.
Therefore, $F_Q^{(1)}(\ket{r}) = 0$ and 
\begin{equation}
    F_Q(\ket{r}) = F_Q^{(2)}(\ket{r}) = \omega^2 \sinh^2(2r) \,.
\end{equation}

\section{Conservation Laws in closed Gaussian systems}
\label{sec:conserve_laws}

It is well known that, in the case of closed systems, symmetries of the dynamics imply conservation laws, a result often colloquially referred to as Noether's theorem~\cite{Noether1918nachrichten}. 
For instance, in the presence of rotational $\SU(2)$ symmetry, the angular momentum operators remain conserved. More generally, for a symmetry group $G$, if an initial state $\hat{\rho}$ evolves under a symmetry-respecting unitary into a final state $\hat{\rho}'$, then we have  
\begin{equation}\label{eq:noether}
    \tr[\hat{\rho}\hat{U}(g)] = \tr[\hat{\rho}'\hat{U}(g)] \,, \quad\text{for all } g \in G \,.
\end{equation}
If the group is generated by an observable $\widehat{Q}$ such that $\hat{U}(s) = e^{i\widehat{Q} s}$, $s \in \mathbb{R}$, we find that, for arbitrary integer $k$,
\begin{equation}\label{eq:generic_coserve}
    \tr[\hat{\rho}\widehat{Q}^k] = \tr[\hat{\rho}'\widehat{Q}^k] \,.
\end{equation}
In the aforementioned example of $\SU(2)$ symmetry, this equation implies that all expectation values of polynomials of angular momentum operators remain conserved.

When $\hat{\rho}$ is a pure state, Eq.(\ref{eq:noether}) captures all the implications of the symmetry in the following sense: if this conservation law holds for all $g \in G$, then there exists a unitary transformation $\hat{V}$ that respects the symmetry, i.e. $[\hat{V}, \hat{U}(g)] = 0$ for all $g \in G$, and converts between the two states $\hat{\rho}' = \hat{V} \hat{\rho} \hat{V}^\dag$~\cite{marvian2013theory}\footnote{It is worth noting that this does not generally hold when $\hat{\rho}$ is mixed; in fact, the conservation of measures of asymmetry imposes independent conservation laws on the dynamics of the system that are not captured by Eq.(\ref{eq:noether})~\cite{marvian2014extending}.}.

Clearly, if the dynamics of a closed system that respects the symmetry is Gaussian, they satisfy the conservation laws in Eqs.(\ref{eq:noether},\ref{eq:generic_coserve}). This naturally raises the following question: Are there conservation laws that are satisfied by Gaussian closed systems that respect the symmetry, but are not necessarily satisfied by non-Gaussian closed systems that respect the symmetry? In other words, does the presence of symmetries in Gaussian systems impose any additional constraints on the dynamics? In the following, we show that the answer is affirmative, and we present a class of such conservation laws.

The intuition is simple: conservation laws of the form given in Eq.(\ref{eq:generic_coserve}) only consider the total amount of charge in the system. However, roughly speaking, the conserved charge and asymmetry can be decomposed into different components corresponding to central moments of different order, which do not mix under Gaussian dynamics that respect the symmetry. Therefore, by analyzing the contributions to charge and asymmetry in these different central moments, one can identify new independent conserved quantities.

\subsection{Conservation laws for closed Gaussian systems}
\label{sec:gauss_conserve}

Suppose that under a Gaussian unitary evolution described by $\hat{V}$, an arbitrary, possibly non-Gaussian, initial state $\hat{\rho}_1$, with displacement vector $\bmd_1$ and covariance matrix $\sigma_1$, evolves to the final state $\hat{\rho}_2 = \hat{V} \hat{\rho}_1 \hat{V}^\dag $, with displacement vector $\bmd_2$ and covariance matrix $\sigma_2$. 
Suppose further that $\hat{V}$ is invariant under a symplectic Gaussian representation of symmetry $e^{ir\widehat{Q}}$ with quadratic charge
\begin{equation}\label{eq:quad_charge}
    \widehat{Q} = \sum_{j,k} Q_{jk} \hat{r}_j \hat{r}_k \,,\quad Q = Q^{\rm T} \,,
\end{equation}
for all $r \in \mathbb{R}$, which corresponds to representation $e^{-r\Omega Q}$ in the phase space.
Then,  
\begin{equation}\label{eq:generic_gauss_conserve}
    \bmd_1^{\rm T} Q \bmd_1 = \bmd_2^{\rm T} Q \bmd_2 \,,\quad\quad \tr[\sigma_1 Q] = \tr[\sigma_2 Q] \,.
\end{equation}
Crucially, while in a Gaussian charge-conserving closed system each of these quantities is individually conserved, this is not generally the case in non-Gaussian charge-conserving closed systems. We provide a counter-example in Section~\ref{sec:violation}.
Instead, due to Eq.(\ref{eq:generic_coserve}) for $k=1$, $\tr[\hat{\rho}_1\widehat{Q}]=\tr[\hat{\rho}_2\widehat{Q}]$, which, in turn,  implies  the following linear combination remains conserved,
\begin{equation}\label{eq:nonGauss_law}
    \tr\left[\sigma_1 Q\right] + 2 \bmd_1^{\rm T} Q\bmd_1 = \tr\left[\sigma_2 Q\right] + 2 \bmd_2^{\rm T} Q\bmd_2 \,. 
\end{equation}
To derive more general conservation laws, we first establish the following simple result.
\begin{lemma}\label{lem:cons}
    Consider the function  
    \begin{equation}\label{eq:Fk}
        F_B(A) = \tr[f(A \Omega, \Omega B)] \,,
    \end{equation}
    where $A$ and $B$ are $(2n\times 2n)$ matrices, and $f$ is any arbitrary non-commutative polynomial.
    Under any symplectic transformation $V \in \Sp(2n,\mathbb{R})$ that commutes with $\Omega B$, i.e. $[V, \Omega B] = 0$, we have
    \begin{equation}
        F_B(A) = F_B(VAV^{\rm T}) \,.
    \end{equation}
\end{lemma}
\begin{proof}
For any invertible matrix $V$,
\begin{align}
    \tr[f(A \Omega, \Omega B)]
    &=\tr[Vf(A\Omega, \Omega B)V^{-1}] \\ 
    &=\tr[f(V A\Omega V^{-1}, V \Omega B V^{-1})] \\
    &=\tr(f(VAV^{\rm T} \Omega, \Omega B)) \,,
\end{align}
where the second equality follows from the property $Vf(X,Y)V^{-1} = f(VXV^{-1}, VYV^{-1})$ of non-commutative polynomials, and the third from the assumptions that $[V, \Omega B] = 0$ and that $V$ is symplectic.
\end{proof} 
For instance, choosing $f(X,Y) = (XY)^k$, for integer $k$, and $A = \sigma$, we obtain $F_B(A) = \tr[(AB)^k]$.
Then, assuming $V$ commutes with $\Omega Q$ for $Q$ in Eq.(\ref{eq:quad_charge}), choosing $B = Q$ leads to the conserved quantity $F_Q(A) = \tr\left[(A Q)^k\right]$, which in turn implies the conservation law
\begin{equation}
    \tr\left[ (\sigma_1 Q)^k \right] = \tr\left[ (\sigma_2 Q)^k \right] \,.
\end{equation}

Similarly, if the Gaussian unitary $\hat{V}$ respects a group $G$ with symplectic representation $S: G \rightarrow \Sp(2n, \mathbb{R})$, then choosing $B = -\Omega S(g)$ yields the conservation law
\begin{equation}
    \tr\left[ (\sigma_1 \Omega S(g))^k \right] = \tr\left[ (\sigma_2 \Omega S(g))^k \right] \,.
\end{equation}
for every group element $g \in G$.

\subsection*{Example: $\SU(2)$ symmetry and conservation of angular momentum}

Recall the $\SU(2)$ example with the angular operators defined in Eq.(\ref{eq:jzjx}).
Let $\bmd_{j,a}$ and $\sigma_{j,a}$ denote the displacement vector and covariance matrix of mode $a$ in subsystem $j$, and let $\bmd_{j,b}$ and $\sigma_{j,b}$ denote the displacement vector and covariance matrix of mode $b$ in subsystem $j$.

The expectation value $\tr[\hat{\rho}\hat{J}_z]$ remains conserved under general $\SU(2)$--invariant dynamics, and, if the dynamics is also Gaussian, Eq.(\ref{eq:generic_gauss_conserve}) implies the independent conservation of the following two quantities,
\begin{align}
    \sum_{j=1}^n \left( \bmd_{j,a}^{\rm T}\bmd_{j,a} - \bmd_{j,b}^{\rm T}\bmd_{j,b} \right) \,,\quad \sum_{j=1}^n \Big( \tr[\sigma_{j,a}] - \tr[\sigma_{j,b}] \Big) \,.
\end{align}

\subsection{Violation of conservation laws
in non-Gaussian closed systems}
\label{sec:violation}

Consider a single mode with $\Un(1)$ representation $\exp({i\phi \hat{a}^\dag \hat{a}})$, $\phi\in[0,2\pi)$, which corresponds to the conserved charge $\widehat{Q}=\hat{a}^\dag \hat{a}$ and $Q=I_2$ is the $(2 \times 2)$ identity matrix. 
Consider a system initially in the pure coherent state $\ket{\alpha} \coloneqq \hat{D}_{\bmd_1} \ket{0}$, with first moment $\bmd_1 \coloneqq \sqrt{2}({\rm Re}(\alpha), {\rm Im}(\alpha))^{\rm T}$, and covariance matrix $\sigma_1 = I_2$, and suppose it evolves to state $\hat{U}\ket{\alpha}$, where $\hat{U} = \hat{I} - 2\ketbra{0}$ is a non-Gaussian U(1)--invariant unitary. Under this transformation, the state's displacement vector and covariance matrix become
\begin{align}
    \bmd_2 &= c\bmd_1 \,, \label{eq:ngmoment1}\\
    \sigma_2 &= (1+(1-c)\bmd_1^{\rm T}\bmd_1)I_2 + 2c(1-c)\bmd_1\bmd_1^{\rm T} \,, \label{eq:ngmoment2}
\end{align}
where $c \coloneqq 1-2e^{-|\alpha|^2}$.
Although the conservation law in Eq.(\ref{eq:nonGauss_law}) still holds, i.e.
\begin{equation}
    \tr\left[\left(\sigma_k + 2\bmd_k\bmd_k^{\rm T}\right)Q\right]
    = 2(1+2|\alpha|^2) \,,
\end{equation}
for $k = 1,2$, the conservation laws in Eq.(\ref{eq:generic_gauss_conserve}) are not satisfied for non-zero displacement $\alpha$.
It is worth noting that a non-zero displacement at the input causes the output covariance matrix to acquire off-diagonal elements. 
As a result, it no longer commutes with $\Omega_1$, and so, unlike the input covariance matrix, it breaks the $\Un(1)$ symmetry.

\subsection{A complete set of conservation laws for incoherent states}

An important case of interest is when the symmetry is represented by a passive quadratic charge $\widehat{Q}$. 
Let $Q = \sum_{q \in {\rm Eig}(Q)} q\Pi_q$ be the spectral decomposition of the symmetric matrix $Q$, where $\Pi_q$ is the projector on the modes $\{j: q_j = q\}$ with charge $q$.
Then, using Lemma~\ref{lem:cons}, one can directly obtain the conservation laws
\begin{equation}\label{eq:cons}
    \tr\left[(\sigma_1 (\Pi_{q}-\Pi_{-q}))^k \right] = \tr\left[(\sigma_2 (\Pi_{q}-\Pi_{-q}))^k \right] \,,
\end{equation}
for $|q| > 0$ and all integers $k$.
To see this, first note that 
according to Proposition~\ref{prop:Qprop}, $[Q,\Omega]=0$ when $\widehat{Q}$ is passive. 
Thus, any symplectic matrix that commutes with $\Omega Q$, also commutes with $Q^2=-(Q\Omega)^2$, hence also with projectors to the eigenspaces of $Q^2$, namely $\Pi_q+\Pi_{-q}$ for $q \in {\rm Eig}(Q)$. 
This holds because each $\Pi_q+\Pi_{-q}$ can be written as a polynomial $p_q$ of $Q^2$, i.e.
\begin{equation}
    \Pi_q+\Pi_{-q} = p_q(-(Q\Omega)^2) = p_q(Q^2) \,.
\end{equation}
It follows that by choosing $f(X,Y) = (XYp_q(-Y^2))^k$ in Proposition~\ref{lem:cons}, and $A = \sigma$ and $B = Q$, we obtain the above conservation laws.

The conservation laws in Eq.(\ref{eq:cons}) only restrict sectors with $q \neq 0$. 
In the sectors with $q = 0$, the symmetry does not impose any restrictions. 
Nevertheless, the fact that the dynamics is closed and Gaussian imposes non-trivial constraints on the dynamics, namely,
\begin{equation}
    \tr[(\sigma_1 \Omega \Pi_0)^k] = \tr[(\sigma_2 \Omega \Pi_0)^k] \,,
\end{equation}
The above equations can be equivalently stated as equality of symplectic eigenvalues of $\sigma_1\Pi_0$ and $\sigma_2\Pi_0$ on modes with zero charge, and equality of ordinary eigenvalues of $\sigma_1[\Pi_q-\Pi_{-q}]$ and $\sigma_2[\Pi_q-\Pi_{-q}]$ on modes with non-zero charge, i.e.
\begin{align}
    {\rm SympEig}\left(\sigma_1\Pi_0\right) &= {\rm SympEig}\left(\sigma_2\Pi_0\right) \,,\hspace{13pt} q = 0 \,, \label{eq:Qaltcond1} \\
    {\rm Eig}(\sigma_1 (\Pi_{q}-\Pi_{-q})) &= {\rm Eig}(\sigma_2 (\Pi_{q}-\Pi_{-q})) \,,\ q \neq 0 \,. \label{eq:Qaltcond2}
\end{align}
We emphasize that, for any covariance matrix $\sigma$, matrix $\sigma (\Pi_q-\Pi_{-q})$ is indeed diagonalizable, and the number of its positive and negative eigenvalues is equal to the number of modes with charge $|q|$ and $-|q|$, respectively\footnote{In general, for any positive real symmetric matrix $A$ and real symmetric matrix $B$, $AB$ is diagonalizable, and the number of its positive, negative and zero eigenvalues, are equal to those of $B$. To see this consider the  similarity transformation $AB\rightarrow A^{-1/2} AB A^{1/2}$. This implies $AB$ is similar to $A^{1/2} B A^{1/2}$. Then, applying the Sylvester's law of inertia, we find that the number of positive, negative, and zero eigenvalues of $B$ is equal to those of $A^{1/2} B A^{1/2}$,  and thus to those of $AB$. Applying this to $A=(\Pi_q+\Pi_{-q}) \sigma (\Pi_q+\Pi_{-q})$ and $B=Q$ proves the claim.}.

Interestingly, it turns out that for Gaussian states that (i) have a vanishing displacement vector, and (ii) contain no coherence with respect to the charge $\widehat{Q}$, these conservation laws capture all the implications of symmetry: if they hold, then there exists a $G$--invariant Gaussian unitary transformation that converts one into the other.
In particular, we show the following statement.
\begin{restatable}{theorem}{thminterconvert}\label{thm:interconvert}
    Consider a symmetric matrix
    \begin{equation}
        Q = \sum_{q \in {\rm Eig}(Q)} q\Pi_q \,,
    \end{equation}
    such that $[Q, \Omega] = 0$, where $\Pi_q$ projects on the sector with charge $q$, i.e. on modes $\{j: q_j = q\}$.
    
    Suppose $\sigma_1, \sigma_2$ are two covariance matrices such that
    \begin{equation}
        [\sigma_1, \Omega Q] = [\sigma_2, \Omega Q] = 0 \,.
    \end{equation}
    Then, there exists a symplectic transformation $S$ such that $SQS^{\rm T} = Q$ and $S \sigma_1 S^{\rm T} = \sigma_2$ if and only if conditions
    \begin{align}
        \tr[(\sigma_1 \Omega \Pi_0)^k] &= \tr[(\sigma_2 \Omega \Pi_0)^k] \,, \label{eq:Qcond1} \\
        \tr\left[(\sigma_1 (\Pi_{q}-\Pi_{-q}))^k \right] &= \tr\left[(\sigma_2 (\Pi_{q}-\Pi_{-q}))^k \right] \,, \label{eq:Qcond2}
    \end{align}
    hold for all $|q| > 0$ and integers $k$.
\end{restatable}
\noindent The proof is deferred to Appendix~\ref{app:interconversion}.
Note that one only needs to check the conditions of Eqs.(\ref{eq:Qcond1},\ref{eq:Qcond2}) for $k \le n(q)$, where $n(q)$ is the total number of modes $\{j: |q_j| = |q|\}$ with charge $\pm q$.

Theorem~\ref{thm:interconvert} allows one to determine what quantum states are reachable from a given initial state via invariant symplectic transformations.
As an example, suppose a system of $n$ modes is initially prepared in a product of uncorrelated thermal states, with covariance matrix $\bigoplus_{j=1}^n \nu_j I_2$. Then, a Gaussian state with covariance matrix $\sigma$ is reachable from this state via a charge-conserving symplectic transformation if and only if 
\begin{align*}
    {\rm Eig}\left(\sigma  (\Pi_q - \Pi_{-q}) \right) = \left\{\frac{q_{j}}{|q|} \nu_j: |q_j|=|q|\right\} \,,
\end{align*}
for all $|q| > 0$, and ${\rm SympEig}\left(\sigma_1\Pi_0\right) = {\rm SympEig}\left(\sigma_2\Pi_0\right)$.
We conclude that 
\begin{equation}\label{eq:Qcondk1}
    \frac{1}{2}\left(\tr[\sigma_q]-\tr[\sigma_{-q}]\right) = \sum_{j: q_j=|q|} \nu_j \hspace{2pt}- \hspace{-8pt}\sum_{j: q_j=-|q|} \nu_j \,,
\end{equation}
where we have defined $\sigma_{p} \coloneqq \Pi_p \sigma \Pi_p$ to be the restriction of covariance matrix $\sigma$ on the modes with charge $p$.

Similarly, for $k=2$ we obtain
\begin{equation}\label{eq:Qcondk2}
    \|\sigma_{q}\|_2^2+\|\sigma_{-q}\|_2^2-2\|\sigma_{q,-q}\|^2_2 = \sum_{j: |q_j|=|q|} \nu^2_j \,,
\end{equation}
where $\|A\|_2 \coloneqq (\tr[A^2])^{1/2}$, and $\sigma_{q,-q} \coloneqq \Pi_q \sigma \Pi_{-q}$.

Concretely, consider the symmetry generated by $Q = I \oplus -I$ in the $(\hat{x}_1,\hat{p}_1,\hat{x}_2,\hat{p}_2)$ basis.
According to Eq.(\ref{eq:u1M}), the invariant covariance matrices take the form
\begin{equation}
    \sigma = \begin{pmatrix}
        a_{11} I & a_{12}Ze^{\theta\Omega}\\
        a_{12}Ze^{\theta\Omega} & a_{22} I
    \end{pmatrix}
    \,,
\end{equation}
for $a_{11}, a_{22}, a_{12} \ge 0$ and $\theta \in [0,2\pi)$.
Then, the conservation laws as written in Eqs.(\ref{eq:Qcondk1},\ref{eq:Qcondk2}) imply
\begin{equation}\label{eq:2modeQcond}
\begin{aligned}
    a_{11} - a_{22} &= \nu_1 - \nu_2 \,,\\
    a_{11}^2 + a_{22}^2 -2 a_{12}^2 &= \nu_1^2 + \nu_2^2 \,,
\end{aligned}
\end{equation}
which, in particular, combine to give the conservation of the determinant of $\sigma$ under symplectic transformations,
\begin{equation}
 a_{11} a_{22} - a_{12}^2 = \nu_1 \nu_2 \,.
\end{equation}
We conclude that, once any of $a_{11}$, $a_{22}$, or $a_{12}$ is chosen, the other two are uniquely determined.

\subsection{Extension of Williamson's theorem (simultaneous normal mode decomposition)}
\label{sec:nmd}

In this section, we establish a symmetric extension of Williamson's seminal theorem~\cite{williamson1936on,nicacio2021williamson}, which is a critical step in the proof of Theorem~\ref{thm:interconvert}.

The content of Williamson's theorem~\cite{williamson1936on} is that any $(2n \times 2n)$ real symmetric positive-definite matrix $M$ can be diagonalized by congruence using a symplectic matrix $S$, such that the symplectic eigenvalues appear in pairs, i.e. $S M S^{\rm T} = \bigoplus_{j=1}^n \nu_j I_2$ for real positive $\nu_j$, where $I_2$ denotes the $(2 \times 2)$ identity matrix.
It has recently been shown~\cite{kamat2025simultaneous} that two $(2n \times 2n)$ real symmetric positive-definite matrices $M_1, M_2$ can be symplectically diagonalized by congruence using one symplectic matrix if and only if $M_1$ and $M_2$ symplectically commute, i.e. $[\Omega M_1, \Omega M_2] = 0$.
In such a scenario, we say that $M_1$ and $M_2$ admit a simultaneous normal mode decomposition.

For applications in the context of symmetry, we need to consider a different extension of Williamson's theorem, in which two $(2n \times 2n)$ real symmetric matrices $M$ and $Q$ are given, with only $M$ required to be positive-definite, and find necessary and sufficient conditions for $M$ and $Q$ to admit a simultaneous normal-mode decomposition.
\begin{restatable}{theorem}{thmnmdQ}\label{thm:nmd_q}
    Let $M$ and $Q$ be $(2n \times 2n)$ real symmetric matrices, such that $M$ is positive-definite and $[Q, \Omega] = 0$.
    
    Then, there exists a symplectic matrix $S$ that diagonalizes both $Q$ and $M$ by congruence,
    \begin{align}
        S M S^{\rm T} &= \bigoplus_{j=1}^n \nu_j I_2 \,,\quad \nu_j \ge 0 \,, \label{eq:nmd_condition1}\\
        S Q S^{\rm T} &= \bigoplus_{j=1}^n q_j I_2 \,,\quad q_j \in \mathbb{R} \,, \label{eq:nmd_condition2}
    \end{align}
    if and only if $M$ and $Q$ satisfy
    \begin{equation}\label{eq:nmd_comm}
        [\Omega M, \Omega Q] = 0 \,,
    \end{equation}
    where $I_2$ denotes the $(2 \times 2)$ identity matrix.
\end{restatable}
\noindent The proof is deferred to Appendix~\ref{app:nmd_invariant}.
Note that the condition in Eq.(\ref{eq:nmd_comm}) can be equivalently stated as
\begin{align}
    [\Omega M, \Omega Q] = [\Omega M, Q] = [M, \Omega Q] = 0 \,.
\end{align}

If $Q$ is taken as the symmetric matrix that corresponds to a conserved charge $\widehat{Q}$ and $M$ is a Hamiltonian matrix or the covariance matrix of a quantum state, the first commutation relation in Eq.(\ref{eq:nmd_comm}) represents invariance of $M$ under the symmetry generated by $\widehat{Q}$.
The second commutation relation in Eq.(\ref{eq:nmd_comm}) implies that the charge is passive.
The standard Williamson's theorem~\cite{williamson1936on} corresponds to the special case of trivial symmetry, i.e. $Q=0$.

Normal mode decomposition when all charges $q_j$ are equal, i.e. $Q = qI$, $q \in \mathbb{R}$ is an important special case of Theorem~\ref{thm:nmd_q}, and we highlight it here.
\begin{proposition}\label{prop:passive}
    For any Gaussian state $\hat{\rho}$ with covariance matrix $\sigma$, the following statements are equivalent:
    \begin{enumerate}
        \item[(1)] The covariance matrix satisfies $[\sigma, \Omega] = 0$;
        \item[(2)] There exists a passive symplectic Gaussian unitary $\hat{S}$ such that, up to displacement,
        \begin{equation}\label{eq:thermal_decomp}
            \hat{S}\hat{\rho}\hat{S}^\dag = \bigotimes_{j=1}^n \frac{e^{-\beta_j \hat{a}_j^\dag \hat{a}_j}}{\tr[e^{-\beta_j \hat{a}_j^\dag \hat{a}_j}]} \,,
        \end{equation}
        is a product of thermal states with $\beta_j > 0$;
        \item[(3)] The covariance matrix consists of $(2 \times 2)$ blocks describing correlations between modes $j$ and $k$, that have the form $\sigma_{jk} = a_{jk}\exp(\theta_{jk} \Omega)$, with $a_{jk} = a_{kj}$, and $\theta_{jk} = -\theta_{kj}$ for all $j,k = 1, \dots, n$.
    \end{enumerate}
\end{proposition}
\begin{proof}
    $(1) \Rightarrow (2)$. If $[\sigma, \Omega] = 0$, the charge matrix $Q = q\bigoplus_{j=1}^n I_2$
    satisfies both assumptions in Eq.(\ref{eq:nmd_comm}) of Theorem~\ref{thm:nmd_q}, which then implies that there exists a symplectic matrix $S$ such that $SMS^{\rm T} = \bigoplus_{j=1}^n \nu_j I_2$ and $SS^{\rm T} = \bigoplus_{j=1}^n I_2$.
    The RHS in the former equation is exactly the covariance matrix of the thermal states in the RHS of Eq.(\ref{eq:thermal_decomp}) with $\beta_j \coloneqq \log\left(\frac{\nu_j+1}{\nu_j-1}\right)$, while the latter equation ensures that $\hat{S}$ is passive.
    
    $(2) \Rightarrow (3)$. This follows immediately from the form of all U(1)--invariant matrices given in Eq.(\ref{eq:u1M}).
    
    $(3) \Rightarrow (1)$. The form of $\sigma_{jk}$ directly yields $(1)$.
\end{proof}
\noindent Proposition~\ref{prop:passive} is the main result of Ref.~\cite{john2025complete}, where it was derived with no reference to symmetry considerations.

\section{Covariant Gaussian dilation}
\label{sec:dilation}

Any Gaussian channel can be realized via a Gaussian unitary that acts on the system and on an environment prepared in a Gaussian pure state~\cite{caruso2008multi}, a result which extends the Stinespring dilation theorem~\cite{stinespring1955positive}. 
In this section, we generalize this result to show that any covariant Gaussian channel can be realized via a symmetry-respecting Gaussian unitary that acts on the system and on an environment prepared in a Gaussian pure state that also respects the symmetry. 
To prove this generalization, we develop an approach for constructing such dilations, which is  different from both the proof of the aforementioned result for Gaussian channels in the absence of symmetries~\cite{caruso2008multi} and the proof of the covariant Stinespring dilation theorem for finite-dimensional systems, which relies on the Kraus decomposition~\cite{marvian2014symmetry}.
Naturally, the symplectic representations under consideration are assumed to admit invariant states, which, by Proposition~\ref{prop:squ}, implies that they are squeezed passive representations. 
We start by a quick review of Gaussian purifications which play a key role in our main proof of the dilation theorem (see Fig.~\ref{fig:dilation} for a proof sketch).

\subsection{Invariant Gaussian purification}

Consider an $n$--mode system $A$ and a symmetry group $G$ with symplectic Gaussian representation $\hat{S}_A(g): g \in G$ in Hilbert space $\H_A$.
Given any mixed Gaussian state $\hat{\rho} \in \B(\H_A)$ that is invariant under representation $\hat{S}_A$, does there exist a purification $\ket{\psi} \in \H_A \otimes \H_{\overbar{A}}$ of $\hat{\rho}$ that is Gaussian, and invariant under tensor representation $\hat{S}_A \otimes \hat{S}_{\overbar{A}}$, for an ancillary system $\overbar{A}$ and an appropriately chosen representation $\hat{S}_{\overbar{A}}$?
The answer is positive as we show in Appendix~\ref{app:purification}, and, analogously to the case of discrete-variable systems~\cite{nielsen2010quantum}, we prove in Appendix~\ref{app:purif_equiv} that it is unique up to a symmetry-respecting Gaussian unitary on the ancilla.
We sketch the construction of such purification in Fig.~\ref{fig:purification}.
\begin{figure}[ht]
    \centering{\includegraphics[width=0.49\textwidth]{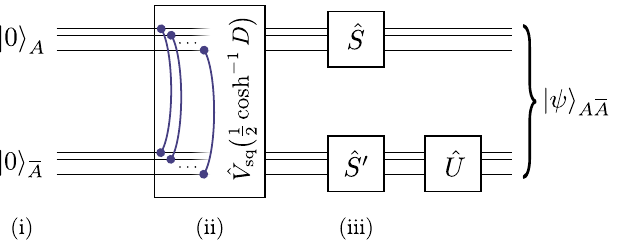}}
    \caption{\textbf{Circuit for purification}.
    We consider an invariant Gaussian state $\hat{\rho}_A$ with covariance matrix $\sigma = VDV^{\rm T}$, for symplectic $V$ and diagonal $D \ge I$.
    Registers $A$ and $\overbar{A}$ both comprise $n$ modes.
    The circuit operates in three steps: (i) preparation of the vacuum state on both registers, (ii) coupling of the two registers via $n$ 2--mode squeezers, each depending on a symplectic eigenvalue in $D$ and depicted by an arched line inside the box, and (iii) symplectic rotation on each of the registers $A$ and $\overbar{A}$ ($\hat{V}$ and $\hat{V}'$ are associated with symplectic matrices $V$ and $V\Omega$, respectively).
    Overall this circuit realizes the purification $\ket{\psi}$ with covariance matrix $\sigma_{\psi} = (V \oplus V\Omega) V_{\rm 2sq}(\tfrac{1}{2}\cosh^{-1}D)V_{\rm 2sq}(\tfrac{1}{2}\cosh^{-1}D)^{\rm T} (V \oplus V\Omega)^{\rm T}$, as given in Eq.(\ref{eq:purif_cov}).
    Lemma~\ref{lem:inv_gauss_purification} proves that this is an invariant Gaussian purification of $\hat{\rho}$ and guarantees that it is unique up to an invariant Gaussian unitary $\hat{U}$ applied on the ancillary register $\overbar{A}$ according to Eq.(\ref{eq:psi12U}).
    We note that steps (i), (ii) and (iii) are generally not individually invariant, but Proposition~\ref{prop:pure_interconversion} gives a $G$--invariant Gaussian unitary that prepares $\sigma_{\psi}$ from a $G$--invariant state that is uncorrelated across $A$ and $\overbar{A}$ which, however, requires the implementation of the globally correlated symplectic transformation $\sqrt{\sigma_{\psi}}$.}
    \label{fig:purification}
\end{figure}

\begin{lemma}[Invariant Gaussian purification]\label{lem:inv_gauss_purification}
    Consider a system with a finite number of bosonic modes and Hilbert space $\H_A$ equipped with symplectic Gaussian representation $\hat{S}_A(g): g \in G$ of a group $G$.
    Let $\hat{\rho}_A \in \B(\H_A)$ be an invariant state, i.e. $\hat{S}_A(g)\hat{\rho}_A\hat{S}_A(g)^\dag = \hat{\rho}_A$ for all $g \in G$.
    Then, the following statements hold true.
    \begin{enumerate}
        \item[(i)] There exists a Gaussian pure state $\ket{\psi}_{A\overbar{A}} \in \H_A \otimes \H_{\overbar{A}}$ for ancilla $\overbar{A}$ with Hilbert space $\H_{\overbar{A}} \cong \H_A$ equipped with representation $\hat{S}_{\overbar{A}}(g) \coloneqq \hat{S}^*_{A}(g): g\in G$, where the complex conjugate is defined in the common eigenbasis of position operators, such that $\ket{\psi}_{A\overbar{A}}$ is a purification of $\hat{\rho}_A$,
        \begin{equation}
            \hat{\rho}_A = \tr_{\overbar{A}}[\ketbra{\psi}{\psi}_{A\overbar{A}}] \,,
        \end{equation}
        and is $G$--invariant, i.e. for all $g \in G$,
        \begin{equation}
            (\hat{S}_A(g) \otimes \hat{S}_{\overbar{A}}(g))\ket{\psi}_{A\overbar{A}} = \ket{\psi}_{A\overbar{A}} \,.
        \end{equation}
        \item[(ii)] Assuming $\bmd$ and $\sigma$ are the displacement vector and covariance matrix of state $\hat{\rho}_A$, respectively, then, $\ket{\psi}$ can be chosen to have displacement vector $\bmd \oplus \bmo$, and covariance matrix given in basis $\hat{\bmr}_A \oplus \hat{\bmr}_{\overbar{A}}$ by
        \begin{equation}\label{eq:purif_cov}
            \hspace{15pt}
            \begin{pmatrix} \sigma &\hspace{-10pt} -\sqrt{-\sigma\Omega\sigma\Omega - I}\ \Omega Z_A \\[5pt] Z_A \Omega \sqrt{-\Omega\sigma\Omega\sigma - I} &\hspace{-10pt} Z_A \sigma Z_A \end{pmatrix} \,,
        \end{equation}
        where $\bmo$ is the zero vector on $\overbar{A}$, $Z_A \coloneqq (I_A \otimes Z)$, and $\Omega$ is the symplectic form on $A$.
        \item[(iii)] Suppose that $\hat{\rho}_A$ admits two Gaussian purifications $\ket{\psi_1}_{AB}$ and $\ket{\psi_2}_{AB}$ for some ancilla $B$ with Hilbert space $\H_B$ equipped with symplectic Gaussian representation $\hat{S}_B(g): g \in G$.
        If $\ket{\psi_1}_{AB}$ and $\ket{\psi_2}_{AB}$ are invariant, i.e. $(\hat{S}_A(g) \otimes \hat{S}_B(g))\ket{\psi_k}_{AB} = \ket{\psi_k}_{AB}$ for $k = 1,2$ and all $g \in G$, then there exists a Gaussian unitary $\hat{U}_B$ that commutes with representation $\hat{S}_B(g): g \in G$ such that
        \begin{equation}\label{eq:psi12U}
            \ket{\psi_2}_{AB} = \left(\hat{I}_A \otimes \hat{U}_B\right) \ket{\psi_1}_{AB} \,.
        \end{equation}
    \end{enumerate}
\end{lemma}
\noindent We note that the purification given in Eq.(\ref{eq:purif_cov}) is a slight modification of the Gaussian purification\footnote{The original Gaussian purification~\cite{holevo2001evaluating} considers covariance matrix $Z_A \sigma_{\psi} Z_A$, which results in a symplectic form $\Omega_A \oplus -\Omega_A$ that is reflected on the second subsystem. Other Gaussian purifications have been considered, e.g. in Ref.~\cite{weedbrook2012gaussian}, but they generally do not respect the symmetries of the original state.} previously found by Holevo and Werner~\cite{holevo2001evaluating}.
We also emphasize that invariance is considered at the state vector level, i.e. $(\hat{S}_A(g) \otimes \hat{S}_B(g))\ket{\psi}_{AB} = \ket{\psi}_{AB}$, for all $g\in G$.
This is a stronger condition than the density operator being $G$--invariant, requiring that the state vector itself must carry the trivial representation of the group, as opposed to a general non-trivial one-dimensional representation.

\subsection{Covariant Gaussian dilation}

Equipped with invariant Gaussian purifications, we are ready to establish covariant Gaussian dilations of covariant Gaussian channels. 
 Recall that Gaussian channels are completely Gaussianity-preserving CPTP operations (see Appendix~\ref{app:channel_def}), i.e. CPTP maps which send Gaussian states to Gaussian states, even when acting only on a subsystem of a larger Gaussian system.

\begin{restatable}[Covariant Gaussian dilation]{theorem}{thmdilation}\label{thm:cov_gauss_dilation}
    Consider a pair of systems $A$ and $B$ with finite numbers of bosonic modes and Hilbert spaces $\H_A$ and $\H_B$. 
    Suppose each system is equipped with a symplectic Gaussian representation of a group $G$ that admits an invariant state (as characterized in Proposition~\ref{prop:squ}), denoted by $\hat{S}_A(g): g\in G$, and $\hat{S}_B(g): g\in G$, respectively.
    
    For any Gaussian channel $\E: \B(\H_A) \rightarrow \B(\H_B)$ that is covariant, i.e. $\E(\hat{S}_A(g) (\cdot) \hat{S}_A(g)^\dag) = \hat{S}_B(g) \E(\cdot) \hat{S}_B(g)^\dag$ for all $g \in G$, there exists a dilation of the form
    \begin{equation}\label{eq:stine_dilation}
        \E(\cdot) = \tr_{A\overbar{B}}\left[ \hat{V}_{AB\overbar{B}} \big((\cdot) \otimes \ketbra{\eta}{\eta}_{B\overbar{B}}\big) \hat{V}_{AB\overbar{B}}^\dag \right] \,,
    \end{equation}
    where ancilla $\overbar{B}$ has Hilbert space $\H_{\overbar{B}} \cong \H_B$ equipped with representation $\hat{S}_{\overbar{B}}(g) \coloneqq \hat{S}^*_{B}(g): g\in G$, where the complex conjugate is defined in the common eigenbasis of position operators, $\ket{\eta}_{B\overbar{B}}$ is an invariant pure Gaussian state, i.e. $(\hat{S}_B(g) \otimes \hat{S}_{\overbar{B}}(g))\ket{\eta}_{B\overbar{B}} = \ket{\eta}_{B\overbar{B}}$ for all $g \in G$, and $\hat{V}_{AB\overbar{B}}$ is a symplectic Gaussian unitary that commutes with representation $\hat{S}_A(g) \otimes \hat{S}_B(g) \otimes \hat{S}_{\overbar{B}}(g): g\in G$. 
    
    Furthermore, $\ket{\eta}_{B\overbar{B}}$ can be chosen to be uncorrelated between $B$ and $\overbar{B}$.
\end{restatable}
\noindent 
Fig.~\ref{fig:dilation} provides a sketch of the proof based on the notion of invariant purification and decoupling. The complete proof is deferred to Appendix~\ref{app:dilation_proof}, where we also present a phase space version of the statement and the converse of the theorem. Note that when the symmetry representations are trivial, we obtain the dilation of Gaussian channels, which was previously shown~\cite{caruso2008multi,serafini2017quantum} using a different approach.

\begin{figure}[t]
    \centering{\includegraphics[width=0.37\textwidth]{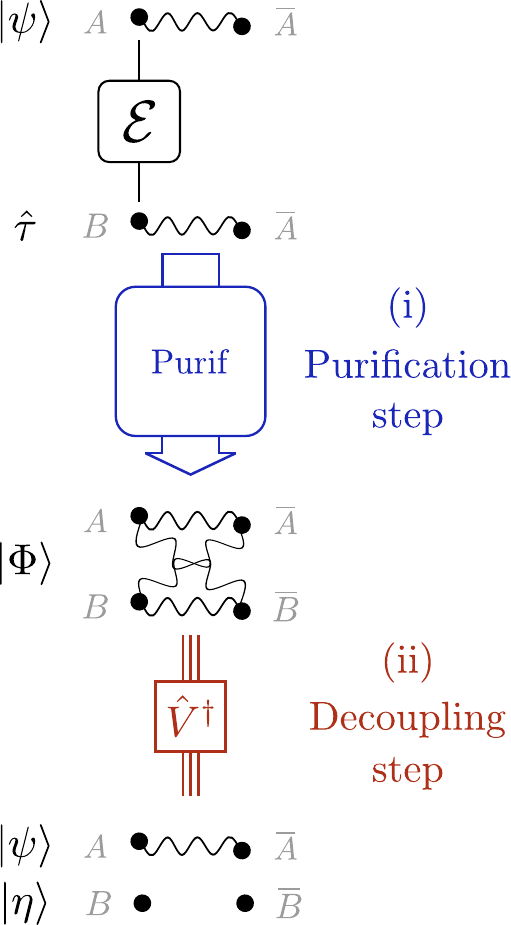}}
    \caption{
    \textbf{Gaussian Stinespring dilation via purification and decoupling}.
    Our proof of Theorem~\ref{thm:cov_gauss_dilation} provides a constructive method for implementing a covariant Gaussian channel $\mathcal{E}$ using $G$--invariant Gaussian unitaries that couple the system to an environment in a $G$--invariant Gaussian pure state.
    First, the channel is applied to the subsystem $A$ of a composite system $A\overbar{A}$, initially prepared in a $G$--invariant Gaussian pure state $\ket{\psi}_{A\overbar{A}}$, resulting in state $\hat{\tau}_{B\overbar{A}}$.
    In step~${\rm (i)}$, we construct a $G$--invariant Gaussian purification of $\hat{\tau}_{B\overbar{A}}$, denoted $\ket{\Phi}_{A\overbar{A}B\overbar{B}}$, whose existence is guaranteed by Lemma~\ref{lem:inv_gauss_purification}.  
    In step~${\rm (ii)}$, we decouple $\ket{\Phi}_{A\overbar{A}B\overbar{B}}$ via a $G$--invariant Gaussian unitary $\hat{V}^\dag_{AB\overbar{B}}$, obtaining $\ket{\psi}_{A\overbar{A}} \otimes \ket{\eta}_{B\overbar{B}}$, which is an equivalent purification of $\hat{\rho}_{\overbar{A}} \coloneqq \tr_A[\ketbra{\psi}_{A\overbar{A}}]$. 
    Lemma~\ref{lem:inv_gauss_purification} ensures that such a unitary exists.  Therefore, if $\hat{\rho}_{\overbar{A}}$ is full rank, we obtain the channel equality in Eq.~(\ref{eq:stine_dilation}). 
    Note that $\overbar{X}$ denotes a system whose Hilbert space is isomorphic to that of system $X$ but carries the complex conjugate representation of symmetry.
    This approach to dilation can also be applied in the setting of discrete variables.
    }
    \label{fig:dilation}
\end{figure}

Theorem~\ref{thm:cov_gauss_dilation} demonstrates that, for any symplectic Gaussian representation, a (covariant) Gaussian channel mapping $n_A$ to $n_B$ modes is guaranteed to have a (covariant) Gaussian dilation with ancillary qubits prepared in (pure) vacuum using at most $2n_B$ ancillary modes.
Therefore, similar to the dilation of general Gaussian channels, a covariant Gaussian channel from $n_A$ to $n_B$ modes can be realized by discarding an environment with at most $n_A + n_B$ modes.

The fact that the ancilla state $\ket{\eta}_{B\overbar{B}}$ can be chosen unentangled across the bipartition $B:\overbar{B}$ follows from Proposition~\ref{prop:squ}, which guarantees that the symplectic Gaussian representation $\hat{S}_B$ admits a $G$--invariant pure state, together with Proposition~\ref{prop:pure_interconversion} which states that $G$--invariant Gaussian transformations act transitively on the space of $G$--invariant Gaussian pure states.

It is worth noting that the proof of Theorem~\ref{thm:cov_gauss_dilation} provides a method to construct the dilating unitary $\hat{V}$.
One has freedom to choose an initial pure state $\ket{\psi}_{A\overbar{A}}$ provided that its reduced state on $A$ is full-rank and invariant under input representation $\hat{S}_A$.
Then, the purification step (i) outlined in Fig.~\ref{fig:dilation} constructs $G$--invariant Gaussian state $\ket{\Phi}_{A\overbar{A}B\overbar{B}}$, e.g. using the purification in Eq.(\ref{eq:purif_cov}).
The decoupling step (ii) consists of identifying a $G$--invariant symplectic Gaussian unitary $\hat{V}_{AB\overbar{B}}$ that maps between the two purifications of marginal state $\tr_A[\ketbra{\psi}_{A\overbar{A}}]$.
This is done by first identifying an intertwiner that maps between the modes in $AB\overbar{B}$ that are correlated with modes in $\overbar{A}$, and then completing it into a full $G$--invariant symplectic transformation (see Appendix~\ref{app:completion}).

As an example, consider a Gaussian channel $\E$ that acts on a system $A$ with $n$ modes equipped with $\Un(1)$ representation $\hat{O}(\theta) = \bigotimes_{j=1}^n \exp(i\theta q_j \hat{a}^\dag_j\hat{a}_j)$ with integer charges $q_j$.
If $\E$ is covariant with respect to the $\Un(1)$ symmetry, i.e.
\begin{equation}\label{eq:U1covdil}
    \E\left(\hat{O}(\theta)(\cdot)\hat{O}(\theta)^\dag\right) = \hat{O}(\theta)\E(\cdot)\hat{O}(\theta)^\dag \,,
\end{equation}
then Theorem~\ref{thm:cov_gauss_dilation} implies that it can be realized by coupling $A$ to $n$--mode ancillary systems $B$ with charges $\{q_j: j=1, \dots, n\}$, and $\overbar{B}$ with charges $\{-q_j: j=1, \dots, n\}$, all initially prepared in vacuum $\ket{0}^{\otimes 2n}$, using U(1)--invariant Gaussian unitaries.

We obtain covariant Gaussian dilations for all 1--mode phase-insensitive, or, equivalently, energy-conserving, channels via the constructive proof of Theorem~\ref{thm:cov_gauss_dilation} in Appendix~\ref{app:dilation_1mode}. 
In particular, we directly find dilating unitaries for all attenuator and amplifier channels, which are slight modifications of the standard dilating unitaries~\cite{weedbrook2012gaussian}, guaranteeing $\Un(1)$--invariance.

\section{Summary \& Outlook}
\label{sec:summary}

We develop a framework that characterizes symmetry-respecting Gaussian dynamics in bosonic systems.
Within this framework, we demonstrate that asymmetry splits into different types and derive families of monotones that are tractable for Gaussian states, quantifying asymmetry in the state's displacement vector and covariance matrix.
We also derive a family of conservation laws for Gaussian closed dynamics, some of which are violated when the dynamics is non-Gaussian.
Finally, we develop a new approach to channel dilation which relies only on an appropriate equivalence of purifications.
In the presence of symmetries, our dilation provides a constructive method of realizing covariant Gaussian channels via invariant Gaussian unitaries and pure states.

It is worth exploring the restrictions that phase space representations of symmetry impose on the characteristic functions of quantum states, which may lead to practical asymmetry witnesses.
In particular, the characteristic function gives rise to the $s$--ordered representations~\cite{cahill1969density}, such as the Wigner, Glauber, and Hussimi representations, which are reconstructible in a lab~\cite{banaszek1996direct,Leibfried1996experimental}.

Our statement of channel dilation can be quantified beyond the context of Gaussian covariant channels, by asking the question of whether channels with limited non-Gaussianity or asymmetry can be dilated via unitaries and ancillary states that themselves contain limited non-Gaussianity or asymmetry as measured by a chosen monotone.
For instance, one could employ a monotone such as the Wigner negativity~\cite{kenfack2004negativity} of non-Gaussian states, and ask what quantum channels may be implemented via Gaussian unitaries and ancillary states of fixed maximum non-Gaussianity as quantified by the chosen monotone.

Fermions are the counterpart of bosons, following canonical anti-commutation relations~\cite{tong2016lectures}.
Besides their natural significance, fermionic simulation is a promising path towards quantum advantage~\cite{ortiz2001quantum}, in particular due to the Jordan-Wigner transform~\cite{jordan1928paulische} which allows fermionic simulation on qubit devices.
Much like for bosonic systems, fermionic Gaussian systems are the most easily engineered, and so an analog of our analysis of symmetric Gaussian dynamics can be useful in the context of fermions as well.
It may also be of interest to investigate a resource theory of asymmetry on hybrid setups of continuous and discrete variables.

In general, specific quantum information tasks, such as property testing for Gaussian states~\cite{girardi2025gaussian}, motivate the study of specific symmetry representations.
We hope that our framework paves the way to exploring CV quantum information tasks that hinge on Gaussian representations of symmetry.

\section{Acknowledgments}

We thank David Jakab for numerous helpful discussions.
We also acknowledge support from NSF Phy-2046195, and NSF QLCI grant OMA-2120757.

\newpage
\onecolumngrid
\appendix
\newpage
\renewcommand{\tocname}{Appendix: Contents}
\tableofcontents

\resumetoc

\section{Notation and background on Gaussian CV}
\label{app:background}

\subsection{Notation}
\label{app:notation}

The complex conjugate is denoted by star ($^*$).
The zero matrix is denoted by $0$.
The inequality $A \ge B$ between square matrices $A,B$ of the same dimension means that $A-B$ has non-negative eigenvalues.
The space of bounded linear operators for Hilbert space $\H$ is written as $\B(\H)$.
The qubit Pauli operators are denoted by $I, X, Y, Z$.
We denote an operator $\hat{a}$ acting on the Hilbert space by a hat, and an operator $a$ acting in phase space without the hat.
We use the shorthand $\hat{a} \pm \text{H.c.} = \hat{a} \pm \hat{a}^\dag$ for operator $\hat{a}$.
Vectors of operators, $\bma$ or $\hat{\bma}$, are denoted by boldface.

We follow Serafini's notation~\cite{serafini2017quantum} for calculations involving vectors of operators.
We write the inner product as $\hat{\bma}^{\rm T}\hat{\bmb} = \sum_{j} \hat{a}_j \hat{b}_j$, and the outer product as $\hat{\bma}\hat{\bmb}^{\rm T}$, with components $\left[\hat{\bma}\hat{\bmb}^{\rm T}\right]_{jk} = \hat{a}_j\hat{b}_k$.
Note that in general $\left(\hat{\bma}\hat{\bma}^{\rm T}\right)^{\rm T} \neq \hat{\bma}\hat{\bma}^{\rm T}$ since operators $\hat{a}_j$ and $\hat{a}_k$ may not commute for $j \neq k$.
The commutators and anti-commutators of a list of operators can thus be succinctly expressed as 
\begin{equation}
    \left[\hat{\bma},\hat{\bma}^{\rm T}\right] = \hat{\bma}\hat{\bma}^{\rm T} - \left(\hat{\bma}\hat{\bma}^{\rm T}\right)^{\rm T} \,, \quad\quad\quad \left\{\hat{\bma},\hat{\bma}^{\rm T}\right\} = \hat{\bma}\hat{\bma}^{\rm T} + \left(\hat{\bma}\hat{\bma}^{\rm T}\right)^{\rm T} \,.
\end{equation}

\subsection{Elements of the Gaussian formalism}
\label{app:elements}

\stoptoc

\subsubsection{Bases for CV systems}

The choice of basis does not affect operators in the Hilbert space, but it determines matrix representations in the phase space.
We often make use of two common bases for CV systems.
The real basis is represented by the vector of position and momentum (quadrature) operators,
\begin{equation}\label{eq:basis_canonical}
    \hat{\bmr} \coloneqq (\hat{x}_1, \hat{x}_n, \dots, \hat{p}_1, \hat{p}_n)^{\rm T} \,.
\end{equation}
The complex basis is represented by the vector of creation and annihilation operators,
\begin{equation}\label{eq:basis_complex}
    \hat{\bma} \coloneqq (\hat{a}_1, \dots, \hat{a}_n, \hat{a}_1^\dag, \dots, \hat{a}_n^\dag)^{\rm T} \,.
\end{equation}
We use symbol $\overset{\bmr}{=}$ or $\overset{\bma}{=}$ in expressions when the RHS is a matrix representation of the LHS given in the real basis (\ref{eq:basis_canonical}) or the complex basis (\ref{eq:basis_complex}), respectively.

Both basis operator vectors satisfy the canonical commutation relation (CCR),
\begin{equation}\label{eq:ccr}
    \left[ \hat{\bmr}, \hat{\bmr}^{\rm T} \right] = \left[ \hat{\bma}, \hat{\bma}^{\rm T} \right] = i\Omega_n \,,\quad\text{where}\quad \Omega_n \overset{\bmr}{=} \bigoplus_{j=1}^n \begin{pmatrix} 0 & 1 \\ -1 & 0 \end{pmatrix} \overset{\bma}{=} -i \begin{pmatrix} 0 & I \\ -I & 0 \end{pmatrix} \,,
\end{equation}
is the $(2n \times 2n)$ symplectic form~\hspace{-3pt}\footnote{It is common in the literature (e.g.~\cite{serafini2017quantum,adesso2014continuous}) to consider the commutation relation $\left[ \hat{\bma}, \hat{\bma}^\dag \right]$ instead, leading to the mirrored symplectic form 
$\Omega_n \overset{\bma}{=} -i \begin{pmatrix} I & 0 \\ 0 & -I \end{pmatrix}$ in the complex basis.} and $I$ here denotes the $(n \times n)$ identity matrix.\\

Phase space objects that appear in Gaussian CV theory transform differently under a change of basis as demonstrated in Table~\ref{tab:transf_rules} according to their tensor type.
$(1,1)$--tensors, such as a symplectic matrix $V$, transform by similarity, while $(0,2)$--tensors, such as a Hamiltonian matrix $H$, and $(2,0)$--tensors, such as the covariance matrix $\sigma$ of a quantum state, transform by congruence.
All $(2,0)$--tensors and $(0,2)$--tensors can contract into $(1,1)$--tensors.
For example, $\Omega H$ and $\sigma\Omega^{-1}$ are $(1,1)$--tensors.
As a consequence, $\Omega H \Omega$ is a $(2,0)$--tensor and transforms identically to $\sigma$.

Any $(1,1)$--tensor $A$ expressed in terms of $(n \times n)$ real blocks as $A \overset{\bmxi}{=} \begin{pmatrix} A_1 & A_2 \\ A_3 & A_4 \end{pmatrix}$ in the permuted (``$xp$--ordered'') real basis $\hat{\bmxi} = (\hat{x}_1, \dots, \hat{x}_n, \hat{p}_1, \dots, \hat{p}_n)$, takes the following block form in the complex basis,
\begin{equation}\label{eq:real2complex}
    A \overset{\bma}{=} \begin{pmatrix} B_1 & B_2 \\ B_2^* & B_1^* \end{pmatrix} \,,\quad\text{where}\quad B_1 = \frac{A_1+A_4}{2} - i\frac{A_2-A_3}{2} \,,\quad
    B_2 = \frac{A_1-A_4}{2} + i\frac{A_2+A_3}{2} \,.
\end{equation}
A $(2,0)$--tensor or $(0,2)$--tensor, such as $\Omega_n$, requires a subsequent ``left-right mirroring'', achieved by right-multiplying with $X \otimes I_n$.

\subsubsection{Second-order Hamiltonians}

Recall that a Gaussian CV system of $n$ bosonic modes is described by a second order Hamiltonian $\hat{H}$ as given in Eq.(\ref{eq:ham}), where additive constants are ignored as they amount to unobservable global phases in the dynamics.
Using Eq.(\ref{eq:real2complex}), we express the Hamiltonian matrix $H$ in the complex basis in terms of two $(n \times n)$ blocks\footnote{In the main text, we use $H^{\text{p}}$ instead of $F$ and $H^{\text{np}}$ instead of $G$.} $F$ and $G$ as
\begin{equation}\label{eq:Hmatrix}
    H \overset{\bma}{=} \begin{pmatrix} G & F \\ F^* & G^* \end{pmatrix} \,,
\end{equation}
where we require that $G = G^{\rm T}$, and $F = F^\dag$ because $H$ is symmetric in the real basis.

For purely quadratic $\hat{H}_A, \hat{H}_B, \hat{H}_C$, the relation $[i\hat{H}_B, i\hat{H}_C] \overset{\bmr}{=} i\frac{1}{2}\hat{\bmr}^{\rm T} \Big( \Omega [-\Omega H_B, -\Omega H_C] \Big) \hat{\bmr}$ leads to the equivalence
\begin{equation}\label{eq:Hcomm}
    i\hat{H}_A = [i\hat{H}_B, i\hat{H}_C] \quad\Longleftrightarrow\quad -\Omega H_A = [-\Omega H_B, -\Omega H_C] \,.
\end{equation}
This allows us to prove Proposition~\ref{prop:Qprop} on passive purely quadratic observables.
\propQ*
\begin{proof}
    $(1) \Leftrightarrow (3)$. $[\widehat{Q}, \hat{H}_{\rm free}] = 0$ implies $[e^{i \widehat{Q} s}, \hat{H}_{\rm free}] = 0$ for all $s\in\mathbb{R}$.
    Conversely, by differentiating $[e^{i \widehat{Q} s}, \hat{H}_{\rm free}] = 0$ with respect to $s$ and setting $s=0$, we get $[\widehat{Q}, \hat{H}_{\rm free}] = 0$.

    $(3) \Leftrightarrow (4)$. This follows directly from Eq.(\ref{eq:Hcomm}) with $H_A = 0$, $H_B = H_{\rm free}$ and $H_C = Q$.
    
    $(4) \Leftrightarrow (2)$. Matrix $e^{-\Omega Q s}$ is the symplectic matrix associated with operator $e^{i \widehat{Q} s}$, for all $s\in\mathbb{R}$.
    Note that $\left(e^{-\Omega Q s}\right)^{\rm T} = e^{Q\Omega s}$.
    Then, $e^{-\Omega Q s}\left(e^{-\Omega Q s}\right)^{\rm T} = I$ for all $s\in\mathbb{R}$ if and only if $[Q,\Omega] = 0$.
    
    $(3) \Rightarrow (5)$ and
    $(5) \Rightarrow (2)$. The vacuum state is a non-degenerate eigenstate of $\hat{H}_{\rm free}$.
    Hence, if $\widehat{Q}$ commutes with $\hat{H}_{\rm free}$, it is also an eigenstate of $\widehat{Q}$.
    Conversely, if the vacuum state is an eigenstate of $\widehat{Q}$, then $e^{-\Omega Q s}\left(e^{-\Omega Q s}\right)^{\rm T} = I$ in order to preserve its covariance matrix.
\end{proof}

\subsubsection{Gaussian unitaries}

A purely quadratic Hamiltonian $\hat{H}$ generates a symplectic Gaussian unitary $\hat{S} = e^{i\hat{H}} = e^{\tfrac{i}{2}\hat{\bmxi}^{\raisebox{-0.3ex}{\small $\scriptstyle T$}} H \hat{\bmxi}}$, which acts on the basis $\hat{\bmxi}$ as $\hat{S}^\dag\hat{\bmxi}\hat{S} = S\hat{\bmxi}$, where the symplectic matrix $S$ associated with operator $\hat{S}$ is
\begin{equation}\label{eq:H2S}
    S = e^{-\Omega H} \overset{\bma}{=} \begin{pmatrix} U & V \\ V^* & U^* \end{pmatrix} \,,
\end{equation}
for $(n \times n)$ blocks $U,V$ satisfying $UU^\dag - VV^\dag = I$ and $UV^{\rm T} - VU^{\rm T} = 0$, in order to ensure that the transformed basis $S\hat{\bmxi}$ satisfies the CCR in Eq.(\ref{eq:ccr}).
Passive symplectic Gaussian unitaries do not mix creation and annihilation operators, so their associated symplectic matrices simplify to
$S \overset{\bma}{=} U \oplus U^*$, where $U$ is now a $(n \times n)$ unitary.

The real symplectic group $\Sp(2n,\mathbb{R})$ is connected, but not compact, so not all elements can be expressed as a single exponential $S = e^{-\Omega H}$.
Instead, every $S \in \Sp(2n,\mathbb{R})$ admits a singular value decomposition~\cite{arvind1995real,braunstein2005squeezing} in terms of compact and non-compact symplectic transformations,
\begin{equation}\label{eq:ozo}
    S \overset{\bmr}{=} O' \left(\bigoplus_{j=1}^n \begin{pmatrix} e^{r_j} & 0 \\ 0 & e^{-r_j} \end{pmatrix}\right) O^{\rm T} \overset{\bma}{=} O' \begin{pmatrix} \cosh{r} & \sinh{r} \\ \sinh{r} & \cosh{r} \end{pmatrix} O^{\rm T} \,,
\end{equation}
where $O, O' \in \Sp(2n,\mathbb{R}) \cap \Or(2n)$ and $r = \diag(r_1, \dots, r_n) \in \mathbb{R}^n$.
Phase-shifters and beam-splitters generate the maximal compact subgroup $\Sp(2n,\mathbb{R}) \cap \Or(2n) \cong \Un(n)$~\cite{bloch1962canonical}.
Together with the 1--mode squeezer, they generate the entire symplectic group $\Sp(2n,\mathbb{R})$~\cite{knill2001scheme}.

Finally, recall that every Gaussian unitary can be brought into the standard form $\hat{U} = \hat{D}_{\bmxi} \hat{S}$ for some Weyl displacement operator $\hat{D}_{\bmxi}$, up to an unobservable global phase.

\subsubsection{Gaussian states}

The set of Gaussian states is the collection of thermal and ground states of positive semi-definite second-order Hamiltonians.
They are fully characterized by the first two statistical moments, as given in Eq.(\ref{eq:moments}).
The covariance matrix takes the form 
\begin{equation}\label{eq:cov}
    \sigma \overset{\bma}{=} \begin{pmatrix} G & F \\ F^* & G^* \end{pmatrix} \,,
\end{equation}
for $(n \times n)$ blocks $G = G^{\rm T}$, and $F = F^\dag$ given by
\begin{equation}\label{eq:sigmaFG}
    F_{jk} = \langle \hat{a}_j\hat{a}^\dag_k + \hat{a}^\dag_k\hat{a}_j \rangle_{\hat{\rho}} - 2\langle\hat{a}_j\rangle_{\hat{\rho}} \langle\hat{a}^\dag_k\rangle_{\hat{\rho}} \,,\quad\quad
    G_{jk} = \langle \hat{a}_j\hat{a}_k + \hat{a}_k\hat{a}_j \rangle_{\hat{\rho}} - 2\langle\hat{a}_j\rangle_{\hat{\rho}} \langle\hat{a}_k\rangle_{\hat{\rho}} \,,
\end{equation}
where $\langle (\cdot) \rangle_{\hat{\rho}} \coloneqq \tr[(\cdot)\hat{\rho}]$.
The covariance matrix $\sigma$ satisfies the uncertainty relation $\sigma + i\Omega \ge 0$ in the real basis~\hspace{-2pt}\footnote{In our complex basis, the uncertainty relation looks like $(\sigma + i\Omega)(I \otimes X) \ge 0$ for Pauli $X$.}, and its symplectic eigenvalues are the ordinary eigenvalues of $|i\Omega\sigma|$. 
The uncertainty relation is equivalent to the symplectic eigenvalues being greater than 1, i.e. $|i\Omega\sigma|^2 \ge I$, giving the relation
\begin{equation}\label{eq:swsw}
    -\sigma\Omega\sigma\Omega - I \ge 0 \,,
\end{equation}
with equality if and only if $\sigma$ is the covariance matrix of a pure state, i.e. all symplectic eigenvalues are 1.

Three notable 1--mode states that we use throughout our discussion are the thermal state $\hat{\rho}_{\rm th}(\bar{n})$ of the free Hamiltonian in Eq.(\ref{eq:ham_free}) with $\bar{n} \coloneqq (e^\beta-1)^{-1}$, the coherent state $\ket{\alpha} \coloneqq \hat{D}_{\alpha}\ket{0}$ and the squeezed state $\ket{r} \coloneqq \hat{V}_{\rm 1sq}\ket{0}$.
Their decompositions~\cite{glauber1963coherent,stoler1970equivalence} in the eigenbasis of the number operator $\hat{a}^\dag\hat{a}$ (Fock basis) $\{\ket{k}: k \in \mathbb{N}\}$ are
\begin{equation}\label{eq:std_decomps}
    \hat{\rho}_{\rm th}(\bar{n}) = \sum_{k=0}^{\infty} \frac{\bar{n}^k}{(\bar{n}+1)^{k+1}} \ketbra{k}{k} \,,\quad 
    \ket{\alpha} = e^{-\frac{|\alpha|^2}{2}}\sum_{k=0}^\infty \frac{\alpha^k}{\sqrt{k!}}\ket{k} \,,\quad 
    \ket{r} = \frac{1}{\sqrt{\cosh r}} \sum_{k=0}^{\infty} \frac{(-\tanh r)^k}{\sqrt{(2k)!}} \ket{2k} \,.
\end{equation}

\resumetoc

\subsection{Engineering second-order Hamiltonians with negative energies}
\label{app:neg_freq}

We address the question of applicability of $\Un(1)$ representations with negative charges in the context of quantum thermodynamics, where physical Hamiltonians have positive frequencies.

We can engineer negative frequencies via modulation~\cite{caves1985new,tsang2012evading}.
Consider a pair of harmonic oscillators with frequencies $\omega_c \pm \delta > 0$, where $\omega_c > 0$ is called the carrier frequency and $\delta > 0$ is called the detuning from the carrier frequency. 
The total Hamiltonian is
\begin{equation}
    \hat{H}(t) = (\omega_{c}+\delta) \hat{a}_1^\dag\hat{a}_1 +   (\omega_{c}-\delta)  \hat{a}_2^\dag\hat{a}_2 + f(t)i\left( e^{i2\omega_{c}t} \hat{a}_1^\dag\hat{a}_2^\dag - e^{-i2\omega_{c}t} \hat{a}_1\hat{a}_2 \right) \,,
\end{equation}
where $f(t)$ is a time-dependent real function.
In the rotating frame defined by the free Hamiltonian $\omega_c (\hat{a}_1^\dag\hat{a}_1 + \hat{a}_2^\dag\hat{a}_2)$, $\hat{H}(t)$ is transformed to
\begin{equation}
    \hat{H}_0 + f(t)\hat{H}_{\text{int}} \coloneqq \delta \hat{a}_1^\dag\hat{a}_1 - \delta \hat{a}_2^\dag\hat{a}_2 + f(t) i\left( \hat{a}_1^\dag\hat{a}_2^\dag - \hat{a}_1\hat{a}_2 \right) \,,
\end{equation}
which corresponds to two harmonic oscillators with frequencies $\pm \delta$.
Since the interaction Hamiltonian $f(t)\hat{H}_{\text{int}}$ commutes with the intrinsic Hamiltonian $\hat{H}_0$, it generates an energy-conserving unitary. 
More precisely, the overall realized unitary in the rotating frame is
\begin{equation}
    \hat{U}(t)=\exp(-i t \hat{H}_0)\exp(-i \hat{H}_{\text{int}}\int_0^t ds f(s)) \,.
\end{equation}

\section{Symplectic symmetry representations}
\label{app:gauss_repr}

The central objects in the theory of asymmetry are operators that remain invariant under the action of a symmetry representation $\hat{U}$ and the operations that are covariant with the action of $\hat{U}$.
\begin{definition}\label{def:invariance_covariance}
    Let $G$ be a symmetry group with unitary representation $\{\hat{U}_A(g): g \in G\}$ in Hilbert space $\H_A$.
    
    We call an operator $\hat{V}: \H_A \rightarrow \H_A$ invariant under representation $\hat{U}_A$ if
    \begin{equation}
        [\hat{V}, \hat{U}_A(g)] = 0 \,,\quad\quad \text{ for all } g \in G \,.
    \end{equation}
    Let $\{\hat{U}_B(g): g \in G\}$ be a unitary representation of $G$ in Hilbert space $\H_B$.
    
    We call a quantum channel $\E: \B(\H_A) \rightarrow \B(\H_B)$ covariant under the input/output representations $\hat{U}_A, \hat{U}_B$ if
    \begin{equation}
        \E\left(\hat{U}_A(g)(\cdot)\hat{U}^\dag_A(g)\right) = \hat{U}_B(g)\E(\cdot)\hat{U}_B(g)^\dag \,,\quad\quad \text{ for all } g \in G \,.
    \end{equation}
\end{definition}
\noindent We may write $G$--invariant/$G$--covariant if the group representations are clear from context, or, simply invariant/covariant if the group is also clear from context.

\subsection{Squeezed passive representations in phase space}
\label{app:passive}

A general Gaussian transformation with non-zero displacement $\hat{U} = \hat{D}_{\bmxi}\hat{S}$ transforms the basis $\hat{\bmxi}$ as
\begin{equation}\label{eq:gauss_basis_transformation_xi}
    \hat{U}^\dag\hat{\bmxi}\hat{U} = S\hat{\bmxi} - \bmxi \,.
\end{equation}
Therefore, every Gaussian representation $g \mapsto \hat{U}(g) = \hat{D}_{\bmxi(g)}\hat{S}(g)$ in the Hilbert space corresponds to a semi-direct product in the phase space,
\begin{equation}\label{eq:gauss_repr_ps}
    S(g_2g_1) = S(g_2)S(g_1) \quad\quad\text{and}\quad\quad \bmxi(g_2g_1) = S(g_2)\bmxi(g_1) + \bmxi(g_2) \,, \quad\quad\text{for all } g_1, g_2 \in G \,.
\end{equation}
When interested in symplectic representations, i.e. representations with $\bmxi(g) = 0$ for all $g \in G$.

Recall that every symplectic representation that admits an invariant state is a squeezed passive representation $g \mapsto S(g) \overset{\bma}{=} R(U(g) \oplus U(g)^*)R^{-1}$ due to Proposition~\ref{prop:squ} which is equivalent to a passive representation up to a symplectic matrix $R$.
We thus focus on passive representations.

Every passive representation $g \mapsto \hat{S}(g)$ in the Hilbert space induces a unitary representation in the phase space via
\begin{equation}\label{eq:unitary_repr_ps}
    g \mapsto S(g) \overset{\bma}{=} U(g) \oplus U(g)^* \,,
\end{equation}
where $U(g)$ is unitary, for all $g \in G$.
In order to make the unitarity of the passive representation apparent, we work in the complex basis using the block form of Eq.(\ref{eq:unitary_repr_ps}).
To unify notation for this discussion, consider the $(1,1)$--tensor
\begin{equation}\label{eq:block_repr}
    V \overset{\bma}{=} \begin{pmatrix} F & G \\ G^* & F^* \end{pmatrix} \,,
\end{equation}
which may, for instance, represent a symplectic matrix $V$, a Hamiltonian matrix $V = \Omega H$, or a covariance matrix $V = \sigma \Omega^{-1}$.
Hence, blocks $F,G$ may need to satisfy additional constraints imposed by the represented tensor, e.g. $V\Omega V^{\rm T} = \Omega$, $H = H^{\rm T}$ or $\sigma = \sigma^{\rm T}$.

Using the transformation rule for $(1,1)$--tensors listed in Table~\ref{tab:transf_rules}, we find that block $F$ transforms by similarity, $F \rightarrow U(g)FU(g)^\dag$, while $G$ transforms by congruency, $G \rightarrow U(g)GU(g)^{\rm T}$,
\begin{equation}\label{eq:H_passive_transformation}
    \begin{pmatrix} F & G \\ G^* & F^* \end{pmatrix} \rightarrow \begin{pmatrix} U(g) & 0 \\ 0 & U(g)^* \end{pmatrix} \begin{pmatrix} F & G \\ G^* & F^* \end{pmatrix} \begin{pmatrix} U(g) & 0 \\ 0 & U(g)^* \end{pmatrix}^\dag = \begin{pmatrix} U(g)FU(g)^\dag & U(g)GU(g)^{\rm T} \\ \left( U(g)GU(g)^{\rm T} \right)^* & \left( U(g)FU(g)^\dag \right)^* \end{pmatrix} \,,
\end{equation}
for all $g \in G$.
Therefore, a $(1,1)$--tensor $V$ is invariant if and only if, for all $g \in G$,
\begin{equation}\label{eq:conditions_invariance}
\begin{split}
    U(g)FU(g)^\dag &= F \,,\\
    U(g)GU(g)^{\rm T} &= G \,.
\end{split}
\end{equation}
Importantly, had we chosen a $(2,0)$--tensor $\begin{pmatrix} G & F \\ F^* & G^* \end{pmatrix}$ instead, such as $\Omega H \Omega$ or $\sigma$, we would have arrived at the same invariance conditions in Eq.(\ref{eq:conditions_invariance}) for $F$ and $G$, because a $(2,0)$--tensor transforms by congruence rather than similarity.

\stoptoc

\subsection*{Invariance of transposed (1,1)--tensors and inverse (2,0)--tensors}

It is clear that if a $(1,1)$--tensor $V$ is invariant under a squeezed passive representation, then $V^{-1}$ is also invariant.
However, it is less clear whether $V^{\rm T}$ is invariant.
If $V$ corresponds to symplectic Gaussian unitary $\hat{V}$, then $V^{\rm T} = \Omega^{\rm T} V^{-1} \Omega$, with the transposes taken in the real basis, corresponds to symplectic Gaussian unitary
\begin{equation}
    \hat{V}' = \exp(-i\frac{\pi}{2} \hat{H}_{\rm free}) \hat{V}^{-1} \exp(i\frac{\pi}{2} \hat{H}_{\rm free})
\end{equation}
In general, if $\hat{V}$ respects a passive representation of symmetry $g \mapsto O(g)$, $\hat{V}'$ also respects the symmetry, i.e. $[V, O(g)] = 0$ for all $g \in G$ implies $[V^{\rm T}, O(g)] = 0$ for all $g \in G$.
However, this is not generally true for squeezed passive representations. 
In fact, $[V, S(g)] = 0$ for all $g \in G$ implies that $[V^{\rm T}, S'(g)] = 0$ for all $g \in G$, where $g \mapsto S'(g) = \Omega^{\rm T} S(g) \Omega$ is an equivalent, but, in general, different representation of symmetry.
We provide a condition that guarantees invariance of $V^{\rm T}$ under squeezed passive representations.
\begin{proposition}\label{prop:VT}
    Consider a group $G$ with squeezed passive representation $g \mapsto S(g) = R O(g) R^{-1}$, where $R$ is a symplectic matrix and $g \mapsto O(g)$ is a passive representation.
    Denote by $C \coloneqq {\rm comm}\{O(g): g \in G\}$ the commutant of the passive representation and suppose that $R^{\rm T}R$ is in the normalizer of $C$, i.e. $(R^{\rm T}R)C(R^{\rm T}R)^{-1} = C$.
    Then, a $(1,1)$--tensor $V$ is invariant only if $V^{\rm T}$ is also invariant under $S$.
\end{proposition}
\begin{proof}
    $V$ is invariant under $S$, i.e. $[V, R O(g) R^{-1}] = 0$, or, equivalently,
    \begin{equation}
        O(g) (R^{-1} V R) = (R^{-1} V R) O(g) \,,
    \end{equation}
    for all $g \in G$, which implies that $R^{-1}VR \in C$.
    Assuming that $R^{\rm T}R$ normalizes the commutant $C$, we also have
    \begin{equation}
        R^{\rm T}VR^{\rm -T} = (R^{\rm T}R) (R^{-1}VR) (R^{\rm T}R)^{-1} \in C \,.
    \end{equation}
    Therefore, $O(g) (R^{\rm T} V R^{-\rm T}) = (R^{\rm T} V R^{-\rm T}) O(g)$ for all $g \in G$, which implies, by taking the transpose, that
    \begin{equation}\label{eq:R-1VTR}
        O(g)^{\rm T} (R^{-1} V^{\rm T} R) = (R^{-1} V^{\rm T} R) O(g)^{\rm T} \,,
    \end{equation}
    for all $g \in G$.
    Since $g \mapsto O(g)$ is an orthogonal representation, we have $O(g)^{\rm T} = O(g^{-1})$ for all $g \in G$, so, by replacing $g^{-1}$ with $g$ in Eq.(\ref{eq:R-1VTR}), we obtain
    \begin{equation}
        O(g) (R^{-1} V^{\rm T} R) = (R^{-1} V^{\rm T} R) O(g) \,,
    \end{equation}
    for all $g \in G$, which is equivalent to
    $[V^{\rm T}, R O(g) R^{-1}] = 0$ for all $g \in G$, i.e. $V^{\rm T}$ is invariant under $S$.
\end{proof}
\noindent The statement of Proposition~\ref{prop:VT} holds also for a $(2,0)$--tensor $V$ and its inverse $V^{-1}$ instead of its transpose.
The assumption of Proposition~\ref{prop:VT} is trivially satisfied when $R^{\rm T}R \in C$, for example when $R = I$.
On the contrary, consider $R = V_{\rm 1sq}(r)$ for some $r \neq 0$ and a passive representation $g \mapsto O(\theta)$ of $\Un(1)$ symmetry on 1 mode (see Fig.~\ref{fig:repr}), so that $\theta \mapsto RO(\theta)R^{-1}$ admits $RR^{\rm T}$ as an invariant covariance matrix, corresponding to some squeezing parameter $r$.
Due to Eq.(\ref{eq:u1M}), the relevant commutant is $C = \{ m\exp(\theta\Omega): m \in \mathbb{R} \,, \theta \in [0,2\pi)\}$, so $R^{\rm T}R = V_{\rm 1sq}(2r)$ does not normalize $C$.
As a consequence, the squeezed state with opposite squeezing $-r$ is not invariant because its covariance matrix is the inverse $(RR^{\rm T})^{-1}$.

\subsection*{Complex conjugate representation}

\resumetoc

The complex conjugate representation $g \mapsto \hat{S}(g)$ has associated phase space representation $g \mapsto S(g)$ given by
\begin{equation}\label{eq:conj_repr_app}
    \widetilde{S}(g) \overset{\bmr}{=} (I_n \otimes Z)S(g)(I_n \otimes Z) \overset{\bma}{=} (X \otimes I_n)S(g)(X \otimes I_n) \overset{\bma}{=} U(g)^* \oplus U(g) \,,
\end{equation}
where $I_n \otimes Z$ maps $\hat{x}_j \mapsto \hat{x}_j$ and $\hat{p}_j \mapsto -\hat{p}_j$ for all $j$ in the real basis, and $X,Z$ are the qubit Pauli operators.

\subsection{Invariance in phase space}
\label{app:ps_invariance}

The transformation rules in Table~\ref{tab:transf_rules} determine the invariance conditions under symplectic representations.

\begin{proposition}\label{prop:symmetry_phasespace}
    Let $G$ be a symmetry group with symplectic Gaussian representation $g \mapsto \hat{S}(g)$ in Hilbert space $\H$, and associated representation $g \mapsto S(g)$ in phase space.
    Then,
    \begin{enumerate}
        \item[(i)] a second-order Hamiltonian $\hat{H} = \frac{1}{2} \hat{D}_{\bmxi}^\dag \hat{\bmxi}^{\raisebox{-0.5ex}{$\scriptstyle T$}} H \hat{\bmxi} \hat{D}_{\bmxi}$ is invariant if and only if
        \begin{equation}\label{eq:invariance_conditions_H}
            S(g) \bmxi = \bmxi \,,\quad\text{ and }\quad S(g) \Omega H \Omega S(g)^{\rm T} = \Omega H \Omega \,,\quad\quad \text{ for all } g \in G \,,
        \end{equation}
        \item[(ii)] a Gaussian state $\hat{\rho} \in \B(\H)$ is invariant if and only if its displacement vector $\bmd$ and covariance matrix $\sigma$ satisfy
        \begin{equation}\label{eq:invariance_conditions_rho}
            S(g) \bmd = \bmd \,,\quad\text{ and }\quad S(g) \sigma S(g)^{\rm T} = \sigma \,,\quad\quad \text{ for all } g \in G \,,
        \end{equation}
        \item[(iii)] a Gaussian unitary $\hat{U} = \hat{D}_{\bmxi}\hat{V}: \B(\H) \rightarrow \B(\H)$ is invariant if and only if
        \begin{equation}\label{eq:invariance_conditions_U}
            S(g) \bmxi = \bmxi \,,\quad\text{ and }\quad S(g) V S(g)^{-1} = V \,,\quad\quad \text{ for all } g \in G \,.
        \end{equation}
    \end{enumerate}
    Let $G$ have squeezed passive representations $g \mapsto \hat{S}_A(g)$ and $g \mapsto \hat{S}_B(g)$ in input and output Hilbert spaces $\H_A, \H_B$, with associated representations $g \mapsto S_A(g)$ and $g \mapsto S_B(g)$ in phase space.
    Then,
    \begin{enumerate}
        \item[(iv)] a Gaussian channel $\E: \B(\H_A) \rightarrow \B(\H_B)$ with phase space action as given in Eq.(\ref{eq:xychannel}) is covariant if and only if
        \begin{equation}\label{eq:covariance_conditions}
            S_B(g) \bmxi = \bmxi \,,\quad XS_A(g) = S_B(g)X \,,\quad\text{ and }\quad S_B(g)YS_B(g)^{\rm T} = Y \,,\quad\quad \text{ for all } g \in G \,.
        \end{equation}
    \end{enumerate}
\end{proposition}
\begin{proof}
    First, note that a displacement operator $\hat{D}_{\bmxi}$ is invariant under a representation $g \mapsto S(g)$ if and only if vector $\bmxi$ is invariant, because $\hat{S}(g) \hat{D}_{\bmxi} \hat{S}(g)^\dag = \hat{D}_{S(g)\bmxi}$, for all $g \in G$.
    
    Statement \emph{(i)} follows directly from $\hat{S}(g)^\dag \hat{H} \hat{S}(g) = \frac{1}{2} \hat{D}_{S(g)\bmxi}^\dag \hat{\bmxi}^{\raisebox{-0.5ex}{$\scriptstyle T$}} S(g)^{\rm -T}HS(g) \hat{\bmxi} \hat{D}_{S(g)\bmxi}$, for all $g \in G$ and the transformation rule for $H$ listed in Table~\ref{tab:transf_rules}.

    Statement \emph{(ii)} follows directly from the fact that Gaussian state $\hat{\rho}$ is uniquely specified by its displacement vector $\bmd$ and covariance matrix $\sigma$ and the transformation rules for $\bmd, \sigma$ listed in Table~\ref{tab:transf_rules}.

    The direct part of statement \emph{(iii)}, i.e. the part where we assume that $\hat{U}$ is invariant, follows immediately from the transformation rule for $V$ listed in Table~\ref{tab:transf_rules}.
    The converse needs care, since, in general, a symplectic matrix $V$ corresponds to a symplectic Gaussian unitary $\hat{V} = e^{i\hat{H}_1}e^{i\hat{H}_2}$ that cannot be expressed as a single exponential due to the non-compactness of the real symplectic group.
    Consider the singular value decomposition in Eq.(\ref{eq:ozo}), which can be re-expressed as $V = O\sqrt{V^{\rm T}V}$ for an orthogonal symplectic $O$ and positive-definite symplectic $\sqrt{V^{\rm T}V}$.
    First, note that we can write $\sqrt{V^{\rm T}V}$ in terms of an exponential, i.e. $\sqrt{V^{\rm T}V} = e^{-\Omega H}$, because, by Eq.(\ref{eq:ozo}), $V^{\rm T}V \overset{\bmr}{=} O\left(\bigoplus_{j=1}^n \diag\left(e^{2r_j}, e^{-2r_j}\right)\right)O^{\rm T}$, and taking the square root gives
    \begin{equation}
        \sqrt{V^{\rm T}V} = \exp(O\left(\bigoplus_{j=1}^n \diag\left(r_j,-r_j\right)\right)O^{\rm T}) = e^{-\Omega H} \,,
    \end{equation}
    for $H \coloneqq O\left(\bigoplus_{j=1}^n \Omega\diag\left(r_j,-r_j\right)\right)O^{\rm T} = H^{\rm T}$ where we used that $O$ is orthogonal symplectic so $\Omega_n O = O\Omega_n$.
    
    Next, consider the case when the representation $S$ is passive.
    Then, according to Proposition~\ref{prop:VT}, $V^{\rm T}$ is invariant, and as a consequence $\sqrt{V^{\rm T}V}$ and $O$ are also invariant.
    Therefore, we can write $O = e^{-\Omega H_1}$ and $\sqrt{V^{\rm T}V} = e^{-\Omega H_2}$, for some invariant $H_1, H_2$.
    Then, $\hat{H}_1, \hat{H}_2$ are both invariant under $g \mapsto \hat{S}(g)$, and so is $\hat{V} = e^{i\hat{H}_1}e^{i\hat{H}_2}$.
    
    If the representation $g \mapsto S(g) = RO(g)R^{-1}$ is squeezed passive for a symplectic Gaussian unitary $\hat{R}$ and passive representation $g \mapsto O(g)$, then $R^{-1}VR$ is invariant under $g \mapsto O(g)$, so there exist $\hat{R}^\dag\hat{H}_1\hat{R}$ and $\hat{R}^\dag\hat{H}_2\hat{R}$ that are invariant under $g \mapsto \hat{O}(g)$ such that $\hat{R}^\dag\hat{V}\hat{R} = \exp(i\hat{R}^\dag\hat{H}_1\hat{R})\exp(i\hat{R}^\dag\hat{H}_2\hat{R})$ is also invariant.
    Therefore, $\hat{H}_1, \hat{H}_2$ and $\hat{V} = e^{i\hat{H}_1}e^{i\hat{H}_2}$ are invariant under $g \mapsto \hat{S}(g)$.
    
    Statement \emph{(iv)} in the case of equal input/output dimensions follows directly from the fact that Gaussian channel $\hat{\E}$ is uniquely specified by $\bmxi, X, Y$ and their transformation rules listed in Table~\ref{tab:transf_rules}.
    The assumption of squeezed passive representation guarantees, due to Proposition~\ref{prop:squ}, that the representations admit invariant states, which is relevant in the direct part, i.e. when the channel is assumed covariant.
    
    To prove statement \emph{(iv)} in the general case where input and output dimensions may not be equal, let $\hat{\rho}$ be a Gaussian quantum state with displacement vector $\bmd$ and covariance matrix $\sigma$.
    Consider the action of group element $g \in G$ followed by the channel,
    \begin{equation}
    \begin{split}
        &\bmd \xrightarrow{\hat{S}_A(g)} S_A(g)\bmd \xrightarrow{\E} XS_A(g)\bmd + \bmxi \eqqcolon \bmd_1' \,, \text{ and } \\
        &\sigma \xrightarrow{\hat{S}_A(g)} S_A(g)\sigma S_A(g)^{\rm T}  \xrightarrow{\E} XS_A(g)\sigma S_A(g)^{\rm T}X^{\rm T} + Y \eqqcolon \sigma_1' \,.
    \end{split}
    \end{equation}
    Now consider the action of the channel followed by group element $g \in G$,
    \begin{equation}
    \begin{split}
        &\bmd \xrightarrow{\E} X\bmd + \bmxi \xrightarrow{\hat{S}_B(g)} S_B(g)X\bmd + S_B(g)\bmxi \eqqcolon \bmd_2' \,, \text{ and } \\
        &\sigma \xrightarrow{\E} X\sigma X^{\rm T} + Y \xrightarrow{\hat{S}_B(g)} S_B(g)X\sigma X^{\rm T}S_B(g)^{\rm T} + S_B(g)YS_B(g)^{\rm T} \eqqcolon \sigma_2' \,.
    \end{split}
    \end{equation}
    Assume that $\E$ is covariant according to Definition~\ref{def:invariance_covariance}. 
    Then, we must have $\bmd_1' = \bmd_2'$ for arbitrary initial $\bmd$, leading to the conditions $S_B(g) \bmxi = \bmxi$ and $XS_A(g) = S_B(g)X$.
    We must also have $\sigma_1' = \sigma_2'$ for arbitrary initial $\sigma$, which leads to the condition $S_B(g)YS_B(g)^{\rm T} = Y$.
    
    Conversely, assume that the conditions in Eq.(\ref{eq:covariance_conditions}) hold. 
    Then, $\bmd_1' = \bmd_2'$ and $\sigma_1' = \sigma_2'$ for an arbitrary initial quantum state with displacement vector $\bmd$ and covariance matrix $\sigma$, so the Gaussian quantum channel $\E$ is covariant according to Definition~\ref{def:invariance_covariance}.
\end{proof}

It is clear that if either Hamiltonian $\hat{H}$ or state $\hat{\rho} = e^{i\beta\hat{H}} / \tr[e^{i\beta\hat{H}}]$ at any temperature $\beta^{-1} > 0$ is invariant, then the other one is as well.
Moreover, if any one of the $(1,1)$--tensors, $\Omega H$ or $\sigma\Omega^{-1} = i\coth\left(\frac{1}{2}i\Omega H\right)$ is invariant, then the other one is as well.
As a corollary of Proposition~\ref{prop:symmetry_phasespace}, proving one of the above invariances either in phase space or in Hilbert space, guarantees the rest.

\subsection{Irrep decomposition of squeezed passive representations}
\label{app:irreps}

Let $G$ be a symmetry group with representation $\hat{U}(g)$ for all $g \in G$ in Hilbert space $\H \cong L^2(\mathbb{R}^n)$ of a bosonic system.
Under the action of a representation of $G$, the Hilbert space admits a decomposition of the form
\begin{equation}\label{eq:Hilbert_space_decomp}
    L^2(\mathbb{R}^n) \cong \bigoplus_{\lambda} \mathbb{C}^{d_\lambda} \otimes \M_\lambda \,,
\end{equation}
where there may be arbitrarily many inequivalent irreps $\lambda$, while $\mathbb{C}^{d_\lambda}$ and $\M_\lambda$ can be finite or infinite dimensional subspaces.
To characterize the asymmetry of general systems, we would need to understand the asymmetry in all these, potentially infinite, irreps.

However, the CV systems under consideration are described by a finite number $n$ of modes in phase space.
A passive operator $\hat{U}$ has an associated symplectic matrix $S \overset{\bma}{=} U \oplus U^*$ according to Eq.(\ref{eq:unitary_repr_ps}).
Since passive representations do not mix creation and annihilation operators, the phase space decomposes identically within the block of creation operators and within the block of annihilation operators. 
Therefore, we introduce a decomposition of symmetry in the linear span of annihilation operators, which is isomorphic to $\mathbb{C}^n$,
\begin{equation}\label{eq:phase_space_decomp}
    \mathbb{C}^n \cong \bigoplus_{\lambda \in \Lambda} \mathbb{C}^{d_\lambda} \otimes \M_\lambda \cong \bigoplus_{\lambda \in \Lambda} \mathbb{C}^{d_\lambda} \otimes \mathbb{C}^{m(\lambda)} \,,
\end{equation}
where $\Lambda$ is the set of inequivalent irreps, $\mathbb{C}^{d_\lambda}$ is the $d_\lambda$--dimensional invariant subspace of irrep $\lambda$, and $\M_\lambda$ is the subspace associated with all $m(\lambda)$ multiplicities of $\lambda$.
According to this decomposition, by a proper change of basis, $U(g)$ decomposes in irreps as
\begin{equation}
    U(g) = \bigoplus_{\lambda \in \Lambda} U_{\lambda}(g) \otimes I_{m(\lambda)} \,,\quad \text{for all } g \in G \,.
\end{equation}

Consider a $(1,0)$--tensor $\bmd \overset{\bma}{=} \bmalpha \oplus \bmalpha^*$ in phase space.
Invariance under representation $S \overset{\bma}{=} U \oplus U^*$ means that $\bmalpha = U(g)\bmalpha$ for all $g \in G$.
This occurs if and only if $\bmalpha$ contains non-zero components within the trivial subrepresentation $\lambda_{\rm trivial}$, which is a 1-dimensional irrep.
Therefore, $\bmalpha$ is invariant if and only if
\begin{equation}
    \bmalpha = \bmo \oplus \bmalpha_{\lambda_{\rm trivial}} \,,
\end{equation}
where $\bmalpha_{\lambda_{\rm trivial}}$ is an arbitrary complex vector of $m(\lambda_{\rm trivial})$ components in the multiplicity space of $\lambda_{\rm trivial}$, and $\bmo$ is the zero vector of $n-m(\lambda_{\rm trivial})$ in the compliment space $\Lambda\setminus\{\lambda_{\rm trivial}\}$.

In order to fully characterize invariance of larger phase space tensors under passive representations, we need to introduce the notion of a dual irrep.
For any irrep $\lambda \in \Lambda$, its dual $\lambda^* \in \Lambda^* \coloneqq \{\mu^*: \mu \in \Lambda\}$ is defined such that irrep $U_{\lambda^*}$ is equivalent to $U^*_{\lambda}$, i.e. there exists a matrix $B_\lambda$ such that
\begin{equation}\label{eq:Blambda}
    U_{\lambda^*}(g)^* = B_{\lambda}^{-1}U_{\lambda}(g)B_{\lambda} \,,
\end{equation}
for all $g \in G$.
Dual irreps have equal dimensions, $d_{\lambda^*} = d_\lambda$, so there exists a bijection between the bases of their corresponding invariant subspaces, but the corresponding multiplicity spaces are unconstrained.
Since $U_{\lambda}(g), U_{\lambda^*}(g)$ are unitary, $B_{\lambda}$ can in general be chosen to be a multiple of a unitary matrix~\cite{fulton2013representation}.

We can provide a classification of irreps in terms of their duals:
\begin{enumerate}
    \item[(1)] Irrep $\lambda \in \Lambda$ is called \emph{self-dual} if $\lambda^* \cong \lambda$, i.e. there exists an invertible matrix $B_{\lambda}$ such that $U^*_{\lambda}(g) = B_{\lambda}^{-1}U_{\lambda}(g)B_{\lambda}$ for all $g \in G$.
    One can check using Schur's lemma that, for self-dual $\lambda$, $B_{\lambda}$ is either symmetric or anti-symmetric, and, additionally, if $B_{\lambda}$ is symmetric, then $U_{\lambda}$ is equivalent to an irrep with real matrix elements~\cite{fulton2013representation, bouchard2020ma}.
    
    If $B_{\lambda}$ is symmetric, the self-dual irrep $\lambda$ is called \emph{real}.
    We denote the set of real irreps by $\Lambda_{\rm real}$.
    Any symmetry that is reducible over the reals has real irreps, i.e. $U^*_{\lambda}(g) = U_{\lambda}(g)$ for all $g \in G$, with the permutation group $\operatorname{S}_n$ serving as an example. 
    The adjoint representation of a simple compact Lie group is also always a real irrep~\cite{fulton2013representation, hall2013lie}.
    
    If $B_{\lambda}$ is anti-symmetric, the self-dual irrep $\lambda$ is called \emph{pseudo-real}.
    We denote the set of pseudo-real irreps by $\Lambda_{\rm pseu}$.
    A prominent example of a pseudo-real irrep is $\SU(2)$ with total angular momentum $1/2$ in which case $U^*_{\lambda}(g) = YU_{\lambda}(g)Y$ for all $g \in G$, where Pauli matrix $Y$ is anti-symmetric.
    
    \item[(2)] Irrep $\lambda \in \Lambda$ is called \emph{complex} if $\lambda^* \not\cong \lambda$.
    It is then relevant whether its inequivalent dual irrep $\lambda^*$ belongs to the set of inequivalent irreps $\Lambda$.
    We denote the irreps with no dual in $\Lambda$, and the irreps with an inequivalent dual pair, respectively, as
    \begin{equation}
        \Lambda_{\rm nd} \coloneqq \Lambda \setminus \Lambda^* \,,\quad\quad \Lambda_{\rm pd} \coloneqq (\Lambda \cap \Lambda^*) \setminus (\Lambda_{\rm real} \sqcup \Lambda_{\rm pseu}) \,.
    \end{equation}
    The $\Un(1)$ symmetry is an important example of a symmetry with complex irreps. 
    Specifically, $\Lambda$ contains inequivalent dual irreps of $\Un(1)$ if the representation of $\Un(1)$ is generated by a charge operator that contains charges of equal magnitude but opposite signs.
    Additionally, the defining representation of $\SU(3)$ is a complex irrep, i.e. is not equivalent to its complex conjugate.
\end{enumerate}
According to this characterization, the set of inequivalent irreps naturally partitions in four disjoint subsets,
\begin{equation}\label{eq:Lambda_partition}
    \Lambda = \Lambda_{\rm nd} \sqcup \Lambda_{\rm real} \sqcup \Lambda_{\rm pseu} \sqcup \Lambda_{\rm pd} \,.
\end{equation}
Let us now consider an invariant $(1,1)$--tensor represented by blocks $F, G$ according to Eq.(\ref{eq:block_repr}).
Our goal is to decompose $F$ and $G$ into the irreps induced by representation $U$.
Recall that $(2,0)$--tensors and $(0,2)$--tensors can be converted into $(1,1)$--tensor by appropriate contractions with the symplectic form.

The invariance condition for component $F$ as given in Eq.(\ref{eq:H_passive_transformation}) is $F = U(g)FU(g)^\dag$.
Denoting by $F_{\lambda\lambda'}$ the component of $F$ that maps irrep $\lambda'$ to irrep $\lambda$, Schur's lemma requires that
\begin{equation}
    F_{\lambda\lambda'} \propto I_{d_\lambda} \delta_{\lambda'\lambda} \,,
\end{equation}
for all $\lambda,\lambda' \in \Lambda$.
This decomposition is independent of the nature of the dual irrep $\lambda^*$.
Accounting for multiplicities, $F$ takes the following explicit form,
\begin{equation}
    F = \bigoplus_{\lambda \in \Lambda} I_{d_\lambda} \otimes F_{\lambda} = \sum_{\substack{\lambda \in \Lambda\\s,\alpha,\alpha'}} F_{\lambda, \alpha, \alpha'} \ketbra{\lambda, s, \alpha}{\lambda, s, \alpha'} \,,
\end{equation}
where index $s = 1, \dots, d_{\lambda}$ labels a basis for irrep $\lambda$ of dimension $d_{\lambda}$, and indices $\alpha, \alpha' = 1, \dots, m(\lambda)$ label a basis in the multiplicity space of irrep $\lambda$ of dimension $m(\lambda)$.
This fully characterizes an invariant $(1,1)$--tensor in the case of $G = 0$, which may correspond to absence of squeezing in an appropriate context.

The invariance condition for component $G$ as given in Eq.(\ref{eq:H_passive_transformation}) is $G = U(g)GU(g)^{\rm T}$.
Denoting by $G_{\lambda\lambda'}$ the component of $G$ that maps irrep $\lambda'$ to irrep $\lambda$, Schur's lemma requires that
\begin{equation}
    G_{\lambda\lambda'} \propto 
    \begin{cases}
        0 \,, &\text{for } \lambda \in \Lambda_{\rm nd} \,,\\
        B_{\lambda}\delta_{\lambda'\lambda} \,, &\text{for } \lambda \in \Lambda_{\rm real} \sqcup \Lambda_{\rm pseu} \,,\\
        B_{\lambda}\delta_{\lambda'\lambda^*} \,, &\text{for } \lambda \in \Lambda_{\rm pd} \,.
    \end{cases}
\end{equation}
Block $B_{\lambda}$ is diagonal, symmetric and unitary for each real irrep $\lambda \in \Lambda_{\rm real}$, and diagonal, anti-symmetric and unitary for each pseudo-real irrep $\lambda \in \Lambda_{\rm pseu}$. 
For each inequivalent dual pair $\lambda,\lambda^* \in \Lambda_{\rm pd}$, unitary block $B_{\lambda} = B_{\lambda^*}^{\rm T}$ is off-diagonal with respect to the irrep bases of $\lambda$ and $\lambda^*$.
Accounting for multiplicities, $G$ takes the following explicit form,
\begin{equation}
    G = \hspace{-12pt}\sum_{\substack{\lambda \in \Lambda_{\rm real} \sqcup \Lambda_{\rm pseu}\\s,s',\alpha,\alpha'}}\hspace{-15pt} B_{\lambda, s, s'}G_{\lambda, \alpha, \alpha'} \ketbra{\lambda, s, \alpha}{\lambda, s', \alpha'} \hspace{5pt}+\hspace{-6pt} \sum_{\substack{\lambda,\lambda^* \in \Lambda_{\rm pd}\\s, s',\alpha,\alpha'}}\hspace{-7pt} \left( B_{\lambda,s,s'}G_{\lambda, \alpha, \alpha'} \ketbra{\lambda, s, \alpha}{\lambda^*, s', \alpha'} + B_{\lambda,s,s'}^{\rm T}G_{\lambda^*, \alpha, \alpha'} \ketbra{\lambda^*, s', \alpha'}{\lambda, s, \alpha} \right) \,,
\end{equation}
where $s, s' = 1, \dots, d_{\lambda}$, $\alpha = 1,\dots,m(\lambda)$, and $\alpha' = 1,\dots,m(\lambda^*)$.
Note that for a self-dual irrep $\lambda$, $m(\lambda^*) = m(\lambda)$.

We summarize the above analysis in the following statement, where we provide a full characterization of the irrep decompositions of blocks $F$ and $G$ under passive representations.
\begin{proposition}\label{prop:irreps}
    Let $G$ be a symmetry group with passive representation $g \mapsto S(g) \overset{\bma}{=} U(g) \oplus U(g)^*$ in phase space. 
    
    There exists a basis, called the irrep basis, such that the unitary $U(g)$ decomposes in irreps as
    \begin{equation}\label{eq:uirreps}
        U(g) = \left( \bigoplus_{\lambda \in \Lambda\setminus\Lambda_{\rm pd}} U_{\lambda}(g) \otimes I_{m(\lambda)} \right) \oplus \left( \bigoplus_{\lambda,\lambda^* \in \Lambda_{\rm pd}} \begin{pmatrix} U_{\lambda} \otimes I_{m(\lambda)} & 0 \\[5pt] 0 & U_{\lambda^*} \otimes I_{m(\lambda^*)} \end{pmatrix} \right) \,,
    \end{equation}
    where every inequivalent irrep $\lambda \in \Lambda$ has dimension $d_\lambda$, and multiplicity $m(\lambda)$.

    Every invariant $(1,0)$--tensor $\bmd \overset{\bma}{=} \bmalpha \oplus \bmalpha^*$ is expressed in the irrep basis as
    \begin{equation}\label{eq:airreps}
        \bmalpha = \bmo \oplus \bmalpha_{\lambda_{\rm trivial}} \,,
    \end{equation}
    where $\bmalpha_{\lambda_{\rm trivial}}$ is an arbitrary vector in the trivial irrep $\lambda_{\rm trivial}$, and $\bmo$ is the zero vector in $\Lambda \setminus \{\lambda_{\rm trivial}\}$.
    
    Every invariant $(1,1)$--tensor $\begin{pmatrix} F & G \\ G^* & F^* \end{pmatrix}$ and $(2,0)$--tensor $\begin{pmatrix} G & F \\ F^* & G^* \end{pmatrix}$ are represented by blocks $F$ and $G$ expressed in the irrep basis as
    \begin{equation}\label{eq:FGinv}
        F = \bigoplus_{\lambda \in \Lambda} I_{d_\lambda} \otimes F_{\lambda} \,,\quad\quad
        G = \left( \bigoplus_{\lambda \in \Lambda_{\rm real} \sqcup \Lambda_{\rm pseu}} B_{\lambda} \otimes G_{\lambda} \right) \oplus \left( \bigoplus_{\lambda,\lambda^* \in \Lambda_{\rm pd}} \begin{pmatrix} 0 & B_{\lambda} \otimes G_\lambda \\[5pt] B_{\lambda}^{\rm T} \otimes G_{\lambda^*} & 0 \end{pmatrix} \right) \,,
    \end{equation}
    where $F_{\lambda}, G_\lambda$ and $G_{\lambda^*}$ are complex matrices of appropriate dimensions and $B_{\lambda}$ are determined by the symmetry conditions $U_{\lambda^*}(g) = B_{\lambda}^{-1}U_{\lambda}(g)B_{\lambda}$ for all $\lambda$ such that $\lambda^* \in \Lambda$.
\end{proposition}
Note that this decomposition can be readily employed when dealing with squeezed passive representations of the form $R (U(g) \oplus U(g)^*) R^{-1}$, where $R$ is a symplectic matrix that may contain squeezing, as equivalent representations have the same irrep decomposition up to a basis change by $R$.

\subsection{Irrep decomposition of invariant Gaussian symplectic Gaussian unitaries and quantum states}
\label{app:state_irreps}

Proposition~\ref{prop:irreps} has direct implications on the structure of invariant symplectic and covariance matrices under a squeezed passive representation $g \mapsto R(U(g) \oplus U(g)^*)R^{-1}$ of a symmetry group $G$, with symplectic $R$, and in particular on the entanglement between modes.
We focus on passive representation $g \mapsto U(g) \oplus U(g)^*$ and recall that the invariant unitaries and states of the squeezed passive representation can be obtained by the invariant states of the passive representation via a basis transformation by $\hat{R}$.

Any symplectic Gaussian unitary $\hat{S}$ cannot correlate modes carrying irreps of different type among the four types that appear in Eq.(\ref{eq:Lambda_partition}).
The symplectic matrix $S \overset{\bma}{=} \begin{pmatrix} F & G \\ G^* & F^* \end{pmatrix}$ associated with $\hat{S}$ is expressed in the irrep basis according to Eq.(\ref{eq:FGinv}), where, in order to satisfy the symplectic matrix constraints $FF^\dag - GG^\dag = I$ and $FG^{\rm T} - GF^{\rm T} = 0$, we require that
\begin{enumerate}
    \item[(i)] $F_{\lambda}$ is a unitary matrix for all $\lambda \in \Lambda_{\rm nd}$\,,
    \item[(ii)] $F_{\lambda}F_{\lambda}^\dag - G_\lambda G_{\lambda}^\dag = I_{m(\lambda)}$ for all $\lambda \in \Lambda\setminus\Lambda_{\rm nd}$\,,
    \item[(iii)] $F_{\lambda}G_{\lambda}^{\rm T} - G_{\lambda}F_{\lambda}^{\rm T} = 0$ for all $\lambda \in \Lambda_{\rm real}$\,,
    \item[(iv)] $F_{\lambda}G_{\lambda}^{\rm T} + G_{\lambda}F_{\lambda}^{\rm T} = 0$ for all $\lambda \in \Lambda_{\rm pseu}$\,,
    \item[(v)] $F_{\lambda} G_{\lambda^*}^{\rm T} - G_{\lambda} F_{\lambda^*}^{\rm T} = 0$ for all $\lambda,\lambda^* \in \Lambda_{\rm pd}$\,.
\end{enumerate}

Analogously, any invariant state $\hat{\rho}$ does not contain correlations between modes carrying irreps of different type among the four types that appear in Eq.(\ref{eq:Lambda_partition}). 
Concretely,
\begin{equation}
    \hat{\rho} = \tr_{\Lambda \setminus \Lambda_{\rm nd}}[\hat{\rho}] \otimes \tr_{\Lambda \setminus \Lambda_{\rm real}}[\hat{\rho}] \otimes \tr_{\Lambda \setminus \Lambda_{\rm pseu}}[\hat{\rho}] \otimes \tr_{\Lambda \setminus \Lambda_{\rm pd}}[\hat{\rho}] \,.
\end{equation}
The symmetry necessitates entanglement between inequivalent dual irreps, and allows for arbitrary entanglement within multiplicity subspaces, but all other irrep subspaces are uncorrelated with each other.

Consider the decomposition of $U(g)$ in the irrep basis given in Eq.(\ref{eq:uirreps}).
The displacement vector $\bmd \overset{\bma}{=} \bmalpha \oplus \bmalpha^*$ of invariant state $\hat{\rho}$ is expressed in the irrep basis via Eq.(\ref{eq:airreps}), meaning that it is non-zero only on modes where the symmetry acts trivially.
The covariance matrix $\sigma \overset{\bma}{=} \begin{pmatrix} G & F \\ F^* & G^* \end{pmatrix}$ of $\hat{\rho}$ is expressed in the irrep basis according to Eq.(\ref{eq:FGinv}), where, in order to satisfy the covariance matrix constraints $F = F^\dag$ and $G = G^{\rm T}$, we require that 
\begin{enumerate}
    \item[(i)] $F_{\lambda}$ is a $(m(\lambda) \times m(\lambda))$ Hermitian matrix for all $\lambda \in \Lambda$\,,
    \item[(ii)] $G_{\lambda}$ is a $(m(\lambda) \times m(\lambda))$ complex symmetric matrix for all $\lambda \in \Lambda_{\rm real}$\,,
    \item[(iii)] $G_{\lambda}$ is a $(m(\lambda) \times m(\lambda))$ complex anti-symmetric matrix for all $\lambda \in \Lambda_{\rm pseu}$\,,
    \item[(iv)] $G_{\lambda} = G_{\lambda^*}^{\rm T}$ is a $(m(\lambda) \times m(\lambda^*))$ complex matrix, for pair $\lambda,\lambda^* \in \Lambda_{\rm pd}$\,.
\end{enumerate}
We can therefore decompose the covariance matrix in the irrep basis as
\begin{equation}
    \sigma = \left( \bigoplus_{\lambda \in \Lambda_{\rm nd}} \sigma_{\lambda} \right) \oplus \left( \bigoplus_{\lambda \in \Lambda_{\rm real} \sqcup \Lambda_{\rm pseu}} \sigma_{\lambda} \right) \oplus \left( \bigoplus_{\lambda,\lambda^* \in \Lambda_{\rm pd}} \sigma_{\lambda} \right) \,,
\end{equation}
where the restriction of $\sigma$ within each disjoint irrep subspace takes different forms.

For $\lambda \in \Lambda_{\rm nd}$, the invariance of $\sigma$ ensures that there is no squeezing, so
\begin{equation}\label{eq:normal_modes_nd}
    \sigma_\lambda \overset{\bma}{=} \begin{pmatrix} 0 & I_{d_\lambda} \otimes F_\lambda \\ I_{d_\lambda} \otimes F_\lambda^* & 0 \end{pmatrix} \,.
\end{equation}

For $\lambda \in \Lambda_{\rm real} \sqcup \Lambda_{\rm pseu}$, squeezing is generally allowed, 
\begin{equation}\label{eq:normal_modes_sd}
    \sigma_\lambda \overset{\bma}{=} \begin{pmatrix} B_\lambda \otimes G_\lambda & I_{d_\lambda} \otimes F_\lambda \\ I_{d_\lambda} \otimes F_\lambda^* & B_\lambda^* \otimes G_\lambda^* \end{pmatrix} \,,
\end{equation}

For inequivalent dual pairs $\lambda, \lambda^* \in \Lambda_{\rm pd}$, entanglement between the irrep subspaces is also generally present,
\begin{equation}\label{eq:normal_modes_pd}
    \sigma_{\lambda} 
    \overset{\bma}{=} \begin{pmatrix} 0 & I_{d_\lambda} \otimes F_\lambda & B_{\lambda} \otimes G_\lambda & 0 \\[2pt] I_{d_\lambda} \otimes F_\lambda^* & 0 & 0 & B_{\lambda}^* \otimes G_\lambda^* \\[2pt] B_{\lambda}^{\rm T} \otimes G_\lambda^{\rm T} & 0 & 0 & I_{d_\lambda} \otimes F_{\lambda_*} \\[2pt] 0 & B_{\lambda}^\dag \otimes G_\lambda^\dag & I_{d_\lambda} \otimes F_{\lambda^*}^* & 0 \end{pmatrix}
    \,.
\end{equation}

\subsection{Examples of passive representation decompositions}
\label{app:examples}

For the sake of concreteness, we provide examples of Gaussian operators that are invariant under specific symmetries of interest.
Using Proposition~\ref{prop:irreps}, we can also immediately derive invariant tensors in phase space.

\vspace{10pt}
{\bf Trivial symmetry.}
In the case of no symmetry, the trivial representation corresponds to $\Lambda = \{\lambda_{\rm trivial}\}$, where $\lambda_{\rm trivial}$ is a 1-dimensional, self-dual irrep with multiplicity equal to the number of total modes of the system, $m(\lambda) = n$.

\vspace{10pt}
{\bf $\SU(2)$ symmetry.}
Recall the Schwinger representation $\hat{S}(U): U \in \SU(2)$ from Eq.(\ref{eq:repSU}).
This is a pseudo-real irrep satisfying $U^* = \Omega^{\rm T}U\Omega$.
An $\SU(2)$--invariant covariance matrix $\sigma$ on 2 modes takes the form
\begin{equation}
    \sigma \overset{\bma}{=} \begin{pmatrix} 0 & F_1 I_2 \\ F_1 I_2 & 0 \end{pmatrix} \,,
\end{equation}
where $I_2$ is the $(2 \times 2)$ identity matrix, and $F_1 > 1$ is a real number. 
Recall that in the case of a pseudo-real irrep, the non-passive block is $B_\lambda \otimes G_\lambda$ for an anti-symmetric matrix $G_\lambda$. 
Here, $B_\lambda = \Omega$ and $G_\lambda$ is a 1-dimensional anti-symmetric matrix, so it must equal 0.
Hence, the only 2--mode $\SU(2)$--invariant Gaussian states are two uncorrelated modes at equal temperature.
Similarly, one finds that the only 2--mode $\SU(2)$--invariant Gaussian unitary is the identity operator.

A less trivial example arises when we consider the group $\SU(2)$ on 6 modes, with representation
\begin{equation}
    U = U(\alpha, \beta,\gamma) \otimes I_3 \,.
\end{equation}
Therefore, the 6--mode $\SU(2)$--invariant covariance matrix $\sigma$ takes the form
\begin{align}
    \sigma \overset{\bma}{=} \begin{pmatrix} \Omega \otimes G_3 & I_2 \otimes F_3 \\ I_2 \otimes F_3^* & \Omega \otimes G_3^* \end{pmatrix} \,,
\end{align}
where $F_3 \ge I_3$ is a $(3 \times 3)$ Hermitian matrix and $G_3$ is a $(3 \times 3)$ anti-symmetric matrix.

\vspace{10pt}
{\bf $\Un(1)$ symmetry.}
Consider $\Un(1)$ symmetry on 3 modes, with representation $\hat{U}(t) = e^{it \hat{H}}$, where $\hat{H} = \omega \hat{a}_1^\dag\hat{a}_1 + \omega \hat{a}_2^\dag\hat{a}_2 - \omega\hat{a}_3^\dag\hat{a}_3$ for some rational $\omega > 0$.
The phase space representation $S(t) \overset{\bma}{=} U(t) \oplus U^*(t)$, with
\begin{equation}
    U(t) = \left( 
    \begin{array}{cc|c}
    & & 0 \\
    \multicolumn{2}{c|}{\smash{\raisebox{.5\normalbaselineskip}{$e^{i\omega t} I_2$}}} & 0 \\
    \hline & & \\[-1.0em]
    0 & 0 & e^{-i\omega t}
    \end{array}
    \right) \,,
\end{equation}
in the irrep basis, contains two distinct 1-dimensional irreps $U_\lambda(t) = e^{i\omega t}$ and $U_{\lambda^*}(t) = e^{-i\omega t}$, which are dual to each other. 
Irrep $\lambda$ has multiplicity 2.

For any invariant matrix, we can directly use our analysis to express components $F$ and $G$ in the irrep basis,
\begin{align}
    F = \left( 
    \begin{array}{cc|c}
    & & 0 \\
    \multicolumn{2}{c|}{\smash{\raisebox{.5\normalbaselineskip}{$F_2$}}} & 0 \\
    \hline & & \\[-1.0em]
    0 & 0 & F_1
    \end{array}
    \right),\quad\quad
    G &= \left( 
    \begin{array}{cc|c}
    0 & 0 & \\
    0 & 0 & \multicolumn{1}{c}{\smash{\raisebox{.5\normalbaselineskip}{$G_1$}}} \\
    \hline & & \\[-1.0em]
    \multicolumn{2}{c|}{G_1^{\rm T}} & 0
    \end{array}
    \right) \,,
\end{align}
where $F_2$ is a $(2 \times 2)$ Hermitian matrix, $F_1$ is a real number, and $G_1$ is a $(2 \times 1)$ complex matrix.

\vspace{10pt}
{\bf $\operatorname{S}_3$ and $\mathbb{Z}_3$ symmetries and composite representations.}
Consider an interaction Hamiltonian $\hat{H}_{AB}^{\rm int}$ between system $A$ with $n_A$ modes and representation $g \mapsto S_A(g) \overset{\bma}{=} U_A(g) \oplus U_A(g)^*$, and system $B$ with $n_B$ modes and representation $g \mapsto S_B(g) \overset{\bma}{=} U_B(g) \oplus U_B(g)^*$.
The transformation $\hat{H}_{AB}^{\rm int} \mapsto (\hat{S}_A \otimes \hat{S}_B)\hat{H}_{AB}^{\rm int}(\hat{S}_A \otimes \hat{S}_B)^\dag$ reads
\begin{equation}
    \hat{H}_{AB}^{\rm int} = \sum_{j=1}^{n_A}\sum_{k=1}^{n_B} \left( F_{jk} \hat{a}^\dag_j\hat{b}_k + G_{jk} \hat{a}_j\hat{b}_k \right) + \text{H.c.} 
    \mapsto \sum_{j,j'=1}^{n_A}\sum_{k,k'=1}^{n_B} \left( F_{jk} U_{A,j'j}U^*_{B,k'k}\hat{a}^\dag_j\hat{b}_k + G_{jk} U_{A,j'j}U_{B,k'k} \hat{a}_j\hat{b}_k \right) + \text{H.c.} \,.
\end{equation}
Therefore, $\hat{H}_{AB}^{\rm int}$ is invariant under the symmetry if $U_A(g) F U_B(g)^\dag = F$ and $U_A(g) G U_B(g)^{\rm T} = G$, for all $g \in G$.
These conditions can be vectorized to
\begin{equation}\label{eq:inv_int_cond}
    (U_A(g) \otimes U_B(g)^* - I_A \otimes I_B)\ket{{\rm vec}(F)} = (U_A(g) \otimes U_B(g) - I_A \otimes I_B)\ket{{\rm vec}(G)} = 0 \,,
\end{equation}
where $I_A,I_B$ denote the identity on systems $A,B$.
If $U_B(g) = U_B(g)^*$ is real, the conditions for $F$ and $G$ are identical.

Consider the $\operatorname{S}_3$ symmetry on the composite system highlighted in Fig.~\ref{fig:intro}, with elements $\{e, r, r^2, t_1, t_2, t_3\}$, where $e, r, r^2$ indicate rotations and $t_1, t_2, t_3$ indicate two-element transpositions.
The standard 2-dimensional irreducible representation of $\operatorname{S}_3$ is
\begin{align}
    O_2(e) =
    \begin{pmatrix}
    1 & 0 \\
    0 & 1
    \end{pmatrix} \,,\quad
    O_2(r) =
    \begin{pmatrix}
    -\frac{1}{2} & -\frac{\sqrt{3}}{2} \\[0.5em]
    \frac{\sqrt{3}}{2} & -\frac{1}{2}
    \end{pmatrix} \,,\quad
    O_2(r^2) =
    \begin{pmatrix}
    -\frac{1}{2} & \frac{\sqrt{3}}{2} \\[0.5em]
    -\frac{\sqrt{3}}{2} & -\frac{1}{2}
    \end{pmatrix} \,,\nonumber\\
    O_2(t_1) =
    \begin{pmatrix}
    -1 & 0 \\
    0 & 1
    \end{pmatrix} \,,\quad
    O_2(t_2) =
    \begin{pmatrix}
    \frac{1}{2} & \frac{\sqrt{3}}{2} \\[0.5em]
    \frac{\sqrt{3}}{2} & -\frac{1}{2}
    \end{pmatrix} \,,\quad
    O_2(t_3) =
    \begin{pmatrix}
    \frac{1}{2} & -\frac{\sqrt{3}}{2} \\[0.5em]
    -\frac{\sqrt{3}}{2} & -\frac{1}{2}
    \end{pmatrix} \,.
\end{align}
whereas the 3-dimensional natural representation of $\operatorname{S}_3$ is
\begin{align}
    O_3(e) =
    \begin{pmatrix}
    1 & 0 & 0 \\
    0 & 1 & 0 \\
    0 & 0 & 1
    \end{pmatrix} \,,\quad
    O_3(r) =
    \begin{pmatrix}
    0 & 0 & 1 \\
    1 & 0 & 0 \\
    0 & 1 & 0
    \end{pmatrix} \,,\quad
    O_3(r^2) =
    \begin{pmatrix}
    0 & 1 & 0 \\
    0 & 0 & 1 \\
    1 & 0 & 0
    \end{pmatrix} \,,\nonumber\\
    O_3(t_1) =
    \begin{pmatrix}
    0 & 1 & 0 \\
    1 & 0 & 0 \\
    0 & 0 & 1
    \end{pmatrix} \,,\quad
    O_3(t_2) =
    \begin{pmatrix}
    1 & 0 & 0 \\
    0 & 0 & 1 \\
    0 & 1 & 0
    \end{pmatrix} \,,\quad
    O_3(t_3) =
    \begin{pmatrix}
    0 & 0 & 1 \\
    0 & 1 & 0 \\
    1 & 0 & 0
    \end{pmatrix} \,.
\end{align}
We will ignore the self-interactions of each subsystem, and focus on determining the interaction parts $F_{\rm int}$ and $G_{\rm int}$, where $[F_{\rm int}]_{jk}$ and $[G_{\rm int}]_{jk}$ are the coefficients of terms $\hat{a}_j^\dag\hat{b}_k$ and $\hat{a}_j\hat{b}_k$, respectively.

First, suppose systems $A$ and $B$ have 3 modes, and representations $U_A = U_B = O_3$.
Then, solving Eq.(\ref{eq:inv_int_cond}) gives
\begin{equation}
    F_{\rm int} = f_1 \begin{pmatrix} 1 & 0 & 0 \\ 0 & 1 & 0 \\ 0 & 0 & 1 \end{pmatrix} + f_2 \begin{pmatrix} 0 & 1 & 1 \\ 1 & 0 & 1 \\ 1 & 1 & 0 \end{pmatrix} \,,\quad\quad G_{\rm int} = g_1 \begin{pmatrix} 1 & 0 & 0 \\ 0 & 1 & 0 \\ 0 & 0 & 1 \end{pmatrix} + g_2 \begin{pmatrix} 0 & 1 & 1 \\ 1 & 0 & 1 \\ 1 & 1 & 0 \end{pmatrix} \,,
\end{equation}
for $f_1, f_2, g_1, g_2 \in \mathbb{C}$.
Therefore, every invariant Hamiltonian interaction is of the form
\begin{equation}
    \hat{H}_{AB}^{\rm int} = \sum_{k=1}^{3} \left( f_1 \hat{a}^\dag_k\hat{b}_k + g_1 \hat{a}_k\hat{b}_k \right) + \sum_{k=1}^{3} \left( f_2 \hat{a}^\dag_k\hat{b}_{k+1} + f_2 \hat{a}^\dag_k\hat{b}_{k+2} + g_2 \hat{a}_k\hat{b}_{k+1} + g_2 \hat{a}_k\hat{b}_{k+2} \right) +  \text{H.c.} \,,
\end{equation}
where subscript addition is considered mod 3.
Suppose system $C$ has 2 modes and representation $U_C = O_2$.
Then,
\begin{equation}
    F_{\rm int} = f \begin{pmatrix} 1 & \frac{1}{\sqrt{3}} \\[0.5em] -1 & \frac{1}{\sqrt{3}} \\[0.5em] 0 & -\frac{2}{\sqrt{3}} \end{pmatrix} \,,\quad\quad G_{\rm int} = g \begin{pmatrix} 1 & \frac{1}{\sqrt{3}} \\[0.5em] -1 & \frac{1}{\sqrt{3}} \\[0.5em] 0 & -\frac{2}{\sqrt{3}} \end{pmatrix} \,,
\end{equation}
for $f, g \in \mathbb{C}$, so every invariant Hamiltonian interaction is of the form
\begin{equation}
    \hat{H}_{AC}^{\rm int} = f\left( \hat{a}^\dag_1\hat{b}_1 + \frac{1}{\sqrt{3}}\hat{a}^\dag_1\hat{b}_2 - \hat{a}^\dag_2\hat{b}_1 + \frac{1}{\sqrt{3}}\hat{a}^\dag_2\hat{b}_2 - \frac{2}{\sqrt{3}}\hat{a}^\dag_3\hat{b}_2 \right) + g\left( \hat{a}_1\hat{b}_1 + \frac{1}{\sqrt{3}}\hat{a}_1\hat{b}_2 - \hat{a}_2\hat{b}_1 + \frac{1}{\sqrt{3}}\hat{a}_2\hat{b}_2 - \frac{2}{\sqrt{3}}\hat{a}_3\hat{b}_2 \right) + \text{H.c.} \,.
\end{equation}

Considering the subgroup $\mathbb{Z}_3 = \{e,r,r^2\}$ on system $AC$ leads to different invariant Hamiltonian interaction, with
\begin{equation}
    F_{\rm int} = f_1 \begin{pmatrix} 1 & 0 \\[0.5em] -\frac{1}{2} & \frac{\sqrt{3}}{2} \\[0.5em] -\frac{1}{2} & -\frac{\sqrt{3}}{2} \end{pmatrix} + f_2 \begin{pmatrix} 0 & 1 \\[0.5em] -\frac{\sqrt{3}}{2} & -\frac{1}{2} \\[0.5em] \frac{\sqrt{3}}{2} & -\frac{1}{2} \end{pmatrix} \,,
\end{equation}
for $f_1, f_2 \in \mathbb{C}$, so that all invariant terms of the form $[F_{\rm int}]_{jk} \hat{a}^\dag_j\hat{b}_k$ are
\begin{equation}
    f_1 \hat{a}^\dag_1\hat{b}_1 + f_2 \hat{a}^\dag_1\hat{b}_2 + \left(-\frac{1}{2}f_1 - \frac{\sqrt{3}}{2}f_2\right) \hat{a}^\dag_2\hat{b}_1 + \left(\frac{\sqrt{3}}{2}f_1 - \frac{1}{2}f_2\right) \hat{a}^\dag_2\hat{b}_2 + \left(-\frac{1}{2}f_1 + \frac{\sqrt{3}}{2}f_2\right) \hat{a}^\dag_3\hat{b}_1 + \left(-\frac{\sqrt{3}}{2}f_1 - \frac{1}{2}f_2\right) \hat{a}^\dag_3\hat{b}_2 + \text{H.c.} \,.
\end{equation}
Here, we are not considering the non-passive terms $G_{\rm int}$.

\subsection{Local decomposition of U(1)--invariant Gaussian unitaries}
\label{app:local}

We prove that U(1)--invariant Gaussian unitaries decompose to 2--local U(1)--invariant Gaussian unitaries.
\local*
\begin{proof}
    Recall the general $\Un(1)$ representation $\hat{O}(\theta) = \bigotimes_{j=1}^n e^{i\theta q_j \hat{a}^\dag_j\hat{a}_j}$ with integer charges $q_j \neq 0$.
    A Hamiltonian $\hat{H} = -i\sum_{j,k} \left( F_{jk} \hat{a}_j^\dag\hat{a}_k + G_{jk} \hat{a}_j\hat{a}_k \right) + \text{H.c.}$ is invariant under $\hat{O}$ if and only if $F_{jk} = e^{i(q_k-q_j)} F_{jk}$ and $G_{jk} = e^{i(q_k+q_j)} G_{jk}$, so we can decompose any invariant Hamiltonian as
    \begin{align}
        i\hat{H} 
        = \sum_{j} F_{jj} \hat{a}_j^\dag\hat{a}_j\ 
        + \sum_{\substack{j,k: j \neq k \\ q_j = q_k}} F_{jk} \left( \hat{a}_j^\dag\hat{a}_k - \hat{a}_j\hat{a}_k^\dag \right)
        + \sum_{q_j = -q_k} G_{jk} \left( \hat{a}_j^\dag\hat{a}_k^\dag - \hat{a}_j\hat{a}_k \right) \,.
    \end{align}
    The Lie algebra of invariant Hamiltonians is therefore spanned and hence generated by phase-shifters, beam-splitters between sectors of equal charge and 2--mode squeezers between sectors of opposite charge.
    According to Proposition~\ref{prop:symmetry_phasespace}, the group of $G$--invariant symplectic transformations is connected.
    As a consequence, the Lie group of U(1)--invariant symplectic matrices is generated by exponentials of multiples of these generators.
\end{proof}

As an example, suppose we have two equal charges, i.e. $Q = q(I \oplus I)$.
Then, any invariant covariance matrix $\sigma$ as given by Eq.(\ref{eq:u1M}) decomposes as
\begin{align*}
    \sigma = (V_{\rm ps}(\phi) \hspace{-1pt}\oplus\hspace{-1pt} I)V_{\rm bs}(\phi')\hspace{-2pt}
    \begin{pmatrix}
        \nu_1 I & \hspace{-4pt}0 \\
        0 & \hspace{-4pt}\nu_2 I
    \end{pmatrix}\hspace{-2pt}
    V_{\rm bs}(\phi')^{\rm T}(V_{\rm ps}(\phi) \hspace{-1pt}\oplus\hspace{-1pt} I)^{\rm T} \,,
\end{align*}
for symplectic eigenvalues $\nu_1, \nu_2 > 1$, and parameters $\phi = \tan^{-1}(\theta_{12})$, and $\phi' = \frac{1}{2}\tan^{-1}\left(\frac{2a_{12}^2}{a_{11} - a_{22}}\right)$.

On the other hand, suppose that we have two opposite charges, i.e. $Q = q(I \oplus -I)$.
Then, $\sigma$ decomposes as
\begin{align*}
    \sigma = (V_{\rm ps}(\phi) \hspace{-0.5pt}\oplus\hspace{-0.5pt} I)V_{\rm 2sq}(r)\hspace{-2pt}
    \begin{pmatrix}
        \nu_1 I & \hspace{-4pt}0 \\
        0 & \hspace{-4pt}\nu_2 I
    \end{pmatrix}
    \hspace{-2pt}V_{\rm 2sq}(r)^{\rm T}\hspace{-1.5pt}(V_{\rm ps}(\phi) \hspace{-0.5pt}\oplus\hspace{-0.5pt} I)^{\rm T} \,,
\end{align*}
for symplectic eigenvalues $\nu_1, \nu_2 > 1$, and parameters $\phi = -\tan^{-1}(\theta_{12})$, and $r = -\frac{1}{2}\tanh^{-1}\left(\frac{2a_{12}^2}{a_{11} + a_{22}}\right)$.

We highlight that in the case of equal charges, a beam-splitter was used for the decomposition, whereas in the case of opposite charges, a 2--mode squeezer was used instead.
It is natural to ask if there is a systematic way to decompose U(1)--invariant symplectic matrices on an arbitrary number of modes.
If all charges $q_j = q$ are equal, then invariant symplectic matrices coincide with passive symplectic matrices, and there exists a known algorithm~\cite{reck1994experimental} that decomposes them into phase-shifters and beam-splitters.
In the context of symmetries, this algorithm hinges on the fact that the symplectic form $\Omega_n$ is proportional to the charge $Q\Omega_n = q\Omega_n$, so the conditions for a matrix to be symplectic and U(1)--invariant coincide.
This reduces the problem into decomposing an $(n \times n)$ unitary that mixes annihilation operators into $(2 \times 2)$ unitaries and phases.
If at least two charges have different signs, then the conditions for a matrix to be symplectic and U(1)--invariant no longer coincide, so the algorithm cannot be easily extended to the general case of $\Un(1)$ symmetry.
We leave the task of finding such an algorithm as an open question.

\subsection{Squeezed passive representations as symplectic representations that admit invariant states}
\label{app:squ_repr}

We prove that all symplectic representations that admit an invariant state are equivalent to a passive representation, and further admit a pure invariant state.
As noted in the main text, the first statement on compact groups is an immediate consequence of the Cartan-Iwasawa-Malcev theorem~\cite{malcev1945theory,iwasawa1949some,onishchik2012lie}. 
Here, we present a self-contained proof.
\squrepr*
\begin{proof}
    We first prove statements (i) and (iii).
    Compact groups always admit an invariant quantum state,
    \begin{equation}
        \sigma = r\int {\rm d}g\ S(g)S(g)^{\rm T} \,.
    \end{equation}
    Indeed, $\sigma$ is symmetric, $G$--invariant, i.e. $S(g) \sigma S(g)^{\rm T} = \sigma$ for all $g \in G$, and satisfies the uncertainty relation $\sigma + i\Omega \ge 0$ for sufficiently high $r>0$.
    
    In both cases of a compact and non-compact group $G$, given a covariance matrix $\sigma$ that remains invariant under the action of symmetry, we can obtain a covariance matrix $\sigma_\infty$ of a pure state by applying the transformation $\lim_{\alpha \rightarrow \infty} \sigma_{\alpha} \coloneqq \coth(\alpha \coth^{-1}(\sigma i\Omega))i\Omega$.
    Here, if $\sigma$ is the covariance matrix of Gaussian state $e^{-\beta \hat{H}}/\tr[e^{-\beta \hat{H}}]$, then $\sigma_{\alpha}$ is the covariance matrix of a Gaussian state with the same Hamiltonian $\hat{H}$ and temperature $(\alpha\beta)^{-1}$.
    Finally, $\sqrt{\sigma_\infty}$ is a symplectic transformation as discussed in the proof of Proposition~\ref{prop:pure_interconversion}. 
    Choosing $R = \sqrt{\sigma_\infty}$ completes the proof.

    We now prove statement (ii).
    The symplectic transformation $Q = \widetilde{R}^{-1}R$ satisfies $\widetilde{O}(g) = Q O(g) Q^{-1}$. 
    Consider the Bloch-Messiah decomposition $Q = O_2 D O_1$ (according to Eq.(\ref{eq:ozo})), where $D$ is a positive-definite diagonal symplectic matrix, and $O_1$ and $O_2$ are orthogonal symplectic. 
    Define the orthogonal transformation $I(g) \coloneqq O_1 O(g) O_1^{\rm T}$.
    Then, 
    \begin{equation}\label{eq:ogodo}
        \widetilde{O}(g) = (O_2 D O_1) O(g) (O_2 D O_1)^{-1} = O_2 D I(g) D^{-1} O_2^{\rm T}
    \end{equation}
    is orthogonal if and only if $D I(g) D^{-1}$ is orthogonal, i.e.
    \begin{equation}
        (D I(g) D^{-1}) (D I(g) D^{-1})^{\rm T} = I \,,
    \end{equation}
    or equivalently,
    \begin{equation}
        D I(g) D^{-2} I(g)^{\rm T} D = I \,,
    \end{equation}
    where we have used the fact that $D$ is diagonal, and therefore $D=D^{\rm T}$.
    This equivalently means that $[D^{-2}, I(g)]=0$, or, since $D$ is positive-definite, $[D, I(g)]=0$, for all $g \in G$.
    Together with Eq.(\ref{eq:ogodo}), we get
    \begin{equation}
        \widetilde{O}(g) = O_2 D I(g) D^{-1} O_2^{\rm T} = O_2 I(g) O_2^{\rm T} = (O_2O_1) O(g) (O_2O_1)^{\rm T} \,,
    \end{equation}
    where $T \coloneqq O_2O_1$ is orthogonal and symplectic, concluding the proof.
\end{proof}

An important corollary of Proposition~\ref{prop:squ} is that symplectically equivalent representations have the same irrep multiplicities.
\begin{proposition}\label{prop:mlambda}
    Consider two $G$--invariant symplectic subspaces $F_1$ and $F_2$ with equal dimensions, symplectic forms $\Omega_1$ and $\Omega_2$, and squeezed passive representations $S_1$ and $S_2$, respectively.
    Suppose there exists a symplectic intertwiner $V: F_1 \rightarrow F_2$, such that $V^{\rm T} \Omega_2 V = \Omega_1$, and $V S_1(g) = S_2(g) V$ for all $g \in G$. 
    Then, for each irrep $\lambda$ of group $G$ in the associated unitary representation in the space of annihilation (or creation) operators, the multiplicity $m(\lambda)$, as defined by Eq.(\ref{eq:uirreps}), is identical for $S_1$ and $S_2$.
    
    Conversely, if all multiplicities $m(\lambda)$ are identical for two symplectic representations $S_1$ and $S_2$, then the representations are equivalent up to a similarity transformation by a symplectic matrix $V$, i.e. $S_2(g) = V^{-1} S_1(g) V: g\in G$. If $S_1$ and $S_2$ are both symplectic and orthogonal, then $V$ can also be chosen to be symplectic and orthogonal.
\end{proposition}
\begin{proof}
    Consider the decomposition $S_1(g) = R_1 O_1(g) R_1^{-1}$ for symplectic matrix $R_1$ and orthogonal symplectic representation $O_1(g): g \in G$, with decomposition  
    \begin{equation}
        O_1(g) \overset{\bma}{=} U_1(g) \oplus U_1(g)^* \,.
    \end{equation}
    A similar decomposition exists for $S_2(g)$, with $S_2(g) = R_2 O_2(g) R_2^{-1}$, and $O_2(g) \overset{\bma}{=} U_2(g) \oplus U_2(g)^*$. 
    
    Since $S_2(g) = VS_1(g)V^{-1} = (V R) O_1(g) (VR)^{-1} $, and $V$ is symplectic, the second part of Proposition~\ref{prop:squ} implies that there exists an orthogonal symplectic transformation $\widetilde{O}$, such that
    \begin{equation}
        O_2(g) = \widetilde{O} O_1(g) \widetilde{O}^{\rm T} \,,
    \end{equation}
    for all $g\in G$. 
    Writing this in the complex basis, we find that there exists a unitary transformation $\tilde{U}$ such that
    \begin{equation}
        U_2(g) = \tilde{U} U_1(g) \tilde{U}^\dag \,,
    \end{equation}
    for all $g\in G$.
    This implies that $U_1$ and $U_2$ are equivalent unitary representations, and thus have the same isotypic decomposition,
    \begin{equation}
        U_2(g) = \tilde{U} U_1(g) \tilde{U}^\dag = \bigoplus_{\lambda \in \Lambda} U_{\lambda}(g) \otimes I_{m(\lambda)} \,,\quad \text{for all } g \in G \,.
    \end{equation}
    In particular, this means that for each irrep $\lambda$, the  multiplicity $m(\lambda)$ is identical for representations $U_1$ and $U_2$.
\end{proof}

\subsection{U(1) symmetry breaking for Gaussian systems}
\label{app:u1breaking}

In Gaussian systems, $\Un(1)$ symmetry breaking is equivalent to $\mathbb{Z}_p$ symmetry breaking for some $p>2$.
\begin{proposition}
    Let $\hat{O}(\theta) = \bigotimes_{j=1}^n \exp(i\theta q_j \hat{a}^\dag_j\hat{a}_j)$, with integer $q_j \neq 0$ for $j = 1, \dots, n$, be a representation of $\Un(1)$ on $n$ modes.
    A Gaussian Hamiltonian breaks the $\Un(1)$ symmetry if and only if it breaks a $\mathbb{Z}_p$ symmetry for some integer $p > 2\max_{j \in \mathbb{Z}_n}|q_j|$.
\end{proposition}
\begin{proof}
    Since no mode carries the trivial representation, i.e. $q_j \neq 0$ for all $j$, we can assume that there is no displacement, otherwise both symmetries are broken.
    
    Any purely quadratic Hamiltonian on $n$ modes takes the form
    \begin{equation}
        \hat{H} = \sum_{j,k=1} F_{jk} \hat{a}_j^\dag\hat{a}_k + G_{jk} \hat{a}_j\hat{a}_k + \text{H.c.} \,.
    \end{equation}
    For an arbitrary representation of $\Un(1)$ on $n$ modes, as given by Eq.(\ref{eq:u1_repr_hat}), the subgroup $\mathbb{Z}_p$ admits subrepresentation
    \begin{equation}
        \hat{O}(r) = \bigotimes_{j=1}^n e^{i r\frac{2\pi}{p} q_j \hat{a}^\dag_j\hat{a}_j} \,,\quad \text{for all } r\in\mathbb{Z}_p \,.
    \end{equation}
    The Hamiltonian transforms under the action of $\Un(1)$ or $\mathbb{Z}_p$ as
    \begin{equation}
    \begin{split}
        \hat{O}(\theta)\hat{H}\hat{O}(\theta)^\dag = \sum_{j,k=1} F_{jk} e^{i\theta(q_k - q_j)} \hat{a}_j^\dag\hat{a}_k + G_{jk} e^{i\theta(q_k + q_j)} \hat{a}_j\hat{a}_k + \text{H.c.} \,,
    \end{split}
    \end{equation}
    where $\theta \in [0, 2\pi)$ if the symmetry is U(1), and $\theta \in \left\{r\frac{2\pi}{p}: r \in \mathbb{Z}_p \right\}$ if the symmetry is $\mathbb{Z}_p$.
    Operator $\hat{H}$ breaks the symmetry if and only if any of the phases $e^{i\theta(q_k - q_j)}, e^{i\theta(q_k + q_j)}$ is non-zero for any permissible $\theta$.
    
    In the case of $\mathbb{Z}_p$ symmetry with $p < \max\{q_k - q_j, q_j + q_k\}$, there may be phases that are zero independently of $r$.
    However, for $p > 2\max_{j \in \mathbb{Z}_n}|q_j|$, it is guaranteed that $p > q_k - q_j$ and $p > q_j + q_k$, so $\hat{H}$ has the same non-zero terms under the action of both $\Un(1)$ and $\mathbb{Z}_p$.
\end{proof}
This result applies to any symplectic Gaussian unitary generated by / any quantum state prepared as the thermal state of a second-order Hamiltonian that breaks the $\Un(1)$ symmetry.

\subsection{Mode-coupling and entanglement in the presence of symmetries}
\label{app:entanglement}

We introduce a primitive that entangles systems in a symmetry-respecting manner.
Starting with a simple scenario, let $\hat{a}$ and $\hat{b}$ be the annihilation operators on two 1--mode systems $A$ and $B$ respectively.
Then, passive operations, such as the beam-splitter cannot create correlations from the vacuum state $\ket{0}^{\otimes 2}$.
On the other hand, the 2--mode squeezer $\hat{V}_{\rm{2sq}}(r)$, with $r \neq 0$, acts on the vacuum as
\begin{equation}\label{eq:2mode_squeezing}
    \hat{V}_{\rm{2sq}}(r) \ket{0}^{\otimes 2} = \sech{r}\sum_{j=0}^\infty (-\tanh r)^j \ket{j} \otimes \ket{j} \,.
\end{equation}
This state is known as the EPR state and it is entangled whenever $r \neq 0$, as its reduced 1--mode states are thermal states with non-zero von Neumann entropy (Eq.(\ref{eq:vonneumann})).
We generalize the 2--mode squeezer $\hat{V}_{\rm{2sq}}(r)$ to obtain an operator that entangles 2 systems of arbitrary numbers of modes, and we characterize the conditions on the squeezing correlation matrix $r$ that ensure invariance of the operator under a general passive representation of a symmetry.
\begin{proposition}\label{prop:entangling}
    Suppose that $A$ and $B$ are two bosonic systems of $n$ modes, and let $G$ be a symmetry group with squeezed passive representation $g \mapsto \hat{S}_A(g) = \hat{R}^{-1}\hat{O}_A(g)\hat{R}$ in Hilbert space $\H_A$, where $\hat{R}$ is a symplectic Gaussian unitary and $\hat{O}_A$ has associated phase space representation $g \mapsto O_A(g) \overset{\bma}{=} U(g) \oplus U(g)^*$.
    
    Define the following symplectic Gaussian unitary,
    \begin{equation}\label{eq:inv_entangle_operator}
        \hat{V}_{\rm sq}(r) \coloneqq \hat{R}\ \exp\left( \sum_{j,k=1}^n r_{jk} \hat{a}^\dag_j \hat{b}^\dag_k - r_{jk}^* \hat{a}_j \hat{b}_k \right) \hat{R}^{-1} \,,
    \end{equation}
    for any $(n \times n)$ complex, invertible matrix $r$.
    Then, operator $\hat{V}(r)$ has associated symplectic matrix
    \begin{equation}\label{eq:inv_entangle_matrix}
        V_{\rm sq}(r) \overset{\bma}{=} R \begin{pmatrix} \cosh\sqrt{r r^\dag} & 0 & 0 & {\displaystyle \frac{\sinh\sqrt{r r^\dag}}{\sqrt{r r^\dag}}\ r} \\ 0 & \left(\cosh\sqrt{r r^\dag}\right)^* & {\displaystyle \left(\frac{\sinh\sqrt{r r^\dag}}{\sqrt{r r^\dag}}\ r\right)^*} & 0 \\ 0 & {\displaystyle \left(\frac{\sinh\sqrt{r^\dag r}}{\sqrt{r^\dag r}}\ r^\dag \right)^*} & \left(\cosh\sqrt{r^\dag r}\right)^* & 0 \\ {\displaystyle \frac{\sinh\sqrt{r^\dag r}}{\sqrt{r^\dag r}}\ r^\dag }  & 0 & 0 & \cosh\sqrt{r^\dag r} \end{pmatrix} R^{-1} \,,
    \end{equation}
    where the basis is $(\hat{a}_1, \dots, \hat{a}_n, \hat{a}_1^\dag, \dots, \hat{a}_n^\dag, \hat{b}_1,\dots, \hat{b}_n, \hat{b}_1^\dag, \dots, \hat{b}_n^\dag)$.
    
    Moreover, $\hat{V}(r)$ is invariant under the tensor representation $g \mapsto \hat{S}_A(g) \otimes \hat{S}_B(g)$ if and only if $r = U_A(g) r U_B(g)^{\rm T}$ for all $g \in G$.
    In particular, if $\hat{S}_B(g) = \hat{S}^*_A(g)$, this condition becomes $[r, U_A(g)] = 0$ for all $g \in G$.
\end{proposition}
\noindent The correlation matrix $r$ describes how operator $\hat{V}_{\rm sq}(r)$ entangles the annihilation (creation) operators on system $A$ with the annihilation (creation) operators on system $B$. 
Operator $\hat{V}_{\rm sq}(r)$ is generated by the Hamiltonian matrix $H_{\rm sq}(r) \overset{\bma}{=} r^* \oplus r \oplus r^* \oplus r$ and it is an $n$--mode generalization of the 2--mode squeezer $\hat{V}_{\rm{2sq}}(r)$, which can be obtained by setting $n=1$ and $r \in \mathbb{R}$.
\begin{proof}
    Throughout the proof, we assume that the representations $\hat{S}_A$ and $\hat{S}_B$ are passive, which can be achieved by simply changing the basis by $\hat{R}$.
    
    First, we prove the invariance condition for operator $\hat{V}(r)$,
    \begin{equation}
    \begin{split}
        (\hat{S}_A(g) \otimes \hat{S}_B(g)) \hat{V}_{\rm sq}(r) (\hat{S}_A(g) \otimes \hat{S}_B(g))^\dag
        &= \exp\left( \sum_{j,k=1}^n \left( r_{jk} \left(\hat{S}_A(g) \hat{a}_j^\dag \hat{S}_A(g)^\dag\right)\left(\hat{S}_B(g) \hat{b}_k^\dag \hat{S}_B(g)^\dag\right) - \text{H.c.} \right) \right) \\
        &\overset{\bma}{=} \exp\left( \sum_{j,k=1}^n \left( r_{jk} \sum_{j',k'=1}^n [U_A(g)]_{j'j} \hat{a}_{j'}^\dag [U_B(g)]_{k'k} \hat{b}_{k'}^\dag - \text{H.c.} \right) \right) \\
        &\overset{\bma}{=} \exp\left( \sum_{j',k'=1}^n \left( \left[U_A(g) r U_B(g)^{\rm T}\right]_{j'k'} \hat{a}_{j'}^\dag \hat{b}_{k'}^\dag - \text{H.c.} \right) \right) \\
        &= \hat{V}_{\rm sq}(U_A(g) r U_B(g)^{\rm T}) \,,
    \end{split}
    \end{equation}
    so $[\hat{S}_A(g) \otimes \hat{S}_B(g),\hat{V}(r)]=0$ if and only if $r = U_A(g) r U_B(g)^{\rm T}$, which holds if $S_B(g) \overset{\bma}{=} \widetilde{S}_A(g)$ and $[r, U_A(g)]=0$, for all $g \in G$.
    
    Next, we derive the symplectic matrix $V_{\rm sq}(r)$ associated to the operator $\hat{V}_{\rm sq}(r)$ given in Eq.(\ref{eq:inv_entangle_operator}). 
    To this end, we evaluate the operator action on $\hat{\bmxi} = (\hat{a}_1, \dots, \hat{a}_n, \hat{a}_1^\dag, \dots, \hat{a}_n^\dag, \hat{b}_1,\dots, \hat{b}_n, \hat{b}_1^\dag, \dots, \hat{b}_n^\dag)$, such that $\hat{V}_{\rm sq}(r)^\dag \hat{\bmxi} \hat{V}_{\rm sq}(r) = V_{\rm sq}(r)\hat{\bmxi}$.
    
    We first derive the action of $\hat{V}(r)$ on $\hat{\bmxi}$ by considering the singular value decomposition (SVD)
    \begin{equation}
        r = U^\dag D W \,,
    \end{equation}
    for some unitary matrices $U,W$ and $D = \diag(\lambda_1, \dots, \lambda_n)$, where $\lambda_\ell \ge 0$ for $\ell = 1, \dots, n$.
    The symplectic Gaussian unitary $\hat{V}_{\rm sq}(r)$ can be re-expressed as
    \begin{equation}
        \hat{V}_{\rm sq}(r) = \exp\left( \sum_{\ell} \lambda_\ell \left( \tilde{a}^\dag_\ell \tilde{b}^\dag_\ell - \tilde{a}_\ell \tilde{b}_\ell \right) \right) \,,
    \end{equation}
    where the new annihilation operators can be expressed in terms of the original annihilation operators as
    \begin{equation}\label{eq:decoupled_basis}
    \begin{split}
        \tilde{a}_\ell &\coloneqq \sum_{k=1}^n U_{\ell k} \hat{a}_k \,, \\
        \tilde{b}_\ell &\coloneqq \sum_{k=1}^n W^*_{\ell k} \hat{b}_k \,.
    \end{split}
    \end{equation}
    The modes are locally decoupled in the new basis, so the action of $\hat{V}(r)$ is independent on each pair of modes,
    \begin{equation}
    \begin{split}
        \hat{V}_{\rm sq}(r)^\dag\tilde{a}_\ell\hat{V}_{\rm sq}(r) 
        &= \cosh{\lambda_\ell}\ \tilde{a}_\ell + \sinh{\lambda_\ell}\ \tilde{b}^\dag_\ell \,, \\
        \hat{V}_{\rm sq}(r)^\dag\tilde{b}_\ell\hat{V}_{\rm sq}(r) 
        &= \cosh{\lambda_\ell}\ \tilde{b}_\ell + \sinh{\lambda_\ell}\ \tilde{a}^\dag_\ell \,.
    \end{split}
    \end{equation}
    This is the action of a 2--mode squeezer.
    
    Converting back to the original basis using the transformations in Eq.(\ref{eq:decoupled_basis}),
    \begin{align}
        \hat{V}_{\rm sq}(r)^\dag\hat{a}_\ell\hat{V}_{\rm sq}(r) 
        &= \sum_{k=1}^n \left( \left[ U^\dag \cosh{D}\ U \right]_{\ell k} \hat{a}_k + \left[ U^\dag \sinh{D}\ W \right]_{\ell k} \hat{b}^\dag_k \right) \nonumber\\
        &= \sum_{k=1}^n \left( \left[ \cosh{\sqrt{rr^\dag}} \right]_{k\ell}\ \hat{a}_k + \left[ \left(\sinh{\sqrt{rr^\dag}}\right)(rr^\dag)^{-\frac{1}{2}}r \right]_{k\ell}\ \hat{b}^\dag_k \right) \,, \label{eq:proof_atransform1}\\
        \hat{V}_{\rm sq}(r)^\dag\hat{b}_\ell\hat{V}_{\rm sq}(r) 
        &= \sum_{k=1}^n \left( \left[ W^\dag \cosh{D}\ W \right]_{k\ell} \hat{b}_k + \left[ U^\dag \sinh{D}\ W \right]_{k\ell} \hat{a}^\dag_k \right) \nonumber\\
        &= \sum_{k=1}^n \left( \left[ \cosh{\sqrt{r^\dag r}} \right]^*_{\ell k}\ \hat{b}_k + \left[ \left( \sinh{\sqrt{r^\dag r}} \right)(r^\dag r)^{-\frac{1}{2}}r^\dag \right]^*_{\ell k}\ \hat{a}^\dag_k \right) \,, \label{eq:proof_btransform1}
    \end{align}
    where we have used that $(r r^\dag)^{-\frac{1}{2}} = ({U^\dag D W (U^\dag D W)^{\dag}})^{-1/2} = U^\dag D^{-1} U$, which implies 
    \begin{equation}
        U^\dag W= [U^\dag D^{-1} U] U^\dag D W =(r r^\dag)^{-\frac{1}{2}}r \,.
    \end{equation}

    The associated symplectic matrix $V_{\rm sq}(r)$ provided in Eq.(\ref{eq:inv_entangle_matrix}) can be read directly from Eq.(\ref{eq:proof_atransform1}) and Eq.(\ref{eq:proof_btransform1}) up to the basis change induced by $\hat{R}$.
\end{proof}

\section{Resource theory of Gaussian asymmetry}

Here, we provide proofs and derivations for Sections~\ref{sec:res_theory} and~\ref{sec:monotones}.

\subsection{Transitivity of free unitaries on free pure states}
\label{app:res_transverse}

We prove the existence of a free unitary that converts between any two free pure states in the resource theory of Gaussian asymmetry and provide an explicit form for it in terms of the states' statistical moments.
\pureinterconvert*
\begin{proof}
    We first consider the passive representation $g \mapsto \hat{O}(g)$ and make use of the irrep decomposition of an invariant state $\ket{\psi}$ (see Section~\ref{app:state_irreps}) to find an invariant Gaussian unitary $\hat{V}$ that maps the vacuum state $\ket{0}^{\otimes n}$ to $\ket{\psi}$.
    
    Let $S(g) \overset{\bma}{=} U(g) \oplus U(g)^*$ be the associated phase space representation.
    Let $\bmd$ and $\sigma$ be the displacement vector and covariance matrix of state $\ket{\psi}$. 
    The displacement operator $\hat{D}_{\bmd}$ is invariant as proven in Proposition~\ref{prop:symmetry_phasespace}.
    Due to Eq.(\ref{eq:ozo}), the covariance matrix of any pure state can be written as $\sigma \overset{\bmr}{=} SS^{\rm T} \overset{\bmr}{=} O\left(\bigoplus_{j=1}^n \diag\left(z_j^2, z_j^{-2}\right)\right)O^{\rm T}$, where $O$ is orthogonal symplectic.
    Then, we construct the matrix
    \begin{equation}
        V \coloneqq \sqrt{\sigma} \overset{\bmr}{=} O\left(\bigoplus_{j=1}^n \diag\left(z_j, z_j^{-1}\right)\right)O^{\rm T} \,,
    \end{equation}
    which is invariant under $S$, symplectic, i.e. $\sqrt{\sigma} \Omega_n \sqrt{\sigma}^{\rm T} = \Omega_n$.
    Then, $\hat{D}_{\bmd}\hat{V}$ maps $\ket{0}^{\otimes n}$ to $\ket{\psi}$.
    
    For squeezed passive representations of the form $\hat{R}\hat{S}(g)\hat{R}^\dag$, where $\hat{S}$ is passive, we can simply use the transformation Table~\ref{tab:transf_rules}, to deduce that an invariant symplectic matrix that sends $\hat{R}\ket{0}^{\otimes n}$ to $\ket{\psi}$ is $V = R\sqrt{R^{-1} \sigma R^{\rm -T}}R^{-1}$.
    
    Applying first $\hat{V}_1^\dag$ to send $\ket{\psi_1}$ to the vacuum, then $\hat{V}_2$ to send the vacuum to $\ket{\psi_2}$, we obtain Eq.(\ref{eq:rho1to2}).
\end{proof}

\subsection{Interconversion between types of asymmetry}
\label{app:types}

Here we provide the proof of Proposition~\ref{prop:types} on the convertibility between type--1 and type--2 asymmetry.
\asymtypes*
\begin{proof}
    As mentioned in the main text, both types of asymmetry remain distinct under Gaussian covariant channels, i.e. the output state contains each type of asymmetry, only if the input state contains it.
    
    Consider a quantum state with zero type--1 asymmetry.
    In such a case, the state's first moment is non-zero only in the trivial irreps of the representation, thus the state respects the $\mathbb{Z}_2$ symmetry.
    If a quantum channel acts on the state increasing its type--1 asymmetry, then the new first moment is non-zero in at least one non-trivial irrep, and the state no longer respects the $\mathbb{Z}_2$ symmetry.
    Therefore, such a channel cannot be $G$--covariant.

    To prove that a pure Gaussian state $\ket{\psi}$ with non-zero type--1 asymmetry can generate type--2 asymmetry via a $G$--covariant operation, consider the passive representation $g \mapsto \hat{S}(g)$ of a group $G$ on $n$ modes, with associated phase space representation $g \mapsto S(g) \overset{\bma}{=} U(g) \oplus U(g)^*$.
    Without loss of generality, we assume that it does not contain a trivial subrepresentation, otherwise we can increase type--2 asymmetry arbitrarily on modes with trivial symmetry.
    
    Firstly, we can convert the covariance matrix $\sigma$ of $\ket{\psi}$ to the identity $I$ by applying the $G$--invariant symplectic transformation $\sqrt{\sigma}$ (see Proposition~\ref{prop:pure_interconversion}).
    Therefore, we can write $\ket{\psi}$ as an $n$--mode coherent state $\ket{\psi} = \ket{\alpha} = \bigotimes_{j=1}^n \ket{\alpha_j}$, for vector $\alpha = (\alpha_1, \dots, \alpha_n)$. 
    We then act on it with the non-Gaussian $G$--invariant unitary $\hat{U} \coloneqq \hat{I} - 2\ketbra{0}{0}$.
    According to Eqs.(\ref{eq:ngmoment1},\ref{eq:ngmoment2}), we can write the output covariance matrix in the complex basis as
    \begin{equation}
        \sigma \overset{\bma}{=} \begin{pmatrix} G & F \\ F^* & G^* \end{pmatrix} \,,\quad\quad F = I + 2(1-c^2) \alpha\alpha^\dag\,,\quad\quad G = 2c(1-c) \alpha\alpha^{\rm T} \,.
    \end{equation}
    Whenever the vector $\alpha$ is not invariant under $g \mapsto U(g)$, the matrix $G \propto \alpha\alpha^{\rm T}$ is also not invariant in the sense of Eq.(\ref{eq:conditions_invariance}), i.e. $U\alpha \neq \alpha$ implies $U G U^{\rm T} \propto U\alpha(U\alpha)^{\rm T} \neq \alpha\alpha^{\rm T}$ for any unitary $U$.
\end{proof}

In fact, if $\alpha$ is arbitrarily high and the symmetry group $G$ is compact, we can output any arbitrary state $\hat{\tau}$ with arbitrarily high probability.
Define the measure-and-prepare operation $\E_{\hat{\tau}}: \B(\H) \rightarrow \B(\H)$ that acts on a state $\hat{\rho}$ as
\begin{equation}\label{eq:mnp_channel}
    \E_{\hat{\tau}}(\hat{\rho}) = \int {\rm d}g\ \tr[\hat{\rho}\ketbra{\alpha_g}] \hat{S}(g) \hat{\tau} \hat{S}(g)^\dag + \tr\left[\hat{\rho}\left(\hat{I} - \int {\rm d}g\ \ketbra{\alpha_g}\right)\right] \ketbra{0} \,,
\end{equation}
where $\hat{\tau}$ is a fixed state, and $\ket{\alpha_g} \coloneqq \hat{S}(g)\ket{\alpha} = \ket{U(g)\alpha}$.
Choose $\hat{\rho} = \ketbra{\alpha}$ and label the POVM implicitly defined in Eq.(\ref{eq:mnp_channel}) as $\{\hat{M}_g: g \in G \sqcup \{0\}\}$, with $g=0$ corresponding to the preparation of the vacuum.
Then, the density function,
\begin{equation}
    p(\alpha|g) = 
    \begin{cases}
        |\braket{\alpha}{\alpha_g}|^2 \,, &g \in G \,, \\
        1 - \int {\rm d}g\ |\braket{\alpha}{\alpha_g}|^2 \,, &g = 0 \,,
    \end{cases}
\end{equation}
satisfies $p(\alpha|g) \neq p(\alpha|g')$ for $g \neq g'$, because the representation acts non-trivially on $\ket{\alpha}$, and it is uniquely maximized at the identity element $g = e$.
Consider the limit of high $\|\alpha\|_2$ which, in the absence of a trivial representation, can be taken in any direction $|\alpha_i|$.
In this limit, and for $\hat{\rho} = \ketbra{\alpha}$, $p$ decays exponentially in $\|\alpha\|_2^2$ tending to a sharp distribution around $g = e$ (Dirac delta $\delta(g-e)$ if $G$ is continuous or Kronecker delta $\delta_{g,e}$ if $G$ is finite) due to the approximate orthogonality of coherent states, and channel $\E_{\hat{\tau}}$ outputs state $\hat{\tau}$ with arbitrarily high probability.
This holds for any input state with first moment defined by $\alpha \neq 0$.

In fact, the maximum likelihood estimator can be chosen~\cite{newey1994large} to estimate $e$.
Importantly, the density function $p$ is continuous over the compact parameter space $G \sqcup \{0\}$, and identifiable, i.e. $g \neq g'$ implies $p(\alpha|g) \neq p(\alpha|g')$, so the maximum likelihood estimator converges in probability to the unique maximum $g = e$ for a sequence of input states with increasingly high $\|\alpha\|_2$.

\subsection{First moments of reachable states under Gaussian covariant channels}
\label{app:suppd}

We provide the proof of Proposition~\ref{prop:suppd} which describes the complete set of first moments of states that are reachable from a fixed state under Gaussian covariant channels.
\suppsubset*
\begin{proof}
    First, we note that for a general covariant channel the displacement vector transforms as $\bmd_B = X\bmd_A + \bmxi$, where $X$ satisfies $X S_A(g)=S_B(g) X$ for all $g\in G$, and $ \bmxi$ is zero in all modes except the mode carrying the trivial representation. It follows that the projection of $\bmd_B$ to the trivial modes is fully unrestricted and can be independent of the input. Therefore, in the following, we focus on modes carrying non-trivial representations of symmetry and find the condition for the existence of an invariant $X$ satisfying $\bmd_B = X\bmd_A+\bmxi$.  
        
    Consider input and output representations
    \begin{equation}
        K_A S_A(g) K_A^{-1} = \bigoplus_{\mu \in M_A} U_\mu \otimes I_{m(\mu)} \,,\quad\quad K_B S_B(g) K_B^{-1} = \bigoplus_{\nu \in M_B} U_\nu \otimes I_{m'(\nu)} \,,
    \end{equation}
    where $K_A, K_B$ are defined according to Eq.(\ref{eq:uustar_decomp}), $m(\mu)$ and $m'(\nu)$ are the multiplicities of irrep $\mu$ in the input irrep decomposition and of irrep $\nu$ in the output irrep decomposition, respectively.
    
    Given a complex matrix $X$ and $g \in G$, we have $XS_A(g) = S_B(g)X$ if and only if $X' \coloneqq K_B^{-1} X K_A$ satisfies
    \begin{equation}\label{eq:KXK}
        X' \left( \bigoplus_{\mu \in M_A} U_\mu \otimes I_{m(\mu)} \right) = \left( \bigoplus_{\nu \in M_B} U_\nu \otimes I_{m'(\nu)} \right) X' \,,
    \end{equation}
    which leads to the following form for $X'$ according to Schur's lemma,
    \begin{equation}\label{eq:X'}
        X' = \bigoplus_{\mu \in M_A \cap M_B} I_{d_\mu} \otimes X'_\mu \,,
    \end{equation}
    for some $(m'(\mu) \times m(\mu))$ matrices $X'_\mu$.

    Defining $\bmd'_A \coloneqq K_A \bmd_A$ and $\bmd'_B \coloneqq K_B \bmd_B$, we now determine the necessary and sufficient condition for existence of $X'$ in the form given in Eq.(\ref{eq:X'}), which satisfies (i) $X'\bmd'_A = \bmd'_B$, and (ii) $X = K_B X' K^{-1}_A$ is real.  

    First, note that there exists $X'$ in the above form satisfying $X'\bmd'_A=\bmd'_B$ if and only if for all $\mu \in M_A\cap M_B$, there exists $X'_\mu$ such that
    \begin{equation}\label{eq:X'Pi}
        (I_{d_\mu} \otimes X'_\mu) \Pi_\mu \bmd'_A= \Pi_\mu \bmd'_B \,.
    \end{equation}
    According to Lemma~\ref{lem:supp}, such $X_\mu$ exists if and only if
    \begin{equation}
        {\rm supp}\left( \tr_{\mathcal{N}_\mu}[\Pi_\mu \bmd'_B {\bmd'_B}^{\dag} \Pi_\mu]  \right) \subseteq {\rm supp}\left( \tr_{\mathcal{N}_\mu}[\Pi_\mu \bmd'_A {\bmd'_A}^{\dag} \Pi_\mu]  \right) \,,
    \end{equation}
    for all $\mu \in M_A\cap M_B$, or, equivalently,
    \begin{equation}
          {\rm supp}\left( T^{(\mu)}\left(\bmd_B\bmd_B^{\dag}\right) \right) \subseteq {\rm supp}\left(T^{(\mu)}\left(\bmd_A\bmd_A^\dag\right)\right) \,,
    \end{equation}
    for all $\mu \in M_A\cap M_B$, where 
    \begin{equation}
        T^{(\mu)}(\cdot) \coloneqq \tr_{\mathcal{N}_\mu}[\Pi_\mu K(\cdot)K^\dag \Pi_\mu] \,.
    \end{equation}
    This proves the necessity of the condition in Eq.(\ref{eq:suppsubseq}).
    To prove the sufficiency of this condition, consider a set of $\{X'_\mu\}$ satisfying Eq.(\ref{eq:X'Pi}) for all $\mu\in M_A \cap M_B$, and consider the corresponding $X'$ obtained from Eq.(\ref{eq:X'}). 
    Then, $X = K_B X' K^{-1}_A$ satisfies the desired equation $X\bmd_A = \bmd_B$.
    However, in general, this $X$ will not be real.
    To show that $X$ can be chosen to be real, first we note that, because $\bmd_A$ and $\bmd_B$ are real, then $X^*$, the complex conjugate of $X$ also satisfies this equation. Furthermore, because $S_A$ and $S_B$ are real, it also respects the symmetry. It follows that $\tfrac{1}{2}(X+X^*)$ is a $G$--invariant real matrix satisfying $\tfrac{1}{2}(X+X^*) \bmd_A = \bmd_B$. This proves the necessity and sufficiency of the condition in Eq.(\ref{eq:suppsubseq}).
\end{proof}

We prove a standard result for bipartite vectors, needed in the proof of Proposition~\ref{prop:suppd}.
\begin{lemma}\label{lem:supp}
    Given bipartite vectors $\ket{\psi},\ket{\psi'} \in \H_A \otimes \H_B$, there exists a linear operator $N_B$ acting on $\H_B$ such that $(I_A \otimes N_B)\ket{\psi} = \ket{\psi'}$ if and only if
    \begin{equation}\label{eq:supp_condition}
        {\rm supp}\left(\tr_B[\ketbra{\psi'}{\psi'}]\right) \subseteq {\rm supp}\left(\tr_B[\ketbra{\psi}{\psi}]\right) \,.
    \end{equation}
\end{lemma}
\begin{proof}
    Arbitrary bipartite vectors $\ket{\psi}$ and $\ket{\psi'}$ can be written as $\ket{\psi} = \left(\sqrt{M_A} \otimes U_B\right) \sum_k \ket{k}\otimes\ket{k}$ and $\ket{\psi'} = \left(\sqrt{M'_A} \otimes U'_B \right) \sum_k \ket{k}\otimes\ket{k}$, where $M_A = \tr_B[\ket{\psi}\bra{\psi}]$ and $M'_A = \tr_B[\ket{\psi'}\bra{\psi'}]$, and $U_B, U'_B$ are arbitrary unitaries.
    It follows that there exists an operator $N_B$ acting on $\H_B$ such that $(I_A \otimes N_B)\ket{\psi} = \ket{\psi'}$, if and only if $N_B U_B\sqrt{M_A}^{\rm T} = U'_B \sqrt{M'_A}^{\rm T}$, or, equivalently, $\tilde{N}_B \coloneqq U^{\prime\dag}_B N_B U_B$ satisfies $\sqrt{M_A}\tilde{N}_B^{\rm T} = \sqrt{M'_A}$.     
    Such $\tilde{N}_B$ exists only if the support of $M'_A$ is contained in $M_A$, which can be seen, e.g. by noting that $M'_A=\sqrt{M_A}\tilde{N}_B^{\rm T} \tilde{N}^*_B \sqrt{M_A}$. 
    Conversely, suppose this condition is satisfied. 
    Then, by defining $M_A^{-1}$ to be the inverse of $M_A$ on its support and zero in the orthogonal subspace, we find $\tilde{N}_B^{\rm T}=\sqrt{M^{-1}_A}\sqrt{M'_A}$ satisfies the above equation. This completes the proof.
\end{proof}

\subsection{SU(2)--covariant channels}
\label{app:su2cov}

Consider the passive irreducible representation of $\SU(2)$ given by
\begin{equation}
    U \mapsto S(U) \overset{\bma}{=} U \oplus U^* = (I \oplus \Omega^{\rm T}) (U \oplus U)(I \oplus \Omega) \,,
\end{equation}
where $\Omega = \begin{pmatrix} 0 & I \\ -I & 0 \end{pmatrix}$ throughout this section.
Suppose a channel characterized by its action on the phase space via Eq.(\ref{eq:xychannel}), is covariant under input and output representation $S$.
Then, its displacement vector $\bmxi$ must be 0 because representation $S$ does not contain the trivial irrep, while matrices $X,Y$ must satisfy
\begin{equation}\label{eq:appXYsu2}
    [X, S(U)] = 0 \,, \quad S(U) Y S(U)^{\rm T} = Y \,,
\end{equation}
for all $U \in \mathrm{SU}(2)$.
The conditions in Eq.(\ref{eq:appXYsu2}) are equivalent to
\begin{equation}
    \left[(I \oplus -\Omega) X (I \oplus \Omega), U \oplus U\right] = 0 \,,\quad [(I \oplus -\Omega) Y (\Omega \oplus I), U \oplus U] = 0 \,,
\end{equation}
for all $U \in \SU(2)$, where we have used $U^{\rm T} = \Omega U^\dag \Omega^{\rm T} = -\Omega U^\dag \Omega$.
Expressing $X \overset{\bma}{=} \begin{pmatrix} X_1 & X_2 \\ X_3 & X_4 \end{pmatrix}$ and $Y \overset{\bma}{=} \begin{pmatrix} Y_1 & Y_2 \\ Y_3 & Y_4 \end{pmatrix}$ in block form, the conditions in Eq.(\ref{eq:appXYsu2}) are further equivalent to
\begin{align}
    \begin{pmatrix} U & 0 \\ 0 & U \end{pmatrix}
    \begin{pmatrix} X_1 & X_2 \Omega \\ \Omega X_3 & \Omega X_4 \Omega \end{pmatrix}
    \begin{pmatrix} U^\dag & 0 \\ 0 & U^\dag \end{pmatrix} 
    &= \begin{pmatrix} X_1 & X_2 \Omega \\ \Omega X_3 & \Omega X_4 \Omega \end{pmatrix} \,,\\
    \begin{pmatrix} U & 0 \\ 0 & U \end{pmatrix}
    \begin{pmatrix} Y_1\Omega & Y_2 \\ \Omega Y_3 \Omega & \Omega Y_4 \end{pmatrix}
    \begin{pmatrix} U^\dag & 0 \\ 0 & U^\dag \end{pmatrix} 
    &= \begin{pmatrix} Y_1\Omega & Y_2 \\ \Omega Y_3 \Omega & \Omega Y_4 \end{pmatrix} \,.
\end{align}
Since $U$ is an irrep of $\SU(2)$, Schur's lemma implies that
\begin{equation}
    X_1, X_4, Y_2, Y_3 \propto I \,,\quad X_2, X_3, Y_1, Y_4 \propto \Omega \,,
\end{equation}
where the proportionality constants are any complex numbers provided that $Y = Y^{\rm T}$ and that $X$ and $Y$ are of the form of Eq.(\ref{eq:real2complex}) to ensure that they are real in the real basis, i.e.
\begin{align}
    X &\overset{\bma}{=} \begin{pmatrix} aI & b\Omega \\ b^*\Omega & a^*I \end{pmatrix} \overset{\bmr}{=} e^{\phi \Omega} \begin{pmatrix} x_+ I & x_-Ze^{\theta \Omega} \\ -x_-Ze^{\theta \Omega} & x_+ I \end{pmatrix} \,,\\
    Y &\overset{\bma}{=} \begin{pmatrix} 0 & yI \\ yI & 0 \end{pmatrix} \overset{\bmr}{=} y I_4 \,,
\end{align}
where $a, b \in \mathbb{C}$, and $y \in \mathbb{R}$ are arbitrary, and $x_+ \coloneqq |a|, x_- \coloneqq |b|, \phi \coloneqq -\arg(a), \theta \coloneqq \arg(b)-\arg(a)$.
This form for $X$ is equivalent to Eq.(\ref{eq:Xchar}).
The channel uncertainty relation $Y + i(\Omega \oplus \Omega - X (\Omega \oplus \Omega) X^{\rm T}) \ge 0$ takes the form
\begin{equation}
    \begin{pmatrix} yI + i(1 - x_+^2 + x_-^2)\Omega & 2ix_+x_- (Z\Omega) e^{(\theta-\phi)\Omega)} \\ -2ix_+x_- e^{-(\theta-\phi)\Omega)}(Z\Omega) & yI + i(1 - x_+^2 + x_-^2)\Omega \end{pmatrix} \ge 0 \,,
\end{equation}
which holds if and only if
\begin{equation}
    y \ge \sqrt{\left((x_+ - 1)^2 + x_-^2\right)\left((x_+ + 1)^2 + x_-^2\right)} \,.
\end{equation}

\subsection{Relative entropy of asymmetry on 1 mode}
\label{app:monotone}

Consider $\Un(1)$ symmetry on 1 mode, with representation $\phi \mapsto \hat{V}_{\rm ps}(\phi)$, and uniform distribution $p(\phi) = \tfrac{1}{2\pi}$.
We can calculate the relative entropy of asymmetry given in Eq.(\ref{eq:relent_asym}) of a coherent state $\ket{\alpha} \coloneqq D_{\alpha}\ket{0}$ as follows,
\begin{align}
    \Gamma_p(\ketbra{\alpha}) 
    &\overset{(1)}{=} S_v\left( \frac{1}{2\pi} \int {\rm d}\phi\ \hat{V}_{\rm ps}(\phi) \ketbra{\alpha} \hat{V}_{\rm ps}(\phi)^\dag \right) - S_v(\ketbra{\alpha}) \nonumber\\
    &\overset{(2)}{=} S_v\left( \frac{1}{2\pi} \int {\rm d}\phi\ e^{i\phi(j-k)} e^{-|\alpha|^2}\sum_{j,k=0}^\infty \frac{\alpha^j\alpha^{*k}}{\sqrt{j!k!}}\ketbra{j}{k} \right) \nonumber\\
    &\overset{(3)}{=} S_v\left( e^{-|\alpha|^2}\sum_{k=0}^\infty \frac{|\alpha|^{2k}}{k!}\ketbra{k}{k} \right) \nonumber\\
    &\overset{(4)}{=} |\alpha|^2 \log\frac{e}{|\alpha|^2} + e^{-|\alpha|^2} \sum_{k=0}^{\infty} \frac{|\alpha|^{2k} \log(k!)}{k!}
    \overset{|\alpha| \gg 0}{=} \log(\sqrt{2\pi e} |\alpha|) + O(|\alpha|^{-2}) \,,
\end{align}
where $(1)$ follows from the monotone definition in Eq.(\ref{eq:relent_asym}), $(2)$ from the fact that pure states have zero entropy (e.g. because they have unit symplectic eigenvalues), and from the decomposition of a coherent state $\ket{\alpha}$ in the Fock basis given in Eq.(\ref{eq:std_decomps}), $(3)$ from explicit calculation of the integral, and $(4)$ from the definition of the von Neumann entropy given in Eq.(\ref{eq:vonneumann}).
The asymptotic behavior of $\Gamma_p(\ketbra{\alpha})$ corresponds to the well-known Poisson distribution.

The Holevo asymmetry monotone of a 1--mode squeezed state $\ket{r} \coloneqq V_{\rm 1sq}(r)\ket{0}$ is
\begin{align}
    \Gamma_p(\ketbra{r}) 
    &\overset{(1)}{=} S_v\left( \frac{1}{2\pi} \int {\rm d}\phi\ \hat{V}_{\rm ps}(\phi) \ketbra{r} \hat{V}_{\rm ps}(\phi)^\dag \right) - S_v(\ketbra{r}) \\
    &\overset{(2)}{=} S_v\left( \frac{1}{2\pi} \int {\rm d}\phi\ e^{i2\phi(j-k)} \frac{1}{\cosh(r)} \sum_{j,k=0}^{\infty} \left( \frac{\tanh(r)}{2} \right)^{j+k} \frac{\sqrt{(2j)!(2k)!}}{j!k!} \ketbra{2j}{2k} \right) \nonumber\\
    &\overset{(3)}{=} S_v\left( \sech(r) \sum_{k=0}^{\infty} \left( \frac{\tanh(r)}{2} \right)^{2k} \binom{2k}{k} \ketbra{2k}{2k} \right) \nonumber\\
    &\overset{(4)}{=} -\log(\sech(r)) - \sinh|r|\log(\frac{\tanh|r|}{2}) - \sech(r) \sum_{k=0}^{\infty} \left( \frac{\tanh(r)}{2} \right)^{2k} \binom{2k}{k} \log{\binom{2k}{k}} \nonumber\\
    &\overset{|r| \gg 0}{=} 2\log(e)\ |r| + \log(2^{-5}\sqrt{\pi}e^{2-\gamma/2}) + o(1)
    \,,
\end{align}
where steps $(1)-(3)$ follow similar justifications as for coherence.

Below, we provide a detailed derivation of step $(4)$ and the asymptotic behavior of $\Gamma_p(\ketbra{r})$, which is not well-known, as dephased squeezed states don't correspond to a known distribution, unlike dephased coherence states that correspond to the Poisson distribution.
The entropy of a dephased squeezed state is
\begin{align}
    &S_v\left( \sech(r) \sum_{k=0}^{\infty} \left( \frac{\tanh(r)}{2} \right)^{2k} \binom{2k}{k} \ketbra{2k}{2k} \right) \\
    \overset{(1)}{=} &-\sum_{k=0}^{\infty} \sech(r)\left(\frac{\tanh(r)}{2} \right)^{2k} \binom{2k}{k} \log\left(\sech(r)\left(\frac{\tanh(r)}{2} \right)^{2k} \binom{2k}{k}\right) \\
    \overset{(2)}{=} &-\sqrt{1-4x} \sum_{k=0}^{\infty} x^k \binom{2k}{k} \left(\frac{1}{2}\log(1-4x) + k\log(x) + \log{\binom{2k}{k}} \right) \\
    \overset{(3)}{=} &-\frac{1}{2}\log(1-4x) - \frac{2x\log(x)}{1-4x} - \sqrt{1-4x} \sum_{k=0}^{\infty} x^k \binom{2k}{k} \log{\binom{2k}{k}} \,. \label{eq:squ_ent_exact}
\end{align}
where $(1)$ follows from the definition of entropy in Eq.(\ref{eq:vonneumann}), $(2)$ follows by simply expanding the logarithm into three terms, and setting $x = \left(\frac{\tanh(r)}{2}\right)^2$, so $\sech^2(r) = 1-4x$, while $(3)$ follows from evaluating the first two sums.
This leaves us with the third sum, which diverges as $|r| \rightarrow \infty$ ($x \rightarrow \frac{1}{4}^-$), and is dominated by terms with high $k$. 
We provide an estimate of the sum to leading order,
\begin{align}
    &\sum_{k=0}^{\infty} x^k \binom{2k}{k} \log{\binom{2k}{k}} \\
    \overset{(1)}{\approx}\ &\frac{2\log(2)}{\sqrt{\pi}} \sum_{k=1}^\infty (4x)^k k^{1/2} - \frac{\log(\pi)}{2\sqrt{\pi}} \sum_{k=1}^\infty (4x)^k k^{-1/2} - \frac{1}{2\sqrt{\pi}} \sum_{k=1}^\infty (4x)^k k^{-1/2} \log(k) \\
    \overset{(2)}{\approx}\ &\frac{2\log(2)}{\sqrt{\pi}} \int_{0}^\infty {\rm d}k\ e^{-zk} k^{1/2} - \frac{\log(\pi)}{2\sqrt{\pi}} \int_{0}^\infty {\rm d}k\ e^{-zk} k^{-1/2} - \frac{1}{2\sqrt{\pi}} \int_{0}^\infty {\rm d}k\ e^{-zk} k^{-1/2} \log(k) \\
    \overset{(3)}{\approx}\ &\frac{2\log(2)}{\sqrt{\pi}z^{3/2}} \Gamma\left(\frac{3}{2}\right) - \frac{\log(\pi)}{2\sqrt{\pi}z^{1/2}} \Gamma\left(\frac{1}{2}\right) - \frac{\log(z)}{2\sqrt{\pi}z^{1/2}} \Gamma\left(\frac{1}{2}\right) + \frac{1}{2\sqrt{\pi}z^{1/2}} \Gamma'\left(\frac{1}{2}\right) \\
    \overset{(4)}{\approx}\ &\frac{\log(2)}{(1-4x)^{3/2}} - \frac{\log(\pi)}{2(1-4x)^{1/2}} - \frac{\log(1-4x)}{2(1-4x)^{1/2}} - \frac{\gamma\log(e) + 2\log(2)}{2(1-4x)^{1/2}} \,, \label{eq:squ_ent_approx}
\end{align}
where in $(1)$ we use Stirling's approximation,
\begin{equation}
    \frac{n!}{\sqrt{2\pi n}\left(\frac{n}{e}\right)^n} = 1 + o(1) \,,
\end{equation}
in $(2)$ we approximate the sums by integrals and set $z = -\ln(4x)$ ($\approx 1-4x$ around $x=1/4$), in $(3)$ we evaluate the integrals, using the Gamma function $\Gamma$ and its derivative $\Gamma'$, and in $(4)$ we evaluate the terms, where $\gamma \approx 0.577$ is the Euler–Mascheroni constant.

Combining Eq.(\ref{eq:squ_ent_exact}) with Eq.(\ref{eq:squ_ent_approx}), we obtain
\begin{align}
    S_v\left( \sech(r) \sum_{k=0}^{\infty} \left( \frac{\tanh(r)}{2} \right)^{2k} \binom{2k}{k} \ketbra{2k}{2k} \right)
    \overset{|r| \gg 0}{\approx}\ 2\log(e)\ |r| + \log(2^{-5}\sqrt{\pi}e^{2-\gamma/2}) \,.
\end{align}

\vspace{10pt}

One can similarly calculate that the 2--mode squeezed state $\ket{r}_{\rm EPR}$ defined in Eq.(\ref{eq:2mode_squeezing}),
under $\Un(1)$ representation $\theta \mapsto e^{i\theta\left(q_1 \hat{a}_1^\dag\hat{a}_1 + q_2 \hat{a}_2^\dag\hat{a}_2\right)}$ with integer charges $q_1 \neq -q_2$, results in
\begin{align}
    \Gamma(\ketbra{r}_{\rm EPR}) 
    &= \cosh^2(r)\log\left(\cosh^2(r)\right) - \sinh^2(r)\log\left(\sinh^2(r)\right) \nonumber\\
    &\hspace{-4pt}\overset{r \gg 0}{=} \log\langle n \rangle - \log{4} \,, \label{eq:asym_2squeezed}
\end{align}
where $\langle n \rangle = \langle \hat{a}_1^\dag\hat{a}_1 + \hat{a}_2^\dag\hat{a}_2 \rangle = 2\sinh^2(r)$.
Since $\ketbra{r}_{\rm EPR}$ can be obtained by 1--mode squeezed states via beam-splitters,
\begin{equation}
    V_{\rm{2sq}}(r) = V_{\rm{bs}}(\pi/4) \big( V_{\rm{ps}}(\pi/2) \oplus I \big) \big( V_{\rm{1sq}}(r) \oplus V_{\rm{1sq}}(r) \big) \big( V_{\rm{ps}}(\pi/2) \oplus I \big)^{\rm T} V_{\rm{bs}}(\pi/4)^{\rm T} \,,
\end{equation}
we can confirm the sub-additivity of the relative entropy of asymmetry,
\begin{equation}
    \Gamma(\ketbra{r}_{\rm EPR}) = \Gamma(\ketbra{r} \otimes \ketbra{r}) < 2\Gamma(\ketbra{r}) \,,
\end{equation}
where the equality follows for $q_1 = q_2$ ensuring that the beam-splitter operations are phase-covariant, and the inequality follows from Eq.(\ref{eq:asym_1squeezed}) and Eq.(\ref{eq:asym_2squeezed}).

\subsection{Gaussian asymmetry monotone based on the Petz-R{\'e}nyi divergence}
\label{app:PRentropy}

Consider a state $\hat{\rho} = e^{-\beta H}/Z_{\sigma}$ as the thermal state of second-order Hamiltonian $\hat{H}$ at temperature $\beta^{-1}$, where we denote its displacement vector and covariance matrix by $\bmd$ and $\sigma$, respectively, and its partition function by
\begin{equation}
    Z_{\sigma} \coloneqq \tr[e^{-\beta H}] = \sqrt{\det[(\sigma + i\Omega)/2]} \,.
\end{equation}
Then, the state $\hat{\rho}^\alpha/\tr[\hat{\rho}^\alpha]$ with $\alpha > 0$, i.e. the thermal state of Hamiltonian $\hat{H}$ at temperature $(\alpha\beta)^{-1}$, has first moment $\bmd$ and covariance matrix given by
\begin{equation}\label{eq:sigma_a}
    \sigma_{\alpha}i\Omega \coloneqq \coth(\alpha\ {\rm coth}^{-1}(\sigma i\Omega)) \,,\quad \text{with } \sigma_1 \equiv \sigma \,.
\end{equation}

For Gaussian states $\hat{\rho}$, with displacement vector $\bmd$ and covariance matrix $\sigma$, and $\hat{\rho}'$, with displacement vector $\bmd'$ and covariance matrix $\sigma'$, the Petz-R{\'e}nyi $\alpha$--divergence, with $\alpha \in (0,1)$, has been calculated~\cite{seshadreesan2018renyi} explicitly in terms of the moments,
\begin{align}
    D_\alpha(\hat{\rho} \| \hat{\rho}') &= \frac{1}{\alpha-1} \log Q_\alpha(\hat{\rho} \| \hat{\rho}') \,,\quad \text{where} \\
    Q_\alpha(\hat{\rho} \| \hat{\rho}') &= \frac{Z_{\sigma_\alpha}Z_{\sigma'_{1-\alpha}}}{Z_{\sigma}^\alpha Z_{\sigma'}^{1-\alpha}\sqrt{\rule{0pt}{2ex}\det[(\sigma_\alpha + \sigma'_{1-\alpha})/2]}} \exp(-(\bmd-\bmd')^{\rm T}(\sigma_\alpha + \sigma'_{1-\alpha})^{-1}(\bmd-\bmd')) \,. \label{eq:quasirelent}
\end{align}
If $\hat{\rho}' = \hat{S}\hat{\rho}\hat{S}^\dag$ for some symplectic Gaussian unitary $\hat{S}$, then $\bmd' = S\bmd$ and $\sigma' = S\sigma S^{\rm T}$.
Therefore, the exponent in Eq.(\ref{eq:quasirelent}) becomes $\bmd^{\rm T}(I - S)^{\rm T}\left(\sigma_\alpha + S\sigma_{1-\alpha}S^{\rm T}\right)^{-1}(I - S)\bmd$, giving the expression of $f_{\alpha,g}^{(1)}(\hat{\rho})$ in Eq.(\ref{eq:f1}).

To obtain the expression of $f_{\alpha,g}^{(2)}(\hat{\rho})$ in Eq.(\ref{eq:f2}), first note that $Z_{S\sigma S^{\rm T}} = Z_{\sigma}$ because $\det[S] = 1$ for symplectic $S$, so
\begin{equation}
    Z_{\sigma}^\alpha Z_{S\sigma S^{\rm T}}^{1-\alpha} = Z_{\sigma} \,.
\end{equation}
Then, we can calculate
\begin{align}
    \frac{Z_{\sigma_\alpha}Z_{\sigma_{1-\alpha}}}{Z_{\sigma}} 
    &= \sqrt{\det[(\sigma_\alpha + i\Omega)(\sigma + i\Omega)^{-1}(\sigma_{1-\alpha} + i\Omega)/2]}
    = \sqrt{\det[(I + \sigma_\alpha i\Omega)(I + \sigma i\Omega)^{-1}(I + \sigma_{1-\alpha} i\Omega) i\Omega/2]} \nonumber\\
    &= \sqrt{\det[(\sigma_\alpha + \sigma_{1-\alpha})/2]} \,, \label{eq:Zratio}
\end{align}
where the first step follows by properties of the determinant, and in the third step we used Eq.(\ref{eq:sigma_a}) and the well-known trigonometric formula
\begin{equation}
    \frac{\Big(1+\coth(a\coth^{-1}x)\Big)\Big(1+\coth((1-a)\coth^{-1}x)\Big)}{1+x} = \coth(a\coth^{-1}x) + \coth((1-a)\coth^{-1}x) \,,
\end{equation}
which holds whenever $|x| > 1$.
Inserting Eq.(\ref{eq:Zratio}) into Eq.(\ref{eq:quasirelent}), we can identify the expression of $f_{\alpha,g}^{(2)}(\hat{\rho})$ in Eq.(\ref{eq:f2}) as minus the logarithm of the fraction in front of the exponential in Eq.(\ref{eq:quasirelent}).

\subsection{Second derivative of Gaussian asymmetry monotone}
\label{app:smallt}

We derive the expressions for $F_Q^{(1)}, F_Q^{(2)}$ in Eqs.(\ref{eq:FQ1},\ref{eq:FQ2}), when the group is generated by the family of Gaussian unitaries $\hat{S}(t) = e^{i\widehat{Q}t}$ with associated symplectic matrices $S(t) = e^{-\Omega Qt}$ for $t \in \mathbb{R}$.

We first focus on the first moment quantity, and obtain the zeroth, first, and second derivatives of $f_{1/2,t}^{(1)}(\hat{\rho}) = \bmd^{\rm T}(I-S(t))^{\rm T} \left(\sigma_{1/2} + S(t)\sigma_{1/2}S(t)^{\rm T}\right)^{-1}(I-S(t))\bmd$ at $t=0$,
\begin{align}
    f_{1/2,0}^{(1)}(\hat{\rho}) 
    &= 0 \,,\\
    \frac{d}{dt}\left. f_{1/2,t}^{(1)}(\hat{\rho}) \right|_{t=0} 
    &= 0 \,,\\
    \frac{d^2}{dt^2} \left. f_{1/2,t}^{(1)}(\hat{\rho}) \right|_{t=0} &= -\bmd^{\rm T} Q\Omega \sigma^{-1} \Omega Q \bmd \eqqcolon F_Q^{(1)}(\hat{\rho}) \,,
\end{align}
where the zeroth and first derivatives are 0 because all terms contain a factor of $\left. (I-S(t))\right|_{t=0} = 0$.
Then, the quantity $f_{1/2}^{(1)}$ takes the form,
\begin{align}
    f_{1/2,t}^{(1)}(\hat{\rho}) = \sum_{k=0}^2 \frac{t^k}{k!}\left.\frac{d^k f_{1/2,t}^{(1)}(\hat{\rho})}{dt^k}\right|_{t=0} + O(t^3) = \frac{t^2}{2} F_Q^{(1)}(\hat{\rho}) + O(t^3)\,.
\end{align}

We now turn to the type--2 monotone. 
Let $A(t) \coloneqq \left(\sigma_{1/2} + e^{-\Omega Qt}\sigma_{1/2}e^{Q \Omega t}\right)/2$.
We find that
\begin{align}
    A(0) &= \sigma_{1/2} \,,\\
    \frac{d}{dt} \left. A(t) \right|_{t=0} &= \frac{1}{2} \left( -\Omega Q \sigma_{1/2} + \sigma_{1/2} Q \Omega \right) \,,\\
    \frac{d^2}{dt^2} \left. A(t) \right|_{t=0} &= \frac{1}{2} \left( \Omega Q \Omega Q \sigma_{1/2} - 2 \Omega Q \sigma_{1/2} Q \Omega + \sigma_{1/2} Q \Omega Q \Omega \right) \,.
\end{align}
Using Jacobi's derivative formulas,
\begin{align}
    \frac{d}{dt} \log\det[A] = \tr\left[A^{-1} \frac{d}{dt} A\right] \,,\quad\quad
    \frac{d^2}{dt^2} \log\det[A] = \tr\left[A^{-1}\frac{d^2A}{dt^2} \right]-\tr\left[\left(A^{-1}\frac{dA}{dt} \right)^2\right] \,,
\end{align}
we can obtain the zeroth, first, and second derivatives of $f_{1/2,t}^{(2)}(\hat{\rho}) \coloneqq -\frac{1}{2}\log\left(\det[\sigma_{1/2}]\right) +\frac{1}{2}\log\left( \det[A(t)] \right)$ at $t=0$,
\begin{align}
    f_{1/2,0}^{(2)}(\hat{\rho}) 
    &= 0 \,,\\
    \frac{d}{dt}\left. f_{1/2,t}^{(2)}(\hat{\rho}) \right|_{t=0} 
    &= 0 \,,\\
    \frac{d^2}{dt^2} \left. f_{1/2,t}^{(2)}(\hat{\rho}) \right|_{t=0} &= \frac{1}{4} \tr\Big[ \Omega Q \Omega Q - \sigma_{1/2}^{-1} \Omega Q \sigma_{1/2} Q \Omega \Big] \eqqcolon F_Q^{(2)}(\hat{\rho}) \,.
\end{align}
Then, the monotone $f_{1/2}^{(2)}$ takes the form,
\begin{align}
    f_{1/2,t}^{(2)}(\hat{\rho}) = \sum_{k=0}^2 \frac{t^k}{k!}\left.\frac{d^k f_{1/2,t}^{(2)}(\hat{\rho})}{dt^k}\right|_{t=0} + O(t^3) = \frac{t^2}{2} F_Q^{(2)}(\hat{\rho}) + O(t^3)\,.
\end{align}

\section{Conservation laws in closed Gaussian dynamics}
\label{app:conserve}

\subsection{Invariant normal mode decomposition in the presence of a passive conserved charge}
\label{app:nmd_invariant}

Here, we prove Theorem~\ref{thm:nmd_q}, an extension of Williamson's theorem in the presence of conserved charges.
\thmnmdQ*
\noindent Note that the conditions $[\Omega M,\Omega Q] = [Q, \Omega] = 0$ are equivalent to $[M,\Omega Q] = [Q, \Omega] = 0$.
The theorem is a corollary of the following lemma, which is of independent interest.
\begin{lemma}\label{lem:anti-sym_canonical}
    Let $B$ be a $(2n \times 2n)$ real symmetric positive-definite matrix such that
    \begin{equation}
        [B, \Omega_nQ] = 0 \,,
    \end{equation}
    where $Q = \bigoplus_{j=1}^n q_j I_2$, $q_1, \dots, q_n \in \mathbb{R}$.
    Let $A$ be defined by
    \begin{equation}
        A \coloneqq B \Omega_n B \,.
    \end{equation}
    There exists a $(2n \times 2n)$ orthogonal matrix $O$ such that
    \begin{equation}\label{eq:oaoT}
        O A O^{\rm T} = \bigoplus_{j=1}^n a_j \Omega_1 \,,\quad\quad [O, \Omega_nQ] = 0 \,,
    \end{equation}
    for real positive $a_1, \dots, a_n$.
\end{lemma}
Using this lemma, we now prove Theorem~\ref{thm:nmd_q} via an argument similar to the proof of the standard Williamson's theorem, as presented in Ref.~\cite{serafini2020gaussian}.
\begin{proof}(Theorem~\ref{thm:nmd_q})
    For clarity, we use $\Omega_n$ to denote symplectic form on $n$ modes.

    The direct statement follows because, for any symplectic matrix $S$,
    \begin{equation}
        \left[\Omega_n M, \Omega_n Q\right] = S^{\rm T}S^{- \rm T}\left[\Omega_n M, \Omega_n Q\right]S^{\rm T}S^{-\rm T} = S^{\rm T}\left[\Omega_n SMS^{\rm T}, \Omega_n SQS^{\rm T}\right]S^{-\rm T} = 0 \,.
    \end{equation}
    
    To prove the converse, first off, assume that $Q = \bigoplus_{j=1}^n q_j I_2$.
    Matrix $M^{-1/2}$ is real symmetric and well-defined because $M$ is positive-definite.
    Consider the anti-symmetric matrix
    \begin{equation}
        A \coloneqq M^{-1/2} \Omega_n M^{-1/2} \,.
    \end{equation}
    In Lemma~\ref{lem:anti-sym_canonical}, we show that for any matrix $M$ satisfying the assumptions of the lemma, $A = M^{-1/2} \Omega_n M^{-1/2}$ can be block-diagonalized using an orthogonal transformation that commutes with $Q\Omega_n$, i.e. there exists a $(2n \times 2n)$ orthogonal matrix $O$ such that
    \begin{align}
        O A O^{\rm T} &= \bigoplus_{j=1}^n \nu_j^{-1} \Omega_1 \,, \label{eq:oaot}\\
        [O, Q\Omega_n] &= 0 \,,\label{eq:oq0}
    \end{align}
    for some real positive $\nu_1, \dots, \nu_n$.
    Let $D \coloneqq \bigoplus_{j=1}^n \nu_j I_2$.
    Then, the matrix $S \coloneqq D^{1/2} O M^{-1/2}$ is symplectic, because
    \begin{align}
        S\Omega_nS^{\rm T} = D^{1/2} \left( \bigoplus_{j=1}^n \nu_j^{-1} \Omega_1 \right) D^{-1/2} = \Omega_n \,,
    \end{align}
    where we used Eq.(\ref{eq:oaot}).
    Furthermore, it brings $M$ to the canonical form in Eq.(\ref{eq:nmd_condition1}),
    \begin{align}
        SMS^{\rm T} = D = \bigoplus_{j=1}^n \nu_j I_2 \,,
    \end{align}
    and satisfies the invariance condition of Eq.(\ref{eq:nmd_condition2}).
    This follows from the fact that $O$, $D$, and $M$ all commute with $Q\Omega_n = \Omega_nQ$, which in turn implies $S \coloneqq D^{1/2} O M^{-1/2}$ commutes with $\Omega_n Q = Q \Omega_n$, and so
    \begin{equation}
        S Q S^{\rm T} = S Q \Omega_n S^{-1} \Omega_n^{\rm T}=Q \Omega_n \Omega_n^{\rm T} = Q \,, 
    \end{equation}
    where the first identity follows from $S \Omega_n S^{\rm T} = \Omega_n$, and the second from the fact that $S$ commutes with $Q\Omega_n$.

    This proves the claim for the special case $Q=\bigoplus_{j=1}^n q_j I_2$. Next, we generalize the result to arbitrary symmetric $Q$ that commutes with $\Omega_n$.
    Let $Q$ be any $(2n \times 2n)$ real symmetric matrix such that $\left[Q, \Omega_n \right] = 0$, or, equivalently, $\left[Q\Omega_n, \Omega_n \right] = 0$.
    Using Lemma~\ref{lem:anti-sym_canonical} for the special case of $B = I_{2n}$ and thus $A = \Omega_n$, we find a $(2n \times 2n)$ orthogonal matrix $O$ such that
    \begin{equation}
        O \Omega_n O^{\rm T} = \Omega_n \,,\quad\quad O Q\Omega_n O^{\rm T} = \bigoplus_{j=1}^n q_j \Omega_1 \,,
    \end{equation}
    which, in particular, implies that $O$ is symplectic.
    
    Then, condition $[M,Q\Omega_n]=0$ is equivalent to $[\widetilde{M},\widetilde{Q}\Omega_n]=0$, where $\widetilde{Q} \coloneqq \bigoplus_{j=1}^n q_j I_2$ is block-diagonal and $\widetilde{M} \coloneqq OMO^{\rm T}$.
    Hence, according to the above argument, there exists a symplectic matrix $\widetilde{S}$ such that
    \begin{align}
        \widetilde{S} \widetilde{M} \widetilde{S}^{\rm T} = \bigoplus_{j=1}^n \nu_j I_2 \,,\quad\quad
        \quad\quad \widetilde{S} \widetilde{Q} \widetilde{S}^{\rm T} = \widetilde{Q} \,,
    \end{align}
    and choosing $S = \widetilde{S} O$ concludes the proof.
\end{proof}

The charge operator $\widehat{Q} = \sum_j q_j \hat{a}_j^\dag\hat{a}_j$ partitions the modes into sectors of constant charge $q$ in the sense that a strictly positive purely quadratic Hamiltonian $\hat{H}$ commuting with $\widehat{Q}$ can be decomposed, due to Theorem~\ref{thm:nmd_q}, as
\begin{equation}
    \hat{H}=\sum_{q\in \text{Eig}(Q)} \sum_{\alpha}\omega_{q,\alpha}\ \hat{b}^\dag_{q,\alpha}\hat{b}_{q,\alpha} \,,
\end{equation}
for some frequencies $\omega_{q,\alpha} > 0$ with a multiplicity index $\alpha$, where each creation operator is expressed in terms of the original operators $\{\hat{a}_j\}$ as
\begin{align}
    \hat{b}_{|q|,\alpha}^{\dag} &= \sum_{j: q_j=|q|} r^+_{j,\alpha} \hat{a}^\dag_j+\sum_{j: q_j=-|q|} s^+_{j,\alpha} \hat{a}_j \,,\\
    \hat{b}_{-|q|,\alpha}^{\dag} &= \sum_{j: q_j=|q|} r^-_{j,\alpha} \hat{a}_j+\sum_{j: q_j=-|q|} s^-_{j,\alpha} \hat{a}^\dag_j
\end{align}
for some constants $r^{\pm}_{j,\alpha}, s^{\pm}_{j,\alpha}$.
In particular, $\hat{b}_q^\dag$ contains terms $a^\dag_j$ when $q_j$ and $q$ have the same sign, and terms $a_j$ when $q_j$ and $q$ have opposite sign.
Therefore, $\hat{b}_q^{\dag}$ shifts the charge operator $\widehat{Q}$ by $q \hat{I}$,
\begin{equation}
    [\widehat{Q}, \hat{b}^\dag_{q,\alpha}] = q\hat{b}_{q,\alpha}^\dag \,,
\end{equation}
so that $\hat{b}^\dag_{q,\alpha}\hat{b}_{q,\alpha}$ stays within the sector of charge $q$.

Theorem~\ref{thm:nmd_q} also implies a symmetric decoupling for symplectic matrices. 
Suppose symplectic matrix $V = e^{H\Omega_n}$ is generated by a Hamiltonian matrix $H$ such that the assumption in Eq.(\ref{eq:nmd_comm}) holds.
Then, there exists a symplectic matrix $S$ which simultaneously diagonalizes $Q$ by congruence, $SQS^{\rm T} =  \bigoplus_{j=1}^n q_j I_2$ and $V$ by similarity,
\begin{equation}
    SVS^{-1} = e^{SHS^{\rm T} \Omega_n} = \bigoplus_{j=1}^n e^{\nu_j \Omega_1} \,,\quad \nu_j \in \mathbb{R} \,.
\end{equation}

\stoptoc

\subsection*{Proof of Lemma~\ref{lem:anti-sym_canonical}}

\resumetoc

\begin{proof}(Lemma~\ref{lem:anti-sym_canonical})
    First, we note that since $B$ is symmetric, $A$ is anti-symmetric.
    For any $q>0$, let
    \begin{equation}
        Q_q = \bigoplus_{j: |q_j|=q} q_j \Omega_1
    \end{equation}
    be the restriction of $\Omega_nQ$ to the subspace $\{j: |q_j| = q\}$. 
    Note that $Q_q$ can be written as a polynomial of $Q\Omega_n$.
    To see this, observe that $(Q\Omega_n)^2 = -\sum_{j} q_j^2 I_2$, which implies that the projector $\bigoplus_{j: |q_j|=q} I_2$ can be expressed as a polynomial of $(Q\Omega_n)^2$.
    Then, $Q_q = Q\Omega_n \bigoplus_{j: |q_j|=q} I_2$  can also be expressed as a polynomial of $Q\Omega_n$.
    
    Since $\Omega_nQ$ and $A$ commute and are both anti-symmetric, according to Lemma~\ref{lem:anti-sym_commute}, there exists an orthogonal matrix $O$ such that
    \begin{equation}\label{sar}
        O A O^{\rm T} = \bigoplus_{j=1}^n a_j \Omega_1 \,,\quad\quad [O, \Omega_nQ] = 0 \,,
    \end{equation}
    where $a_j$ are not necessarily positive.
    Then, 
    \begin{equation}
        \left( O A O^{\rm T} \right) Q_q = \left( \bigoplus_{j: |q_j|=q} a_j \Omega_1 \right) \bigoplus_{j: |q_j|=q} q_j \Omega_1 = -\bigoplus_{j: |q_j|=q} a_j q_j I_2 \, .
    \end{equation}
    Furthermore, since $Q_q$ can be expressed as a polynomial of $\Omega_n Q=Q\Omega_n$, and $O$ commutes with $\Omega_nQ$, it also commutes with $Q_q$ which, in turn, implies that $O A O^{\rm T} Q_q = O A Q_q O^{\rm T}$.
    
    Hence, the eigenvalues of $O A O^{\rm T} Q_q = O A Q_q O^{\rm T}$ are equal to the eigenvalues of $A Q_q$. 
    We conclude that the non-zero eigenvalues of $-A Q_q$ are exactly the multi-set $$\{a_j q_j : |q_j|=q\} \,,$$ with an additional multiplicity 2 for all elements.  
    
    Next, we show that the number of positive and negative eigenvalues of $-A Q_q$ are indeed equal to the number of positive and negative elements of the multi-set $$\{q_j : |q_j|=q\}$$ with an additional multiplicity 2.
    
    To show this, we note that, by assumption, $B$ commutes with $\Omega_n Q$, and therefore it also commutes with $Q_q$, which in turn implies that $ A \coloneqq B \Omega_n B$ satisfies
    \begin{align}
        A Q_q 
        = B \Omega_n B Q_q=B \Omega_n Q_q B  
        = -B \left(\bigoplus_{j: |q_j|=q} q_j I_2 \right) B \,.
    \end{align}
    
    Since $B$ is positive, it is symmetric and full-rank. 
    Then, applying Sylvester's law of inertia, the number of positive and negative eigenvalues of $-AQ_q$ is equal to the number of positive and negative eigenvalues of $\bigoplus_{j: |q_j|=q} q_j I_2$, which are twice the number of positive and negative elements of multi-set $\{q_j: |q_j|=q\}$.  
    
    In summary, by looking at the eigenvalues of $-A Q_q$, we found that  the number of positive and negative elements of multi-set $\{q_j: |q_j|=q\}$ is equal to the number of positive and negative elements of multi-set $\{a_j q_j: |q_j|=q\}$. 
    This means that, for all $j$ appearing in this sector, either $a_j$ are all positive, or there exists an even number of $j$ for which $a_j$ is negative. 
    Furthermore, for exactly half of such $j$, denoted by $S_+$, $q_j$ is positive, and for the other half, denoted by $S_-$, $q_j$ is negative. 
    
    Based on this observation, we can  construct another orthogonal transformation $\widetilde{O}$, which ensures that $a_j$ are positive. Recall that the projector $\bigoplus_{j: |q_j|=q} I_2$ can be expressed as a polynomial of $(Q\Omega_n)^2$. Since $O$ commutes with $Q\Omega_n$, it also commutes with this projector and therefore it is block-diagonal with respect to subspaces $\{j: |q_j|=q\}$.
    Let $O_q$ be the component of $O$ in the subspace $\{j: |q_j|=q\}$. 
    Then, define
    \begin{equation}
        \widetilde{O}_q = \text{SWAP}_{\pm} \left( \bigoplus_{j \in S_+ \sqcup S_-} X(j) \right) O_q \eqqcolon W_q O_q \,,
    \end{equation}
    where $X(j)$ is the Pauli $X$ operator on mode $j$, and $\text{SWAP}_{\pm}$, exchanges the modes in $S_+$ and $S_-$ (recall that they have equal size).
    With this definition, we obtain
    \begin{align}
        W_q Q_q W_q^{\rm T}
        &= \bigoplus_{\substack{j: |q_j|=q, \\ q\notin S_+ \sqcup S_-}} q_j \Omega_1 \oplus \left(\text{SWAP}_{\pm} \left(\bigoplus_{j\in S_+ \sqcup S_-}  q_j X(j) \Omega_1 X(j)\right) \text{SWAP}_{\pm}\right) \nonumber\\
        &= \bigoplus_{\substack{j: |q_j|=q, \\ q\notin S_+ \sqcup S_-}} q_j \Omega_1 \oplus (-1)\left(\text{SWAP}_{\pm} \left(\bigoplus_{j\in S_+ \sqcup S_-} q_j  \Omega_1 \right) \text{SWAP}_{\pm}\right) \nonumber\\
        &= \bigoplus_{\substack{j: |q_j|=q, \\ q\notin S_+ \sqcup S_-}} q_j \Omega_1 \oplus \bigoplus_{j\in S_+ \sqcup S_-}  q_j  \Omega_1 =Q_q \,,
    \end{align}
    where the second equality follows from the fact that $X(j) q_j\Omega_1 X(j)=-q_j\Omega_1$.  
    
    Then, the orthogonal transformation $\widetilde{O} = \bigoplus_{q\ge 0} \widetilde{O}_q = W O$, where $W \coloneqq \bigoplus_q W_q$ commutes with $Q\Omega_n$, i.e.
    \begin{equation}
        \widetilde{O} Q\Omega \widetilde{O}^{\rm T} = W O Q\Omega O^{\rm T} W^{\rm T} = W Q\Omega W^{\rm T}=\bigoplus_{q \ge 0} W_q Q_q W_q^{\rm T} = \bigoplus_{q \ge 0} Q_q = Q\Omega \,.
    \end{equation}
    Furthermore, it brings $A$ to the form required in Eq.(\ref{sar}) with positive $a_j$, i.e.
    \begin{equation}
        \widetilde{O} A \widetilde{O}^{\rm T}
        = W O A O^{\rm T} W^{\rm T}
        = W \left( \bigoplus_{j=1}^n a_j \Omega_1 \right) W^{\rm T}
        = \bigoplus_{j=1}^n |a_{\pi(j)}| \Omega_1 \,,
    \end{equation}
    where $\pi$ is an unspecified permutation.
    The case of $q = 0$ corresponds to the standard canonical decomposition of an anti-symmetric matrix within that subspace.
\end{proof}

\begin{lemma}\label{lem:anti-sym_commute}
    Consider a set of real anti-symmetric matrices $\{A_1, \dots, A_N\}$ that pairwise commute.
    There exists a real orthogonal matrix $O$ that simultaneously brings them into canonical form,
    \begin{equation}\label{eq:can_A}
        O A_j O^{\rm T} = \left( \bigoplus_{k=1}^{{\rm rank}(A_j)/2} a_{jk} \Omega_1 \right) \oplus 0_{{\rm dim}(A_j)-{\rm rank}(A_j)} \,,\quad\quad \text{for all } j = 1, \dots, N \,,
    \end{equation}
    where $\{\pm ia_{jk}: k = 1, \dots, {\rm rank}(A_j)/2\}$ are the eigenvalues of $A_j$ and $0_n$ is the $(n \times n)$ zero matrix.
\end{lemma}
\begin{proof}
    For $N=1$, the statement reduces to the regular canonical decomposition of anti-symmetric matrix $A_1$.
    Consider the case of $N=2$.
    Since $A_1$ and $A_2$ commute and are anti-symmetric, matrices $A_1A_2 = A_2A_1$, $A_1^2$ and $A_2^2$ are all symmetric and pairwise commute. 
    Therefore, the algebra generated by $A_1A_2$, $A_1^2$, and  $A_2^2$ is commutative and decomposes the vector space $\V$ into orthogonal subspaces, $\V = \bigoplus_k \V_k$, such that $A_1A_2$, $A_1^2$, and $A_2^2$ are constant over each $\V_k$, and for each distinct pair of subspaces $\V_k, \V_{k'}$, at least one of $A_1A_2$, $A_1^2$, and $A_2^2$ takes different values between $\V_k$ and $\V_{k'}$. 
    
    Let $\Pi_k$ be the projector to $\V_k$. 
    Since matrices $A_1A_2$, $A_1^2$, and $A_2^2$ are all symmetric, and $\Pi_k$ lives in the algebra generated by these matrices, it is also symmetric.
    
    Now since $A_1$ and $A_2$ commute with $A_1^2, A_2^2$, and $A_1A_2$, they also commute with all projectors $\Pi_k$.
    Then, 
    \begin{equation}
        \Pi_k A_1 \Pi_k = \Pi_k A_1 = A_1 \Pi_k \,,
    \end{equation}
    are anti-symmetric and commute with $\Pi_k A_2 \Pi_k$, for all $k$. 
    
    Since $A^2_1$ is constant over subspace $\V_k$, the symmetric positive operator $-\Pi_k A^2_1 \Pi_k = a_{1,k}^2 \Pi_k$ has a non-negative eigenvalue $a_{1,k}^2$. 
    Then, the eigenvalues of $A_1 \Pi_k$ are all in the form $\pm i a_{1,k}$. 
    The fact that $\Pi_k A_1 \Pi_k$ is anti-symmetric implies that there exists an orthogonal transformation on $\V_k$ that brings $\Pi_k A_1 \Pi_k$ to block-diagonal form 
    \begin{equation}\label{eq:a1}
        O(\Pi_k A_1 \Pi_k)O^{\rm T}= i a_{1,k} \Pi_k \bigoplus_j c_j \Omega_1 \,,
    \end{equation}
    where $c_j \in \{-1,1\}$.
    
    Now recall that $O(\Pi_k A_2 \Pi_k)O^{\rm T}$ is also anti-symmetric and restricted to $\V_k$. 
    Furthermore,
    \begin{equation}
        \Pi_k A_1 A_2 \Pi_k = b \Pi_k \,, 
    \end{equation}
    for some real number $b$.
    This implies that
    \begin{equation}\label{eq:a2}
        O (\Pi_k A_1 \Pi_k) O^{\rm T} O (\Pi_k A_2 \Pi_k) O^{\rm T} = O (\Pi_k A_1 A_2 \Pi_k) O^{\rm T} =  b \Pi_k \,, 
    \end{equation}
    and, consequently,
    \begin{equation}
     O(\Pi_k A_2 \Pi_k)O^{\rm T} 
     = b \Pi_k O (\Pi_k A_1 \Pi_k)^{-1} O^{\rm T}
     = i b a_{1,k}^{-1} \Pi_k \bigoplus_j c_j \Omega_1 \,,
    \end{equation}
    where the inverse is defined  on the support of $\Pi_k$. This means that the same orthogonal transformation $O$ also brings $A_2$ to block-diagonal form.
    
    Extending the proof to $N>2$ is straightforward and involves first bringing $A_1$ to its canonical form in each subspace as in Eq.(\ref{eq:a1}), then repeating the step in Eq.(\ref{eq:a2}) $N-1$ times for each matrix $A_j$, $j=2, \dots, N$.
\end{proof}

\subsection{Conditions for state interconversion in the presence of a passive conserved charge}
\label{app:interconversion}

\thminterconvert*
\begin{proof}
    On modes that carry the trivial representation, i.e. $q=0$, $\sigma_1\Pi_0$ and $\sigma_2\Pi_0$ are related by a symplectic transformation if and only if they have equal symplectic eigenvalues, i.e. Eq.(\ref{eq:Qcond1}) holds.
    Hence, we now focus on Eq.(\ref{eq:Qcond2}).

    We first prove the necessity of Eq.(\ref{eq:Qcond2}) when $|q| > 0$.
    The conditions $[M, \Omega_n Q] = [Q, \Omega_n] = 0$ imply that
    \begin{equation}
        [M, Q^2] = -[M, \Omega_n Q \Omega_n Q] = 0 \,,
    \end{equation}
    for any matrix $M$.
    Considering the decomposition
    \begin{equation}
        Q^2 = \sum_{q \in {\rm Eig}(Q)} q^2 (\Pi_q + \Pi_{-q}) \,,
    \end{equation}
    we also get $[M, \Pi_q + \Pi_{-q}] = 0$, for all $q \in {\rm Eig}(Q)$. 
    Similarly, $S Q S^{\rm T}=Q$ implies that $S$ commutes with  $Q\Omega_n$, and hence with $\Pi_q + \Pi_{-q}$.
    
    We now assume that there exists symplectic matrix $S$ such that $SQS^{\rm T} = Q$ and $S \sigma_1 S^{\rm T} = \sigma_2$.
    In particular, this implies that $[S,Q\Omega_n] = 0$, leading to
    \begin{equation}\label{eq:sigmaQOmega}
        \sigma_2Q\Omega_n
        = S\sigma_1S^{\rm T}Q\Omega_n 
        = S\sigma_1Q\Omega_n S^{\rm T} \,,
    \end{equation}
    where the second equality follows from the facts that $[Q,\Omega_n]=0$ and $Q=Q^{\rm T}$ which imply that $[S^{\rm T},Q\Omega_n] = 0$.
    Right-multiplying Eq.(\ref{eq:sigmaQOmega}) with $\Omega_n(\Pi_q + \Pi_{-q})$, we obtain
    \begin{equation}
        \sigma_2Q(\Pi_q + \Pi_{-q})
        = S\sigma_1Q(\Pi_q + \Pi_{-q})S^{-1} \,,
    \end{equation}
    where in the right-hand side we again used the fact that $S$ commutes with $Q\Omega_n$, and hence with $Q^2 = -(Q\Omega_n)^2$, which in turn implies it commutes with $\Pi_q + \Pi_{-q}$.
    Equivalently this can be written as 
    \begin{equation}
        q\sigma_2(\Pi_q - \Pi_{-q})
        = qS\sigma_1(\Pi_q - \Pi_{-q})S^{-1} \,,
    \end{equation}

    Hence, $\sigma_1(\Pi_q - \Pi_{-q})$ and $\sigma_2(\Pi_q - \Pi_{-q})$ have equal ordinary eigenvalues when $q \neq 0$, so Eq.(\ref{eq:Qcond2}) holds for all $k$.
    
    Now we proceed to show the converse.
    According to Theorem~\ref{thm:nmd_q}, there exist symplectic matrices $S_1, S_2$ such that
    \begin{align}
        S_1 \sigma_1 S_1^{\rm T} = \bigoplus_{j=1}^n \nu_j I_2 \,, \quad\quad
        S_2 \sigma_2 S_2^{\rm T} = \bigoplus_{j=1}^n \widetilde{\nu}_j I_2 \,, \quad\quad
        S_1 Q S_1^{\rm T} = S_2 Q S_2^{\rm T} = \bigoplus_{j=1}^n q_j I_2 \,, \label{eq:simul_diag}
    \end{align}
    where $\nu_j, \widetilde{\nu}_j \ge 0$ (in fact $\nu_j, \widetilde{\nu}_j \ge 1$), and
    \begin{equation}
        {\rm SympEig}(Q) \coloneqq \{q_j: 1, \dots, n\} = {\rm Eig}(Q) \,.
    \end{equation}
    The fact that the symplectic and ordinary eigenvalues of $Q$ are identical is a consequence of $[\Omega_n, Q]=0$ and can be seen, e.g. using Theorem~\ref{thm:nmd_q} with $M = I_{2n}$.
    The equations in Eq.(\ref{eq:simul_diag}) imply that
    \begin{align}
        S_1\sigma_1 \Omega S_1^{-1} = \bigoplus_j \nu_j \Omega_1 \,,\quad\quad
        S_1 Q \Omega S_1^{-1} = \bigoplus_j q_j \Omega_1 \,. \label{eq:simul_diag2}
    \end{align}
    Using the fact that $\Pi_q+\Pi_{-q}$ can be written as a polynomial of $Q^2=- (Q\Omega)^2$, the latter equation implies that
    \begin{align}
        S_1 (\Pi_q+\Pi_{-q}) S_1^{-1} = \bigoplus_{j: |q_j|=|q|} I_2 \,.
    \end{align}
    Consider $q\sigma_1 (\Pi_q-\Pi_{-q}) = \sigma_1 Q (\Pi_q+\Pi_{-q})$.
    Sandwiching this between $S_1$ and $S_1^{-1}$ we obtain 
    \begin{align}
        q S_1 \sigma_1 (\Pi_q-\Pi_{-q}) S^{-1}_{1} 
        &= S_1 \left(\sigma_1 Q (\Pi_q+\Pi_{-q})\right)S_1^{-1} 
        = -S_1 \left(\sigma_1 \Omega Q\Omega (\Pi_q+\Pi_{-q})\right)S_1^{-1} \\
        &= -\left(\bigoplus_j \nu_j \Omega_1\right) \left(\bigoplus_{j: |q_j|=q} q_j \Omega_1\right) 
        = \bigoplus_{j: |q_j|=q} \nu_j q_j I_2 \\ 
        &= |q| \bigoplus_{j: |q_j|=|q|} \nu_j \frac{q_j}{|q|} I_2 \,, 
    \end{align}
    where in the second line we used Eq.(\ref{eq:simul_diag2}) and in the third line we assumed $q \neq 0$.
    We conclude that, for $|q|>0$, $\sigma_1 (\Pi_q-\Pi_{-q})$ is diagonalizable and its positive and negative eigenvalues are, respectively, $\{\nu_j: q_j=|q|\}$ and $\{-\nu_j: q_j=-|q|\}$, with an additional multiplicity of 2, and similarly for $\sigma_2 (\Pi_q-\Pi_{-q})$. 
    Therefore, the assumption of the theorem that ${\rm Eig}(\sigma_1 (\Pi_q-\Pi_{-q})) = {\rm Eig}(\sigma_2 (\Pi_q-\Pi_{-q}))$ implies that for $|q| > 0$, 
    \begin{align}
        \{\nu_j: q_j = |q|\} = \{\widetilde{\nu}_j: q_j = |q|\} \quad\text{and}\quad
        \{-\nu_j: q_j = -|q|\} = \{-\widetilde{\nu}_j: q_j = -|q|\} \,.
    \end{align}
    Consider the set of modes $\{j: q_j=q\}$ where the diagonal matrix $S_1 Q S_1^{\rm T} = S_2 Q S_2^{\rm T}$ takes the value $q \neq 0$.
    We have shown that on these modes, $S_1 \sigma_1 S_1^{\rm T}$ and $S_2 \sigma_2 S_2^{\rm T}$ are diagonal and, up to a possible permutation of modes, have an identical set of diagonal elements. 
    Therefore, by simply reordering these modes, we can convert $S_1 \sigma_1 S_1^{\rm T}$ to $S_2 \sigma_2 S_2^{\rm T}$. 
    That is, there exists an orthogonal symplectic transformation $O=\bigoplus_q O_q$, where $O_q$ permutes modes appropriately within $\{j: q_j=q\}$ and commutes with $S_1 Q S_1^{\rm T}$, such that $O S_1 \sigma_1 S_1^{\rm T} O^{\rm T}= S_2 \sigma_2 S_2^{\rm T}$.
    This, in turn, implies that the matrix $S \coloneqq S^{-1}_2 O S_1$ is symplectic, and satisfies $\sigma_2 = S\sigma_1S^{\rm T}$ and $Q = SQS^{\rm T}$.
\end{proof}

\section{Invariant Gaussian purification}

In this section, we provide the proof of Lemma~\ref{lem:inv_gauss_purification}, which we split into two parts, first proving the existence of an invariant Gaussian purification, and then showing that all such purifications are equivalent up to application of a Gaussian unitary on the ancillary system.

\subsection{A standard form for invariant Gaussian purifications}
\label{app:purification}

Here we provide the proof of the invariant Gaussian purification as stated in Lemma~\ref{lem:inv_gauss_purification}, which is a slight modification of the original Gaussian purification by Holevo and Werner~\cite{holevo2001evaluating}.
\begin{lemma}\label{lem:inv_gauss_purification1}
    Consider a system with a finite number of bosonic modes and Hilbert space $\H_A$ equipped with symplectic Gaussian representation $\hat{S}_A(g): g \in G$ of a group $G$ on $\H_A$.
    
    For any Gaussian state $\hat{\rho}_A \in \B(\H_A)$ which is invariant under $\hat{S}_A$, i.e. $\hat{S}_A(g)\hat{\rho}_A\hat{S}_A(g)^\dag = \hat{\rho}_A$ for all $g \in G$, there exists a Gaussian pure state $\ket{\psi}_{A\overbar{A}} \in \H_A \otimes \H_{\overbar{A}}$ where ancilla $\overbar{A}$ has Hilbert space $\H_{\overbar{A}} \cong \H_A$ and is equipped with representation $\hat{S}_{\overbar{A}}(g) \coloneqq \hat{S}^*_{A}(g): g\in G$, such that $\ket{\psi}_{A\overbar{A}}$ is a purification of $\hat{\rho}_A$,
    \begin{equation}
        \hat{\rho}_A = \tr_{\overbar{A}}[\ketbra{\psi}{\psi}_{A\overbar{A}}] \,,
    \end{equation}
    and is invariant, i.e. $\ket{\psi}_{A\overbar{A}} = (\hat{S}_A(g) \otimes \hat{S}_{\overbar{A}}(g))\ket{\psi}_{A\overbar{A}}$ for all $g \in G$.

    Assuming $\bmd$ and $\sigma$ are the displacement vector and covariance matrix of state $\hat{\rho}_A$, respectively, then, $\ket{\psi}$ can be chosen to have displacement vector $\bmd_\psi = \bmd \oplus \bmo$, and covariance matrix given in basis $\hat{\bmr} = \hat{\bmr}_A \oplus \hat{\bmr}_{\overbar{A}}$ by
    \begin{equation}\label{eq:purif_cov_app}
        \sigma_{\psi} = \begin{pmatrix} \sigma &\hspace{-8pt} -\sqrt{-\sigma\Omega\sigma\Omega - I}\ \Omega Z_A \\[5pt] Z_A \Omega \sqrt{-\Omega\sigma\Omega\sigma - I} &\hspace{-8pt} Z_A \sigma Z_A \end{pmatrix} \,,
    \end{equation}
    where $\bmo$ is the zero vector on $\overbar{A}$, $Z_A \coloneqq (I_A \otimes Z)$, and $\Omega$ is the symplectic form on $A$.
\end{lemma}

\begin{proof}
    The displacement vector $\bmd_\psi$ is invariant under any choice of $S_{\overbar{A}}$ provided that $\bmd$ is invariant under $S_A$.

    Regarding the covariance matrix, we first note that $\sigma_{\psi}$ is a well-defined covariance matrix, for any covariance matrix $\sigma$ on system $A$.
    Write $\Sigma \coloneqq -\sigma\Omega\sigma\Omega - I$, and note that $\Sigma \ge 0$ according to Ineq.(\ref{eq:swsw}).
    Then, the square root of $\Sigma$ and its transpose $\sqrt{\Sigma^{\rm T}} = \sqrt{\Sigma}^{\rm T}$ are well-defined.
    Therefore, $\sigma_{\psi}$ is a well-defined, symmetric matrix.
    
    The Schur complement $\sigma_\psi/\sigma$ of the top left block $\sigma > 0$ is
    \begin{align}
        \sigma_\psi/\sigma 
        = -Z_A \sigma Z_A - \left(Z_A\Omega\sqrt{\Sigma^{\rm T}}\right)\sigma^{-1}\left(-\sqrt{\Sigma}\Omega Z_A\right)
        = (Z_A\Omega)\sigma^{-1} (Z_A\Omega)^{\rm T}
        > 0 \,.
    \end{align} 
    This implies that $\sigma_\psi > 0$ by the Schur complement condition~\cite{johnson1985matrix}, and $\det[\sigma_\psi] = \det[\sigma]\ \det\left[(Z_A\Omega) \sigma^{-1} (Z_A\Omega)^{\rm T}\right] = 1$.

    Moreover, all symplectic eigenvalues of $\sigma_\psi$ are exactly equal to 1, i.e. $- \sigma_\psi (\Omega \oplus \Omega) \sigma_\psi (\Omega \oplus \Omega) - I = 0$, which can by shown be explicit calculation,
    \begin{align}
        - \sigma_\psi (\Omega \oplus \Omega) \sigma_\psi (\Omega \oplus \Omega) - I
        = - \left( \begin{pmatrix} \sigma & -\sqrt{-\sigma\Omega\sigma\Omega - I}\ \Omega Z_A \\[5pt] Z_A \Omega \sqrt{-\Omega\sigma\Omega\sigma - I} & Z_A \sigma Z_A \end{pmatrix} \begin{pmatrix} \Omega & 0 \\[5pt] 0 & \Omega \end{pmatrix} \right)^2
        - \begin{pmatrix} I  & 0 \\[5pt] 0 & I \end{pmatrix} = 0 \,.
    \end{align}
    Therefore, $\sigma_\psi$ is the covariance matrix of a pure state and it satisfies the uncertainty relation, $\sigma_\psi + i(\Omega \oplus \Omega)\ge 0$.

    The invariance of $\sigma_\psi$ under representation $g \mapsto S_A(g) \oplus S_{\overbar{A}}(g) = S_A(g) \oplus \widetilde{S}_A(g)$ follows directly,
    \begin{align}
        (S_A(g) \oplus \widetilde{S}_A(g)) \hspace{-2pt}\begin{pmatrix} \sigma &\hspace{-20pt} -\sqrt{-\sigma\Omega\sigma\Omega - I}\ \Omega Z_A \\[5pt] Z_A \Omega \sqrt{-\Omega\sigma\Omega\sigma - I} &\hspace{-20pt} Z_A \sigma Z_A \end{pmatrix}\hspace{-2pt} (S_A(g) \oplus \widetilde{S}_A(g))^{\rm T}
        \hspace{-2pt}=\hspace{-2pt} \begin{pmatrix} \sigma &\hspace{-20pt} -\sqrt{-\sigma\Omega\sigma\Omega - I}\ \Omega Z_A \\[5pt] Z_A \Omega \sqrt{-\Omega\sigma\Omega\sigma - I} &\hspace{-20pt} Z_A \sigma Z_A \end{pmatrix}\hspace{-2pt} \,,
    \end{align}
    where we have made use of the invariance of $\sigma$ under $S_A$, the fact that $S_A(g)$ is symplectic, and the definition of the time-reversed representation, $Z_AS_A(g)Z_A$ .
\end{proof}

In Fig.~\ref{fig:purification} of the main text, we illustrate a procedure that prepares the purification in Eq.(\ref{eq:purif_cov_app}), which we describe here for the case of a passive representation.
As usual, performing a symplectic basis change generalizes the result to any squeezed passive representation.
Assume system $A$ has $n$ modes, and consider the normal mode decomposition of the covariance matrix of $\hat{\rho}$,
\begin{equation}\label{eq:nmd}
    \sigma \overset{\bma}{=} V \begin{pmatrix} 0 & D \\ D & 0 \end{pmatrix} V^{\rm T} \,,
\end{equation}
for some symplectic matrix $V$ and diagonal $D \ge I \equiv I_n$.
\begin{enumerate}
    \item[(i)] Prepare the vacuum on both subsystems $A$ and $\overbar{A}$, whose covariance matrix is expressed in basis $\hat{\bma}_A \oplus \hat{\bma}_{\overbar{A}}$ as
    \begin{equation}
        \sigma_0 \overset{\bma}{=} \begin{pmatrix} 0 & I & 0 & 0 \\ I & 0 & 0 & 0 \\ 0 & 0 & 0 & I \\ 0 & 0 & I & 0 \end{pmatrix} \,.
    \end{equation}
    \item[(ii)] Apply the symplectic transformation $V_{\rm sq}(r)$ from Proposition~\ref{prop:entangling} on $\sigma_0$ with
    \begin{equation}
        r = \frac{1}{2}\cosh^{-1}{D}\,,
    \end{equation}
    which is well-defined for $D \ge I$.
    $V_{\rm sq}(r)$ represents $n$ 2--mode squeezers coupling each mode in $A$ with a mode in $\overbar{A}$, leading to the covariance matrix
    \begin{equation}
        \sigma_1 
        = V_{\rm sq}(r) \sigma_0 V_{\rm sq}(r)^{\rm T}
        \overset{\bma}{=} \begin{pmatrix} 0 & D & \sqrt{D^2 - I} & 0 \\ D & 0 & 0 & \sqrt{D^2 - I} \\ \sqrt{D^2 - I} & 0 & 0 & D \\ 0 & \sqrt{D^2 - I} & D & 0 \end{pmatrix} \,.
    \end{equation}
    \item[(iii)] Apply the symplectic transformation $V \oplus V\Omega$ on $\sigma_1$, to obtain
    \begin{equation}
        \sigma_2 
        = (V \oplus V\Omega) \sigma_1 (V \oplus V\Omega)^{\rm T} = \begin{pmatrix} \sigma & -\sqrt{-\sigma\Omega\sigma\Omega - I}\ \Omega \\[5pt] \Omega\sqrt{-\Omega\sigma\Omega\sigma - I} & \sigma \end{pmatrix} \,,
    \end{equation}
    where we have used that $\sqrt{-\sigma\Omega\sigma\Omega - I} = V \big( \sqrt{D^2 - I} \oplus \sqrt{D^2 - I} \big) V^{-1}$.
\end{enumerate}
Finally, we apply $Z_A$ on the second subsystem to restore the symplectic form $\Omega \oplus \Omega$, resulting in the purification of Eq.(\ref{eq:purif_cov_app}).
We can obtain any other $G$--invariant purification of $\hat{\rho}$ by acting on $\overbar{A}$ with a $G$--invariant Gaussian unitary.

\subsection{Equivalence of Gaussian purifications up to local Gaussian operations}
\label{app:purif_equiv}

Here, we prove that all $G$--invariant Gaussian purifications of a $G$--invariant Gaussian state are equivalent up to a Gaussian unitary on the environment. 
In the absence of symmetry constraints, this result has been shown previously~\cite{botero2003modewise}.
Since our proof relies on this special case, we include it here.
\begin{lemma}\label{lem:inv_gauss_purification0}
    Consider a pair of systems $A$ and $B$ with finite numbers of bosonic modes and Hilbert spaces $\H_A$ and $\H_B$.
    Suppose two pure Gaussian states $\ket{\psi_1}, \ket{\psi_2} \in \H_A \otimes \H_B$ have the same reduced state on 
    $A$, i.e. $\tr_B[\ketbra{\psi_1}_{AB}] = \tr_B[\ketbra{\psi_2}_{AB}]$. 
    Then, there exists a Gaussian unitary $\hat{U}_B$ acting on $B$ that transforms $\ket{\psi_1}_{AB}$ to $\ket{\psi_2}_{AB}$, i.e. $\ket{\psi_2}_{AB} = (\hat{I}_A \otimes \hat{U}_B) \ket{\psi_1}_{AB}$.
\end{lemma}
\noindent To reach this result, first note that the first moments of $\ket{\psi_k}_{AB}$ are of the form $\bmd_A \oplus \bmd_{k,B}$, $k=1,2$, ensuring that the restriction on $A$ is the same.
Their covariance matrices $\sigma_{k,AB}$, $k=1,2$, are related by $\sigma_{2,AB} = (I_A \oplus V_B) \sigma_{1,AB} (I_A \oplus V_B)^{\rm T}$ for some symplectic matrix $V_B$ acting on $B$~\cite{botero2003modewise,serafini2020gaussian}.
Therefore, we can relate the two purifications as follows,
\begin{equation}
    \ket{\psi_2}_{AB} = (\hat{I}_A \otimes \hat{D}_{\bmd_{2,B}-\bmd_{1,B}}\hat{V}_B) \ket{\psi_1}_{AB} \eqqcolon (\hat{I}_A \otimes \hat{U}_B) \ket{\psi_1}_{AB} \,.
\end{equation}
Next, we show how this result extends in the presence of symmetry constraints.
\begin{lemma}\label{lem:inv_gauss_purification2}
    Consider a pair of systems $A$ and $B$ with finite numbers of bosonic modes, Hilbert spaces $\H_A$ and $\H_B$, and symplectic Gaussian representations of a group $G$, denoted by $\hat{S}_A(g): g\in G$ and $\hat{S}_B(g): g\in G$, respectively.

    Then, two pure Gaussian states $\ket{\psi_1}, \ket{\psi_2} \in \H_A \otimes \H_B$ that are $G$--invariant, i.e. $(\hat{S}_A(g) \otimes \hat{S}_B(g))\ket{\psi_k} = \ket{\psi_k}$ for $k = 1,2$ and all $g \in G$, have the same reduced state on $A$, i.e. $\tr_B[\ketbra{\psi_1}_{AB}] = \tr_B[\ketbra{\psi_2}_{AB}]$, if and only if there exists a $G$--invariant Gaussian unitary $\hat{U}_B$ acting on $B$ that transforms $\ket{\psi_1}_{AB}$ to $\ket{\psi_2}_{AB}$, i.e. 
    \begin{equation}
        \ket{\psi_2}_{AB} = (\hat{I}_A \otimes \hat{U}_B) \ket{\psi_1}_{AB} \,,\qquad [\hat{U}_B, \hat{S}_B(g)] = 0 \text{ for all } g \in G \,.
    \end{equation}
\end{lemma}
\begin{proof}
    The necessity of the equality of reduced states is obvious.
    Here, we prove the converse.
    From Lemma~\ref{lem:inv_gauss_purification0} we know that there exists, a possibly symmetry-breaking, Gaussian unitary $\hat{U}_B$ acting on $B$ that converts  $\ket{\psi_1}$ to $\ket{\psi_2}$.
    In the following, we prove that $\hat{U}_B$ can be chosen to be $G$--invariant.
    
    The overall unitary can be realized via $\hat{U}_B = \hat{D}_{\bmd_{2,B}-\bmd_{1,B}}\hat{V}_B$, where $\hat{V}_B$ is a purely symplectic Gaussian unitary, and $\bmd_{k,B}$ is the restriction of the displacement vector $\bmd_k$ of state $\ket{\psi_k}$ on $B$ for $k=1,2$. 
    Since both states $\ket{\psi_1}, \ket{\psi_2}$ are $G$--invariant, their displacement vectors $\bmd_1, \bmd_2$ are also $G$--invariant, which in turn implies that $\hat{D}_{\bmd_{2,B}-\bmd_{1,B}}$ is invariant under $S_B$.
    Without loss of generality, we can thus assume that the displacement vectors of $\ket{\psi_1}$ and $\ket{\psi_2}$ are zero and proceed to show that the symplectic Gaussian unitary $\hat{V}_B$ can also be chosen to be $G$--invariant.
    
    The fact that there exists a symplectic Gaussian unitary $\hat{V}_B$ such that $\ket{\psi_2}_{AB} = (\hat{I}_A \otimes \hat{V}_B) \ket{\psi_1}_{AB}$ implies the existence of a symplectic matrix $V_B$ such that
    \begin{equation}\label{eq:cov2Vcov1V}
        \sigma_2 = (I \oplus V_B)\sigma_1(I \oplus V_B)^{\rm T} \,.
    \end{equation}
    With respect to this direct sum decomposition, the covariance matrix can be written as 
    \begin{equation}\label{eq:cov_bipart}
            \sigma_k = \begin{pmatrix} \sigma_{k,A} & \sigma_{k,BA}^{\rm T} \\ \sigma_{k,BA} & \sigma_{k,B} \end{pmatrix} \,,\ k=1,2 \,,
    \end{equation}
    with $\sigma_{k,BA}: A \rightarrow B$ mapping from the subspace associated with modes in $A$ to the subspace associated with modes in $B$.
    The invariance of $\sigma_k$ under representation $g \mapsto S_A(g) \oplus S_B(g)$ implies that
    \begin{equation}\label{eq:sigmaAB}
        S_B(g) \sigma_{k,BA} S_A(g)^{\rm T} = \sigma_{k,BA} \,,
    \end{equation}
    for $k=1,2$ and all $g \in G$.
    Then,
    \begin{align}
        S_B(g) V_B S_B(g)^{-1} \sigma_{1,BA} 
        &= S_B(g) V_B \sigma_{1,BA} S_{A}(g)^{\rm T}
        = S_B(g) \sigma_{2,BA} S_{A}(g)^{\rm T}
        = \sigma_{2,BA} \nonumber\\
        &= V_B \sigma_{1,BA} \,, \label{eq:VBsigmaAB}
    \end{align}
    for all $g \in G$, where the first and third equalities follow from the invariance of $\sigma_{1,BA}$ and $\sigma_{2,BA}$, respectively, as given in Eq.(\ref{eq:sigmaAB}), while the second and fourth equalities follow from Eq.(\ref{eq:cov2Vcov1V}). 
    
    As a consequence, $S_B(g) V_B S_B(g)^{-1}$ and $V_B$ are equal in the image of $\sigma_{1,BA}$, denoted by $B_c$, i.e. for all $g \in G$,
    \begin{equation}\label{eq:svsp}
        S_B(g) V_B \Pi_c = V_B S_B(g) \Pi_c \,,
    \end{equation}
    where $\Pi_c$ is any, possibly non-symmetric, projector to the image of $\sigma_{1,BA}$, i.e. any matrix satisfying the properties (i) $\Pi_c^2=\Pi_c$, (ii) $\Pi_c \bmd\in B_c$ for all $\bmd\in B$, and (iii) $\Pi_c \bmd = \bmd$ for all $\bmd \in B_c$, which, in particular, imply $\Pi_c \sigma_{1,BA} = \sigma_{1,BA}$.
    Moreover, left-multiplying both sides of Eq.(\ref{eq:sigmaAB}) with $\Pi_c$ we find that
    \begin{equation}
        \Pi_c S_B(g) \Pi_c \sigma_{1,BA} S_A(g)^{\rm T} = \Pi_c \sigma_{1,BA} = \sigma_{1,BA} = S_B(g) \Pi_c\sigma_{1,BA} S_A(g)^{\rm T}
    \end{equation}
    for all $g \in G$, which implies that, for all $g \in G$,
    \begin{equation}\label{eq:psp_sp}
        S_c(g) \coloneqq \Pi_c S_B(g) \Pi_c = S_B(g) \Pi_c
    \end{equation}
    or, equivalently, $(I-\Pi_c) S_B(g) \Pi_c= 0$. 
    Therefore, the subspace $B_c$ is invariant under the action of the symmetry representation $S_c(g): g\in G$.
    Defining $W_c \coloneqq V_B \Pi_c$ and combining Eq.(\ref{eq:psp_sp}) with Eq.(\ref{eq:svsp}), we get
    \begin{equation}\label{eq:wpsp_sw}
        W_c S_c(g) = S_B(g) W_c \,,
    \end{equation}
    for all $g \in G$, i.e. $W_c$ is the intertwiner for representations $S_c(g): g \in G$ and $S_B(g): g \in G$.
    It also satisfies
    \begin{equation}\label{eq:bilinear}
        W^{\rm T}_c \Omega_B W_c = \Pi_c^{\rm T} V_B^{\rm T} \Omega_B V_B \Pi_c = \Pi_c^{\rm T} \Omega_B \Pi_c \eqqcolon \Omega_c \,.
    \end{equation}
    While $\Omega_c$ is always an anti-symmetric matrix, it may, in general, vanish or have rank smaller than the dimension of $B_c$.
    However, this cannot occur when $\sigma_{k}$ is the covariance matrix of a pure state (note that the assumption of state purity has not been used up to this point).
    According to Refs.~\cite{botero2003modewise,serafini2020gaussian}, in this case, the image $B_c$ of $\sigma_{BA}$ is a subspace that respects the symplectic structure.
    
    More precisely, the covariance matrices $\sigma_1$ and $\sigma_2$ induce decompositions of the $(2n_B)$--dimensional space $B$ as
    \begin{equation}
        B = B_c \oplus B_u = B'_c \oplus B'_u \,,
    \end{equation}
    where $B_c$ and $B_u$ are symplectic subspaces that are symplectically orthogonal, with $B_c$ the image of $\sigma_{1,BA}$.
    With respect to the decomposition $(A\oplus B_c) \oplus B_u$, $\sigma_1$ takes the form
    \begin{equation}
    \sigma_1=\omega_{1}\oplus \nu_1 \,, 
    \end{equation}
    where $\omega_1$ is restricted to subspace $A\oplus B_c$ and $\nu_1$ is restricted to $B_u$.
    In particular, $B_c$ corresponds to the modes that are correlated with system $A$ for the state $\ket{\psi_1}$. 

    Similarly, $B'_c$ and $B'_u$ are symplectically orthogonal symplectic subspaces, with $B'_c$ being the image of $\sigma_{2,BA}$. With respect to the decomposition $(A\oplus B'_c) \oplus B'_u$, $\sigma_2$ takes the form
    \begin{equation}
        \sigma_2 = \omega_{2} \oplus \nu_2 \,. 
    \end{equation}
    In the above, we established that the $G$--invariant symplectic intertwiner $W_c \coloneqq V_B \Pi_c$ maps $B_c$ to $B'_c$ while preserving the symplectic form.
    Choosing $\Pi_c$ to be the symplectic projector to $B_c$, $W_c=V\Pi_c$ maps $\sigma_1$ to $\omega_2$, i.e.
    \begin{equation}
        (I_A\oplus W_c) \sigma_1 (I_A\oplus W_c)^{\rm T}=  (I_A\oplus V_B\Pi_c) \sigma_{1} (I_A\oplus V_B \Pi_c)^{\rm T} = (I_A\oplus V_B\Pi_c) \omega_{1} (I_A\oplus V_B \Pi_c)^{\rm T} = \omega_{2} \,.
    \end{equation}

    We now proceed to show that there exists a full $G$--invariant symplectic transformation acting on $B$ that maps $\sigma_1$ to $\sigma_2$.
    Firstly, as shown in Lemma~\ref{prop:iso_completion}, any such intertwiner $W_c$ that preserves the symplectic form can be extended to a full $G$--invariant symplectic transformation $W_B$ satisfying
    \begin{equation}
        W_B \Pi_c = W_c = V_B \Pi_c \,, \qquad
        W_B^{\rm T} \Omega_B W_B = \Omega_B \,, \qquad
        W_B S_B(g) = S_B(g) W_B \qquad \text{for all } g \in G \,.
    \end{equation}
   Since $W_B$ maps the subspace $B_c$ to $B'_c$, it should also map $B_u$, namely the symplectically orthogonal complement of $B_c$, to $B'_u$, namely the symplectically orthogonal complement of $B'_c$.
    Furthermore, with respect to decomposition $(A \oplus B_c) \oplus B_u$, $\sigma_1$ is block-diagonal, hence $\tilde{\sigma}_1 \coloneqq W_B \sigma_1 W_B^{\rm T}$ is also block-diagonal with respect to $(A \oplus B'_c) \oplus B'_u$, i.e.
    \begin{equation}
        \tilde{\sigma}_1 \coloneqq (I_A\oplus W_B) \sigma_1  (I_A\oplus W_B)^{\rm T} = \omega_2 \oplus \tilde{\nu}_1 \,,
    \end{equation}
    where $\tilde{\nu}_1$ is restricted to $B'_u$.
    Roughly speaking, this means that state $\tilde{\sigma}_1$ is equal to the target state $\sigma_2$, when restricted to the modes in  $A$ and modes of $B$ that are correlated with $A$.
    It is left to show that there exists a $G$--invariant symplectic transformation acting only on $B$ which transforms $\tilde{\nu}_1$ to $\nu_2$, hence transforming $\tilde{\sigma}_1$ to $\sigma_2$.

    The fact that $\tilde{\sigma}_1$ and $\sigma_2$ are covariance matrices of pure states immediately implies that $\tilde{\nu}_1$ and $\nu_2$ should also correspond to covariance matrices of pure states. This can be seen, e.g. by noting that, as covariance matrices of pure Gaussian states, all symplectic eigenvalues of $\tilde{\sigma}_1$ and $\sigma_2$ are 1.
    Since $\tilde{\sigma}_1 = \omega_2 \oplus \tilde{\nu}_1$ is block-diagonal with respect to decomposition $(A \oplus B'_c) \oplus B'_u$, where subspaces $A \oplus B'_c$ and $B'_u$ are symplectic, its symplectic eigenvalues are the union of the symplectic eigenvalues of $\omega_2$ and of $\tilde{\nu}_1$, which in turn,  implies that all symplectic eigenvalues of $\tilde{\nu}_1$ should be 1.
    Similarly, $\sigma_2 = \omega_2 \oplus \nu_2$ with respect to the same decomposition, so all symplectic eigenvalues of $\tilde{\nu}_2$ should be 1.

    Based on this observation and the fact that $B'_c$ and $B'_u$ are both $G$--invariant symplectic subspaces we construct a $G$--invariant symplectic transformation that maps $\tilde{\sigma}_1$ to $\sigma_2$.

    First, it is useful to apply a change of basis by a symplectic transformation $T$ on $B$, such that  the subspaces $T B'_c$ and $T B'_u$ are mutually orthogonal with respect to both the Euclidean inner product and the symplectic form.
    In other words, they can be thought of as subspaces corresponding to independent bosonic modes.
    This transformation brings the covariance matrices to the form
    \begin{equation}\label{eq:kappa}
        (I_A \oplus T) \tilde{\sigma}_1 (I_A \oplus T)^{\rm T} = \kappa_{c} \oplus \kappa_{u} \,,\qquad (I_A\oplus T)\tilde{\sigma}_2 (I_A\oplus T)^{\rm T} = \kappa_{c} \oplus \kappa'_{u} \,,
    \end{equation}
    where $\kappa_{u}$ and $\kappa'_{u}$ are both pure state covariance matrices, and the direct sum on the right-hand sides is with respect to the orthogonal subspaces $A \oplus T B'_c$ and $T B'_u$.
    
    This transformation also brings the representation of symmetry to the form
    \begin{equation}
        \tilde{S}_{c}(g) \oplus \tilde{S}_{u}(g) \coloneqq TS_B(g)T^{-1}: g \in G \,,
    \end{equation}
    The covariance matrices $\kappa_{u}$ and $\kappa'_{u}$ are invariant under the representation $\tilde{S}_{u}$. Then, according to part (i) of Proposition~\ref{prop:squ}, the representation $\tilde{S}_{u}$ is a squeezed passive representation, i.e. it is equivalent to an orthogonal symplectic representation, up to a similarity transformation by a symplectic matrix.  
    
    Since subspace $T B_u'$ corresponds to a collection of bosonic modes, we can apply Proposition~\ref{prop:pure_interconversion} on subspace $T B_u'$, which guarantees that there exists a symplectic transformation $\tilde{N}_u$ acting on $T B_u'$, such that (i) it commutes with $\tilde{S}_{u}(g)$ for all $g \in G$, and (ii) it maps $\kappa_{u}$ to  $\kappa'_{u}$, i.e.
    \begin{equation}
        \tilde{N}_u \kappa_{u} \tilde{N}_u^{\rm T} = \kappa'_{u} \,.
    \end{equation}
    We extend this to a symplectic transformation on $TB$ by acting with the identity operator on $TB'_c$, i.e. we consider
    \begin{equation}
        \tilde{N}_B \coloneqq I_{TB'_c} \oplus \tilde{N}_u \,,
    \end{equation}
    which clearly commutes with $\tilde{S}_{c}(g) \oplus \tilde{S}_{u}(g)$ for all $g \in G$, and satisfies
    \begin{equation}\label{eq:Ntilde}
        (I_A \oplus \tilde{N}_B) (\kappa_{c} \oplus \kappa_{u}) (I_A \oplus \tilde{N}_B)^{\rm T} =(\kappa_{c} \oplus \kappa'_{u}) \,.
    \end{equation}
    Returning to the original basis, with respect to decomposition $B = B'_c \oplus B'_u$, we obtain the symplectic transformation
    \begin{equation}\label{eq:TNT}
        N_B = T^{-1} \tilde{N}_B T = T^{-1}(I_{TB'_c} \oplus \tilde{N}_u)T \,,
    \end{equation}
    which acts on $B$, is $G$--invariant, i.e. it commutes with representation $S_B(g) : g \in G$, and transforms $\tilde{\sigma}_1$ to $\sigma_2$, i.e.
    \begin{equation}
        \left(I_A \oplus N_B\right) \tilde{\sigma}_1 \left(I_A \oplus N_B\right)^{\rm T} = (I_A \oplus T^{-1})( \kappa_{c} \oplus \kappa'_{u}) (I_A\oplus T^{-1})^{\rm T} = \sigma_2 \,.
    \end{equation}
    where the first equality follows from Eqs.(\ref{eq:Ntilde},\ref{eq:TNT}) and the second equality follows from Eq.(\ref{eq:kappa}).
    Hence, the $G$--invariant symplectic transformation $I_A \oplus N_BW_B$ acts non-trivially only on $B$ and maps $\sigma_1$ to $\sigma_2$, i.e.
    \begin{equation}
        (I_A \oplus N_B W_B) \sigma_1 (I_A \oplus N_B W_B)^{\rm T} = \sigma_2 \,,
    \end{equation}
    concluding the proof.
\end{proof}

\subsection{Completion of symplectic intertwiner}
\label{app:completion}

We now complete the proof of equivalence between invariant Gaussian purifications, by providing a procedure that completes a symplectic intertwiner, i.e. an intertwiner that preserves the symplectic form, into a full $G$--invariant symplectic transformation.
\begin{lemma}\label{prop:iso_completion}
    Consider a system of $n$ bosonic modes, with phase space $\mathbb{R}^{2n}$, and symplectic form 
    \begin{equation}
        \Omega \equiv \Omega_n = \bigoplus_{j=1}^n \begin{pmatrix} 0 & 1 \\ -1 & 0 \end{pmatrix} \,.
    \end{equation}
    Let $S: G\rightarrow  \Sp(2n,\mathbb{R})$ be a symplectic representation that admits a $G$--invariant state, i.e.
    \begin{equation}
        S(g) \Omega S(g)^{\rm T} = \Omega \,, \qquad S(g) \sigma S(g)^{\rm T} = \sigma \,,
    \end{equation}
    for all $g \in G$ and some covariance matrix $\sigma$.
    Let $F_{\text{in}}$ and $F_{\text{out}}$ be $G$--invariant symplectic subspaces of equal dimension. 
    Suppose there exists an intertwiner $V: F_{\text{in}} \rightarrow F_{\text{out}}$ preserving the symplectic form, i.e. 
    \begin{equation}
        V S(g) = S(g) V \,,
        \qquad
        \Pi_{F_{\text{in}}}^{\rm T} V^{\rm T} \Omega V \Pi_{F_{\text{in}}} = \Pi_{F_{\text{in}}}^{\rm T}\Omega\Pi_{F_{\text{in}}} \,,
    \end{equation}
    where $\Pi_{F_{\text{in}}}$ is any projector onto $F_{\text{in}}$.
    Then, $V$ can be extended to a global $G$--invariant symplectic transformation $W\in \Sp(2n,\mathbb{R})$, i.e. $W$ satisfies
    \begin{equation}
        W^{\rm T} \Omega W = \Omega \,,
        \qquad
        W S(g) = S(g) W \,,
        \qquad
        W\Pi_{F_{\text{in}}} = V \,.
    \end{equation}
\end{lemma}
\begin{proof}
    For any symplectic subspace $F \subset \mathbb{R}^{2n}$, the restriction of the symplectic form $\Omega_n$ to $F$ is, by definition, non-degenerate. 
    It follows that 
    \begin{equation}
        \mathbb{R}^{2n} = F \oplus F^\perp \,,
    \end{equation}
    where $F^\perp$ denotes the symplectic complement of $F$,  i.e. the set of vectors that are symplectically orthogonal to $F$, and satisfies $\dim(F^\perp) = 2n - \dim(F)$. 
    Moreover, there exists a symplectic transformation $T$ (i.e. $T^{\rm T} \Omega_n T = \Omega_n$) such that the subspaces
    \begin{equation}\label{eq:tre}
        A = T F \,,\qquad B = T F^\perp \,,
    \end{equation}
    are mutually orthogonal with respect to both the Euclidean inner product and the symplectic form $\Omega_n$.
    
    It follows that, relative to the decomposition $\mathbb{R}^{2n}=A\oplus B$, $\Omega_n$ takes the standard block-diagonal form 
    \begin{equation}
        \Omega_n = \Omega_{a} \oplus \Omega_{n-a} \,,
    \end{equation}
    where $2a = \dim(A)=\dim(F)$. In other words, $A$ and $B$ describe independent bosonic modes. 
    
    Under this change of basis, the original representation $S(g): g \in G$, is transformed to 
    \begin{equation}
        \widetilde{S}(g) = T S(g) T^{-1} \,: g \in G \,.
    \end{equation}
    The assumptions that $S(g)$ leaves $F$ invariant and $S(g)$ is symplectic imply that $S(g)$ also leaves $F^\perp$ invariant. This, in turn, implies that $\widetilde{S}(g)$ leaves both $A$ and $B$ invariant. Then, the restrictions of $\widetilde{S}(g)$ to these subspaces are also symplectic representations of $G$, i.e. $\widetilde{S}(g)$ and can be decomposed as 
    \begin{equation}
        \widetilde{S}(g) = \widetilde{S}_A(g) \oplus \widetilde{S}_B(g) \,,
    \end{equation}
    for all $g\in G$, where the direct sum is with respect to the decomposition $\mathbb{R}^{2n} = A \oplus B$.
    Furthermore, the assumption that $S(g)$ admits a $G$--invariant state implies that these restricted representations also have invariant states. Since $A$ and $B$ are both symplectic and orthogonal complements, i.e. they correspond to independent bosonic modes, we can apply part (i) of Proposition~\ref{prop:squ} to each of them. Thus, there exist symplectic transformations $R_A$ and $R_B$ acting on $A$ and $B$, respectively, such that
    \begin{align}
         O_A(g) = R_A\, \widetilde{S}_A(g)\, R_A^{-1} \,,\qquad O_B(g)= R_B\, \widetilde{S}_B(g)\, R_B^{-1} \,,
    \end{align}
    where $O_A(g)$ and $O_B(g)$ are orthogonal and symplectic for all $g \in G$. 
    
    In summary, there exists a symplectic transformation
    \begin{equation}\label{eq:KAKBT}
        K \coloneqq (R_A \oplus R_B)\ T \,,
    \end{equation}
    such that (i) it maps the subspaces $F$ and $F^\perp$ to fully independent sets of bosonic modes, and (ii) it transforms the representation $S(g)$ into an orthogonal symplectic (passive) representation on each subsystem,
    \begin{equation}\label{eq:KSK}
        K S(g) K^{-1} = O(g) = O_A(g) \oplus O_B(g) \,,
    \end{equation}
    where the direct sum in both Eq.(\ref{eq:KAKBT}) and Eq.(\ref{eq:KSK}) is taken with respect to $\mathbb{R}^{2n} = A \oplus B$. 
    
    Since $A$ and $B$ are orthogonal  with respect to both Euclidean and symplectic forms, there exists a basis 
    $(\bmr_1, \bmr_2, \dots, \bmr_{2a}, \dots, \bmr_{2n})$
    where $\{\bmr_1, \dots, \bmr_{2 a}\}$ and $\{\bmr_{2a+1}, \dots, \bmr_{2n}\}$  span $A$ and $B$, respectively, and satisfy 
    \begin{equation*}
        \bmr_{2i}^{\mathrm{T}} \Omega_n \bmr_{2j} 
        = \bmr_{2i-1}^{\mathrm{T}} \Omega_n \bmr_{2j-1} = 0 \,, 
        \qquad 
        \bmr_{2i}^{\mathrm{T}} \Omega_n \bmr_{2j-1} = \delta_{ij} \,, 
        \quad i,j = 1,\dots,n \,.
    \end{equation*}
    Based on this basis, one can define the complex basis ${\bma}$, as  
    \begin{equation}
    {\bma}_j=\frac{\bmr_{2j-1}+i \bmr_{2j}}{\sqrt{2}}\ \ \ \,,\ \ \ \ {\bma}_{n+j}=\bma_j^*=\frac{\bmr_{2j-1}-i \bmr_{2j}}{\sqrt{2}} \,,
    \end{equation}
    for $j=1,\cdots, n$. Relative to this basis, the subspaces $A$ and $B$ correspond to basis elements 
    $\{\bma_j, \bma^*_j={\bma}_{n+j}: j = 1, \dots, a\}$, and $\{\bma_j, {\bma}^*_j={\bma}_{n+j}: j = a+1, \dots, n\}$, respectively. 
    
    Then, relative to this basis we obtain 
    \begin{equation}
        K S(g) K^{-1} = O(g) \overset{\bma}{=} U(g)\oplus U^{*}(g) \,,
    \end{equation}
    where $U(g)$ is a unitary transformation on $\mathbb{C}^{n}$, satisfying
    \begin{equation}
        U(g)=U_A(g)\oplus U_B(g) \,,
    \end{equation}
    for all $g\in G$, and
    \begin{equation}
        O_{A}(g) \overset{\bma}{=} U_{A}(g) \oplus U_{A}(g)^* \,,
    \end{equation}
    where $U_{A}(g)$ is a unitary transformation on $\mathbb{C}^{a}$, for all $g \in G$. 
    A similar decomposition yields $O_{B}(g) \overset{\bma}{=} U_{B}(g) \oplus U_{B}(g)^*$, where $U_{B}(g)$ is a unitary transformation on  $\mathbb{C}^{n-a}$, for all $g \in G$. 
    
    Further decomposing into irreps of $G$, we obtain
    \begin{equation}
        U_A(g) = \bigoplus_{\lambda \in \Lambda} U_\lambda(g) \otimes I_{m_A(\lambda)} \,,
    \end{equation}
    where $\Lambda$ is the set of all irreps of $G$, and $m_A(\lambda)$ denotes the multiplicity of irrep $\lambda$ in $U_A(g)$.
    
    In summary, we found 
    \begin{equation}
        K S(g) K^{-1} = O(g)=O_A(g) \oplus O_B(g) \overset{\bma}{=} U(g)\oplus U^{*}(g) \,,
    \end{equation}
    where $U(g)=U_A(g)\oplus U_B(g)$. It follows that the  multiplicity of $\lambda$ in $U(g): g\in G$ is 
    \begin{equation}
        m(\lambda)=m_{A}(\lambda)+m_{B}(\lambda) \,,
    \end{equation}
    where $m_{A}(\lambda)$ and $m_{B}(\lambda)$ are multiplicities of irrep $\lambda$ in representations $U_{A}$ and $U_B$, respectively. 
    
    Recall that the above decomposition and symplectic transformation $K$ depend on the original symplectic decomposition $\mathbb{R}^{2n}=F\oplus F^\perp$. In particular, choosing different symplectic decompositions 
    $\mathbb{R}^{2n}=F\oplus F^\perp$,  will result in different multiplicities $m_{A}(\lambda)$ and $ m_{B}(\lambda)$. However, part (ii) of Proposition~\ref{prop:squ} guarantees that their sum $m(\lambda)=m_{A}(\lambda)+m_{B}(\lambda)$ is independent of the decomposition.
    \vspace{10pt}
    
    Now we apply this to the subspaces $F_{\text{in}}$ and $F_{\text{out}}$ given in the statement of the lemma. Let $\mathbb{R}^{2n}=A_\text{in}\oplus B_\text{in}$ and $\mathbb{R}^{2n}=A_\text{out}\oplus B_\text{out}$ be the decompositions obtained from  symplectic  subspaces $F_{\text{in}}$ and $F_{\text{out}}$, respectively, such that
    \bes\label{eq:KSKOO}
    \begin{align}
        K_{\text{in}} S(g) K^{-1}_{\text{in}}&=O_{A_{\text{in}}}(g) \oplus O_{B_{\text{in}}}(g) \\ 
        K_{\text{out}} S(g) K^{-1}_{\text{out}}&=O_{A_{\text{out}}}(g) \oplus O_{B_{\text{out}}}(g) \,,
    \end{align}
    \ees
    for all $g\in G$. Then, for all $\lambda\in\Lambda$,
    \begin{equation}\label{eq:msum}
        m_{A_\text{in}}(\lambda) + m_{B_\text{in}}(\lambda) = m_{A_\text{out}}(\lambda) + m_{B_\text{out}}(\lambda)
    \end{equation}
    where $m_{A_\text{in}}(\lambda)$ is the multiplicity of irrep $\lambda$ in 
    $A_\text{in}$, the subspace obtained from $F_{\text{in}}$ via Eq.(\ref{eq:tre}), and $m_{A_\text{out}}(\lambda)$, $m_{A_\text{out}}(\lambda)$, $m_{B_\text{out}}(\lambda)$, are defined in a similar fashion.
    
    To construct the extension of $G$--invariant symplectic intertwiner $V$ to $W$, we consider the chain of transformations
    \begin{equation}
        F_\text{in}\oplus F^\perp_\text{in}
        \xlongrightarrow{K_{\text{in}}} A_\text{in} \oplus B_\text{in}
        \xlongrightarrow{\widetilde{W}} A_\text{out} \oplus B_\text{out}
        \xlongrightarrow{K^{-1}_{\text{out}}} F_\text{out} \oplus F^\perp_\text{out} \,,
    \end{equation}
    where $\widetilde{W}$ is a symplectic, $G$--invariant transformation defined below in Eq.(\ref{eq:WVQ}). The corresponding representations are 
    \begin{equation}
        S(g) 
        \xlongrightarrow{K_{\text{in}}} O_{A_\text{in}}(g) \oplus O_{B_\text{in}}(g)
        \xlongrightarrow{\widetilde{W}} O_{A_\text{out}}(g)\oplus O_{B_\text{out}}(g)
        \xlongrightarrow{K^{-1}_{\text{out}}} S(g) \,,
    \end{equation}
    To construct $\widetilde{W}$, first recall that, by assumption, $V: F_{\text{in}}\rightarrow F_{\text{out}}$ is a symplectic intertwiner commuting with $S(g)$, i.e.
    \begin{equation}
        V S(g) \Pi_{F_{\text{in}}} = S(g) V \Pi_{F_{\text{in}}} \,,
    \end{equation}
    which, in turn, implies
    \begin{equation}
        V K^{-1}_{\text{in}} (K_{\text{in}} S(g) K_{\text{in}}^{-1}) \Pi_{A_{\text{in}}} = K^{-1}_{\text{out}} (K_{\text{out}} S(g) K_{\text{out}}^{-1}) K_{\text{out}} V K^{-1}_{\text{in}} \Pi_{A_{\text{in}}}
    \end{equation}
    or, equivalently,
    \begin{equation}
        \widetilde{V} O_{A_{\text{in}}}(g) = O_{A_{\text{out}}}(g) \widetilde{V} 
    \end{equation}
    for all $g\in G$, where
    \begin{equation}\label{eq:KVK}
        \widetilde{V} \coloneqq K_{\text{out}} V K^{-1}_{\text{in}} \,,
    \end{equation}
    maps $A_{\text{in}}$ to $A_{\text{out}}$, and preserves the symplectic form. That is,
    \begin{equation}
        O_{A_{\text{in}}}(g) \Pi_{A_{\text{in}}} = \widetilde{V}^{-1} O_{A_{\text{out}}}(g) \widetilde{V} \Pi_{A_{\text{in}}}
    \end{equation}
    where $\widetilde{V}^{-1}$ is defined on subspace $A_{\text{out}}$. Then, again part (ii) of Proposition~\ref{prop:squ} implies that the orthogonal symplectic representations $O_{A_{\text{in}}}$ and $O_{A_{\text{out}}}$ are equivalent, up to an orthogonal symplectic transformation, and therefore 
    \begin{equation}
        m_{A_{\text{in}}}(\lambda)=m_{A_{\text{out}}}(\lambda)
    \end{equation}
    for all $\lambda\in\Lambda$.
    Combined with Eq.(\ref{eq:msum}) we find that
    \begin{equation}
        m_{B_{\text{in}}}(\lambda) = m_{B_{\text{out}}}(\lambda) \,,
    \end{equation}
    for all $\lambda\in\Lambda$.
    By Proposition~\ref{prop:mlambda}, this immediately implies that unitary representations $U_{B_{\text{in}}}$ and $U_{B_{\text{out}}}$ are equivalent up to conjugation by a unitary, which, in turn implies that orthogonal symplectic representations $O_{B_{\text{in}}}$ and $O_{B_{\text{out}}}$ are equivalent up to an orthogonal symplectic transformation. That is, there exists an orthogonal symplectic intertwiner $Q: B_{\text{in}} \rightarrow B_{\text{out}}$, such that, for all $g\in G$,
    \begin{equation}
        Q O_{B_{\text{in}}}(g) = O_{B_{\text{out}}}(g) Q \,.
    \end{equation}
    Then, define the symplectic transformation $\widetilde{W} \in \Sp(2n, \mathbb{R})$, as   
    \begin{equation}\label{eq:WVQ}
        \widetilde{W} = \widetilde{V} \oplus Q
    \end{equation}
    where $\widetilde{V}: A_{\text{in}} \rightarrow B_{\text{in}}$ and $Q: A_{\text{out}}\rightarrow B_{\text{out}}$. By definition, on subspace $A_{\text{in}}$, $\widetilde{W}$ is equal to 
    $\widetilde{V}$, i.e.
    \begin{equation}\label{eq:tildeWinVin}
        \widetilde{W}\Pi_{A_{\text{in}}} = \widetilde{V}\Pi_{A_{\text{in}}} \,.
    \end{equation}
    Furthermore, $\widetilde{W}$ is a $G$--invariant symplectic intertwiner,  i.e. for all $g\in G$, it satisfies 
    \bes\label{eq:tildeWGinv}
    \begin{align}
        \widetilde{W} \big(O_{A_{\text{in}}}(g) \oplus O_{B_{\text{in}}}(g)\big) 
        &= (\widetilde{V} \oplus Q) \big(O_{A_{\text{in}}}(g)\oplus O_{B_{\text{in}}}(g)\big) \\ 
        &= \big(O_{A_{\text{out}}}(g) \oplus O_{B_{\text{out}}}(g)\big) (\widetilde{V} \oplus Q)\\ 
        &= \big(O_{A_{\text{out}}}(g) \oplus O_{B_{\text{out}}}(g)\big) \widetilde{W} \,.
        \end{align}
    \ees
    Finally, we go back to the original basis and define the symplectic transformation
    \begin{equation}\label{eq:KWK}
        W \coloneqq K_{\text{out}}^{-1} \widetilde{W} K_{\text{in}} \,.
    \end{equation}
    Then, using Eq.(\ref{eq:KVK}) and Eq.(\ref{eq:KWK}), Eq.(\ref{eq:tildeWinVin}) can be rewritten as
    \begin{equation}
        K_{\text{out}} W K^{-1}_{\text{in}} \Pi_{A_{\text{in}}}= K_{\text{out}} V K^{-1}_{\text{in}} \Pi_{A_{\text{in}}} \,.
    \end{equation}
    By multiplying from left with $K_{\text{out}}^{-1}$ and from right with $K_{\text{in}}$, this implies
    \begin{equation}\label{eq:WinVin}
        W \Pi_{F_{\text{in}}} = V \Pi_{F_{\text{in}}} \,,
    \end{equation}
    i.e. $V$ and $W$ act  identically on subspace $F_{\text{in}}$.
    Furthermore, using Eq.(\ref{eq:KSKOO}), we find that the left-hand side of Eq.(\ref{eq:tildeWGinv}) is equal to
    \begin{align}
        \widetilde{W} \big(O_{A_{\text{in}}}(g)\oplus O_{B_{\text{in}}}(g)\big) = \widetilde{W} K_{\text{in}} S(g) K^{-1}_{\text{in}} = K_{\text{out}} W S(g) K^{-1}_{\text{in}} \,,
    \end{align}
    and the right-hand side of Eq.(\ref{eq:tildeWGinv}) is equal to
    \begin{align}
        \big(O_{A_{\text{out}}}(g)\oplus O_{B_{\text{out}}}(g)\big) \widetilde{W} = K_{\text{out}} S(g) K^{-1}_{\text{out}} \widetilde{W}=K_{\text{out}} S(g) W K^{-1}_{\text{in}} \,.
    \end{align}
    We conclude that the symplectic transformation $W$ is $G$--invariant, i.e. it satisfies 
    $S(g) W = W S(g)$ for all $g \in G$, and $W \Pi_{F_{\text{in}}} = V \Pi_{F_{\text{in}}}$.
    In summary, this map is symplectic, $G$--invariant and as shown in Eq.(\ref{eq:WinVin}), its action on $F_{\text{in}}$ is equal to the action of $V$. This completes the proof.
\end{proof}

\section{Covariant Gaussian dilation}
\label{app:dilation}

In this section, we derive the channel dilation as stated in Theorem~\ref{thm:cov_gauss_dilation}.
Our derivation works in both the discrete and continuous setting, and can be extended to hold in the presence of symmetries.
We start our discussion by motivating a rigorous definition of Gaussian channels (in the absence of symmetries) as completely Gaussianity-preserving CPTP operations.
We then prove that a channel is a (covariant) Gaussian channel if and only if it admits a (covariant) Gaussian Stinespring dilation.
We provide examples of our dilation on 1 mode and in a discrete setting.

\subsection{Definition of Gaussian channels}
\label{app:channel_def}

Several recent works~\cite{serafini2017quantum,weedbrook2012gaussian} introduce Gaussian channels by automatically adopting their Gaussian dilation as a defining property. 
We say that a quantum channel $\E: \B(\H_A) \rightarrow \B(\H_B)$ admits a Gaussian dilation if there exists a Gaussian state $\hat{\tau}_{BE} \in \B(\H_B \otimes \H_E)$, where $E$ is the environment, and a symplectic Gaussian unitary $\hat{V}_{ABE}$ acting on the composite system $ABE$, such that
\begin{equation}\label{eq:gauss_dilation}
    \E(\hat{\rho}_A) = \tr_{AE}\left[ \hat{V}_{ABE}(\hat{\rho}_A \otimes \hat{\tau}_{BE})\hat{V}_{ABE}^\dag \right] \,,
\end{equation}
for all input states $\hat{\rho}_A \in \B(\H_A)$.
Proposition~\ref{prop:cov_E2XY} shows any such dilation immediately implies that $\E$ preserves the Gaussian character of quantum states and acts in phase space via matrices $X,Y$ according to Eq.(\ref{eq:xychannel}).

A more axiomatic approach~\cite{demoen1977completely,demoen1979completely,lindblad2000cloning,holevo2001evaluating,caruso2008multi} involves shifting the focus from the Schr{\"o}dinger to the Heisenberg picture by defining the dual $\E^\dag$ of a channel $\E$ via the implicit relation $\tr[\E(\rho)\hat{O}] = \tr[\rho\E^\dag(\hat{O})]$.
Then, one can define Gaussian channels as follows.
\begin{definition}\label{def:gauss_channel}
    A bosonic quantum channel $\E: \B(\H_A) \rightarrow \B(\H_B)$ is called a Gaussian channel if its dual channel $\E^\dag$ maps displacement operators to displacement operators according to
    \begin{equation}\label{eq:dual_channel}
        \E^\dag(D_{\bmd}) \overset{\bmr}{=} D_{X^{\rm T}\bmd}\ \exp(-\frac{1}{2}\bmd^{\rm T} Y \bmd + i\bmd^{\rm T}\bmxi) \,,
    \end{equation}
    for real vector $\bmxi \in \mathbb{R}^{2n_B}$, $(2n_B \times 2n_A)$ real matrix $X$ and $(2n_B \times 2n_B)$ real matrix $Y$, where $n_A$ and $n_B$ are the numbers of modes in systems $A$ and $B$ respectively.
\end{definition}
From this definition, one can directly recover the channel uncertainty relation (Ineq.(\ref{eq:channel_uncertainty})) and the action of $\E$ in phase space (Eq.(\ref{eq:xychannel}))~\cite{lindblad2000cloning,demoen1977completely}.
The phase $f(\bmd) \coloneqq \exp(-\frac{1}{2}\bmd^{\rm T} Y \bmd + i\bmd^{\rm T}\bmxi)$ is a Gaussian function of $\bmd$, so the characteristic function~\cite{barnett2002methods} of output $\E(\rho)$ is Gaussian whenever the characteristic function of input $\rho$ is Gaussian.
As a consequence, $\E$ sends Gaussian states to Gaussian states.
In fact, $\E$ is completely Gaussianity-preserving, i.e. for any environment $E$, $(\E \otimes \I_E) (\hat{\rho}_{AE})$ is Gaussian whenever $\hat{\rho}_{AE}$ is Gaussian, where the phase space action of $\E \otimes \I_E$ is characterized by vector $\bmxi \oplus 0$, and matrices $X \oplus I$ and $Y \oplus 0$~\cite{serafini2017quantum}.

\subsection{Covariant Gaussian Stinespring dilation}
\label{app:dilation_proof}

We first show that every covariant Gaussian dilation is a covariant Gaussian channel.
\begin{proposition}\label{prop:cov_E2XY}
    Consider a pair of systems $A$ and $B$ with a finite number of bosonic modes, and Hilbert spaces $\H_A$ and $\H_B$. 
    Suppose each system is equipped with a symplectic Gaussian representation of a group $G$ that admits an invariant state, denoted as $\hat{S}_A(g): g\in G$, and $\hat{S}_B(g): g\in G$, respectively (Recall that such representations are characterized in Proposition~\ref{prop:squ}).

    Let $\E: \B(\H_A) \rightarrow \B(\H_B)$ be a linear map defined by the dilation
    \begin{equation}\label{eq:cov_dilation}
        \E(\hat{\rho}_A) = \tr_{A\overbar{B}}\left[ \hat{V}_{AB\overbar{B}}(\hat{\rho}_A \otimes \hat{\tau}_{B\overbar{B}})\hat{V}_{AB\overbar{B}}^\dag \right] \,,\quad \text{ for all } \hat{\rho}_A\in\B(\H_A) \,,
    \end{equation}
    where $\hat{\tau}_{B\overbar{B}} \in \B(\H_B \otimes \H_{\overbar{B}})$ is a Gaussian state which commutes with representation $\hat{S}_B \otimes \hat{S}_{\overbar{B}}$ and $\hat{V}_{AB\overbar{B}}$ is a symplectic Gaussian unitary which commutes with representation $\hat{S}_{A} \otimes \hat{S}_B \otimes \hat{S}_{\overbar{B}}$.
    Then, $\E$ is a covariant Gaussian channel under input/output representations $\hat{S}_A, \hat{S}_B$, i.e. $\E(\hat{S}_A(g) (\cdot) \hat{S}_A(g)^\dag) = \hat{S}_B(g) \E(\cdot) \hat{S}_B(g)^\dag$ for all $g \in G$.
    
    Furthermore, there exist real matrices $X,Y$ that satisfy 
    \begin{enumerate}
        \item[(i)] consistency with dilation, $V (\sigma \oplus \sigma_\tau) V^{\rm T} \bigg|_B = X\sigma X^{\rm T} + Y$, where $\sigma_\tau$ is the covariance matrix of state $\hat{\tau}$;
        \item[(ii)] the uncertainty relation, $Y + i(\Omega_B - X \Omega_A X^{\rm T}) \ge 0$, where $\Omega_A, \Omega_B$ are the symplectic forms on $A, B$;
        \item[(iii)] the covariance conditions, $S_B(g)XS_A(g)^{-1} = X$ and $S_B(g) Y S_B(g)^{\rm T} = Y$ for all $g \in G$.
    \end{enumerate}
\end{proposition}
\begin{proof}
    The dilation defined in Eq.(\ref{eq:cov_dilation}) is a Gaussian quantum channel as a composition of Gaussian quantum channels according to Definition~\ref{def:gauss_channel}, and it is also covariant in the sense that
    \begin{align}
        \E(\hat{S}_A\hat{\rho}_A\hat{S}_A^\dag) 
        &= \tr_{A\overbar{B}}\left[ \hat{V}_{AB\overbar{B}}(\hat{S}_A\hat{\rho}_A\hat{S}_A^\dag \otimes \hat{\tau}_{B\overbar{B}})\hat{V}_{AB\overbar{B}}^\dag \right] \nonumber\\
        &= \tr_{A\overbar{B}}\left[ (\hat{S}_{A} \otimes \hat{S}_B \otimes \hat{S}_{\overbar{B}}) \hat{V}_{AB\overbar{B}}(\hat{\rho}_A \otimes \hat{\tau}_{B\overbar{B}})\hat{V}_{AB\overbar{B}}^\dag (\hat{S}_{A} \otimes \hat{S}_B \otimes \hat{S}_{\overbar{B}})^\dag \right] \nonumber\\
        &= \hat{S}_B \tr_{A\overbar{B}}\left[ \hat{V}_{AB\overbar{B}}(\hat{\rho}_A \otimes \hat{\tau}_{B\overbar{B}})\hat{V}_{AB\overbar{B}}^\dag \right] \hat{S}_B^\dag \nonumber\\ 
        &= \hat{S}_B\E(\hat{\rho}_A) \hat{S}_B^\dag \,.
    \end{align}
    The first and last equalities follow from the map definition in Eq.(\ref{eq:cov_dilation}), the second equality follows from the invariance conditions on $\hat{V}_{AB\overbar{B}}$ and $\hat{\tau}_{B\overbar{B}}$, and the third equality follows from properties of the partial trace.

    Since $\E$ is a Gaussian channel, it is fully characterized by two real matrices $X,Y$.
    The dilation condition $(i)$ follows directly by this equivalent characterization of the channel, the uncertainty relation $(ii)$ follows from the Gaussianity of the channel, and the invariance conditions $(iii)$ follow from the covariance of the channel according to statement $(iv)$ in Proposition~\ref{prop:symmetry_phasespace}.
\end{proof}

We now proceed to prove the converse of Proposition~\ref{prop:cov_E2XY}, which informally states that every covariant Gaussian channel $\E: \B(\H_A) \rightarrow \B(\H_B)$ admits a covariant dilation.
The dilation requires an ancilla of $2n_B$ modes, where $n_B$ is the number of modes in the output system $B$, prepared in a pure state.

\thmdilation*
\begin{proof}
    First, consider an ancillary system with the same number of modes as $A$, denoted by $\overbar{A}$. 
    Let $\ket{\psi}_{A\overbar{A}}$ be an arbitrary Gaussian pure state on $A\overbar{A}$ such that its reduced state $\rho_A = \tr_{\overbar{A}}\left[\ketbra{\psi}_{A\overbar{A}}\right]$ on $A$ is full rank.

    Next, act with channel $\E$ on subsystem $A$, to obtain the state
    \begin{equation}
        \hat{\tau}_{B\overbar{A}} = (\E \otimes \I)(\ketbra{\psi}_{A\overbar{A}}) \,,
    \end{equation}
    where $\I(\cdot) = \hat{I}_{\overbar{A}}(\cdot)\hat{I}_{\overbar{A}}$ is the identity channel on $\overbar{A}$.
    Lemma~\ref{lem:inv_gauss_purification} guarantees that $\hat{\tau}_{B\overbar{A}}$ has a Gaussian purification
    \begin{equation}\label{eq:purifPhi}
        \ket{\Phi}_{A\overbar{A}B\overbar{B}} \in \H_A \otimes \H_{\overbar{A}} \otimes \H_B \otimes \H_{\overbar{B}} \,,
    \end{equation}
    such that the reduced state on $B\overbar{A}$ is $\hat{\tau}_{B\overbar{A}}$, i.e.
    \begin{equation}
        \hat{\tau}_{B\overbar{A}} = \tr_{A\overbar{B}}[\ket{\Phi}\bra{\Phi}_{A\overbar{A}B\overbar{B}}] \,,
    \end{equation}
    where $\overbar{B}$ has, at most, the same number of modes as $B$.
    
    Moreover, consider the Gaussian state
    \begin{equation}\label{eq:purifPsi}
        \ket{\Psi}_{A\overbar{A}B\overbar{B}} = \ket{\psi}_{A\overbar{A}} \otimes \ket{\eta}_{B\overbar{B}} \in \H_A \otimes \H_{\overbar{A}} \otimes \H_B \otimes \H_{\overbar{B}} \,,
    \end{equation}
    where $\ket{\eta}_{B\overbar{B}}$ is some pure state.
    The reduced states of $\ket{\Phi}_{A\overbar{A}B\overbar{B}}$ and $\ket{\Psi}_{A\overbar{A}B\overbar{B}}$ on $\overbar{A}$ are the same, so Lemma~\ref{lem:inv_gauss_purification} states that there exists a Gaussian unitary $\hat{V}_{AB\overbar{B}}$ on the complement system $AB\overbar{B}$ such that
    \begin{equation}\label{eq:PhiPsi}
        \ket{\Phi}_{A\overbar{A}B\overbar{B}} = \left(\hat{I}_{\overbar{A}} \otimes \hat{V}_{AB\overbar{B}}\right) \ket{\Psi}_{A\overbar{A}B\overbar{B}} \,.
    \end{equation}
    Finally, define the linear map
    \begin{equation}\label{eq:Fchannel}
        \F(\cdot) = \tr_{A\overbar{B}}\left[ \hat{V}_{AB\overbar{B}}((\cdot) \otimes \ketbra{\eta}{\eta}_{B\overbar{B}})\hat{V}_{AB\overbar{B}}^\dag \right] \,,
    \end{equation}
    and observe that
    \begin{align}
        (\F \otimes \I)(\ketbra{\psi}_{A\overbar{A}})
        &= \tr_{A\overbar{B}}\left[ \left(\hat{I}_{\overbar{A}} \otimes \hat{V}_{AB\overbar{B}}\right) \left(\ketbra{\psi}_{A\overbar{A}} \otimes \ketbra{\eta}{\eta}_{B\overbar{B}}\right) \left(\hat{I}_{\overbar{A}} \otimes \hat{V}_{AB\overbar{B}}\right)^\dag \right] \nonumber\\
        &= \tr_{A\overbar{B}}\left[ \left(\hat{I}_{\overbar{A}} \otimes \hat{V}_{AB\overbar{B}}\right) \ketbra{\Psi}_{A\overbar{A}B\overbar{B}} \left(\hat{I}_{\overbar{A}} \otimes \hat{V}_{AB\overbar{B}}\right)^\dag \right] \nonumber\\
        &= \tr_{A\overbar{B}}\left[ \ketbra{\Phi}_{A\overbar{A}B\overbar{B}} \right] \nonumber\\
        &= \hat{\tau}_{B\overbar{A}} \nonumber\\
        &= (\E \otimes \I)(\ketbra{\psi}_{A\overbar{A}}) \label{eq:PhiIPsiI} \,,
    \end{align}
    where all equalities follow from the definitions given earlier.
    Since the reduced state $\hat{\rho}_A$ of state $\ket{\psi}_{A\overbar{A}}$ is full-rank, we conclude the equality of channels $\E = \F$. 
    This can be seen, for instance, by noting that for any bounded operator $\hat{M}_A$ on system $A$, there exists a bounded operator $\hat{N}_{\overbar{A}}$ on system $\overbar{A}$ such that
    \begin{equation}\label{eq:MANAbar}
        \hat{M}_A = \tr_{\overbar{A}}\left[\left(\hat{I} \otimes \hat{N}_{\overbar{A}}\right)\ketbra{\psi}_{A\overbar{A}}\right] \,.
    \end{equation} 
    Namely, one can choose $\hat{N} = \hat{\rho}^{-1/2} \hat{M}^{\rm T} \hat{\rho}^{-1/2}$, where $\hat{M}^{\rm T}$ is the transpose of $\hat{M}$ with respect to an arbitrary orthonormal basis, e.g. the Fock basis.
    Therefore, for all bounded operators $\hat{M}_A$,
    \begin{equation}
        \E(\hat{M}_A) 
        = \tr_{\overbar{A}}\left[(\E \otimes \I)(\ketbra{\psi}_{A\overbar{A}})\left(\hat{I} \otimes \hat{N}_{\overbar{A}}\right)\right]
        = \tr_{\overbar{A}}\left[(\F \otimes \I)(\ketbra{\psi}_{A\overbar{A}})\left(\hat{I} \otimes \hat{N}_{\overbar{A}}\right)\right]
        = \F(\hat{M}_A) \,,
    \end{equation}
    where the first and third equalities follow from Eq.(\ref{eq:MANAbar}) and the second equality from Eq.(\ref{eq:PhiIPsiI}).
    
    Next, we argue that if $\mathcal{E}$ is $G$--covariant, then all the steps in this construction can be chosen to respect the symmetry.
    To obtain $G$--invariant purifications, we choose the symmetry representations on systems $\overbar{A}$ and $\overbar{B}$ to be the complex conjugate of the representations on $A$ and $B$, respectively, i.e.
    \begin{equation}
        \hat{S}_{\overbar{A}} = \hat{S}^*_A \,,\quad\quad \hat{S}_{\overbar{B}} = \hat{S}^*_B \,,
    \end{equation}
    where the complex conjugate is defined in the common eigenbasis of position operators.
    Then, due to Lemma~\ref{lem:inv_gauss_purification2}, state $\ket{\psi}_{A\overbar{A}}$ can be chosen to be $G$--invariant, i.e.
    \begin{equation}
        \left(\hat{S}_{A}(g) \otimes \hat{S}_{\overbar{A}}(g)\right)\ket{\psi}_{A\overbar{A}} = \ket{\psi}_{A\overbar{A}} \,,
    \end{equation}
    for all $g \in G$, where the reduced state on $A$ is full rank.
    The $G$--covariance of $\mathcal{E}$  implies that the output state $\hat{\tau}_{B\overbar{A}}$ is also $G$--invariant, and therefore, by Lemma~\ref{lem:inv_gauss_purification2},  it has a $G$--invariant Gaussian purification $\ket{\Phi}_{A\overbar{A}B\overbar{B}}$ satisfying Eq.(\ref{eq:purifPhi}).
    
    Consider the second purification $\ket{\Psi}_{A\overbar{A}B\overbar{B}}$ of the reduced state on $\overbar{A}$, provided in Eq.(\ref{eq:purifPsi}).
    As we saw in Proposition~\ref{prop:squ}, since both $\hat{S}_{B}$ and $\hat{S}_{\overbar{B}}$ admit $G$--invariant states, they also admit $G$--invariant pure Gaussian states $\ket{\eta_1}_{B}$ and $\ket{\eta_2}_{\overbar{B}}$ such that
    \begin{equation}
        \hat{S}_{B} \ket{\eta_1}_{B} = \ket{\eta_1}_{B} \,,\quad\quad \hat{S}_{\overbar{B}} \ket{\eta_2}_{\overbar{B}} = \ket{\eta_2}_{\overbar{B}} \,,
    \end{equation}
    for all $g\in G$.
    Choosing $\ket{\eta}_{B\overbar{B}} = \ket{\eta_1}_{B} \otimes \ket{\eta_2}_{\overbar{B}}$ thus gives another $G$--invariant Gaussian purification of $\ket{\Psi}_{A\overbar{A}B\overbar{B}}$.
  
    Finally, according to Lemma~\ref{lem:inv_gauss_purification2}, there exists a symplectic Gaussian unitary $\hat{V}_{AB\overbar{B}}$ that commutes with the representation of symmetry on $AB\overbar{B}$, i.e. $\hat{S}_A \otimes \hat{S}_B \otimes \hat{S}_{\overbar{B}}$, and converts one purification to the other, as described in Eq.(\ref{eq:PhiPsi}).
    Hence, Eq.(\ref{eq:Fchannel}) and Eq.(\ref{eq:PhiIPsiI}) remain valid, so Eq.(\ref{eq:Fchannel}) provides a dilation of the $G$--covariant channel $\mathcal{E}$.
\end{proof}

Note that we are free to choose any (invariant) full-rank state $\hat{\rho}$ to construct the dilation in the proof of Theorem~\ref{thm:cov_gauss_dilation}.
Consider the Stinespring dilation of a channel tensor product $\E_1 \otimes \E_2: \B(\H_{A_1} \otimes \H_{A_2}) \rightarrow \B(\H_{B_1} \otimes \H_{B_2})$,
\begin{align}
    (\E_1 \otimes \E_2)((\cdot)_{A_1A_2}) = \tr_{A_1A_2\overbar{B}_1\overbar{B}_2}\left[ \hat{V}_{A_1A_2B_1B_2\overbar{B}_1\overbar{B}_2}((\cdot)_{A_1A_2} \otimes \ketbra{\eta}{\eta}_{B_1B_2\overbar{B}_1\overbar{B}_2})\hat{V}_{A_1A_2B_1B_2\overbar{B}_1\overbar{B}_2}^\dag \right]\,,
\end{align}
Choosing the full-rank, tensor-product state $\hat{\rho}_{A_1} \otimes \hat{\rho}_{A_2} \in \B(\H_{A_1} \otimes \H_{A_2})$ leads to a tensor product form for the unitary
\begin{equation}
    \hat{V}_{A_1A_2B_1B_2\overbar{B}_1\overbar{B}_2} = \hat{V}^{(1)}_{A_1B_1\overbar{B}_1} \otimes \hat{V}^{(2)}_{A_2B_2\overbar{B}_2} \,.
\end{equation}
The Stinespring dilation of a channel composition $\E' \circ \E$, where $\E: \B(\H_A) \rightarrow \B(\H_B)$ and $\E': \B(\H_B) \rightarrow \B(\H_C)$, can be written as follows,
\begin{align}
    (\E' \circ \E)((\cdot)_A) &= \E'\left( \tr_{A\overbar{B}}\left[ \hat{V}_{AB\overbar{B}}(\hat{\rho}_A \otimes \ketbra{\eta}{\eta}_{B\overbar{B}})\hat{V}_{AB\overbar{B}}^\dag \right] \right) \nonumber\\
    &= \tr_{B\overbar{C}}\left[ \hat{V}_{BC\overbar{C}}' \left( \tr_{A\overbar{B}}\left[ \hat{V}_{AB\overbar{B}}((\cdot)_A \otimes \ketbra{\eta}{\eta}_{B\overbar{B}})\hat{V}_{AB\overbar{B}}^\dag \right] \otimes \ketbra{\eta}{\eta}_{C\overbar{C}} \right)\hat{V}_{BC\overbar{C}}'^\dag \right] \nonumber\\
    &= \tr_{AD}\left[ \hat{V}''_{ACD}((\cdot)_A  \otimes \ketbra{\eta}{\eta}_{CD}) \hat{V}_{ACD}''^\dag \right] \,, \label{eq:dilation_composition}
\end{align}
where the ancilla system is $D = B\overbar{B}\overbar{C}$ and $\hat{V}''_{ACD} \coloneqq \hat{V}_{BC\overbar{C}}'  \hat{V}_{AB\overbar{B}}$.

\subsection{Phase-insensitive Gaussian channels on 1 mode}
\label{app:dilation_1mode}

We use our dilation theorem to construct dilations for all 1--mode phase-insensitive channels\footnote{Throughout this section, every covariance matrix $\sigma_{\hat{\rho}}$ of composite state $\hat{\rho}_{A_1 \dots A_n}$ is represented in the basis $\bigoplus_{j=1}^n \hat{\bmr}_{A_j}$ .}.
Any such channel $\E: \B(\H_A) \rightarrow \B(\H_B) \cong \B(\H_A)$, is defined by real matrices $X = xI$ and $Y = yI$, where $y \ge |x^2-1|$~\cite{holevo2007one,giovannetti2014ultimate}.
Pick the full-rank thermal density operator
\begin{equation}\label{eq:thermal}
    \hat{\rho}_A = \frac{\exp(-\beta \hat{a}^\dag\hat{a})}{\tr[\exp(-\beta \hat{a}^\dag\hat{a})]} = \sum_{k=0}^{\infty} \frac{\bar{n}^k}{(\bar{n}+1)^{k+1}} \ketbra{k}{k} \,,
\end{equation}
with zero first moment and $\sigma = cI$, where we set $c \equiv \cosh{(2r)} = 2\bar{n}+1$ and $s \equiv \sinh{(2r)} = 2\sqrt{\bar{n}(\bar{n}+1)}$, and $\bar{n} =(e^\beta-1)^{-1}$.
Consider the purification $\ket{\psi}_{A\overbar{A}}$ of $\hat{\rho}_A$ according to Eq.(\ref{eq:purif_cov}) with covariance matrix
\begin{equation}
    \sigma_{\psi} = \begin{pmatrix} cI & sZ \\[5pt] sZ & cI \end{pmatrix} \,,
\end{equation}
where the symplectic form is $\Omega \oplus \Omega$.

Applying $\E$ on the first subsystem, we get $\hat{\tau}_{B\overbar{A}} = (\E \otimes \I)(\ketbra{\psi}{\psi}_{A\overbar{A}})$ with covariance matrix
\begin{align}
    \sigma_{\tau}
    = (X \oplus I) \sigma_{\psi} (X \oplus I)^{\rm T} + (Y \oplus 0)
    = \begin{pmatrix} (cx^2 + y)I & sxZ \\[5pt] sxZ & cI \end{pmatrix} \,.
\end{align}
Purifying $\hat{\tau}_{B\overbar{A}}$ according to Eq.(\ref{eq:purif_cov}), we get $\ket{\Phi}_{B\overbar{A}\overbar{B}A}$ with covariance matrix
\begin{align}
    \sigma_{\Phi}
    = \begin{pmatrix} (cx^2 + y)I & sxZ & uZ & tI \\[5pt] sxZ & cI & -tI & vZ \\[5pt] uZ & -tI & (cx^2 + y)I & -sxZ \\[5pt] tI & vZ & -sxZ & cI \end{pmatrix} \,,
\end{align}
where we have calculated the off-diagonal block $M$ as follows,
\begin{align}
    M 
    &= \sqrt{-\sigma_{\tau}(\Omega \oplus \Omega)\sigma_{\tau}(\Omega \oplus \Omega) - (I \oplus I)} 
    = \sqrt{\begin{pmatrix} ((cx^2+y)^2 - s^2x^2 - 1)I & sx(c(1-x^2)-y)Z \\[5pt] sx(c(x^2-1)+y)Z & s^2(1 - x^2)I \end{pmatrix}} \nonumber\\
    &= \begin{pmatrix} uI & tZ \\[5pt] -tZ & vI \end{pmatrix} \,,
\end{align}
for parameters $u,v,t$ defined as
\begin{align}
    u &= \frac{1}{2} \left(Q + \frac{(cx^2+y)^2 - c^2}{Q}\right) \,,\\
    v &= \frac{1}{2} \left(Q - \frac{(cx^2+y)^2 - c^2}{Q}\right) \,, \text{ and}\\
    t &= \frac{sx(c(1-x^2)-y)}{Q} \,, \text{ where}\\
    Q &= \sqrt{(cx^2+y)^2 - 2s^2x^2 + s^2 - 1 + 2s\sqrt{((cx^2+y)^2 - s^2x^2 - 1)(1 - x^2)+x^2(c(1-x^2)-y)^2}} \,.
\end{align}

The state $\ket{\Psi}_{B\overbar{A}\overbar{B}A} = \ket{\psi}_{A\overbar{A}} \otimes \ket{0}^{\otimes 2}_{B\overbar{B}}$ has covariance matrix
\begin{align}
    \sigma_{\Psi}
    = \begin{pmatrix} I & 0 & 0 & 0 \\ 0 & cI & 0 & sZ \\ 0 & 0 & I & 0 \\ 0 & sZ & 0 & cI \end{pmatrix} \,.
\end{align}

To construct the dilation, we need to find a symplectic matrix $V_{AB\overbar{B}}$ such that $(I_{\overbar{A}} \oplus V_{AB\overbar{B}}) \sigma_{\Psi} (I_{\overbar{A}} \oplus V_{AB\overbar{B}})^{\rm T} = \sigma_{\Phi}$, which we first do separately for attenuator and amplifier channels on 1 mode.

Attenuators $\E^{\rm att}_{\alpha,\beta}$ are characterized by $x = \cos{\alpha}$ and $y = (2n_\beta+1) \sin^2{\alpha}$, where $n_\beta$ is the average number of thermal excitations.
To construct its dilation, we can choose $\bar{n} = n_\beta$ for the full-rank state that we define in Eq.(\ref{eq:thermal}), so that $y = \cosh(2r) \sin^2{\alpha}$.
Since the sign of $x$ does not affect the definition of the channel, we can also choose $\alpha \in [0, \pi/2)$.
As a consequence, $Q = 2\sinh(2r)\sin\alpha$, hence $u = v = \sinh(2r)\sin\alpha$, $t=0$, and
\begin{align}
    \sigma_{\Phi} = \begin{pmatrix} \cosh(2r)I & \sinh(2r)\cos\alpha Z & \sinh(2r)\sin\alpha Z & 0 \\[5pt] \sinh(2r)\cos\alpha Z & \cosh(2r)I & 0 & \sinh(2r)\sin\alpha Z \\[5pt] \sinh(2r)\sin\alpha Z & 0 & \cosh(2r)I & -\sinh(2r)\cos\alpha Z \\[5pt] 0 & \sinh(2r)\sin\alpha Z & -\sinh(2r)\cos\alpha Z & \cosh(2r)I \end{pmatrix} \,,
\end{align}
Indeed, this covariance matrix $\sigma_{\Phi}$ is equivalent to $\sigma_{\Psi}$, up to the symplectic matrix
\begin{equation}
    V^{\rm att}_{\alpha,\beta} \coloneqq V_{AB\overbar{B}} = V_{\rm bs}(-\alpha + \pi/2)_{AB}V_{\rm 2sq}\left(\frac{1}{2}\cosh^{-1}(2n_\beta+1)\right)_{B\overbar{B}} \,,
\end{equation}
which is invariant under any symmetry representation $S_A(g) \oplus S_B(g) \oplus S_{\overbar{B}}(g)$, provided that $S_A(g) = S_B(g) = ZS_{\overbar{B}}(g)Z$ for all $g \in G$.

Amplifiers $\E^{\rm amp}_{\alpha,\beta}$ are characterized by $x = \cosh{\alpha}$ and $y = (2n_\beta+1)\sinh^2{\alpha}$, where $n_\beta$ is the average number of thermal excitations.
Similarly, we choose $\alpha \in [0, \pi/2)$ and $n = n_\beta$, so that $y = \cosh(2r) \sinh^2{\alpha}$.
As a consequence, $Q = \cosh(2r)\sinh(2\alpha)$, hence $u = \cosh(2r)\sinh(2\alpha)$, $v = 0$, $t = -\sinh(2r)\sinh\alpha$, and
\begin{align}
    \sigma_{\Phi} = \begin{pmatrix} \cosh(2r)\cosh(2\alpha)I & \sinh(2r)\cosh\alpha Z & \cosh(2r)\sinh(2\alpha)Z & -\sinh(2r)\sinh\alpha I \\[5pt] \sinh(2r)\cosh\alpha Z & \cosh(2r)I & \sinh(2r)\sinh\alpha I & 0 \\[5pt] \cosh(2r)\sinh(2\alpha)Z & \sinh(2r)\sinh\alpha I & \cosh(2r)\cosh(2\alpha)I & -\sinh(2r)\cosh\alpha Z \\[5pt] -\sinh(2r)\sinh\alpha I & 0 & -\sinh(2r)\cosh\alpha Z & \cosh(2r)I \end{pmatrix} \,,
\end{align}
Indeed, this covariance matrix $\sigma_{\Phi}$ is equivalent to $\sigma_{\Psi}$, up to the symplectic matrix
\begin{equation}
    V^{\rm amp}_{\alpha,\beta} \coloneqq V_{AB\overbar{B}} = V_{\rm 2sq}(\alpha)_{B\overbar{B}}V_{\rm bs}(\pi/2)_{AB}V_{\rm 2sq}\left(\frac{1}{2}\cosh^{-1}(2n_\beta+1)\right)_{B\overbar{B}} \,,
\end{equation}
which is invariant under any symmetry representation $S_A(g) \oplus S_B(g) \oplus S_{\overbar{B}}(g)$, provided that $S_A(g) = S_B(g) = ZS_{\overbar{B}}(g)Z$ for all $g \in G$.

Any 1--mode phase-covariant channel $\E$ is unitarily equivalent to the composition of a pure loss channel $\E^{\rm att}_{\alpha,\infty}: \B(\H_A) \rightarrow \B(\H_B)$ and a quantum-limited amplifier $\E^{\rm amp}_{\alpha,\infty}: \B(\H_B) \rightarrow \B(\H_C)$~\cite{holevo2007one}, therefore one can use the composition rule in Eq.(\ref{eq:dilation_composition}) to express the Stinespring dilation of $\E$ in terms of the symplectic matrix
\begin{align}
    V^{U(1)}_{\alpha,\alpha'} 
    = \lim_{\beta \rightarrow \infty}V^{\rm amp}_{\alpha,\beta}V^{\rm att}_{\alpha',\beta}
    = V_{\rm 2sq}(\alpha)_{C\overbar{C}} V_{\rm bs}\left(\frac{\pi}{2}\right)_{BC} V_{\rm bs}\left(-\alpha' + \frac{\pi}{2}\right)_{AB} \,.
\end{align}

\end{document}